\documentclass[a4paper,11pt]{memoir}

\usepackage{graphicx}
\usepackage{algorithm}
\usepackage[noend]{algpseudocode}
\usepackage{amsmath}
\usepackage{amssymb}
\usepackage{amsthm}
\usepackage{amstext}
\usepackage{amsfonts}
\usepackage{braket}
\usepackage[german,english]{babel}
\usepackage{centernot}
\usepackage[babel=true]{csquotes}
\usepackage[T1]{fontenc}
\usepackage[utf8]{inputenc}
\usepackage{lmodern}
\usepackage[scr=dutchcal]{mathalfa}
\usepackage{multibib}
\usepackage{multicol}
\usepackage{pbox}
\usepackage[usenames,dvipsnames]{xcolor}
\usepackage{xargs}
\usepackage{xspace}

\newcommand{\mythesistitle}{Quantum Security of\\Cryptographic Primitives}

\ifdefined\printversion
	\newcommand{\url}[1]{#1}
\else
	\usepackage{hyperref}
	\definecolor{UrlPurple}{RGB}{56,4,112}
	\definecolor{CiteGreen}{RGB}{0,97,58}
	\hypersetup{
		colorlinks=true,
		linktoc=page,
		urlcolor=UrlPurple,
		linkcolor=MidnightBlue,
		citecolor=CiteGreen,
		pdfauthor={Tommaso Gagliardoni},
		pdftitle={\mythesistitle}
	}
\fi

\newcites{my}{List of Publications}

\IfFileExists{version}{
	\newread\file
	\openin\file=version
	\read\file to\version
	\closein\file
	\usepackage{eso-pic}
	\AddToShipoutPicture{\raisebox{256mm}{\hspace{54mm}\textnormal{\texttt{\tiny draft \version}}}}
}{}

\newlength{\titleoffset}
\ifdefined\printversion
	\calccentering{\titleoffset}
\else
	\setlrmargins{*}{*}{1}
	\checkandfixthelayout
\fi

\ifdefined\change\else
	\newcommand{\change}[1]{$\blacksquare$ {\color{Red}#1}}
\fi

\setpnumwidth{1.8em}
\setrmarg{1.8em}
\nouppercaseheads
\let\mp\marginpar\renewcommand{\marginpar}[1]{\mp{\raggedright\footnotesize\sffamily\color{red}#1}}
\makechapterstyle{madsenplain}{%

\def\printchaptertitle##1{\raggedleft\Huge\sffamily\bfseries ##1}

}
\newsubfloat{figure}
\captionnamefont{\small}
\captiontitlefont{\small}
\loosesubcaptions
\newtheorem{definition}{Definition}[chapter]
\newtheorem{proposition}[definition]{Proposition}
\newtheorem{construction}[definition]{Construction}
\newtheorem{theorem}[definition]{Theorem}
\newtheorem{lemma}[definition]{Lemma}
\newtheorem{corollary}[definition]{Corollary}
\newtheorem{experiment}[definition]{Experiment}
\newcommand{\A}{\ensuremath{\mathcal{A}}\xspace}
\newcommand{\B}{\ensuremath{\mathcal{B}}\xspace}
\newcommand{\C}{\ensuremath{\mathcal{C}}\xspace}
\newcommand{\D}{\ensuremath{\mathcal{D}}\xspace}
\newcommand{\E}{\ensuremath{\mathcal{E}}\xspace}
\newcommand{\F}{\ensuremath{\mathcal{F}}\xspace}
\newcommand{\G}{\ensuremath{\mathcal{G}}\xspace}
\renewcommand{\H}{\ensuremath{\mathcal{H}}\xspace}
\newcommand{\I}{\ensuremath{\mathcal{I}}\xspace}
\newcommand{\J}{\ensuremath{\mathcal{J}}\xspace}
\newcommand{\K}{\ensuremath{\mathcal{K}}\xspace}
\renewcommand{\L}{\ensuremath{\mathcal{L}}\xspace}
\newcommand{\M}{\ensuremath{\mathcal{M}}\xspace}

\renewcommand{\O}{\ensuremath{\mathcal{O}}\xspace}
\renewcommand{\P}{\ensuremath{\mathcal{P}}\xspace}

\newcommand{\R}{\ensuremath{\mathcal{R}}\xspace}
\renewcommand{\S}{\ensuremath{\mathcal{S}}\xspace}
\newcommand{\T}{\ensuremath{\mathcal{T}}\xspace}
\newcommand{\U}{\ensuremath{\mathcal{U}}\xspace}
\newcommand{\V}{\ensuremath{\mathcal{V}}\xspace}
\newcommand{\W}{\ensuremath{\mathcal{W}}\xspace}
\newcommand{\X}{\ensuremath{\mathcal{X}}\xspace}
\newcommand{\Y}{\ensuremath{\mathcal{Y}}\xspace}
\newcommand{\Z}{\ensuremath{\mathcal{Z}}\xspace}
\renewcommand{\a}{\ensuremath{\mathscr{a}}\xspace}

\renewcommand{\c}{\ensuremath{\mathscr{c}}\xspace}
\renewcommand{\d}{\ensuremath{\mathscr{d}}\xspace}

\newcommand{\f}{\ensuremath{\mathscr{f}}\xspace}
\newcommand{\g}{\ensuremath{\mathscr{g}}\xspace}
\newcommand{\h}{\ensuremath{\mathscr{h}}\xspace}

\renewcommand{\k}{\ensuremath{\mathscr{k}}\xspace}
\renewcommand{\l}{\ensuremath{\ell}\xspace}
\newcommand{\m}{\ensuremath{\mathscr{m}}\xspace}

\newcommand{\p}{\ensuremath{\mathscr{p}}\xspace}
\newcommand{\q}{\ensuremath{\mathscr{q}}\xspace}
\renewcommand{\r}{\ensuremath{\mathscr{r}}\xspace}
\newcommand{\s}{\ensuremath{\mathscr{s}}\xspace}
\renewcommand{\t}{\ensuremath{\mathscr{t}}\xspace}

\newcommand{\x}{\ensuremath{\mathscr{x}}\xspace}
\newcommand{\y}{\ensuremath{\mathscr{y}}\xspace}

\renewcommand{\epsilon}{\varepsilon}

\renewcommand{\phi}{\varphi}

\newcommand{\QS}{\ensuremath{\mathbf{QS}}}
\newcommand{\secpa}{\ensuremath{n}}
\newcommand{\secpar}{\ensuremath{n}\xspace}
\newcommand{\secparam}{\ensuremath{1^n}\xspace}
\newcommand{\xor}{\oplus}
\newcommand{\negl}{\ensuremath{\mathsf{negl}}\xspace}
\newcommand{\poly}{\ensuremath{\mathsf{poly}}\xspace}
\newcommand{\half}{\ensuremath{\frac{1}{2}}\xspace}
\newcommand{\sqhalf}{\ensuremath{\frac{1}{\sqrt{2}}}\xspace}
\renewcommand{\iff}{\ensuremath{\Leftrightarrow}\xspace}
\newcommand{\nimplies}{\ensuremath{\centernot\implies}\xspace}
\newcommand{\foral}{\, \forall \: }
\newcommand{\hc}{\ensuremath{\mathsf{hc}}\xspace}
\newcommand{\im}{\ensuremath{\mathrm{i}}\xspace} 
\renewcommand{\set}[1]{\ensuremath{ \left\{ #1 \right\} }\xspace}
\newcommand{\bin}{\ensuremath{\set{0,1}\xspace}}
\newcommand{\words}{\ensuremath{\bin^*\xspace}}
\newcommand{\Supp}[1]{\ensuremath{\text{Supp}\left(#1\right) }\xspace}
\newcommand{\card}[1]{\ensuremath{ \left| #1 \right| }\xspace}
\newcommand{\family}[2][n]{\ensuremath{ \left( {#2}_{#1} \right)_{#1} }\xspace}
\renewcommand{\emptyset}{\varnothing}
\newcommand{\project}[1]{\Big|_{#1}}
\newcommand{\NN}{\mathbb{N}}

\newcommand{\ZZ}{\mathbb{Z}}

\newcommand{\RR}{\mathbb{R}}
\newcommand{\CC}{\mathbb{C}}
\newcommand{\from}{\ensuremath{\leftarrow}}
\newcommand{\fromdist}[1]{\ensuremath{\stackrel{{#1}}{\longleftarrow}\xspace}}
\newcommand{\rand}{\raisebox{-1pt}{\ensuremath{\,\xleftarrow{\raisebox{-1pt}{$\scriptscriptstyle\$$}}}\xspace}}
\DeclareMathOperator{\tr}{tr}
\renewcommand{\vec}{\mathbf}
\newcommand{\Hilbert}{\ensuremath{{\mathfrak{H}}}\xspace}
\newcommand{\Hilb}[1]{\ensuremath{{\Hilbert_{#1}}}\xspace}
\newcommand{\states}[1]{\ensuremath{\mathfrak{D}\left({#1}\right)}\xspace}
\newcommand{\conj}[1]{\ensuremath{{\overline{#1}}}\xspace}
\newcommand{\zerovec}{\ensuremath{{\mathbf{0}}}\xspace}
\newcommand{\Id}{\ensuremath{{\mathbb{I}}}\xspace}
\newcommand{\nullop}{\ensuremath{{\mathbb{O}}}\xspace}
\newcommand{\ketbra}[1]{\ensuremath{\ket{{#1}}\!\!\bra{{#1}}}\xspace}
\newcommand{\pauli}{\ensuremath{{\mathfrak{P}}}\xspace}
\newcommand{\game}{\ensuremath{\mathsf{Game}}}
\newcommand{\adv}{\ensuremath{\mathsf{Adv}}}
\newcommand{\lookuptable}{\ensuremath{\mathsf{LookupTable}}}
\newcommand{\DPT}{\ensuremath{\mathsf{DPT}}\xspace}
\newcommand{\PPT}{\ensuremath{\mathsf{PPT}}\xspace}
\newcommand{\QPT}{\ensuremath{\mathsf{QPT}}\xspace}
\newcommand{\PPP}{\ensuremath{\mathsf{P}}\xspace}
\newcommand{\NP}{\ensuremath{\mathsf{NP}}\xspace}
\newcommand{\BPP}{\ensuremath{\mathsf{BPP}}\xspace}
\newcommand{\BQP}{\ensuremath{\mathsf{BQP}}\xspace}
\newcommand{\PRNG}{\ensuremath{\G}\xspace}
\newcommand{\PRF}{\ensuremath{\F}\xspace}
\newcommand{\PRP}{\ensuremath{\P}\xspace}
\newcommand{\Gen}{\ensuremath{\mathsf{Gen}}\xspace}

\newcommand{\Invert}{\ensuremath{\mathsf{Invert}}\xspace}
\newcommand{\Eval}{\ensuremath{\mathsf{Eval}}\xspace}
\newcommand{\RO}{\ensuremath{{\O_\h}}\xspace}
\newcommand{\state}{\ensuremath{\mathsf{state}}\xspace}
\newcommand{\Enc}{\ensuremath{\mathsf{Enc}}\xspace}
\newcommand{\Dec}{\ensuremath{\mathsf{Dec}}\xspace}
\newcommand{\KGen}{\ensuremath{\mathsf{KGen}}\xspace}
\newcommand{\sk}{\ensuremath{\mathsf{sk}}\xspace}
\newcommand{\pk}{\ensuremath{\mathsf{pk}}\xspace}
\newcommand{\Kp}{\ensuremath{\K^{\p}}\xspace}
\newcommand{\Ks}{\ensuremath{\K^{\s}}\xspace}
\newcommand{\gameSEMA}[1][\A]{\ensuremath{\game^{\mathsf{SEM}}_{\E,{#1}}}\xspace}
\newcommand{\gameSEMS}[1][\S]{\ensuremath{\game^{\mathsf{SEM*}}_{\E,{#1}}}\xspace}
\newcommand{\gameIND}[1][\A]{\ensuremath{\game^{\mathsf{IND}}_{\E,{#1}}}\xspace}
\newcommand{\advIND}{\ensuremath{\adv^{\mathsf{IND}}_{\E,\A}}\xspace}
\newcommand{\gameINDCPA}[1][\A]{\ensuremath{\game^{\mathsf{IND-CPA}}_{\E,{#1}}}\xspace}
\newcommand{\advINDCPA}{\ensuremath{\adv^{\mathsf{IND-CPA}}_{\E,\A}}\xspace}
\newcommand{\gameINDCCA}{\ensuremath{\game^{\mathsf{IND-CCA1}}_{\E,\A}}\xspace}
\newcommand{\advINDCCA}{\ensuremath{\adv^{\mathsf{IND-CCA1}}_{\E,\A}}\xspace}
\newcommand{\gameINDCCAA}{\ensuremath{\game^{\mathsf{IND-CCA2}}_{\E,\A}}\xspace}
\newcommand{\advINDCCAA}{\ensuremath{\adv^{\mathsf{IND-CCA2}}_{\E,\A}}\xspace}
\newcommand{\Sigscheme}{\ensuremath{\mathscr{Sig}}\xspace}
\newcommand{\Sign}{\ensuremath{\mathsf{Sign}}\xspace}
\newcommand{\sig}{\ensuremath{\mathsf{sig}}\xspace}
\newcommand{\SVer}{\ensuremath{\mathsf{SigVerify}}\xspace}
\newcommand{\gameEUFCMA}[1][\A]{\ensuremath{\game^{\mathsf{EUF-CMA}}_{\Sigscheme,{#1}}}\xspace}
\newcommand{\advEUFCMA}{\ensuremath{\adv^{\mathsf{EUF-CMA}}_{\Sigscheme,\A}}\xspace}
\newcommand{\gameEUFCMARO}[1][\A]{\ensuremath{\game^{\mathsf{EUF-CMA-RO}}_{\Sigscheme,{#1}}}\xspace}
\newcommand{\advEUFCMARO}{\ensuremath{\adv^{\mathsf{EUF-CMA-RO}}_{\Sigscheme,\A}}\xspace}
\newcommand{\com}{\ensuremath{\mathsf{com}}\xspace}
\newcommand{\ch}{\ensuremath{\mathsf{ch}}\xspace}
\newcommand{\resp}{\ensuremath{\mathsf{resp}}\xspace}
\newcommand{\sigmaproto}{\ensuremath{{\left(\P,\V\right)}}\xspace}
\newcommand{\FSSigma}{\ensuremath{\mathscr{Sig}_{\mathsf{FS}}^\RO\sigmaproto }\xspace}
\newcommand{\Inst}{\ensuremath{{\mathsf{Inst}}}\xspace}
\newcommand{\hardL}{\ensuremath{{\L_{\W,\R,\Inst}}}\xspace}
\newcommand{\idscheme}{\sigmaproto}
\newcommand{\bsize}{{\ensuremath{n_\mathsf{blk}}}\xspace}
\newcommand{\dbsize}{{\ensuremath{n_\mathsf{db}}}\xspace}
\newcommand{\dsize}{{\ensuremath{n_\mathsf{dat}}}\xspace}
\newcommand{\msize}{{\ensuremath{n_\mathsf{msg}}}\xspace}
\newcommand{\tsize}{{\ensuremath{n_\mathsf{tree}}}\xspace}
\newcommand{\ksize}{{\ensuremath{n_\mathsf{tag}}}\xspace}
\newcommand{\zsize}{{\ensuremath{n_\mathsf{bkt}}}\xspace}
\newcommand{\maxsize}{{\ensuremath{n_\mathsf{Max}}}\xspace}
\newcommand{\oram}{\ensuremath{\mathsf{ORAM}}\xspace}
\newcommand{\init}{\ensuremath{\mathsf{Init}}\xspace}
\newcommand{\access}{\ensuremath{\mathsf{Access}}\xspace}
\newcommand{\op}{\mathsf{op}}
\newcommand{\ap}{\mathsf{ap}}
\newcommand{\DB}{\texttt{DB}}
\newcommand{\pathoram}{\texttt{PathORAM}\xspace}
\newcommand{\leaf}{\texttt{Leaf}}
\newcommand{\node}{\texttt{Node}}
\newcommand{\treeroot}{\texttt{Root}}
\newcommand{\branch}{\texttt{Branch}}
\newcommand{\stash}{\texttt{Stash}}
\newcommand{\newbranch}{\texttt{NewBranch}}
\newcommand{\dr}{\mathsf{dr}}
\newcommand{\block}{\texttt{block}}
\newcommand{\data}{\mathsf{data}}
\newcommand{\Data}{\mathsf{Data}}
\newcommand{\gameORAM}[1][\A]{\ensuremath{\game_{\oram,#1}^{\mathsf{AP-IND-CQA}}}\xspace}
\newcommand{\gamePORAM}[1][\A]{\ensuremath{\game_{\pathoram,#1}^{\mathsf{AP-IND-CQA}}}\xspace}
\newcommand{\advORAM}[1][\A]{\ensuremath{\adv^{\mathsf{AP-IND-CQA}}_{\oram,#1}}\xspace}
\newcommand{\QRO}{\ensuremath{{\ket{\O_\h}}}\xspace}
\newcommand{\scDist}{\ensuremath{{\U^\delta}}\xspace}
\newcommand{\scQRO}{\ensuremath{{\ket{\O_\h^\delta}}}\xspace}
\newcommand{\PRNGBM}{\ensuremath{\G_{BM}}\xspace}
\newcommand{\gamepqIND}[1][\A]{\ensuremath{\game^{\mathsf{pq-IND}}_{\E,{#1}}}\xspace}
\newcommand{\advpqIND}{\ensuremath{\adv^{\mathsf{pq-IND}}_{\E,\A}}\xspace}
\newcommand{\gamepqINDCPA}[1][\A]{\ensuremath{\game^{\mathsf{pq-IND-CPA}}_{\E,{#1}}}\xspace}
\newcommand{\advpqINDCPA}{\ensuremath{\adv^{\mathsf{pq-IND-CPA}}_{\E,\A}}\xspace}

\newcommand{\gamepqEUFCMA}[1][\A]{\ensuremath{\game^{\mathsf{pq-EUF-CMA}}_{\Sigscheme,{#1}}}\xspace}
\newcommand{\advpqEUFCMA}{\ensuremath{\adv^{\mathsf{pq-EUF-CMA}}_{\Sigscheme,\A}}\xspace}
\newcommand{\gameEUFCMAQRO}[1][\A]{\ensuremath{\game^{\mathsf{EUF-CMA-QRO}}_{\Sigscheme,{#1}}}\xspace}
\newcommand{\advEUFCMAQRO}{\ensuremath{\adv^{\mathsf{EUF-CMA-QRO}}_{\Sigscheme,\A}}\xspace}
\newcommand{\Com}{\ensuremath{\mathsf{Com}}\xspace}
\newcommand{\crs}{\ensuremath{\mathsf{crs}}\xspace}
\newcommand{\SmplRnd}{\ensuremath{\mathsf{SmplRnd}}\xspace}
\newcommand{\PL}{\ensuremath{\P_\Lambda}\xspace}
\newcommand{\VL}{\ensuremath{\V_\Lambda}\xspace}
\newcommand{\lambdaproto}{\ensuremath{\left(\PL,\VL\right)}\xspace}
\newcommand{\FSLambda}{\ensuremath{\mathscr{Sig}_{\mathsf{FS}}^\RO\lambdaproto }\xspace}
\newcommand{\pathoramBM}{\texttt{PathORAM}_{BM}\xspace}
\newcommand{\qPRF}{\ensuremath{\F}\xspace}
\newcommand{\qPRP}{\ensuremath{\P}\xspace}
\newcommand{\gameINDqCPA}[1][\A]{\ensuremath{\game^{\mathsf{IND-qCPA}}_{\E,{#1}}}\xspace}
\newcommand{\advINDqCPA}{\ensuremath{\adv^{\mathsf{IND-qCPA}}_{\E,\A}}\xspace}
\newcommand{\Encq}{\ensuremath{{\ket{\Enc_k}_{\!(2)}}}\xspace}
\newcommand{\Encqd}{\ensuremath{{\ket{\Enc_k}_{\!(2)}^{\dagger}}}\xspace}
\newcommand{\Decq}{\ensuremath{{\ket{\Dec_k}_{\!(2)}}}\xspace}

\newcommand{\Encqdr}{\ensuremath{{\bra{\Enc_k}_{(2)}^{\dagger}}}\xspace}

\newcommand{\Decqdr}{\ensuremath{{\bra{\Dec_k}_{(2)}^{\dagger}}}\xspace}
\newcommand{\gameqIND}[1][\A]{\ensuremath{\game^{\mathsf{qIND}}_{\E,{#1}}}\xspace}
\newcommand{\advqIND}{\ensuremath{\adv^{\mathsf{qIND}}_{\E,\A}}\xspace}
\newcommand{\gameqINDqCPA}[1][\A]{\ensuremath{\game^{\mathsf{qIND-qCPA}}_{\E,{#1}}}\xspace}
\newcommand{\advqINDqCPA}{\ensuremath{\adv^{\mathsf{qIND-qCPA}}_{\E,\A}}\xspace}
\newcommand{\gamewqINDqCPA}[1][\A]{\ensuremath{\game^{\mathsf{wqIND-qCPA}}_{\E,{#1}}}\xspace}
\newcommand{\advwqINDqCPA}{\ensuremath{\adv^{\mathsf{wqIND-qCPA}}_{\E,\A}}\xspace}
\newcommand{\desc}[1]{\ensuremath{\mathsf{Dsc}({#1})}\xspace}
\newcommand{\qbuild}{\ensuremath{\mathsf{Qbuild}}\xspace}
\newcommand{\gamewqSEMA}[1][\A]{\ensuremath{\game^{\mathsf{wqSEM}}_{\E,{#1}}}\xspace}
\newcommand{\gamewqSEMS}[1][\S]{\ensuremath{\game^{\mathsf{wqSEM*}}_{\E,{#1}}}\xspace}
\newcommand{\QEnc}{\ensuremath{\mathsf{QEnc}}\xspace}
\newcommand{\QDec}{\ensuremath{\mathsf{QDec}}\xspace}
\newcommand{\HX}{\ensuremath{{\Hilbert_\X}}\xspace}
\newcommand{\HY}{\ensuremath{{\Hilbert_\Y}}\xspace}
\newcommand{\Env}{\ensuremath{\mathsf{Env}}\xspace}
\newcommand{\HEnv}{\ensuremath{{\Hilbert_\Env}}\xspace}
\newcommand{\QX}{\ensuremath{\states{\HX}}\xspace}
\newcommand{\QY}{\ensuremath{\states{\HY}}\xspace}
\newcommand{\QEnv}{\ensuremath{\states{\HEnv}}\xspace}
\newcommand{\QOTP}[1][k]{\ensuremath{\mathsf{QOTP}_{{#1}}}\xspace}
\newcommand{\gameQIND}[1][\A]{\ensuremath{\game^{\mathsf{QIND}}_{\E,{#1}}}\xspace}
\newcommand{\advQIND}{\ensuremath{\adv^{\mathsf{QIND}}_{\E,\A}}\xspace}
\newcommand{\gameQINDCPA}[1][\A]{\ensuremath{\game^{\mathsf{QIND-CPA}}_{\E,{#1}}}\xspace}
\newcommand{\advQINDCPA}{\ensuremath{\adv^{\mathsf{QIND-CPA}}_{\E,\A}}\xspace}
\newcommand{\gameQINDCCA}{\ensuremath{\game^{\mathsf{QIND-CCA1}}_{\E,\A}}\xspace}
\newcommand{\advQINDCCA}{\ensuremath{\adv^{\mathsf{QIND-CCA1}}_{\E,\A}}\xspace}
\newcommand{\qoram}{\ensuremath{\mathsf{QORAM}}\xspace}
\newcommand{\QDB}{\ensuremath{{\ket{\texttt{QDB}}}}\xspace}
\newcommand{\qdr}{\ensuremath{\ket{\mathsf{qdr}}}\xspace}
\newcommand{\qcom}{\ensuremath{{\ket{\mathsf{qcom}}}}\xspace}
\newcommand{\QData}{\ensuremath{\ket{\mathsf{QData}}}\xspace}
\newcommand{\qinit}{\ensuremath{\mathsf{QInit}}\xspace}
\newcommand{\qaccess}{\ensuremath{\mathsf{QAccess}}\xspace}
\newcommand{\safex}{\ensuremath{\B}\xspace}
\newcommand{\qap}{\ensuremath{{\ket{\mathsf{qap}}}}\xspace}
\newcommand{\pathqoram}{\ensuremath{\texttt{PathQORAM}}\xspace}
\newcommand{\qleaf}{\ensuremath{\ket{\texttt{QLeaf}}}\xspace}
\newcommand{\qnode}{\ensuremath{\ket{\texttt{QNode}}}\xspace}
\newcommand{\qtreeroot}{\ensuremath{\ket{\texttt{QRoot}}}\xspace}
\newcommand{\qbranch}{\ensuremath{\ket{\texttt{QBranch}}}\xspace}
\newcommand{\qstash}{\ensuremath{\ket{\texttt{QStash}}}\xspace}
\newcommand{\newqbranch}{\ensuremath{\ket{\texttt{NewQBranch}}}\xspace}
\newcommand{\gameQORAM}[1][\A^\safex]{\ensuremath{\game_{\qoram,#1}^{\mathsf{QAP-IND-CQA}}}\xspace}

\newcommand{\advQORAM}[1][\A^\safex]{\ensuremath{\adv^{\mathsf{QAP-IND-CQA}}_{\qoram,#1}}\xspace}
\begin{document}
\frontmatter
\chapterstyle{madsenplain}
\begin{titlingpage}
\begin{adjustwidth}{\the\titleoffset - 2cm}{-\the\titleoffset - 2cm}
\begin{center}
{\scshape\huge \mythesistitle \par}
\vspace{4.5em}

Faculty of Computer Science

\vspace{.2em}

of the Technical University of Darmstadt, Germany

\vspace{1.5em}

\textbf{Dissertation}

\vspace{1.5em}

for the achievement of the title

\vspace{.2em}

Doctor rerum naturalium (Dr.\@ rer.\@ nat.)

\vspace{.2em}

of

\vspace{1.5em}

\textbf{Tommaso Gagliardoni, M.Sc.}

\vspace{.2em}

born in Perugia, Italy

\vspace{10em}

\begin{tabular}{rl}
Supervisor: & Prof.\@ Dr.\@ Marc Fischlin \\
Second reviewer:  & Prof.\@ Dr.\@ Christian Schaffner \\ 
            & \\
Submission Date: & 2016-12-16 \\
Defense Date: & 2017-02-13 \\
 \hspace*{6cm} & \hspace*{6cm}\\
\end{tabular}

\vspace{2.6em}

Darmstadt, 2017
\end{center}
\end{adjustwidth}
\end{titlingpage}

\ifdefined\printversion
	\newcommand{\tturl}[1]{\texttt{#1}}
	\newcommand{\href}[2]{#2}
\else
	\newcommand{\tturl}[1]{\href{#1}{\texttt{#1}}}
\fi
\thispagestyle{empty}
{

	\hbadness=10000
	\vspace*{12.5cm}
	\footnotesize
	\noindent This document is an electronic version with minor modifications of the original, published through the E-Publishing-Service of the TU Darmstadt. \\
	\tturl{http://tuprints.ulb.tu-darmstadt.de} \\
	\href{mailto:tuprints@ulb.tu-darmstadt.de}{\texttt{tuprints@ulb.tu-darmstadt.de}}

	\bigskip

	\bigskip

	\noindent This document is released under the following Creative Commons license:\\
	Attribution -- NonCommercial -- NoDerivatives -- International 4.0\\
	\tturl{http://creativecommons.org/licenses/by-nc-nd/4.0/}

	\bigskip

	\noindent \includegraphics[width=2.5cm]{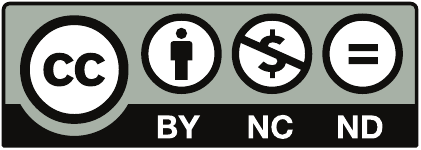}
	\vfill
}

\chapter*[Acknowledgments]{Acknowledgments}

Being a PhD student is a strange experience. I am sure that everyone who goes through this experience has their own personal stories, difficult moments to remember, and funny anecdotes to tell. I, for one, can truly say that these last five years have been exciting, funny, and productive. In short, they have been intense, and I can really say at the end that I have grown up a lot, both from an academic and from a personal perspective.

All of this I owe to my advisor, Marc Fischlin. If I could travel back in time and I were given the choice of applying as a PhD student again, at {\em any one} research group I possibly wished for, I would still spam ruthlessly my application to Marc. He taught me a lot of things which go well beyond academic matters, and I value his guidance immensely. When I was accepted in Marc's group in 2011, I was not aware at the time of how privileged I was. Now I am, and for this I will owe forever a debt of gratitude to Marc.

Being part of the group was a great experience, and I really would like to thank a lot my present and former colleagues for this. I am particularly grateful to Andrea for being always there to help me with the bureaucracy, to Giorgia for helping me to support the thesis that Hawaii Pizza is a mortal sin, to Özgür for taking care of me during my first months in Darmstadt, and to Paul for sharing with me a lot of good time, laugh, and hate for pigeons. I am also very grateful to Arno, Chris, Christian, Cristina, Felix, Jacqueline, Pooya, Sogol, and Victoria, for their friendship and support. Thank you all!

I would also like to thank all my coauthors for many successful collaborations and for having helped me a lot into expanding my scientific knowledge. Sometimes collaboration turned into sincere friendship as well, and therefore I would like to thank in particular Gorjan Alagic, Andreas Hülsing, Nikolaos Karvelas, and Christian Schaffner for the priceless time spent together.

Finally, I would like to thank my family for their endless love and support. I will always look at you as an example and a guidance, and I strive to make you proud of me every day of my life. Thanks.

\vspace{1cm}

\hfill Tommaso Gagliardoni

\hfill Darmstadt, December 2016

\chapter*[Abstract]{Abstract}

We call {\em quantum security} the area of IT security dealing with scenarios where one or more parties have access to quantum hardware. This encompasses both the fields of {\em post-quantum cryptography} (that is, traditional cryptography engineered to be resistant against quantum adversaries), and {\em quantum cryptography} (that is, security protocols designed to be natively run on a quantum infrastructure, such as {\em quantum key distribution}). Moreover, there exist also {\em hybrid models}, where traditional cryptographic schemes are somehow `mixed' with quantum operations in certain scenarios. Even if a fully-fledged, scalable quantum computer has yet to be built, recent results and the pace of research in its realization call for attention, lest we suddenly find ourselves one day with an obsolete security infrastructure. For this reason, in the last two decades research in quantum security has experienced an exponential growth in interest and investments.

In this work, we propose the first systematic {\em classification} of quantum security scenarios, and for each of them we recall the main tools and results, as well as presenting new ones. We achieve this goal by identifying four distinct {\em quantum security classes}, or {\em domains}, each of them encompassing the security notions and constructions related to a particular scenario. We start with the class \QS0, which is `classical cryptography' (meaning that no quantum scenario is considered), where we present some classical constructions and results as a preliminary step.

Regarding post-quantum cryptography, we introduce the class \QS1, where we discuss in detail the problems arising when designing a classical cryptographic object meant to be resistant against adversaries with local quantum computing power, and we provide a classification of the possible quantum security reductions in this scenario when considering provable security. Moreover, we present results about the quantum security and insecurity of the {\em Fiat-Shamir transformation} (a useful tool used to turn interactive identification schemes into digital signatures), and ORAMs (protocols used to outsource a database in a private way).

In respect to hybrid classical-quantum models, in the security class \QS2 we discuss in detail the possible scenarios where these scenarios arise, and what a correct formalization should be in terms of {\em quantum oracle access}. We also provide a novel framework for the quantum security (both in terms of indistinguishability and semantic security) of {\em secret-key encryption schemes}, and we give explicit secure constructions, as well as impossibility results.

Finally, in the class \QS3 we consider all those cryptographic constructions designed to run natively on quantum hardware. We give constructions for {\em quantum encryption schemes} (both in the secret- and public-key scenario), and we introduce transformations for obtaining such schemes by conceptually simpler schemes from the class \QS2. Moreover, we introduce a quantum version of ORAM, called {\em quantum ORAM (QORAM)}, aimed at outsourcing in a private way a database composed of quantum data. In proposing a suitable security model and an explicit construction for QORAMs, we also introduce a technique of independent interest which models a quantum adversary able to extract information from a quantum system without disturbing it `too much'.

We believe that the framework we introduce in this work will be a valuable tool for the scientific community in addressing the challenges arising when formalizing sound constructions and notions of security in the quantum world.

\cleardoublepage
\tableofcontents*

\mainmatter

\chapterstyle{madsen}
\DeclareFontShape{T1}{lmr}{bx}{sc}{<->ssub*lmr/bx/n}{}
\DeclareFontShape{T1}{lmss}{bx}{sc}{<->ssub*lmss/bx/n}{}

\chapter{Introduction}

{\em Cryptography} is the subdiscipline of mathematics studying {\em information security}, that is, the processing of information in presence of an adversary. This includes goals such as {\em communication secrecy, message authentication, identity verification, multiparty computation,} and much more. In the modern era of electronic information processing, cryptography is an area of crucial importance, and its applications are ubiquitous.

Modern cryptography is based on {\em provable security}. This is a methodological approach to assessing the security of a cryptosystem, where rigorous mathematical models and proofs are required in order to show that the security of the cryptosystem can be formally validated. Arguably the most important branch of provable security, from a practical standpoint, is {\em computational security}, which aims at {\em reducing} the security of a cryptosystem to some basic {\em hardness assumptions} in a mathematically sound way. Hardness assumptions are inherent to the difficulty of solving certain mathematical problems (such as integer factorization) which, for theoretical or historical reasons, are widely considered to be very hard to solve even with the help of the most powerful supercomputers known today. If a given cryptosystem is computationally secure, this means that on one hand it is always {\em theoretically} possible for an adversary with enough computational resources to break the security of that cryptosystem. But on the other hand, doing so would reguire either an unreasonable amount of time (modern standards of security often refers to many times the age of the universe), or an unreasonable amount of computational resources (storage, memory, power, etc.), or both.

The advantage of having a provably secure cryptosystem is that, as long as the security model used is sound and the underlying hardness assumptions hold, one can stay assured that the cryptosystem cannot be `broken'. This is in stark contrast with the `heuristic' approach to cryptography employed until the '70s, where cryptosystems were designed to be secure according to the intuition of the authors, and the only guarantee of that security was given by the `test of time', in the sense that nobody would find a way to attack the cryptosystem for a long enough time. This approach has turned cryptography from a mere engineering exercise to a logical-deductive discipline.

However, the effectiveness of provable security strongly relies on the hardness assumptions used, which are {\em not guaranteed}. Good hardness assumptions are based on the observation that algorithmical advances on solving the underlying mathematical problem would imply (unlikely) breakthrough results of scientific importance. However, all of these assumptions are also based on the {\em belief} that the future computing technology will never be {\em inherently different} from today's, save for a somewhat expected increase in performance, due to engineering improvements.

\section{Security in a Quantum World}

This is where {\em quantum computers} come into play. Quantum computers~\cite{Feynman} are machines, first theorized by Richard Feynmann in the early '80s, which are not based on the laws of classical physics like traditional computers are, but on the laws of {\em quantum mechanics} instead. Quantum mechanics is a very fundamental scientific theory, which has revolutionized physics since the early 20th century. Despite requiring a quite involved mathematical formalism and leading often to very counterintuitive consequences, it has routinely succeeded in predicting experimental results which classical physics could not explain.

From a formal point of view, a quantum computer is a mathematical model where the laws of quantum mechanics are exploited to perform some kind of computation, in a much more efficient way than traditional computers. 
Quantum computers promise to revolutionize the Age of Information as we know it. The ability to store, transmit, and process quantum data opens a world of new possibilities in the area of information processing. 
Simplified~\cite{Dwave} or limited models~\cite{IBM} of quantum computers have already been built, and everything from the experiments performed so far seems to confirm the validity of the underlying theory and the viability of the technology. Although a fully-fledged, scalable quantum computer has yet to be built, recent results~\cite{Google} and the pace of research in its realization seem to hint at the fact that quantum computing might soon become a reality.

\subsection{Post-Quantum Cryptography}

It turns out that, due to the effects predicted by quantum mechanics, quantum computers can perform tasks which are not possible with any classical computing device, present or future. The breakthrough result in this direction (which sparked a lot of interest for quantum computing in the area of cryptography) is the 1994 work by Peter Shor~\cite{Shor94}, who showed how for a quantum computer it is possible to factor large integers efficiently, a mathematical task considered to be unreasonably difficult until then, and at the base of many modern cryptosystems such as RSA~\cite{RSA}. Subsequent works have shown how to harness the power of quantum computing in order to speed up the search of solutions to problems like the {\em discrete logarithm}~\cite{Watrousalgos} on finite fields and elliptic curves, search on unstructured database~\cite{Grover96}, collision finding~\cite{BHT98}, and many others. Given that these are all hardness assumptions at the base of the security of cryptosystems~\cite{DH,ElGamal,ECDSA} widely adopted in the industrial, banking, and military sectors amongst others, it is clear how the realization of a scalable quantum computer would pose a threat to modern IT infrastructures.

A sound notion of security should be {\em proactive}, i.e., trying to take countermeasures against a reasonable future threat before the threat manifests itself. For this reason, cryptography has tried to address the looming danger of quantum computing since the early '90s. The idea is to find new mathematical problems which are supposed to be `hard' {\em even} for quantum computers, so that new, `quantum-immune' cryptosystems can be constructed by relying on such new quantum computational hardness assumptions. These are problems such as finding short vectors on lattices (which are geometric structures of a certain form), inverting hash functions, decoding certain types of linear codes, and a few others. The branch of cryptography dealing with the mathematical analysis of these assumptions and the construction of new cryptosystems based on such assumptions is called {\em post-quantum cryptography}~\cite{pqcrypto}. Post-quantum cryptography is today a thriving branch of information security, and so far it has been quite successful at designing cryptosystems which are at the same time reasonably efficient on today's hardware, and based on problems which are believed to be quantum-hard.

However, post-quantum cryptography has two fundamental issues.

The first problem is that security proof techniques that have been developed for traditional cryptosystems might {\em fail} when `translated' to the quantum scenario. A typical example is {\em rewinding}, a technique used in the security proofs of many cryptosystems, which roughly consists in modeling a scenario where the adversary is first run once, then rewound, partially reset, and then re-run again, in order to extract two different but related `adversarial transcripts' that are then used somehow in the security proof. The problem is that rewinding often {\em does not work} with quantum adversaries, because the nature of quantum mechanics does not guarantee that a `partial reset' of a quantum computer is always possible.

Proof failures of this kind have often been ignored in the post-quantum community in the past, and there are examples of attempts to `patching' non--post-quantum cryptosystems into post-quantum ones, by merely replacing the underlying hardness assumption with a quantum-hard one, and ignoring the fact that in so doing the security proof might become invalid.

The second problem of post-quantum cryptography is the often incomplete understanding of sound security models in the quantum world. One thing is to say that {\em ``the cryptosystem should be secure against a quantum adversary''}, another thing is to formalize mathematically what this exactly means. Models that are used for classically secure schemes are sometimes not adequate to model quantum security, and this can lead to confusion.

A typical example is the case of the {\em random oracle model (ROM)}, which is a formal paradigm widely used in security proofs. A random oracle is a purely mathematical construct which is completely independent from the type of adversary considered, and there are hence no exotic technical difficulties in adopting such paradigm in security proofs for post-quantum cryptosystems. In fact, such approach has been taken before, and there exist in literature cryptosystems advertised as `post-quantum' just because they are based on quantum-hard problems and provably secure in the ROM.

A random oracle, however, is just an abstraction describing an idealized model of hash function, which is an algorithmic object eventually run on a computing device. As the code for such a hash function is usually public, it is reasonable to assume that an adversary equipped with a quantum computer could run the code on his quantum machine, and therefore would be able to access the hash function in a way which is not modeled anymore by the ROM. For this reason, in a sound post-quantum security analysis, the random oracle model should always be avoided, and replaced by a different, more involved model called {\em quantum random oracle model (QROM)}. It can happen that schemes proven secure in the ROM become insecure in the QROM~\cite{QROM}.

All the above considerations are not intended to mean that the whole idea of post-quantum cryptography relies on a flawed model. In fact, there are plenty of cryptographically sound security analyses, where such problems are carefully taken care of. However, it is often the case that `secure against quantum adversaries' is confused with `relying on quantum-hard assumptions'.

\subsection{Quantum Cryptography}

On one hand, quantum computing poses new challenges for modern cryptography, as many of the currently used cryptographic schemes and protocols base their security on the hardness of certain mathematical problems which are known to be easily solvable by a quantum machine. On the other hand, quantum computers open up new possibilities in secure information processing, as they can also be used `defensively' in order to reach unprecedented levels of privacy, integrity, and trusted authentication. Importantly, it is often the case that such applications do not even require a fully-fledged scalable quantum computer, but only quantum hardware of modest technological engineering difficulty, which is already commercially available and deployed in many applications worldwide.

A typical example is {\em quantum key distribution (QKD)}~\cite{BB84}, where two remote parties aim at establishing a secure communication channel by exchanging a secret key, employing the exchange of elementary quantum information packets ({\em qubits}) through a quantum channel. This can be technologically done, for example, by transmitting polarized photons through an optic fiber channel. QKD is already largely developed~\cite{idquantique}, and it provides levels of security that classical cryptography cannot reach.

Looking into the future, with the advent of more and more advanced quantum hardware, it is easy to envision a world where a large part (if not most or all) of our global IT infrastructure will rely on quantum information processing. Under this scenario, it is important to think how to manage security related to quantum data. Not only it is required to re-model in a quantum way tasks usually performed by classical cryptography, for example {\em encryption of quantum data}~\cite{ABF+16} or {\em quantum authentication}~\cite{quantauth}. But it also means to consider tasks which are {\em inherently impossible} without quantum data, and which only make sense when considering a `fully quantum infrastructure', such as {\em quantum money}~\cite{quantmoney} or {\em delegated quantum computation}~\cite{delegated}.

In general, quantum computers promise to revolutionize the Age of Information as we know it. The ability to store, transmit, and process quantum data opens a world of new possibilities in the area of information processing. {\em Quantum cryptography} is the branch of cryptography which deals with designing secure cryptographic solutions which are natively meant to be run on a quantum hardware - this includes QKD and all of the other examples above, and still others. Quantum cryptography is a relatively recent area of study of modern cryptography, and there is still much to be done in terms of inventing new cryptosystems, creating correct security models, and figuring out the relations between classical and quantum cryptographic constructions.

\section{Contribution and Structure of this Work}

We define {\em `quantum security'} to be the discipline dealing with {\em all} the scenarios where one or more parties have access to quantum hardware. This encompasses both the fields of post-quantum cryptography, quantum cryptography, and also {\em hybrid models}, where traditional cryptographic schemes are somehow `mixed' with quantum operations in certain scenarios. The term `quantum security', 
although having appeared in the scientific literature before, has often been used used inconsistently from one work to another (see, for example,~\cite{ZhandryPRF,UnruhEverlasting,KuwakadoM12,Frodo}), at times denoting `post-quantum' notions of security, and at times denoting something else.

In this work, we provide the first systematic classification of quantum security scenarios, and a new framework for modeling quantum security notions in a sound way. We achieve this by identifying four distinct {\em quantum security classes}, or {\em domains}, each of them encompassing the security notions and constructions related to a particular scenario. We denote these classes by \QS~(standing for `quantum security'), followed by a number identifying the class. For each of these classes we recall known notions and results, as well as providing some results which are new or appearing in one or more of the author's publications. We start with a preliminary section in \textbf{Chapter~\ref{chap:preliminaries}} where we recall some basic concepts and notation, and then we proceed by presenting the four quantum security classes in the following chapters. 

As it often happens in academic research, many of the results presented in the various chapters of this thesis stem from collaborative projects, where each individual achievements can be contributed by several, and most often all, researchers participating in that project. This makes it hard, if not impossible sometimes, to pinpoint who contributed to which specific part of the overall work. At the beginning of chapters~\ref{chap:QS0},~\ref{chap:QS1},~\ref{chap:QS2}, and~\ref{chap:QS3}, we will give an account of the results presented in that chapter which are novel or appearing in some of the author's publications.

\subsection{\QS0}

We start in \textbf{Chapter~\ref{chap:QS0}} with the class \QS0, which is `classical cryptography' (meaning that no quantum scenario is considered), where we present some results about traditional cryptography as a preliminary step. In this chapter we introduce security models for different classical cryptographic primitives, and we also introduce other building blocks and transformations from one primitive to another. More in detail, first we define and analyze in Section~\ref{sec:QS0buildblocks} some of the building blocks used in modern cryptography: {\em pseudorandom number generators, functions, and permutations}.

Then we look at the security models (and some example of constructions) for {\em secret-key and public-key encryption schemes}, in sections~\ref{sec:SKES} and~\ref{sec:pke} respectively. We do it by looking at both the security models of {\em semantic security} and {\em indistinguishability of ciphertexts}.

In Section~\ref{sec:sig}, we discuss {\em digital signature schemes}, both in the standard model and in the ROM, and we show how to obtain secure signature schemes through the {\em Fiat-Shamir transformation} in Section~\ref{sec:FS}.

Finally, in Section~\ref{sec:ORAM}, we introduce {\em oblivious random access machines (ORAMs)}, which are interactive protocols used to privately outsource a large database. We look at PathORAM, one of the most famous of such protocols, by using the formalism introduced in~\cite{GKK17}.

\subsection{\QS1}

In \textbf{Chapter~\ref{chap:QS0}}, we look at post-quantum security, and we call \QS1 the related quantum security domain. We start in Section~\ref{sec:QS1basics} with a detailed discussion of all the issues arising when modeling quantum provable security for classical cryptographic objects, including some examples of how classical proofs can fail when `translated' to the quantum world, and the meaning of {\em quantum access to classical oracles}. We conclude this section with a {\em classification} of possible quantum security reductions which, to the best of the author's knowledge, does not explicitly appear in existing literature.

Then, in Section~\ref{sec:QROM} we introduce the quantum random oracle model, and we give some technical tools to deal with quantum random oracles.

In Section~\ref{sec:pqbuildingblocks}, we see how the security models for the building blocks defined in Section~\ref{sec:QS0buildblocks} change when considering post-quantum scenarios. We also have a look at cryptographic objects which are minimal {\em post-quantum hardness assumptions}, such as {\em post-quantum one-way functions} and {\em post-quantum one-way trapdoor permutations}.

In Section~\ref{sec:QS1enc} we discuss post-quantum security notions for encryption schemes, both in the secret-key and public-key scenario, and we show some basic constructions. Then we discuss post-quantum digital signatures in Section~\ref{sec:QS1sig}. We do this both for the standard post-quantum model and for the quantum random oracle model.

We proceed in Section~\ref{sec:FSQROM} to the analysis of the Fiat-Shamir transformation in the quantum random oracle model. We provide here both a positive and a negative result: if the underlying identification scheme has certain properties, then the Fiat-Shamir transform of that scheme yields a secure signature scheme in the quantum random oracle model. However, if the underlying identification scheme has different properties, it is possible to find an argument (using the technique of {\em meta-reduction}) which shows that security proofs of a certain form cannot be found at all. The surprising result here is that identification schemes having the latter type of properties are usually {\em less desirable} (in terms of security) than the former ones. We exploit this fact by showing a counterintuitive but efficient technique to `strengthen' the quantum security of a signature scheme obtained through the Fiat-Shamir transformation by `weakening' the security of the underlying identification scheme.

Finally, in Section~\ref{sec:pqORAM} we look at post-quantum ORAMs, and at sufficient and necessary conditions to obtain a post-quantum version of PathORAM.

\subsection{\QS2}

In \textbf{Chapter~\ref{chap:QS2}}, we look at {\em superposition-based quantum security}, and we call \QS2 the related quantum security domain. This security class deals with special scenarios, where the cryptosystems studied are still classical (and can hence be run on a classical computer), but {\em extra} security guarantees against quantum adversaries are required in respect to the `post-quantum' definition of security. We model these new scenarios in terms of {\em quantum oracle access capabilities} of the adversaries, explaining when such access is already implied in \QS1 and when instead it leads to new security scenarios covered by \QS2. Such scenarios arise in certain contexts, such as {\em obfuscation} and {\em fault attacks}, as explained in Section~\ref{sec:whysuperposition}. But very often they also stem from ambiguous interpretations of the `post-quantum' setting (as defined in \QS1) sometimes present in the literature. From this point of view, one of the most important contributions of this thesis is to formally clarify the distinction between these two security classes.

In Section~\ref{QS2:bb} we look at what happens when considering cryptographic building blocks in the new scenarios. It turns out that, in respect to the post-quantum scenarios, nothing changes for most of them, with two notable exceptions: quantum secure pseudorandom functions and permutations.

Finally we discuss quantum-resistant encryption schemes in Section~\ref{sec:QS2enc}, with a special emphasis on the secret-key case. For such schemes, we provide new notions of indistinguishability and semantic security, as well as secure constructions and impossibility results.

\subsection{\QS3}

Finally, in \textbf{Chapter~\ref{chap:QS3}}, we leave the realm of classical cryptosystems, and we look at {\em quantum cryptosystems}, that is, cryptosystems meant to be natively run on quantum hardware.

First we look at {\em quantum encryption} (that is, quantum algorithms for the encryption of quantum data) both in the secret-key (Section~\ref{sec:QS3ske}) and public-key (Section~\ref{sec:QS3pke}) scenarios. For both cases we provide security notions, as well as new constructions. We also show a novel technique for building encryption schemes secure in the \QS3 sense starting from encryption schemes secure in the \QS2 sense.

Finally, we introduce {\em quantum ORAMs (QORAMs)} in Section~\ref{sec:QORAM}. This is a new primitive (basically a quantum version of ORAM) which is aimed at outsourcing in a private way a database composed of quantum data. In proposing a new security model and an explicit construction for QORAMs, we also introduce a novel technique of independent interest which models a quantum adversary able to extract information from a quantum system without disturbing its state `too much'.

\section{Related Work}

The idea of {\em quantum security} as defined in this work is to encompass different types of scenarios which have in common the secure management of information in presence of quantum devices. Therefore, the existing related literature in this respect is vast, and we only cite a few key works here.

The term `post-quantum cryptography', as meant in the \QS1 sense, was popularized by Bernstein, Buchmann, and Dahmen in~\cite{pqcrypto}. The QROM was introduced in~\cite{QROM}. Regarding the problems inherent to quantum rewinding, see Watrous~\cite{WatrousZK}, Unruh~\cite{UnruhZK}, and Ambainis et al.~\cite{QuantRewinding}. Song~\cite{Song14} discussed relations between classical and quantum reductions, and Hallgren et al.~\cite{HSS11} discussed classical cryptographic protocols in the quantum world. Post-quantum building blocks and encryption schemes can be constructed from mathematical problems on lattices~\cite{GGH97,Micciancio11a,LPR13}, linear codes~\cite{McEliece}, multivariate equations~\cite{OilVinegar}, and supersingular isogenies~\cite{isogeny}. In addition to the problems just mentioned, post-quantum signature schemes can be constructed from hash functions~\cite{sphincs}.

Superposition-based attacks have been first proposed in~\cite{superposition} in respect to multiparty computation. Quantum-secure pseudorandom functions and pseudorandom permutations have been investigated by Zhandry~\cite{ZhandryPRF,ZhandryPRP}, Kuwakado and Morii~\cite{KuwakadoM10,KuwakadoM12}, and Alagic and Russell~\cite{gorjanPRP}, while secret- and public-key encryption schemes falling in the \QS2 cathegory have been proposed by Boneh and Zhandry in~\cite{BZ13}, where superposition-resistant signature schemes also appear. Signature schemes secure against superposition attacks have also been studied in~\cite{ES15}. Anand et al.~\cite{Anand+}, Kaplan et al.~\cite{Kapl+}, and Santoli and Schaffner~\cite{SS16} extended some attacks against pseudorandom permutations to other block ciphers, modes of operation, and compression functions.
Quantum key distribution was introduced in the seminal works by Wiesner~\cite{Wiesner}, and Bennet and Brassard~\cite{BB84}. Quantum money was introduced by Aaronson~\cite{quantmoney}. Computationally secure quantum encryption was formalized by Broadbent and Jeffery~\cite{BJ15}, while~\cite{AlagicAuth,Barnum+,Garg+16} deal with authentication of quantum information. See~\cite{BS16} for an overview of quantum cryptographic schemes belonging to the \QS3 class, and Vidick and Watrous~\cite{VW16} for an overview of quantum complexity theory and reductions in the quantum world.

\chapter{Preliminaries}\label{chap:preliminaries}

In this chapter we discuss the notation and provide basic definitions used in the rest of this work.

\section{Basic Notions}

We start with a few basic concepts, mathematical notation and terminology. In the rest of this work, `w.l.o.g.' stands for `without loss of generality', `iff' stands for `if and only if', and `classical' means `non-quantum'.

Numbers, strings, and generic atomic objects are denoted by default as lowercase letters, e.g., $a,b,x,y$. In particular, indices for sequences or families will be often denoted by $n,m,i,j,k$. Sometimes inputs and outputs of an algorithm will be denoted by lowercase Sans Serif script, e.g., $\com,\state,\sig$. The security parameter is \secpar, or \secparam when expressed in unary notation.

Special symbols are $\bot$ (usually denoting `error', or `lack of meaning') and the lowercase Roman \im (denoting the imaginary unity, $\sqrt{-1}$). The symbol $\|$ denotes concatenation of {\em bit strings}, and the symbol $0^k$ (resp. $1^k$) denotes a $k$-bit string of zeroes (resp., ones). For a bit string (or natural number) $x$ we denote its bit size (or bit length) as $|x|$. If $x$ is a non-integer number, $|x|$ denotes its absolute value. If $x$ is a complex number, $|x|$ denotes its complex modulus, and $\conj{x}$ its complex conjugate.

Families or collections of objects (sets, functions, probability distributions) are of the form $\family{\A},\family[j,k]{\X}$, where individual elements of the family are indexed, e.g., $\A_n,\X_{j,k}$. However, if there is no ambiguity in the choice of the index (usually this is the security parameter), such families are labeled in short just as \A,\X, etc.

Sets are usually denoted by uppercase letters, e.g., $T,X,Y$, except for special sets such as $\emptyset,\NN,\RR,\CC$, and the set of all permutations on a set $X$, denoted by $S(X)$. The set of all finite bit strings or words is $\words$. However, sets of bit strings will often be presented as families, where each member of the family contains bit strings of the same length. In this case, sets will be denoted by \T,\X,\Y instead, being understood that, e.g., $\X = \family{\X}$, where 
$\X_\secpar$ only contains bit strings of length $\f(\secpar)$ for some positive (usually polynomial) function \f.  
The cardinality (number of elements) of a set $X$ is denoted by $|X|$. Set operations are $\cup$ (union), $\cap$ (intersection), $\setminus$ (set difference), and $\times$ (Cartesian product). If a tuple $(x,y,z) \in \X \times \Y \times \Z$, then single entries of the tuple are isolated by writing, e.g., $(x,y,z)_{\X\Y\Z} \project{\Y} = y$.

Functions (from sets to sets) are denoted by lowercase calligraphic letters, e.g., $\f,\g,\l: X \to Y$. Borrowing a commonly used notation when defining `small' quantities (relative to some parameter), exceptions to this notation are special functions $\epsilon$ and $\delta$.

However, when a function is actually a family (indexed, for example, in terms of the bit size of the input) then it is denoted by uppercase calligraphic letters, e.g., $\F,\G,\L$. Commonly, in this case, domain and target space of these functions are also indexed as families, in relation to the bit size of the function's input. For example, $\F:\X\to\Y$ represents a function \F from set \X to set \Y, which can be seen as a family of functions $\family{\F}$, where $\F_\secpar:\X_\secpar \to \Y_\secpar$.

A (real-valued) function \f is {\em polynomially bounded} iff there exists a polynomial function $\p$ and an element $\bar{x}$ such that $|f(x)| \leq \p(x), \foral x$ with $|x| > |\bar{x}|$. In this case we write $\f = \poly$. A (real-valued) function $\epsilon$ is {\em negligible} iff, for any polynomial function \p, there exists an element $\bar{x}$ such that $|\epsilon(x)| < \frac{1}{\p(x)}, \foral x$ with $|x| > |\bar{x}|$. In this case we write $\epsilon = \negl$.

Lowercase Greek letters denote quantum states, either pure ones when written in bra-ket notation (e.g., $\ket{\phi},\ket{\psi}$) or mixed ones when written without (e.g., $\sigma, \rho$). Exceptions are the symbols $\delta$ and $\epsilon$, as already discussed, and $\lambda$ (used for eigenvalues). Uppercase Greek letters ($\Sigma,\Gamma,\Theta$) are reserved for special purposes, usually to denote quantum channels.

Data structures (trees, blocks) are labeled with Typewriter script, e.g., \texttt{tree, block, node}.

\subsection{Probability}

Distributions are denoted by uppercase calligraphic letters, e.g., $\D,\P,\U$. Distribution ensembles, or families, are denoted by $\family{\D}, \family{\P}$, etc. As usual, if there is no ambiguity in the choice of the index (usually this is the security parameter), such families are labeled in short just as \D,\P, etc., with individual member distributions being $\D_n, \P_n$, etc.

If \D is a distribution over a set \X, then sampling an element $x$ from the distribution is written as $x \fromdist{\D} \X$ (or, a shorthand notation when the domain is clear, just $x \from \D$). Sampling an element uniformly at random from a set \R is written as $r \rand \R$.

The {\em support} of a distribution \D over a set \X is the subset of elements with non-zero probability, i.e., $\set{x \in \X : \Pr[x \from \D] > 0}$. The {\em cardinality of a distribution} is the cardinality of its support.

If \D is a distribution over $\X \times \Y$, then we denote the distribution on \X induced by \D as $\D_\X := \D\project{\X}$, and the sampling as $\D_\X \to x$ where $x:=(x,y)\project{\X}$.

The {\em total variation distance} (or, {\em statistical distance}) of two distributions $\D_0,\D_1$ is defined as:

$$
\left| \D_0 - \D_1\right| := \sum_x \left| \Pr[x \from \D_0] - \Pr[x \from \D_1] \right|.
$$

\subsection{Linear Algebra}

Vectors are denoted either as tuples (e.g., $(x_1,\ldots,x_n)$) or as boldface characters for the notation of the corresponding components (e.g., $\vec{x}$). The zero vector is denoted as \zerovec. Matrices (linear operators between two vector spaces) are denoted by uppercase letters, e.g., $A,B,M$. (unless families, in that case $\A,\B,\M$ etc., as previously explained), except for the special symbols {\em zero matrix (or null operator) over $n$ elements} (denoted by $\nullop_n$), and the {\em identity matrix (or identity operator) over $n$ elements} (denoted by $\Id_n$). If $M$ is an $n \times m$ matrix (which includes the case of vectors or scalars if $n$ or $m$ equals $1$), then $M^T$ denotes its $m \times n$ transpose, $\conj{M}$ denotes its $n \times m$ complex conjugate, and $M^\dagger$ denotes its $m \times n$ Hermitian conjugate (or adjoint) $\conj{M^T} = \conj{M}^T$. If $M$ is an $n\times n$ matrix with non-zero determinant, its unique inverse is denoted by $M^{-1}$. An $n\times n$ matrix (or linear operator) $M$ is {\em Hermitian} if $M = M^\dagger$, and {\em unitary} if $M^\dagger = M^{-1}$. The {\em trace} of a square matrix $M$ is denoted by $\tr(M)$, and it is the sum of the elements on the diagonal.

A {\em complex Hilbert space} is a complex vector space \Hilbert, together with an inner product operation $\braket{.,.}: \Hilbert \times \Hilbert \to \CC$ such that \Hilbert (seen as a metric space) is complete in respect to the metric $\|\vec{x}\|:= \sqrt{\left|\braket{x,x}\right|}$ induced by the inner product. Unless otherwise specified, the inner product adopted here is always the scalar product:

$$
\braket{\vec{x},\vec{y}} := \vec{x} \vec{y}^\dagger = \left( x_1 , \ldots , x_n \right) \left( \begin{array}{c}\conj{y_1}\\ \vdots \\ \conj{y_n} \end{array}\right) = \sum_i x_i \conj{y_i}
$$
The norm induced by the above product is the {\em Euclidean norm}, and it is denoted by $\|\vec{x}\|_2$. The Euclidean distance between two vectors $\vec{x}$ and $\vec{y}$ is hence $\|\vec{x} - \vec{y} \|_2$. The {\em dimension} of a Hilbert space is the cardinality of a minimal set of orthonormal elements spanning the whole space. Such a set is called a {\em basis} for the complex Hilbert space, and it is not unique. In this work we only consider finite-dimensional complex Hilbert spaces.

\section{Classical Computation}

In this section we recall the basic concepts and notation related to classical computation and complexity theory. The topic is of course vast and here we do not cover in depth every aspect of it. For a more complete treatment of the aspects of computation and complexity theory we refer to~\cite{AroraBarak}.

\subsection{Circuits and Algorithms}

The fundamental objects of study of computation theory are {\em algorithms}, which are sequences of elementary operations applied to some input data; the goal is to perform some procedure on those input data to produce some output. The {\em complexity} of an algorithm can refer to the number of elementary steps performed, the running time, the memory consumption, or any other resource used during its execution. Such complexity is expressed in relation to the {\em instance size} of the computation, which is a positive integer expressing the `size' of the computational problem which the algorithm has to solve in order to perform the desired computation; this parameter is usually (related to) the bit size of the input. The complexity of an algorithm is then expressed as a function of the instance size: for example, if an algorithm \A has complexity at most $O(n^2)$ for instance size $n$, we say that \A has `quadratic complexity'. An algorithm is {\em deterministic} if it produces always the same output for the same input, while it is {\em probabilistic} if it also takes an additional input (of size at most polynomial in the instance size) drawn from uniform random bits; its output is hence expressed as a distribution over these `internal random coins'.

In this work we only deal with {\em time complexity}, i.e., we count as complexity the execution time of the algorithm. Time complexity is expressed in terms of the number of elementary operations performed by the algorithm, regardless of their nature, i.e., we assume for simplicity that any elementary operation (be it an addition, logical AND, division, etc.) takes one unit of time to execute. Moreover, as common in cryptography, we call the instance size the {\em security parameter}, denoted by \secpar. \DPT stands for `(Boolean) deterministic polynomial time', while \PPT stands for `(Boolean) probabilistic polynomial time', where `Boolean' refers to the fact that the algorithm operates on bit strings and performs elementary Boolean (bit) operations.

Traditionally, the two most commonly used models used to describe a classical algorithm are {\em Turing machines} and {\em Boolean circuits}.
\begin{itemize}
\item A Turing machine is a mathematical model describing an abstract machine with an internal state, acting on a data tape and performing operations according to a pre-specified set of rules.
\item Boolean circuits are acyclic directed graphs where the nodes are either input bits, output bits, or elementary (Boolean) operations. Complexity in this case is given by the total number of gates in the circuit.
\end{itemize}
In this work, by `algorithm' we mean `a uniform family of circuits', i.e., there exists a Turing machine which, given the security parameter expressed in unary $\secparam$ as input, runs in time at most polynomial in $\secpa$, and outputs a description of the $\secpa$-th member of the circuit family. So, for example, a \PPT algorithm \A is a family of Boolean circuits $\A := \family{\A}$ such that:
\begin{enumerate}
\item there exists a Turing machine \M such that, on input \secparam, \M runs in time $O\left(\poly(\secpa)\right)$ and outputs a description of $\A_\secpa$; and
\item $\A_\secpa$ is a Boolean circuit of size $O\left(\poly(\secpa)\right)$, taking as input a $\poly(\secpa)$-bit value and a $\poly(\secpa)$ many uniformly random bits, and producing a $O\left(\poly(\secpa)\right)$-bit output.
\end{enumerate}
Algorithms, being families of circuits, are denoted by, e.g., $\A := \family{\A}$. When studying an algorithm which is a subroutine of another algorithm, or where we do not want to stress that it is a family, or anyway for clarity of notation, we use math Sans Serif script (e.g., \access, \KGen, \Enc). Every algorithm {\em always} gets as input at least the security parameter, so we will ignore it in the notation, being understood that such input is always present. In order to express that a deterministic algorithm \A, on input a value $x$, produces an output $y$, we write: $y := \A(x)$ or, equivalently, $\A(x) =: y$. For a probabilistic algorithm instead, the notation becomes $y \from \A(x)$ (or, equivalently, $\A(x) \to y$). However, if the output of a probabilistic algorithm \A is written as $\A(x) = y$, that is a shorthand notation for: $\Pr \left[ \A(x) \to y \right] = 1$, where the probability is taken over the internal randomness of \A. If an algorithm's only input is the security parameter (which we omit from the notation, as said), we only write, e.g., $\A =: y$, or $\A \to y$ if probabilistic.

The random coins of a probabilistic algorithm are almost always omitted from its input, so we write simply, e.g., $\A(x) \to y$; however, if for some reason it is necessary to `de-randomize' the algorithm (that is, to consider the deterministic algorithm obtained by fixing a particular choice of randomness $r$), we write this as $\A(x;r) =: y$. If \A is probabilistic, then the notation $\Pr \left[ \A(x) \to y \right]$ is meant as `probability over the randomness of \A, for that particular value $x$', while $\Pr_{x \in \X} \left[ \A(x) \to y \right]$ (or $\Pr_{x \from \X} \left[ \A(x) \to y \right]$) means `over the randomness of \A, averaged over the uniform distribution on \X'. However, if \A is deterministic, then $\Pr_{x \in \X} \left[ \A(x) \to y \right]$ (or $\Pr_{x \from \X} \left[ \A(x) \to y \right]$) is given by the fraction $\frac{\card{\set{x \in \X: \A(x) =: y}}}{\card{\X}}$.

Abusing notation, we express sometimes algorithms as {\em (families of) functions} from (families of) sets of inputs to (families of) sets of outputs. So, for example, $\A := \family{\A} : X \times \Y \to \Z \times \words$ means that, for every $\secpar \in \NN, \A_\secpar$ is a Boolean circuit taking as input one element of $X$ and one element of $\Y_\secpar$, and outputting one element of $\Z_\secpar$ and one extra bit string of unspecified length. If an algorithm only takes as input the security parameter and outputs elements in \X we write just: $\A : \to \X$.

Algorithms can be {\em interactive}, and communicate with each other. A special case is given by {\em stateful} algorithms, which have an internal state variable which can be updated and stored across different executions of the same algorithm (in this sense, the algorithm `communicates with his future self'). In order to represent this communication, three different notations can be used.

\begin{enumerate}
\item Explicit state transport. For example, if $\A = (\A_1,\A_2)$ and one wants the first stage algorithm $\A_1$ to communicate some information to the second stage $\A_2$, we write something like:
\begin{algorithmic}[1]
\State $\A_1(x) \to (y,\state)$
\State $\A_2(z,\state) \to w$
\end{algorithmic}
where \state, when left unspecified, is a bit string of size polynomial in the security parameter, carrying the information to be transmitted.
\item Circuit self-output, used in particular for stateful algorithms. For example, if $\A_0$ is the algorithm in the initial state, then $\A_0$ `outputs $y$ and a description of its own updated state' as: $\A_0(x) \to (y,\A_1)$. If using this notation, from now on $\A_0$ cannot be invoked again anymore. Instead, $\A_1$ is run on some other input $a$ and updates itself as: $\A_1(a) \to (w,\A_2)$. From now on, $\A_1$ cannot be invoked anymore. Instead a fresh invocation can be written as: $\A_2(w,y,b) \to (u,r,\A_3)$, and so on.
\item Communication transcript. In this case, two or more algorithms communicate back and forth through a {\em communication channel} (which is a shared register between the two circuits). The `history' of the content of such register during the execution of two algorithms \A and \B is called {\em communication transcript \com}, and it is usually denoted as: $\com \from \braket{\A(x),\B(y)}$.
\end{enumerate}

If an algorithm \A has {\em oracle access} to another algorithm (or family of functions) \O, this is written as $\A^\O$. In this case, it is understood that \A can communicate with \O through \O's input and output registers solely, while \A does not know anything else about \O's structure, code, or working details. Such communication is called {\em query}: \A queries \O on input value $x$, then \O computes the answer $y \from \O(x)$ and finally $y$ is sent back to \A. In this case, \O's running time is ignored: it is always assumed that one oracle invocation takes one unit time to execute, regardless of \O's running time. Giving \A oracle access to another resource \O models the case where \A is given `extra power' in performing a certain task, without having to deal with the exact way this task is performed.

\vfill

\subsection{Computational Complexity Theory}

{\em Complexity classes} are families of problems with related asymptotic difficulty. Their definition is often given in terms of {\em language verifiers}: a {\em language} is a subset of \words, and a {\em verifier} for a language is an algorithm which checks if a given input bit string belongs to that language (outputs $1$) or not (outputs $0$). In this work we only consider the following three classical complexity classes.
\begin{itemize}
\item \PPP is the set of all languages \L for which there exists a \DPT algorithm \M such that:
	\begin{enumerate}
	\item $\foral x \in \L \implies \M(x) = 1$; and
	\item $\foral x \notin \L \implies \M(x) = 0$.
	\end{enumerate}
Informally, \PPP is the set of all problems which are `easy to solve', in the sense that a solution for a given problem instance of size $n$ can be found deterministically in time at most polynomial in $n$.
\item \BPP is the set of all languages \L for which there exists a \PPT algorithm \M and a positive constant $c$ such that:
	\begin{enumerate}
	\item $\foral x \in \L \implies \Pr [\M(x) \to 1] \geq \half + c$; and
	\item $\foral x \notin \L \implies \Pr [\M(x) \to 0] \geq \half + c$.
	\end{enumerate}
Informally, \BPP is the set of all problems which are `easy to solve with high probability', in the sense that a solution for a given problem instance of size $n$ can be found with high probability in time at most polynomial in $n$. It is currently unknown whether $\PPP \neq \BPP$ or not~\cite{GolBPP}.
\item \NP is the set of all languages \L for which there exists a \DPT algorithm \M and a polynomial \p such that:
	\begin{enumerate}
	\item $\foral x \in \L \ \exists \ y \in \words$ with $|y|\leq \p(\secpar)$ such that $\M(x,y) = 1$; and
	\item $\foral x \notin \L, \foral y \in \words$ with $|y|\leq \p(\secpar) \implies \M(x,y) = 0$.
	\end{enumerate}
Informally, \NP is the set of all problems which admit a `solution easy to check'. in the sense that a candidate solution for a given problem instance of size $n$ can be tested deterministically in time at most polynomial in $n$. It is currently unknown whether $\PPP \neq \NP$ or not~\cite{AroraBarak}.
\end{itemize}

Let $\L \in \NP$ be a language with a (polynomially computable) relation $\R$, i.e., there exists a \DPT algorithm $\mathsf{Rel}$ and a polynomial $\p$ such that $x \in \L$ iff there exists some $w \in \W \subset \words$ such that $(x,w) \in \R$ and $|w| \leq \p(|x|) \forall x$, where $(x,w) \in \R \iff \mathsf{Rel}(x,w) = 1$. We say that $w$ is a {\em witness} for $x \in \L$ (and $x$ is called a {\em theorem} or {\em statement}). We sometimes use the notation $\R_\secpar$ to denote the set of pairs $(x,w)$ in $\R$ of complexity measured in relation to the security parameter, e.g., if $|x|=\secpar$. In this case, with abuse of notation we identify the relation $\R$ with the algorithm testing its membership $\mathsf{Rel}$.

\section{Classical Cryptography}

In this section we briefly recall the basic concepts and terminology used in modern cryptography.

\subsection{Provable Security}\label{sec:provsec}\label{sec:meta}

Traditionally, cryptography has been seen for a long time as a `cat-and-mouse' game, in the sense that the only way to validate the quality of a proposed cryptographic object was to perform some sort of cryptanalysis on it (i.e., `trying to break it'), and then trying to fix the vulnerabilities potentially found, until new flaws were found, and so on. Under this perspective, the criterion to decide whether a cryptographic object should be trusted or not is just `the test of time', in the sense that no new vulnerabilities are being found `for a long time'.

However, this paradigm has shifted radically in the last \textasciitilde{}30 years. The modern approach to defining good practice in cryptography is {\em provable security}, which is a paradigm involving a rigorous mathematical analysis of the cryptographic object, adversarial model, and security assumptions. In provable security, when analyzing a cryptographic scheme, one needs to provide rigorous definitions and models for the following aspects:
\begin{enumerate}
\item the {\em functionality} of the cryptographic object, i.e., what exactly is the goal that the object wants to achieve;
\item the {\em adversary model}, i.e., what does a `reasonable' adversary against the object look like? What does the adversary want to achieve? When can we say that he is `successful'?
\item The {\em security proof}, i.e., a mathematical proof showing that, under the specified model and some basic, commonly accepted assumptions, it is possible to rule out {\em any} successful adversary against the cryptographic object in exam.
\end{enumerate}

It is important to distinguish between two different concepts of security.
\begin{itemize}
\item {\em Information-theoretical (or, statistical) security.} In this case, the proof of security aims at showing that the behavior of the cryptographic object is statistically equivalent (in the sense that it produces a distribution of outputs at most negligibly different) to the behavior of an {\em idealized object}, against which no successful attacker can exist by definition. For example, an information-theoretical secure encryption scheme produces a distribution of ciphertexts which is at most negligibly different from the uniform distribution over all ciphertexts, regardless of the input plaintext. Clearly, information-theoretical security is very strong, because it gives security guarantees {\em regardless of the adversarial model}. However, being so strong, it is also limited in use, as very few cryptographic objects can be shown to be statistically secure.
\item {\em Computational security,} on the other hand, aims at showing that a cryptographic object is secure by relying on the intrinsic computational limitations of a `reasonable' adversary. For example, in a computationally secure (but not statistically secure) encryption scheme, an adversary might be able to break security by testing (`brute-forcing') all the possible encryption keys one after one. However if such an adversary, in so doing, takes an amount of time which exceeds by many orders of magnitude the age of the universe, we would not consider him a threat for the security of the cryptographic scheme. A commonly accepted definition of `computationally bounded adversary' is `polynomial-time bounded' (in the security parameter).
\end{itemize}

In this work we only focus on computational security, but sometimes we refer to statistical security when needed for comparison. The adversary model we consider in classical security is thus some form of \PPT algorithm, possibly with oracle access to additional resources.

The `winning condition' for a given adversary \A is expressed in terms of the outcome of an {\em experiment} (or {\em game}), which is a mathematical model describing the intuitive behavior of an adversary trying to compromise the security of a cryptographic scheme \S. Formally, an experiment is an algorithm (taking as input the security parameter \secpar and, optionally, other parameters) with oracle access to \A and (the components of) \S, and outputting some value (typically a bit) telling whether the experiment was successful (i.e., \A won) or not. The notation used is of the form $\game^{\mathsf{LABEL}}_{\S,\A}$, where $\mathsf{LABEL}$ identifies the particular experiment. The {\em advantage} of an adversary \A running such experiment (denoted by $\adv^{\mathsf{LABEL}}_{\S,\A}$) is the difference between \A's success probability, and the success probability of a `naif' adversary who just guesses at random a possible solution to the problem of breaking \S's security. Then, in order to define \S secure, two possible approaches are considered:
\begin{enumerate}
\item {\em game-based security}. In this case, it is required that the advantage of {\em any} (computationally bounded) adversary is `small' (meaning, negligible in the security parameter); or
\item {\em simulation-based security}. In this case, the success probability of an arbitrary adversary \A in the original experiment is compared to the success probability of the same \A in a {\em different} experiment, describing an idealized, or `simulated' situation where there is basically no possibility that \A can break the security of the underlying scheme. In this case, security requires that for {\em any} (computationally bounded) adversary, the difference between the success probabilities in the `real' and the `ideal' world are roughly the same (meaning, at most negligibly distinct).
\end{enumerate}
Both approaches are widely used in provable security. Usually, simulation-based security better captures the idea of transforming in a rigorous mathematical model what intuitively we want a cryptographic object to achieve; game-based security, however, is often of more immediate formulation and simpler use in security proofs. A common technique in provable security is in fact to show equivalence between an intuitive, rigorous simulation-based security definition, and a simpler, easier-to-use game-based one.

Regarding {\em security proofs}, it must be noticed that such proofs are intuitively very hard to come up with. In fact, it is in theory easy to show that a particular, formally well-described adversary is unable to successfully attack a certain cryptographical scheme. However, the security proofs we need require to rule out {\em every possible adversary}, even those which we do not know yet, or are unable to formalize. Therefore, directly showing security against one adversary does not work, and different techniques are used instead.

A very common technique to show the security of a cryptographic scheme \S is the concept of {\em reduction} to another problem, or primitive \P. Let us assume that \P is hard to solve, or anyway widely believed to be hard. Then one could `show' the security of \S by proving that the problem of breaking \S's security is `at least as hard' as solving \P. This is accomplished by proving that, given an hypothetical, successful adversary \A against \S, such adversary can be turned, {\em constructively and in an efficient way}, into an efficient solver for \P. In this case we say that the security of \S {\em reduces} to the hardness of \P, and the formal proof itself is called {\em reduction}. A typical example of reduction is giving an explicit description of an efficient algorithm \B which solves \P, and which has oracle access to \A (in that case \B is also said to be the reduction itself). We say that a reduction is {\em `black-box'} if such oracle access is the {\em only} interaction between \A and \B, and \B does not have any other clue about \A, such as insights about \A's code or access to oracles which, according to the security model, should be only accessible by \A. However, as it is common practice in provable security, \B is allowed to know a priori an upper bound on \A's running time or number of queries to his oracles.

Finally, another common topic in provable security are {\em impossibility results}, that is, general theorems stating that a certain class of cryptographic object having certain properties {\em cannot} be secure. The most direct way to do it is by providing an explicit attack, i.e., an efficient adversary working against every member of that class. However, this can be hard sometimes, and there are countless examples of cryptographic schemes where a direct attack is {\em not known}, but at the same time {\em no reduction can be found}.

A possible technique to show impossibility results is that of {\em meta-reductions}. Intuitively, a meta-reduction is `a reduction on reductions': the idea is to show that, if a scheme \S admits an efficient reduction \B to some problem \P, then another reduction \M exists, which uses \B to attack another, possibly different hard problem $\P'$. This rules out the existence of \B.

In the case of meta-reductions, since \B needs an efficient adversary \A against \P in order to work, and reductions must always be {\em constructive and efficient}, it should be \M's duty to provide such adversary \A for \B to work with. However, since \M cannot break \P directly (or else this would be a contradiction), the meta-reduction {\em simulates} a `fake' adversary, in such a way that the simulation cannot be used directly to break \P, but at the same time such simulation is undetectable from \B's perspective. So, a meta-reduction technique works like this:
\begin{enumerate}
\item assume the existence of a reduction \B from scheme \S to problem \P.
\item Give an explicit description of {\em any} adversary \A against \S. This adversary does not necessarily need to exist, because \B works regardless of \A's nature. In practice though, it is usually required that \B is a black-box reduction.
\item Give an explicit description of an efficient algorithm \M which can simulate \A (from \B's point of view) and any other resource or oracle that \B needs to access.
\item Execute the reduction \B, and use \B's output to break $\P'$.
\end{enumerate}

\subsection{Hardness Assumptions}\label{sec:assumptions}

{\em Hardness assumptions} relate to mathematical problems which are at the same time easy to formalize (and it is clear what a solver for these problems should accomplice), and such that to date no known general method for solving these problems has been found (and there is evidence that finding such a method is arguably very hard). These assumptions are important, because they identify problems which are very attractive reduce to during security proofs.

Since we are dealing with computational security, a very minimal assumption is that $\PPP \neq \NP$. This is widely believed to be the case~\cite{AroraBarak}; however, finding cryptographic reductions to such a minimal assumption is very hard. In this section, we recall some commonly used hardness assumptions used in cryptography. In what follows, we assume w.l.o.g. that the message space is $\X = \family{\X} := \left( \bin^\secpar \right)_\secpar$. 

One very well studied assumption that we will explicitly use later in this work is the {\em computational hardness of the discrete logarithm problem (DLP)}.

\begin{definition}[Discrete Logarithm Problem]\label{def:dlp}
For a security parameter \secpar, let $(\G,\star)$ be a cyclic group of order exponential in \secpar, with generator $g$, and such that $\star$ is efficiently computable. The {\em discrete logarithm problem (DLP) on \G} is, given $h\rand\G$, to find $x \in \NN$ such that $h = g^x$.
\end{definition}

The {\em DLP hardness assumption} (for a given group $(\G,\star)$) states that no \PPT algorithm exists, which is able to solve the DLP problem with probability better than $\half + c$ for any positive constant $c$ (i.e., the DLP problem is {\em not} in \BPP for many known groups). There exist many different variants of the DLP problem, such as the {\em decisional Diffie-Hellman (DDH) problem} and many others, see~\cite{BonehDDH} for a survey. There exist also many other number-theoretic hardness assumptions, both quantum-insecure (RSA~\cite{RSA} and factorization, elliptic-curve DLP~\cite{ECDSA}, etc.) and (presumably) quantum-resistant (lattice problems~\cite{GGH97}, code-based~\cite{McEliece}, isogenies~\cite{isogeny}, etc.) but we will not address them specifically in this work

Another very minimal hardness assumption that we make heavy use of is the existence of {\em one-way functions}. Intuitively, these are (families of) functions that are `easy' to evaluate on any input, but `hard' to invert on a random output, meaning that no efficient algorithm can find a pre-image for a randomly generated image.

\begin{definition}[One-Way Functions (OWF) and Permutations (OWP)]\label{def:owfowp}
Let $\F = \family{\F}$ be a \DPT algorithm, with $\F_\secpar : \X_\secpar \to \words$. \F is a (family of) {\em one-way functions (OWF)} iff for any \PPT algorithm \A it holds:
$$
\Pr_{x \rand \X} \left[ \A(\F(x)) \to x' : \F(x) = \F(x') \right] \leq \negl.
$$
Moreover, in the special case where $\F_\secpar : \X_\secpar \to \X_\secpar$ are permutations on $\X_\secpar$ for every \secpar, \F is a (family of) {\em one-way permutations (OWP)}.
\end{definition}
The existence of one-way functions would imply $\PPP \neq \NP$, but the converse is not believed to hold~\cite{AroraBarak}. However, one-way functions are considered to be a very minimal assumption for the existence of computationally secure cryptography. In general, reducing the security of a cryptographic object to the existence of one-way functions is a strong indicator of the scheme's security.

Notice the following: Definition~\ref{def:owfowp} does not say anything about individual members of the family being pseudorandom. For example, there might be one-way functions which always fixes certain bits of their output, which can hence be trivially inverted. However, these `easily predictable' bits cannot be `too many', otherwise an adversary \A could invert the whole function by guessing the other bits, against the assumption of one-wayness. Those (Boolean functions of) bits which are {\em not} easily predictable are called {\em hard-core bits} (or {\em hard-core predicates}).

\begin{definition}[Hard-Core Predicate]\label{def:hc}
Let $\F: \X \to \Y$ be a OWF. A polynomial-time computable function $\hc_\F : \X \to \bin$ is a {\em hard-core predicate (or bit) of \F} iff, for any \PPT algorithm \A it holds:
$$
\Pr_{x \rand \X} \left[ \A(\F(x)) \to \hc_\F(x) \right] \leq \half + \negl.
$$
\end{definition}
Whether {\em every} OWF admits hard-core predicates or not is an open problem~\cite{KatzLindell}. But it can be shown that, given any OWF \F, it is always possible to construct another OWF \H such that $\hc_\H$ exists. Moreover, if \F is a OWP, then also \H is.
 
\begin{proposition}[{\cite{HILL99}}]\label{prop:hc}
Let \F be a OWF (resp., OWP). Then 
it is possible to efficiently transform \F into 
a OWF (resp., OWP) \H such that at least one hard-core predicate $\hc_\H$ exists.
\end{proposition}
Given the above, from now on we assume for simplicity that every OWF admits hard-core predicates. In the case that $\F:\X\to\X$ (in particular, if \F is a OWP), the construction of hard-core bits can be iterated to $\hc_{\H^2}, \hc_{\H^3}, \ldots$.

\begin{proposition}[{\cite{HILL99}}]\label{prop:hcmult}
Let $\F:\X\to\X$ be a OWF (resp., OWP) with hard-core predicate $\hc_\F$. \!Then \!$\F^2$ \!is a OWF (resp., OWP) with hard-core predicate $\hc_{\F^2}$.
\end{proposition}

Another very important cryptographic assumption is the existence of {\em one-way trapdoor permutations (OWTP)}. A OWTP is a (family of) permutations which are easy to evaluate but hard to invert, {\em unless} an extra piece of secret information is known (the {\em trapdoor}) which is specific to a certain permutation.

For our scope, it is convenient to express a family of OWTPs as indexed through an {\em index family}, which is efficiently sampleable together with the related trapdoor. We will denote by $\I := \family{\I}$ and $\T := \family{\T}$ the index and trapdoor spaces, respectively. W.l.o.g., we assume $\I_\secpar \subseteq \bin^{\d(\secpar)}$, and $\T_\secpar \subseteq \bin^{\t(\secpar)}$ for security parameter $\secpar \in \NN$, where $\d$ and $\t$ are polynomial functions determined by the OWTP family.

\begin{definition}[One-Way Trapdoor Permutation Family (OWTP)]\label{def:owtp}
A (family of) {\em one-way trapdoor permutations (OWTP)} is a tuple of \PPT algorithms $\P := (\Gen,\Eval,\Invert)$:
\begin{enumerate}
\item $\Gen: \to \I \times \T$;
\item $\Eval: \I \times \X \to \X$;
\item $\Invert: \I \times \T \times \X \to \X \cup \set{\bot}$,
\end{enumerate}
and such that:
\begin{enumerate}
\item for any \PPT algorithm \A it holds:
$$
\Pr_{\substack{x\!\rand\!\!\X \\(i,t) \from \Gen}} \left[ \A(i,\Eval(i,x)) \to x \right] \leq \negl; \text{ and}
$$
\item $\Invert(i,t,y) = \Eval(i,x), \foral x \in \X , \foral (i,t) \from \Gen , \foral y \from \Eval(i,x)$.
\end{enumerate}
\end{definition}

The existence of OWTP is an assumption, like in the case of OWF. It is a stronger assumption, because the existence of OWTP in particular implies the existence of OWF, but the converse is not believed to hold.
 
\begin{proposition}[OWTP$\implies$OWP$\implies$OWF]\label{prop:OWTPtoOWF}
Let $\P := (\Gen,\Eval,\Invert)$ be a OWTP on \X. Then, for all but a negligible fraction of possible sequences $\left( i_\secpar,t_\secpar \right)_\secpar \from \Gen(\secpar) \Rightarrow \Eval(i_\secpar,.)$ is a OWP (and thus a OWF) on $\X = \family{\X}$.
\end{proposition}

Candidates OWTP can be constructed from some hard problems such as factorization, DLP, and many others. 
As an example, it is well known that if factoring large integers is hard, then one can build OWTP using, e.g., the {\em RSA cryptosystem~\cite{RSA}}.

\begin{theorem}[RSA $\implies$ OWTP]\label{thm:DLPOWTP}
If factorization of large integers is computationally hard, then OWTPs exist.
\end{theorem}

\subsection{The Random Oracle Model}\label{sec:rom}

In this section we briefly recall the {\em random oracle model (ROM)} methodology. The subject is quite involved and here we do not discuss it in detail, see~\cite{BellareROM} for an overview. A random oracle (RO) is an abstract mathematical model representing an idealized version of a publicly accessible source of randomness. In practice, a RO is used in security proofs to replace pseudorandom objects, such as hash functions, which would be otherwise too difficult to analyze. The idea is that such objects approximate very well the mathematical model described by the random oracle, so that a security proof given in the ROM is `almost as good' as a security proof given for the real-world implementation. However, it is important to keep in mind that there are cases of {\em ROM uninstantiability}~\cite{CGH98,BFM15}. That is, there exist (artificial) examples of cryptographic schemes which are provably secure in the ROM, but which become insecure whenever the random oracle is replaced by {\em any} hash function.

Formally, a random oracle from a bit string set \X to a bit string set \Y is a function $\O:\X\to\Y$ drawn uniformly at random from the set $\Y^\X$. The description of \O is not explicitly given; instead, \O can only be queried in a black-box way. At the beginning of the security analysis, the oracle is {\em initialized} by drawing a function uniformly at random from the set $\Y^\X$. The function so chosen remains unknown to all the parties involved in the protocol, but all those parties gain oracle access to it.

It is important to notice that, with high probability, a randomly chosen function from \X to \Y does not have a compact representation, so that the mere act of selecting a random function in $\Y^\X$ is not algorithmically defined. Because the security proofs we are interested in must be constructive and efficient, different approaches should be taken when constructing a random oracle. One possibility is {\em lazy sampling}: because the value distribution of a completely random function on a certain point $x$ is independent from the value the function takes on any other point, then the following procedure defines a random function, by adaptively filling a lookup table of values as soon as they are queried for the first time. In terms of pseudocode, a lazy sampling procedure would look as follows:

\vfill

\begin{algorithmic}[1]
\State set $\lookuptable = \emptyset$
\ForAll{query received on element $x$}
	\If{$(x,y') \in \lookuptable$ for some element $y'$}
		\State \textbf{Return:} $y'$
	\Else
		\State sample $y \rand \Y$
		\State set $\lookuptable := \lookuptable \cup \set{(x,y)}$
		\State \textbf{Return:} $y$
	\EndIf
\EndFor
\end{algorithmic}
Another possible method is to instantiate the RO with an efficiently computable {\em pseudorandom function family}, which will be described in Section~\ref{sec:PRF}.

Finally, it is important to mention that a RO can be {\em reprogrammed}, that is, the underlying function can be changed `on the fly' during the security proof. The intuition for this is that, since the RO replaces a hash function, the proof should still hold if we use a certain hash function instead of another one, as long as it does not have exploitable `structures' which are not supposed to be found (with high probability) on a completely random function.

\section{Quantum Computation}

In this section we recall the basic concepts of quantum information theory and quantum computation. We only give here a brief overview, and refer to~\cite{NC00} for a more detailed exposition.

\subsection{Quantum Mechanics}

In quantum mechanics, an isolated physical system (which we denote usually by an uppercase letter, e.g., $A$) is represented by a complex Hilbert space, denoted by $\Hilbert_A$ (or just $\Hilbert$ when the physical system is implied), of dimension suitable to represent all the independent possible physical states of $A$. Using the {\em bra-ket notation}, a completely defined state $\phi$ of the system (also called a {\em pure state}) is represented by (a class of) unitary vectors denoted by $\ket{\phi}$. A set of orthonormal generators for $\Hilbert$ is a {\em basis} for $\Hilbert$; a {\em computational basis} for $\Hilbert$ is a conventionally defined basis where elements are labeled as bit strings (or integers) $\set{\ket{x}: x = 1, \ldots, d}$, where $d := \dim{\Hilbert}$. Every pure state $\ket{\phi}$ can thus be written as:
$$
\ket{\phi} = \sum_x a_x \ket{x},
$$
with $\sum_x |a_x|^2 = 1$. The complex coefficients $a_x$ are the {\em amplitudes of $\ket{x}$}, and we say that $\ket{\phi}$ is a {\em quantum superposition} of states $\ket{x}$. Sometimes, if \X is a set, we use the notation $\Hilbert_\X$ to denote a complex Hilbert space for some physical system such that the computational basis for that space is labeled with elements of \X. That is, $\Hilbert_\X$ is the space generated by $\set{\ket{x}:x\in \X}$. For two pure quantum states $\ket{\phi}=\sum_x a_x \ket{x}$ and $\ket{\psi}= \sum_x b_x \ket{x}$ in superposition in the basis states $\ket{x}$, the Euclidean distance is given by $\big( \sum_x \left| a_x - b_x\right|^2 \big)^\frac{1}{2}$. 

We denote by $\bra{\psi}$ the {\em dual} of a state $\ket{\psi}$, i.e., $\bra{\psi} := \ket{\psi}^\dagger$. By Riesz's Representation Theorem, for every linear functional $\a: \Hilbert \to \CC$ there exists a unique $\ket{\alpha}$ such that $\a(\ket{\phi}) = \braket{\alpha|\phi}, \foral \phi \in \Hilbert$. Notice that, since pure states are represented by classes of unitary vectors, then $|\braket{\psi|\phi}|^2 \in [0,1]$, with $|\braket{\psi|\phi}|^2 = 1$ iff $\ket{\phi} = \ket{\psi}$, and $|\braket{\psi|\phi}|^2 = 0$ iff $\ket{\psi}$ and $\ket{\phi}$ are orthogonal. In particular, $\braket{x_i|x_j} = 0 \foral i \neq j$.

According to the laws of quantum mechanics, two different types of physically valid transformations can be applied to pure states:
\begin{itemize}
\item {\em reversible transformations,} or {\em evolutions}, which are modeled by unitary operators of the form $U:\Hilbert \to \Hilbert$; and
\item {\em measurements}, which allow an observer to extract information from the physical system.
\end{itemize}
In this work, for pure states we only consider {\em measurement in the computational basis}, which works in the following way: let $\ket{\phi} = \sum_x a_x \ket{x}$. Then, measuring such state yields a single real-valued outcome $x$ with probability $|a_x|^2$, and after such measurement the state {\em collapses} to the basis state $\ket{x}$.

The {\em composition} (joint system) of two physical systems $A$ and $B$ is represented by the {\em tensor product} of the respective Hilbert spaces, $\Hilbert_{AB} := \Hilbert_A \otimes \Hilbert_B$. So, for example, if $\set{\ket{x}}_{x \in \X}$ is a basis for $\Hilbert_A$ and $\set{\ket{y}}_{y \in \Y}$ is a basis for $\Hilbert_B$ (for two sets \X and \Y), then $\set{\ket{x} \otimes \ket{y}}_{(x,y) \in \X \times \Y}$ is a basis for $\Hilbert_{AB} = \Hilbert_{\X \times \Y}$. We write equivalently $\ket{x} \otimes \ket{y} = \ket{x}\ket{y} = \ket{x,y}$.

Two fundamental theorems in quantum information theory, which we only mention here informally, are the following.
\begin{theorem}[No-Cloning Theorem]
There does not exist any valid physical process which, given as input an arbitrary state $\ket{\phi}$, produces the state $\ket{\phi} \otimes \ket{\phi}$.
\end{theorem}
\begin{theorem}[No-Signaling Theorem]
There does not exist any valid physical process which allows two parties to transmit information faster-than-light, even though these parties are allowed to perform instantaneous physical action on remote and possibly entangled quantum systems.
\end{theorem}

\subsection{Entanglement}

Notice that not all states of $\Hilbert_{AB}$ are of the form $\ket{\phi} \otimes \ket{\psi}$ for some $\ket{\phi} \in \Hilbert_A$ and $\ket{\psi} \in \Hilbert_B$ - actually, very few of them are. For example, for $2$-dimensional Hilbert spaces $\Hilbert_A$ and $\Hilbert_B$, the following state:
\begin{equation}\label{eq:EPR}
\ket{\rho}_{AB} = \sqrt{\half} \ket{00} + \sqrt{\half} \ket{11}
\end{equation}
cannot be expressed as a tensor product of two pure states, even if it is a pure states itself. In fact, if we consider the state associated to the system $A$ alone (which we denote by $\ket{\rho}_{A}$) we find that it is impossible to write this state as a superposition of $\ket{x}$ elements, and the same applies to $\ket{\rho}_{B}$. When this happens we say that $\ket{\rho}_{A}$ and $\ket{\rho}_{B}$ are {\em entangled} states, otherwise we say that $\ket{\rho}_{AB}$ is {\em separable}. It turns out that, in a composite system, the vast majority of possible quantum states are entangled, and only a small subclass of them are separable. Entangled states cannot be pure states, so a different formalism is required to express them.

The {\em density matrix formalism} is used to represent all those states (including entangled states) which cannot be represented as pure states. We call such states {\em mixed states}, and we drop the bra-ket notation to represent them, in order to highlight the fact that they are not vectors, but matrices. Mixed states can be represented as probability distributions over sets of pure states. If a mixed state $\rho$ is defined as a distribution over elements $\ket{\phi_i}$, each of them occurring with probability $p_i$, then we define:
$$
\rho := \sum_i p_i \ketbra{\phi_i}.
$$
We call the resulting matrix representation of $\rho$ the {\em density matrix (or density operator) representation} of $\rho$. Formally, density matrices are operators $\rho:\Hilbert \to \Hilbert$ such that:
\begin{enumerate}
\item (trace condition) $\tr(\rho) = 1$
\item (positivity condition) $\braket{\phi|\rho|\phi} \geq 0 \foral \phi \in \Hilbert$.
\end{enumerate}
As a consequence, every density operators has diagonal elements in $[0,1]$. We denote the set of all admissible quantum states on a system $A$ (that is, the set of all positive, unitary-trace linear operators on $\Hilbert_A$) as $\states{\Hilbert_A}$.

All the formalism defined for pure states can be reformulated in terms of mixed states, because mixed states describe a statistics on pure states. 
If $\ket{\phi}$ is a pure state, its density matrix is defined just as $\ketbra{\phi}$. If $\rho \in \states{\Hilbert_A}$ and $\sigma \in \states{\Hilbert_B}$, then $\rho\otimes\sigma \in \states{\Hilbert_{AB}}$ is the state of the joint system. A unitary evolution $U$ applied to a mixed state $\rho$ produces another mixed state $U\rho U^\dagger$. Measuring a state $\rho$ in the computational basis yields outcome $x_i$ with probability $p_i$, where $p_i$ is the $i$-th diagonal element of $\rho$; in this case, the system is left in the state $\ketbra{x}$.

If we have two (or more) physical systems $A,B$, and they are jointly in the state $\rho_{AB}$, then the state describing the system $A$ (resp., $B$) alone is denoted by $\rho_A$ (resp, $\rho_B$), which has density operator:
$$
\rho_A := \tr_B(\rho_{AB}),
$$
where $\tr_B$ is the {\em partial trace over $B$}, defined by:
$$
\tr_B(\ket{x_1}\!\!\bra{x_2}_A \otimes \ket{y_1}\!\!\bra{y_2}_B) := \ket{x_1}\!\!\bra{x_2}_A \cdot \tr(\ket{y_1}\!\!\bra{y_2}_B).
$$
The act of taking a state in a joint system and considering only the state in one of its subsystems, `forgetting' about the rest of the system is called {\em tracing out (or, reducing)} to a certain subsystem. W.l.o.g this can be seen as: first measuring the state in the computational basis {\em only} on the subsystem to be `forgotten' (thereby collapsing part of the state and hence obtaining a separable state between the two systems), and then discarding the collapsed state and only consider the state of the subsystem left.

Any physically allowable process in nature, according to quantum mechanics, has to obey the constraints that density operators must be mapped to other density operators. That is, the mathematical transformation describing a physical process must preserve the unitarity of the trace, and the positivity of the operators. We call such `admissible transformations' {\em CPTP maps} (completely positive, trace-preserving maps), or {\em quantum channels}.

\subsection{Quantum Circuits}

The most widely used model for quantum computation is that of {\em quantum circuits}. A quantum circuit is the analogue of a Boolean circuit, with a few differences. For the purpose of this work, we consider the following:
\begin{itemize}
\item instead of acting on register of bits, a quantum circuits operates on {\em quantum registers}, which are physical systems composed of subsystems (called {\em qubits}) described by $2$-dimensional complex Hilbert spaces.
\item Instead of being composed of Boolean gates, quantum circuits are composed of {\em elementary quantum gates}, which are either measurement operators, or transformations on (some subsets of) qubits, described by unitary operators.
\end{itemize}
A quantum circuit takes as input a quantum register in a certain state and produces a quantum output, but we can always consider additional classical inputs and outputs (which can be `embedded' into quantum registers as basis states). The outcome of the quantum computation, however, is usually recovered through a measurement. It turns out that, w.l.o.g., measurements during a quantum computation can always be postponed to the very end of the quantum circuit, without changing the distribution of outcomes.

The number of input and output qubits of a quantum circuit can be different from each other. In fact, even if unitary operators act on the same subspace, a quantum circuit can have additional constant, `hidden input registers' (called {\em ancilla qubits}, usually initialized to $\ket{0}$), and can `delete' or `forget' some register (by tracing them out). However, {\em any} CPTP map can be modeled as a quantum circuit.

For the purpose of this work, we only consider measurement operators in the computational basis. If we have a single qubit in a state $\ket{\phi}$ and we apply a measurement on that qubit in order to obtain a single bit as outcome, we denote this as in Figure~\ref{fig:measure} (the double line denotes a classical output).
\begin{figure}[t]
\begin{center}
\includegraphics[width=0.2\textwidth]{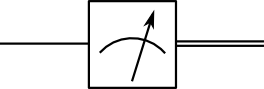}
\end{center}
\caption{Quantum measurement gate.}
\label{fig:measure}
\end{figure}

If $U$ is a single-qubit unitary acting on the $i$-th qubit of an $\secpar$-qubit system, it is denoted by $U_i$ as shown in Figure~\ref{fig:unitary}.
\begin{figure}[t]
\begin{center}
\includegraphics[width=0.2\textwidth]{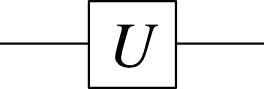}
\end{center}
\caption{Single-qubit unitary gate.}
\label{fig:unitary}
\end{figure}

The most basic single-qubit gate is the identity $\Id$:
$$
\Id :=  \left(\begin{array}{cc} 1&0\\0&1 \end{array}\right).
$$
A very important single-qubit gate is the {\em Hadamard gate}, denoted by $H$, and defined by the unitary matrix:
$$
H := \sqhalf \left(\begin{array}{cc} 1&1\\1&-1 \end{array}\right).
$$
Other useful single-qubit operators are the {\em Pauli matrices} $X,Y,Z$ defined by:
$$
\begin{array}{ccc}
X := \left(\begin{array}{cc} 0&1\\1&0 \end{array}\right) &
Y := \left(\begin{array}{cc} 0&-\im\\\im&0 \end{array}\right) &
Z := \left(\begin{array}{cc} 1&0\\0&-1 \end{array}\right)
\end{array}
$$
Notice how the Pauli matrices are also Hermitian. Moreover they satisfy: $XZ = \im Y$. We define the {\em Pauli group on $1$ qubit} $\mathfrak{P}_1$ as the matrix multiplicative subgroup generated by $\set{\im\Id, X, Y, Z}$. This extends to the {\em Pauli group on \secpar qubits} $\mathfrak{P}_\secpar$ as the subgroup generated by $\set{\im\Id_i, X_i, Y_i, Z_i : i = 1,\ldots,\secpar}$.

Finally, two very important $2$-qubit gates are the {\em controlled-NOT (CNOT)} and the {\em SWAP} gates:
$$
\begin{array}{cc}
\mathsf{CNOT} := \left(\begin{array}{cccc} 1&0&0&0\\0&1&0&0\\0&0&0&1\\0&0&1&0 \end{array}\right) &
\mathsf{SWAP} := \left(\begin{array}{cccc} 1&0&0&0\\0&0&1&0\\0&1&0&0\\0&0&0&1 \end{array}\right) 
\end{array}
$$

\begin{figure}[t]
\begin{center}
\includegraphics[width=0.2\textwidth]{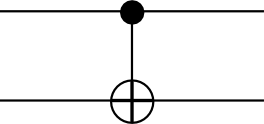}
\end{center}
\caption{CNOT gate.}
\label{fig:cnot}
\end{figure}

\begin{figure}[t]
\begin{center}
\includegraphics[width=0.6\textwidth]{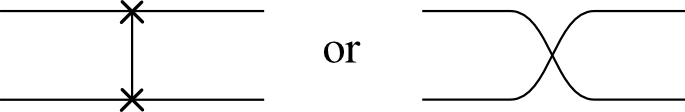}
\end{center}
\caption{SWAP gate.}
\label{fig:swap}
\end{figure}

A {\em quantum algorithm} is a uniform family of quantum circuits, i.e., there exists a (classical) Turing machine which, given the security parameter expressed in unary $\secparam$ as input, runs in time at most polynomial in $\secpa$, and outputs a (classical) description of the $\secpa$-th member of the quantum circuit family. \QPT stands for `quantum polynomial time', so a \QPT algorithm is a uniform family of quantum circuits of size polynomial in the security parameter.

As quantum algorithms are probabilistic by nature, there is no quantum analogue of the classical complexity security class \PPP. However, there is an analogue for \BPP: the complexity \BQP is the set of all languages \L for which there exists a \QPT algorithm \M and a positive constant $c$ such that:
	\begin{enumerate}
	\item $\foral x \in \L \implies \Pr [\M(x) \to 1] \geq \half + c$; and
	\item $\foral x \notin \L \implies \Pr [\M(x) \to 0] \geq \half + c$.
	\end{enumerate}
Informally, \BQP is the set of all problems which are `easy to solve with high probability' on a quantum computer.

`Famous' quantum algorithms include Shor's algorithm~\cite{Shor94} for factoring integers and solving DLP in polynomial time, Simon's algorithm~\cite{Simon97} for recognizing in polynomial time black-box functions of a certain form, and Grover's algorithm~\cite{Grover96} for polynomially speeding up search on unsorted databases, inversion of functions, and general brute-force attacks.

\subsection{Quantum Oracles}\label{sec:quantumoracles}

As in the classical case, the computational capabilities of a quantum algorithm \A can be expanded by giving to the algorithm access to an oracle \O, which we denote by $\A^\O$. The oracle can be classical (with the same meaning as in the classical case), or it can be quantum. In the latter case, we have to distinguish between:
\begin{itemize}
\item {\em (standard) quantum oracle access}. In this case the oracle is a unitary operation $U$ which \A can query on a quantum state $\rho$ at unit time cost in order to receive the response state $U \rho U^\dagger$. Whenever not specified, by `oracle access' we always mean the standard one.
\item {\em Quantum gate access}. In this case the oracle is also a unitary operator, like in the standard oracle access, the only difference is that \A automatically gains access to the inverse operator $U^\dagger$ as well.
\item {\em Quantum circuit access.} In this case the oracle is not necessarily unitary, but an arbitrary CPTP map. This means that the oracle could, e.g., perform measurements, or tracing out qubits, or act on additional quantum registers outside of \A's control.
\end{itemize}

We use the following technical tool in the proof of Theorem~\ref{thm:FSLambda}. Let \A be a quantum algorithm performing quantum queries to an oracle \O, and let $\q_x(\ket{\phi_j})$ be the magnitude squared of basis element $x$ in the $j$-th query, which we call the {\em query probability of $x$ in query $j$}. If we sum over all queries, we get an upper bound on the total query probability of $x$.

\begin{lemma}[{\cite[Theorem 3.3]{BBBV97}}]\label{lem:queryprob}
Let $\A$ be a quantum algorithm running in time $t$ with quantum oracle access to $\O:\X\to\Y$. Let $\epsilon > 0$ and let $S \subseteq \set{1,\ldots,t} \times \X$ be a set of time-string pairs such that $\sum_{(j,x)\in S} \q_x(\ket{\phi_j}) \leq \epsilon$. If we modify $\O$ into an oracle $\O'$ which answers each query $x$ at time $j$ by providing the same string $\bar{x}$ (which has been sampled independently from $\O$) whenever $(j,x) \in S$, then the Euclidean distance between the final states of $\A$ when invoking $\O$ and $\O'$ is at most $\sqrt{t \epsilon}$.
\end{lemma}

\subsection{Distinguishing Quantum States}

A crucial problem in quantum information theory is {\em distinguishing quantum states}. Because quantum states form a continuum, and because the only way we have to extract information from them is by performing measurements, distinguishing different quantum states with certainty is not always possible. In fact, for any practical purpose two quantum states are `the same state' 
if {\em there is no physically admissible process extracting measurement outcomes with different distributions from those states}. In other words, it is only possible to distinguish different quantum states if we can perform operations on them leading to measurement outcome distributions which are themselves distinguishable. Because in this work we only deal with computationally bounded processes, it is clear that the minimal requirement for two states to be distinguishable is that they are (or can be efficiently transformed to) states which yield computationally distinguishable outcome distributions when measured.

The following lemma from~\cite{BV97} upperbounds the statistical distance between the distributions of measurements on two quantum states in terms of their Euclidean distance.
\begin{lemma}[{\cite[Lemma 3.6]{BV97}}]\label{lem:distances}
Let $\ket{\phi},\ket{\psi}$ be pure quantum states with Euclidean distance at most $\epsilon$. Then, performing the same measurement on $\ket{\phi},\ket{\psi}$ yields distributions with statistical distance at most~$4\epsilon$.
\end{lemma}

For mixed states on isolated systems, the {\em trace distance} is a useful mathematical tool which gives directly an upper bound on the probability of distinguishing two states {\em for any physical process}.
\begin{definition}[Trace Distance]
Let $\rho,\sigma \in \states{\Hilbert}$. The {\em trace distance} between $\rho$ and $\sigma$ is defined by:
$$
\left\| \rho - \sigma \right\|_{\tr} := \half \sum_i |\lambda_i|,
$$
where $\lambda_i$ are the eigenvalues of $\rho - \sigma$.
\end{definition}

We call the {\em totally, or maximally mixed (or entangled) state} over a physical system $A$ the mixed state $\tau_A := \frac{\Id}{\dim{\Hilbert_A}}$; it has the property that a measurement over {\em any} possible orthonormal basis on this state always yields the uniform distribution of possible outcomes. This state represents somehow `a state of maximal uncertainty', and a common technique to show that no information can be extracted from a quantum state is to show that such state has `low' trace distance from the maximally mixed state.

However, when trying to distinguish between two CPTP maps, or two possible states on a {\em non-isolated system}, the trace distance is not enough. The reason is that, because of entanglement, two states which are {\em different} on a joint system $AB$ might yield {\em the same} reduced state on a subsystem $A$. In that case, the trace distance on $A$ would be $0$, but a distinguisher with access to $B$ might still be able to tell them apart. In these cases, the {\em diamond norm} is used, which induces a distance between CPTP maps.
\begin{definition}[Diamond Norm]
If $\Phi$ is a CPTP map (quantum channel) from operator spaces $\states{\Hilbert_A}$ to $\states{\Hilbert_B}$, then its {\em diamond norm} is defined by:
$$
\left\| \Phi \right\|_{\diamond} := \sup_{\rho \in \states{\Hilbert_{AK}}} \left\| \left( \Phi \otimes \Id_K \right)(\rho)\right\|_{\tr},
$$
where $\Hilbert_K$ is {\em any} Hilbert space such that $\dim{\Hilbert_K} \geq \dim{\Hilbert_A}$.
\end{definition}
It can be shown that an upper bound to the probability of distinguishing two quantum channels $\Phi$ and $\Psi$ is given by $\| \Phi - \Psi \|_\diamond$.

\chapter{\QS0: Classical Security}\label{chap:QS0}

The first class of cryptographic security notions that we are going to analyze encompasses the weakest notions in the quantum world. Namely: no quantum at all. In our new labeling system, the security class \QS0 refers to all the security notions and concepts which make {\em no mention} of quantum information theory. That is, \QS0 is just classical cryptography, in a sense `pre-(post-quantum)'. Studying this security class is essential in order to understand how the results change when we introduce quantum adversaries.

In this chapter we introduce security models for different classical cryptographic primitives, starting from the very basic ones (such as secret-key encryption schemes) to more elaborated ones. We also introduce other building blocks and transformations from one primitive to another.

A key feature of this part of our work is to perform all this analysis using a formalism which sometimes deviates from the one conventionally used in the existing literature, but which has the advantage of being easier to translate to the quantum world.

\subsection{My Scientific Contribution in this Chapter}

Most of the material in this chapter 
can be found in the existing literature (see for example~\cite{KatzLindell,Goldreich1,Goldreich2}, and 
is part of the preliminary technical results needed to understand the challenges arising when modeling security scenarios in a quantum world. However, to the best my knowledge, the proof of Theorem~\ref{thm:INDCPAnotoINDCCA1} has never been made explicit before. In fact, separation examples between CPA and CCA scenarios in the scientific literature usually refer either to the public-key scenario (where one can exploit group homomorphic properties) or to the separation between CPA and CCA2. Moreover, all the material from Section~\ref{sec:ORAM} first appeared in~\cite{GKK17}, which is a joint work with Nikolaos P. Karvelas and Stefan Katzenbeisser.

\section{Building Blocks}\label{sec:QS0buildblocks}

We start our analysis of classical cryptographic primitives by recalling some basic building blocks which we will use throughout the rest of this work. In what follows, $\X$ and $\Y$ are (sub)sets of binary strings. W.l.o.g., we assume that $\X = \family{\X} := \left( \bin^\secpar \right)_\secpar$. The {\em key space \K} instead, is identified with $\family{\K} \subseteq \bin^{\s(n)}$ for security parameter $\secpar \in \NN$, where $\s$ is a polynomial function determined by the scheme considered. W.l.o.g. we assume that, for security parameter \secpar, keys are of bit size $\secpa$.

\subsection{Pseudorandom Number Generators}\label{sec:PRNG}

A {\em pseudorandom number generator (PRNG)} is a \DPT stateful algorithm which outputs bit strings with a distribution computationally indistinguishable from the uniform distribution over some set. There is no secret key involved, but a secret internal state of the algorithm determines the value to be output next. As the algorithm is deterministic, the same internal state produces the same output value, so the state must be updated after every execution, according to a procedure specified by the algorithm itself. The initial value of the PRNG's state is called the {\em seed}.

Formally, we give a slightly different definition.

\begin{definition}[Pseudorandom Number Generator (PRNG)]\label{dfn:PRNG}
Let $\p$ be a polynomial such that $\p(\secpar) \geq \secpar+1, \forall \secpar \in \NN$. A {\em pseudorandom number generator (PRNG)} with expansion factor $\p$ is a \DPT algorithm $\PRNG$ such that:
\begin{enumerate}
\item given as input a bit string $s \in \bin^\secpar$, (the {\em seed}), outputs a bit string $\PRNG(s) \in \bin^{\p(\secpar)}$; and
\item for any \PPT algorithm $\D$:
$$
\left| \Pr \left[ \D(r) \to 1 \right] - \Pr \left[ \D(\PRNG(s)) \to 1 \right] \right| \leq \negl,
$$
where $r\rand\bin^{\p(\secpar)}, s\rand\bin^\secpar$, and the probabilities are taken over the choice of $r$ and $s$, and the randomness of $\D$.
\end{enumerate}
\end{definition}

However, it is possible to show that with the above definition one can actually also define a procedure to output a stream of polynomially many values of bit size polynomial in \secpar. The idea is to define a bit stream, where some of the $\p(\secpar)$ output bits are used to form the stream, and the others are used to generate a new, updated seed for the \PRNG. Therefore, one usually speaks of {\em PRNG with \secpa-bit output}. Analogously, it is easy to see that bits truncated by \PRNG's output are also pseudorandom.

Moreover, it is possible to prove that the condition of indistinguishability from random is equivalent to the condition of {\em non-predictability}, that is, no \PPT algorithm can reliably guess the next bit output by \PRNG, even by observing polynomially many bits output in the past. Clearly, if one could predict such bits then the pseudorandomness property would be violated. The other direction of the equivalence is known as Yao's Test~\cite{YaoPRNG}.

PRNGs have many useful applications. Given the existence of a PRNG, it is always possible (see, e.g.,~\cite{Goldreich1}) to build a OWF (by encoding the input to the OWF as a seed for the PRNG), but it is also possible to show the converse. Namely, given a OWF \F, one can define a PRNG which outputs a hard-core bit of \F, computed on the seed. This construction can be iterated producing a PRNG which we denote by $\PRNG_\F$, and which outputs polynomially many hard-core bits of $\F,\F^2,\F^3,\ldots$.

\begin{construction}[Goldreich-Levin PRNG~\cite{GL89}]\label{constr:goldreichlevin}
Let $\F:\X \to \Y$ be a OWF with hard-core predicate $\hc_\F$ . Define a stateful \DPT algorithm $\PRNG_\F:\X \to \X$ which, given as input an \secpar-bit seed $x \in \X$, outputs the \secpar-bit string:
$$
\hc_\F(x) \| \hc_{\F^2}(x) \| \ldots \| \hc_{\F^\secpar}(x).
$$
We call $\PRNG_\F$ the {\em Goldreich-Levin construction for OWF \F}.
\end{construction}

\begin{theorem}[{\cite{GL89}}]\label{thm:OWFtoPRNG}
Construction~\ref{constr:goldreichlevin} is a PRNG.
\end{theorem}
It must be noticed that the proof for the above theorem does not make any assumption on the adversary in terms of queries to the OWF. This fact will be important in the next chapter. It follows from Proposition~\ref{prop:hc} that a PRNG can be constructed by any OWF.

\begin{corollary}[OWF $\iff$ PRNG]\label{cor:OWFiffPRNG}
OWFs exist iff PRNGs exist.
\end{corollary}

\subsection{Pseudorandom Functions}\label{sec:PRF}

A (family of) {\em pseudorandom functions (PRF)} from \X to \Y with key space \K is a family of efficiently computable functions $\PRF: \K \times \X \to \Y$ which, without knowledge of the secret key $k \in \K$ indexing the particular member of the family, is computationally indistinguishable from the collection of all functions from \X to \Y (denoted by $\Y^\X$). We identify \PRF as a \DPT algorithm computing \PRF for a specific security parameter \secpar. As a shorthand notation, we write $\PRF_k:\X\to\Y$ meaning the member of the family indexed by $k \in \K$.

\begin{definition}[Pseudorandom Function (PRF)]\label{def:PRF}
A (family of) {\em pseudorandom functions (PRF)} from \X to \Y with key space \K is a \DPT algorithm $\PRF: (k \in \K_\secpar, x \in \X_\secpar) \mapsto y \in \Y_\secpar$ such that for any \PPT algorithm \D it holds:
$$
\left| \Pr_{k \rand \K} \left[ \D^{\PRF_k} \to 1 \right] - \Pr_{\h \rand \Y^\X} \left[ \D^{\O_\h} \to 1 \right] \right| \leq \negl,
$$
where $\O_\h$ is an oracle computing \h (i.e., a random oracle), and the probabilities are over the choice of $k$ and $\h$, and the randomness of $\D$.
\end{definition}

In security reductions, PRFs are usually modeled as random oracles. However, unlike PRNGs, their security depends on the secrecy of the key used, because any party with knowledge of such key can trivially distinguish the PRF from a completely random function.

PRFs, being indistinguishable from random functions, 
can be used as PRNGs (for a vast majority of the possible keys).

\begin{theorem}[PRF $\implies$ PRNG]\label{thm:PRFimPRNG}
If a PRFs exist, then PRNGs exist.
\end{theorem}

Still, one can show that PRFs can be built by using PRNGs, and therefore their existence is equivalent to the existence of OWFs.

\begin{theorem}[{\cite{GGM84}}]\label{thm:PRNGimPRF}
If a PRNGs exist, then PRFs exist.
\end{theorem}

However, unlike in the case of Theorem~\ref{thm:OWFtoPRNG}, the proof does make assumptions on the query capabilities of the adversary.

\begin{corollary}\label{cor:OWFiffPRF}
OWF exist iff PRF exist.
\end{corollary}

\subsection{Pseudorandom Permutations}

{\em Pseudorandom permutations (PRP)} are just PRFs which also happen to be (invertible) permutations on some space \X, for any choice of key. That is, a PRP \P is a family of permutations (and their inverses) which is computationally indistinguishable from the family $S(\X)$ of all the permutations on \X. As in the PRF case, we identify a PRP \PRP with the \DPT algorithm evaluating it, and as a shorthand notation, we write $\PRP_k:\X\to\X$ meaning the member of the (circuit or function) family indexed by $k \in \K$.

We start by defining a {\em weak} PRP, that is, indistinguishable from random to any adversary who does not have oracle access to the inverse permutation.

\begin{definition}[Weak Pseudorandom Permutation (WPRP)]\label{def:WPRP}
A (family of) {\em weak pseudorandom permutations (WPRP)} on \X with key space \K is a pair of \DPT algorithm $(\PRP,\PRP^{-1}): (k \in \K, x \in \X) \mapsto x' \in \X$ such that:
\begin{enumerate}
\item $\forall k \in \K \implies \PRP_k, \PRP^{-1}_k$ are permutations on $\X$;
\item $\forall k \in \K \implies (\PRP_k)^{-1} = \PRP^{-1}_k$; and
\item for any \PPT algorithm \D it holds:
$$
\left| \Pr_{k \rand \K} \left[ \D^{\PRP_k} \to 1 \right] - \Pr_{\p \rand S(\X)} \left[ \D^{\O_\p} \to 1 \right] \right| \leq \negl,
$$
where $\O_\p$ is an oracle for \p, and the probabilities are over the choice of $k$ and $\p$, and the randomness of $\D$.
\end{enumerate}
\end{definition}

In many applications though, we also need the possibility of inverting the permutations, hence we we also require the existence of another \DPT algorithm $\PRP^{-1}$ computing the inverse permutation. A PRP is called {\em strong} if it maintains pseudorandomness also in this setting.

\begin{definition}[Strong Pseudorandom Permutation (SPRP)]\label{def:SPRP}
A (family of) {\em strong pseudorandom permutations (SPRP)} on \X with key space \K is a pair of \DPT algorithms $(\PRP,\PRP^{-1}): (k \in \K, x \in \X) \mapsto x' \in \X$ such that:
\begin{enumerate}
\item $\forall k \in \K \implies \PRP_k, \PRP^{-1}_k$ are permutations on $\X$;
\item $\forall k \in \K \implies (\PRP_k)^{-1} = \PRP^{-1}_k$; and
\item for any \PPT algorithm \D it holds:
$$
\left| \Pr_{k \rand \K} \left[ \D^{\PRP_k,\PRP^{-1}_k} \to 1 \right] - \Pr_{\p \rand S(\X)} \left[ \D^{\O_\p,\O_{\p^{-1}}} \to 1 \right] \right| \leq \negl,
$$
where $\O_\p$ is an oracle for \p, $\O_{\p^{-1}}$ is an oracle for $\p^{-1}$, and the probabilities are over the choice of $k$ and $\p$, and the randomness of $\D$.
\end{enumerate}
\end{definition}

When left unspecified, by `PRP' we mean the strong version. A PRP is clearly also a PRF, but not necessarily the other way around. However, there exist constructions of PRPs from PRFs, such as the {\em Feistel} construction. Therefore, the existence of PRPs is also equivalent to the existence of OWFs.

\begin{theorem}[PRF $\iff$ PRP]\label{thm:PRFiffPRP}
PRFs exist iff PRPs exist.
\end{theorem}

\section{Secret-Key Encryption Schemes}\label{sec:SKES}

A very fundamental object in cryptography is {\em secret-key (or, symmetric-key) encryption schemes (SKES)}. In what follows, $\X$ and $\Y$ represent the {\em plaintext and ciphertext message spaces} respectively, 
while \K is the {\em key space}.

\begin{definition}[Secret-Key Encryption Scheme (SKES)]\label{def:skes}
A {\em secret-key encryption scheme (SKES)} with plaintext space \X, ciphertext space \Y, and key space \K is a tuple of \PPT algorithms $\E := \E_{\K,\X,\Y} := (\KGen,\Enc,\Dec)$:
\begin{enumerate}
\item $\KGen: \to \K$;
\item $\Enc: \K \times \X \to \Y$;
\item $\Dec: \K \times \Y \to \X \cup \set{\bot}$;
\end{enumerate}
such that $\forall\ n \in \NN, x \in \X, k \from \KGen \implies \Dec(k, \Enc(k,x)) = x$.
\end{definition} 

\vfill

Notice the following:

\begin{itemize}
\item $\KGen$ only gets as input a security parameter, but $\Enc,\Dec$ also need as input the correct security parameter related to the second input (the secret key) they receive. In order to lighten notation and w.l.o.g. we just assume that $\secpar$ is also appended to every $k \from \KGen$, so that every key also implicitly contains the security parameter.
\item Strictly speaking, it is not necessary to define $\bot$ as a possible output for $\Dec$. However, this is useful when defining schemes which can also {\em reject} certain ciphertexts (such as {\em CCA2 secure encryption schemes}).
\item As a shorthand notation, we will write $\Enc_k$ meaning the $\Enc$ algorithm with $k \in \K$ fixed as a first input; analogously for $\Dec_k$.
\item $\KGen$ is always assumed to be a nondeterministic algorithm, otherwise the encryption scheme would be trivial.
\item $\Enc$ can be a probabilistic algorithm, so it is certainly possible that two different executions of $\Enc_k(x)$ for fixed $k$ and $x$ yield two different ciphertexts. However, those ciphertexts would still decrypt to the same $x$ through $\Dec_k$.
\item As an immediate consequence of the previous point, it is clear that, for a given $k \in \K$, the image sets of different plaintexts are disjoint. That is: $x\neq x' \implies \Supp{\Enc_k(x)} \cap \Supp{\Enc_k(x')} = \emptyset$.
\item The behavior of $\Dec_k$ is unspecified (and dependent on the SKES considered) if given as input an element of $\Y$ which is not a valid encryption, i.e., of the form $\Enc_k(x)$ for some $x \in \X$.
\item If $\Enc_k$ is nondeterministic for all $k \in \K$, then we say that $\E$ is {\em randomized}, otherwise we say that $\E$ is {\em deterministic}.
\end{itemize}

Finally, notice that Definition~\ref{def:skes} does not say anything about the {\em security} of a SKES. We will study this aspect in the next section. In particular, for a SKES to be considered `secure', the size of $\Supp{\KGen(\secpar)}$ must be superpolynomial in $\secpar$. One of the most basic examples of SKES is the well known {\em one-time pad (OTP)}.

\begin{construction}[One-Time Pad (OTP)]\label{constr:otp}
Let $\X = \K = \Y = \bin^\secpar$. Define the {\em one-time pad (OTP) on \secpar bits} $\E = \E_{\K,\X,\Y}:=(\KGen,\Enc,\Dec)$ as the SKES with key space $\K $, plaintext space \X, and ciphertext space $\Y$, defined as follows:
\begin{enumerate}
\item $\KGen \to k$, with  $k \rand \K$;
\item $\Enc_k(x) := x \xor k$;
\item $\Dec_k(y) := y \xor k$.
\end{enumerate}
\end{construction}

It is well known~\cite{Shannon01} that the OTP is {\em information-theoretically secure}, as long as the key is completely random and only used once.

\subsection{Semantic Security}

In order to analyze the security of a SKES, we first have to define what it means for a SKES to be `secure'. That is, we have to define a `meaning', i.e., a {\em semantics} of the term `security'. Intuitively, we want to formalize the fact that no reasonable adversary should be able, given a ciphertext, to find out any `interesting' information about the underlying plaintext. There are three aspects to consider here.

First of all, we should define what a `reasonable' adversary is. In our case we will consider computationally bounded adversaries, that is, adversaries as \PPT algorithms, because we consider computational security. However, adversaries could be given additional power in the form of oracles. We will see a few examples in the next sections, while in this part we will start with the basic scenario (without oracles).

Secondly, we should define what constitutes `interesting information' about the underlying plaintext. We do not consider `interesting' all that information which is already publicly available, leaked, or manifest. For example, the length (bit size) of the plaintext is usually identifiable by only looking at the length of the ciphertext. Moreover, if some information about the plaintext is known a priori, e.g.: `the message starts with a vowel', we do not consider an adversary succesful if he is only able to tell that the message starts with a vowel, because that fact is already known. We want security to protect the encryption scheme only against those adversaries who can extract `interesting' information from the ciphertexts.

Finally, we should define the `winning conditions' for our adversaries, so that we can define our schemes `secure' if they prevent the adversaries from reaching those conditions. In theory, we could define a scheme to be `secure' if every adversary fails consistently in his goals, regardless of the choice of keys and plaintexts he intends to attack. However, this is not reasonable to expect, for three reasons:

\begin{itemize}
\item the choice of some particular key might influence the adversary's winning probability. For example, what if the message is encrypted with a key that the adversary happens to know as well?
\item The choice of the plaintexts is important as well. On one hand, we need the scheme to be secure even in the worst case scenario (that is, the best case scenario from the adversary's perspective.) On the other hand we cannot leave arbitrary freedom to the adversary in choosing the underlying plaintext - otherwise he could just break the encryption of a message he already knows, but that would not be `interesting' information.
\item The adversary might just get lucky. For example, when trying to decrypt a single bit of the message, he might just guess randomly, and still be succesful $50\%$ of the times.
\end{itemize}

In literature, {\em semantic security} is the well-established golden standard in defining the security of an encryption scheme. Semantic security is a simulation-based security notion, where the success probability of an adversary trying to guess meaningful information about a ciphertext is compared to that of a {\em simulator}, which has the same goal as the adversary but is not allowed to see the ciphertext at all. The probability is taken over the internal randomness of the algorithms (and, hence, over all the keys), and `interesting' and `non-interesting' information is defined in terms of a {\em target function $\f$} and an {\em auxiliary information function $\h$}, respectively (these are functions of the possible plaintexts.) The goal of the adversary/simulator is to guess $\f(x)$ when having access to $\h(x)$, for a certain plaintext $x$ drawn from a chosen distribution. The scheme is considered secure if the adversary and the simulator have roughly the same probability of guessing $\f(x)$. 

There are many, different but equivalent ways to define semantic security for SKES. In this work, we follow the approach from~\cite{Goldreich2}.

\begin{definition}[SEM Adversary, SEM Simulator]\label{def:SEMadvsim}
Let $\E := \E_{\K,\X,\Y}$ 
 be a SKES, and $\f,\h:\bin^* \to \bin^*$ two functions efficiently computable and polynomially bounded in the input bit size. A {\em SEM adversary \A for \E} is a \PPT algorithm $\A: \Y \times \Supp{\h} \to \Supp{\f}$. A {\em SEM simulator \S for \E} is a \PPT algorithm $\S: \Supp{\h} \to \Supp{\f}$.
\end{definition}

Notice that, w.l.o.g., we can assume that $\h(x)$ always includes the bit size of the plaintext $x$. We assume that \h and \f are efficiently computable, but actually, as shown in~\cite{Goldreich2}, this is redundant.

\begin{experiment}[$\gameSEMA$]\label{expt:SEMA}
Let $\E$ be a SKES, and $\A$ a SEM adversary. The {\em SEM experiment} proceeds as follows:
\begin{algorithmic}[1]
\State \textbf{Input:} $n \in \NN$, $\f,\h:\bin^* \to \bin^*$ efficiently computable and polynomially bounded in the input bit size, $\M := \family{\M}$, where $\M_n$ are probability distributions over $\X_\secpar$ with $\card{\M_n} = \poly(n)$
\State $k \from \KGen$
\State $m \from \M_n$
\State $c \from \Enc_k(m)$ \Comment{this is called `SEM challenge query'}
\State $f \from \A(c,\h(m))$
\If{$f = \f(m)$}
	\State \textbf{Output:} $1$
\Else
	\State \textbf{Output:} $0$
\EndIf
\end{algorithmic}
\end{experiment}

\vfill

\begin{experiment}[$\gameSEMS$]\label{expt:SEMS}
Let $\E$ be a SKES, and $\S$ a SEM simulator. The {\em simulated SEM experiment} proceeds as follows:
\begin{algorithmic}[1]
\State \textbf{Input:} $n \in \NN$, $\f,\h:\bin^* \to \bin^*$ efficiently computable and polynomially bounded in the input bit size, $\M := \family{\M}$, where $\M_n$ are probability distributions over $\X_\secparam$ with $\card{\M_n} = \poly(n)$
\State $k \from \KGen$
\State $m \from \M_n$
\State $f \from \S(\h(m))$
\If{$f = \f(m)$}
	\State \textbf{Output:} $1$
\Else
	\State \textbf{Output:} $0$
\EndIf
\end{algorithmic}
\end{experiment}

\begin{definition}[Semantic Security (SEM)]\label{def:SEM}
A SKES $\E$ is {\em semantically secure (SEM)} iff, for any SEM adversary \A there exists a SEM simulator \S such that, for every efficiently computable $\f,\h:\bin^* \to \bin^*$ polynomially bounded in the input bit size, for every probability ensemble $\M := \family{\M}$, where $\M_n$ are probability distributions over $\X_\secpar$ with $\card{\M_n} = \poly(n)$, it holds:
$$
\left| \Pr \left[ \gameSEMA(\M,\f,\h) \to 1\right] - \Pr \left[ \gameSEMS(\M,\f,\h) \to 1\right] \right| \leq \negl,
$$
where the probabilities are taken over the randomness of $\A,\E,\M,\S$.
\end{definition}

Intuitively, the notion of SEM tells us the following: any information about the plaintext the adversary could guess from the ciphertext, could also be guessed by only looking at publicly available information. That means, the ciphertext does not leak any meaningful information about the plaintext. This security notion captures in a very complete way what we want from an encryption scheme, but it has the drawback of being quite involved formally, and cumbersome to use in security proofs. Because of this, different notions of security are often used, which are equivalent to SEM but easier to formalize.

\subsection{Ciphertext Indistinguishability}

Another notion of security for encryption schemes is {\em indistinguishability of ciphertexts (IND)}. Unlike SEM, this notion is game-based instead of simulation-based: there is no simulator at all, and security requires that no reasonable adversary can win a certain security game with probability substantially better than merely guessing. The IND security game consists in distinguishing the encryption of two different plaintexts (chosen by the adversary). Although, unlike in the case of SEM, it is unclear at a first glance that IND captures in a complete way exactly what we require from a `secure' encryption scheme, we will see that the two notions are actually equivalent.

As in SEM, we model IND adversaries as \PPT algorithms, as we are interested in computational security. However, in the IND game it is usually convenient to separate the adversary in two {\em stages}, each one with a specific function. The first stage, the {\em message generator \M}, chooses two messages from the plaintext space -- the idea being that, in order to achieve the strongest security notion, the adversary is allowed to choose the most favourable scenario when playing this game. Then, one of these two messages is selected at random and encrypted with a key unknown to the adversary. Finally, the second stage of the adversary, the {\em distinguisher \D}, receives the resulting ciphertext, and his goal is to guess which one of the two plaintexts was encrypted. Formally, the adversary outputs a bit, and he wins the game if that bit is equal to the secret bit used to select one of the two plaintexts.

More formally, we define an {\em IND adversary} as follows.

\begin{definition}[IND Adversary]\label{def:INDadv}
Let $\E$ be a SKES. An {\em IND adversary \A for \E} is a pair of \PPT algorithms $\A := (\M,\D)$, where:
\begin{enumerate}
\item $\M: \to \X \times \X \times \words$ is the {\em IND message generator};
\item $\D: \Y \times \words \to \bin$ is the {\em IND distinguisher}
\end{enumerate}
\end{definition}

The security experiment related to the IND notion is as follows.

\begin{experiment}[$\gameIND$]\label{expt:IND}
Let $\E$ be a SKES, and $\A:=(\M,\D)$ an IND adversary. The {\em IND experiment} proceeds as follows:
\begin{algorithmic}[1]
\State \textbf{Input:} $\secpar \in \NN$
\State $k \from \KGen$
\State $(m^0,m^1,\state) \from \M$
\State $b \rand\bin$
\State $c \from \Enc_k(m^b)$ \Comment{this is called `IND challenge query'}
\State $b' \from \D(c,\state)$
\If{$b = b'$}
	\State \textbf{Output:} $1$
\Else
	\State \textbf{Output:} $0$
\EndIf
\end{algorithmic}
The {\em advantage of \A} is defined as:
$$
\advIND := \Pr \left[ \gameIND \to 1 \right] - \half .
$$
\end{experiment}

Notice the following:

\begin{itemize}
\item \D and \M are part of the same `entity' (the IND adversary \A), so that they should be allowed to exchange information. In particular, \D should know which are the two original messages generated by \M. In the security game, this is modeled by exchanging a state string $\state$ from \M to \D (obviously this string has bit size at most polynomial in the security parameter since \M is \PPT.)
\item There is no need to impose the condition that the two plaintexts generated by \M must be distinct, as the security notion requires that {\em all} adversaries (including those who choose distinct messages) fail at winning the game.
\item Since there are only two messages to choose from, the adversary can always win with $50\%$ probability by guessing randomly. Therefore, the advantage of the adversary is measured in terms of doing better than merely guessing.
\item The probability is over $b$ and the internal randomness of \A and \KGen.
\end{itemize}

\begin{definition}[Indistinguishability of Ciphertexts (IND)]\label{def:IND}
A SKES $\E$ has {\em indistinguishable encryptions (or, it is IND secure)} iff, for any IND adversary $\A$ it holds that: $\advIND \leq \negl$.
\end{definition}

The advantage of the IND notion is that, being game-based, it is easier to use in cryptographic reductions. At the same time, one can show that it is equivalent to IND.

\begin{theorem}[{\cite{Goldreich2}}]\label{thm:INDiffSEM}
A SKES is IND secure iff it is SEM secure.
\end{theorem}

Moreover, it has to be mentioned that the choice of defining the IND game in terms of two different messages is not compulsory: there are alternative definitions of the game where \M only generates a message, and the other is either chosen randomly or set to $0$, or where \M generates polynomially many messages, and one of them is selected for the encryption. All these notions turn out to be equivalent, with small modifications.

An example of (unconditionally) IND secure SKES is the OTP.

The notions of IND can be augmented, i.e., made {\em stronger}, by granting extra power to the IND adversary in the form of {\em oracles}. Since the adversary acquires additional computational power in so doing, it might be the case that IND secure schemes now become insecure because of this extra power. Therefore, the resulting security notions are (potentially) stronger, and encryption schemes which are resistant against the new, augmented adversaries are automatically resistant to the weaker adversaries as well. The more power is given to the adversaries, the potentially stronger the security notion.

Traditionally, oracles have been used to model attack scenarios not covered by the IND notion alone. Of course, one could simply give the adversary unlimited access to a decryption oracle and make him super powerful. But that would make the security notion so strong to be unachievable -- after all, SKES are not meant to protect by adversaries in possession of the secret key. Instead, other scenarios are considered.

\subsection{Chosen Plaintext Attacks}

In the {\em chosen plaintext attack (CPA)} scenario, the adversary is able to see encryptions of additional messages, in addition to the ones used in the IND game. He is allowed to choose the plaintexts to be encrypted by querying the encryption oracle $\Enc_k$ during the execution of the IND game. Moreover, he can perform the oracle queries in an {\em adaptive} way, i.e., reacting adaptively to the oracle's answers, for a polynomial number of queries, both before and after the IND challenge query. The resulting security game is as follows.

\begin{experiment}[$\gameINDCPA$]\label{expt:INDCPA}
Let $\E$ be a SKES, and $\A:=(\M,\D)$ an IND adversary. The {\em IND-CPA experiment} proceeds as follows:
\begin{algorithmic}[1]
\State \textbf{Input:} $\secpar \in \NN$
\State $k \from \KGen$
\State $(m^0,m^1,\state) \from \M^{\Enc_k}$
\State $b \rand\bin$
\State $c \from \Enc_k(m^b)$
\State $b' \from \D^{\Enc_k}(c,\state)$
\If{$b = b'$}
	\State \textbf{Output:} $1$
\Else
	\State \textbf{Output:} $0$
\EndIf
\end{algorithmic}
The {\em advantage of \A} is defined as:
$$
\advINDCPA := \Pr \left[ \gameINDCPA \to 1 \right] - \half .
$$
\end{experiment}

\begin{definition}[Indistinguishability of Ciphertexts under Chosen Plaintext Attack (IND-CPA)]\label{def:INDCPA}
A SKES $\E$ has {\em indistinguishable encryptions under chosen plaintext attack (or, it is IND-CPA secure)} iff, for any IND adversary $\A$ it holds that: $\advINDCPA \leq \negl$.
\end{definition}

As discussed above, IND-CPA is clearly at least as strong as IND.

\begin{theorem}[IND-CPA $\implies$ IND]\label{thm:INDCPAtoIND}
If a SKES is IND-CPA secure, then it is also IND secure.
\end{theorem}

But the converse is not true. In particular, all the encryption schemes that are not randomized cannot be IND-CPA secure, because then the adversary could always win the security game by first encrypting two messages of his choice, then performing the IND challenge, and then compare the resulting ciphertext with the encryption previously obtained. As an example, the OTP is not IND-CPA secure, despite being IND secure.

\begin{theorem}[IND \nimplies IND-CPA]\label{thm:INDnotoINDCPA}
There exist SKES which are IND secure, but not IND-CPA secure.
\end{theorem}

IND-CPA secure SKES can be constructed in a block-box way using PRFs.

\begin{construction}[{\cite[Construction 5.3.9]{Goldreich2}}]\label{constr:goldreich}
Let $\PRF:\X\to\Y$ be a PRF with key space \K. Define $\E = \E_{\K,\Y,\Y\times\X}:=(\KGen,\Enc,\Dec)$ as a SKES with key space $\K $, plaintext space \Y, and ciphertext space $\Y\times\X$, as follows:
\begin{enumerate}
\item $\KGen \to k$, with  $k \rand \K$;
\item $\Enc_k(x) \to (y,r)$, with $y := x \xor \PRF_k(r)$, where $r \rand \X$;
\item $\Dec_\sk(y,r) := y \xor \PRF_k(r)$.
\end{enumerate}
\end{construction}

\begin{theorem}\label{thm:GoldreichINDCPA}
Construction~\ref{constr:goldreich} is an IND-CPA SKES.
\end{theorem}
\begin{proof}[Proof (sketch)]
The one-time pad is perfectly (statistically) secure if used with random, independent keys. This means that the only way to break the security of \E is to break the security of \PRF. Since a fresh randomness $r$ is chosen for every encryption, and since the image $\PRF_k(r)$ can be recovered by the related plaintext/ciphertext pairs, giving oracle access to $\Enc_k$ for the adversary is equivalent to giving oracle access to $\PRF_k$. However, by Definition~\ref{def:PRF}, this is indistinguishable from a random oracle for any \PPT adversary, so that the security of the one-time pad carries over, although only computationally.
\end{proof}

Then, recalling Corollary~\ref{cor:OWFiffPRNG} and Theorem~\ref{thm:PRNGimPRF}, we can state the following.

\begin{corollary}[IND-CPA SKES from OWF]\label{cor:INDCPAfromOWF}
If OWFs exist, then IND-CPA SKES exist.
\end{corollary}

\subsection{Non-Adaptive Chosen Ciphertext Attacks}

In the {\em non-adaptive chosen ciphertext attack (CCA1)} scenario, in addition to the IND-CPA capabilities, the adversary is able to also see decryptions of certain ciphertexts. As in the CPA case, he is allowed to choose the ciphertexts to be decrypted by querying the decryption oracle $\Dec_k$ during the execution of the IND game. However, unlike in the CPA case, he is only able to interact with this oracle {\em before} the IND challenge query, and not afterward. The adversary is allowed to perform the decryption oracle queries in an adaptive way, for a polynomial number of queries, but only before the IND challenge query, hence the term `non-adaptive'\footnote{Admittedly, this well-established term in the scientific literature is somewhat misleading, because this `non-adaptivity' refers to `in respect to the challenge ciphertext', while the queries to the decryption oracle can actually be performed adaptively.}. Notice, in fact, that if the adversary were able to perform arbitrary decryption queries {\em after} the challenge query as well, this would allow him to decrypt the challenge ciphertext, and therefore it would render the security notion unachievable.

The resulting security game for the CCA1 scenario is as follows. 

\begin{experiment}[$\gameINDCCA$]\label{expt:INDCCA1}
Let $\E$ be a SKES, and $\A:=(\M,\D)$ an IND adversary. The {\em IND-CCA1 experiment} proceeds as follows:
\begin{algorithmic}[1]
\State \textbf{Input:} $\secpar \in \NN$
\State $k \from \KGen$
\State $(m^0,m^1,\state) \from \M^{\Enc_k,\Dec_k}$
\State $b \rand\bin$
\State $c \from \Enc_k(m^b)$ 
\State $b' \from \D^{\Enc_k}(c,\state)$
\If{$b = b'$}
	\State \textbf{Output:} $1$
\Else
	\State \textbf{Output:} $0$
\EndIf
\end{algorithmic}
The {\em advantage of \A} is defined as:
$$
\advINDCCA := \Pr \left[ \gameINDCCA \to 1 \right] - \half .
$$
\end{experiment}

\begin{definition}[Indistinguishability of Ciphertexts under Non-Adaptive Chosen Ciphertext Attack (IND-CCA1)]\label{def:INDCCA1}
A SKES $\E$ has {\em indistinguishable encryptions under non-adaptive chosen ciphertext attack (or, it is IND-CCA1 secure)} iff, for any IND adversary $\A$ it holds that: $\advINDCCA \leq \negl$.
\end{definition}

IND-CCA1 is clearly at least as strong as IND-CPA.

\begin{theorem}[IND-CCA1 $\implies$ IND-CPA]\label{thm:INDCCA1toINDCPA}
If a SKES is IND-CCA1 secure, then it is also IND-CPA secure.
\end{theorem}

But the converse is not true. There are IND-CPA secure SKES where, being able to decrypt different but related ciphertexts, can leak information about the secret key used.

\begin{theorem}[IND-CPA \nimplies IND-CCA1]\label{thm:INDCPAnotoINDCCA1}
There exists a SKES which is IND-CPA secure, but not IND-CCA1 secure.
\end{theorem}
\begin{proof}[Proof (sketch)]
Consider a SKES $\E' = (\KGen',\Enc',\Dec')$ obtained by modifying another, IND-CPA secure SKES $\E = (\KGen,\Enc,\Dec)$ as follows:
\begin{enumerate}
\item $\KGen' \to (k,\overline{m})$,\\where $k \from \KGen$, and $\overline{m}$ is a special message, unknown to the adversary;
\item $\Enc'_k(m) \to \begin{cases}\left(\Enc_k(m),\Enc_k(\overline{m})\right) &\text{ if } m \neq \overline{m},\\ \left(\Enc_k(\overline{m}),k\right) &\text{ otherwise;}\end{cases}$
\item $\Dec'_k(y,z) = \Dec_k(y)$.
\end{enumerate}
The new SKES $\E'$ is still IND-CPA secure, because the probability for any adversary of guessing the plaintext $\overline{m}$ is negligible. However, in the CCA1 scenario it is trivial to break such modified scheme, by first performing a CPA query to obtain a valid ciphertext, then performing a CCA1 decryption query on the ciphertext obtained by swapping the two ciphertext halves, therefore recovering $\overline{m}$, and then performing another CPA query on $\overline{m}$, hence recovering the secret key.
\end{proof}

However, Construction~\ref{constr:goldreich} is also IND-CCA1.

\begin{theorem}\label{thm:GoldreichINDCCA1}
Let $\E$ be the SKES from Construction~\ref{constr:goldreich}. Then $\E$ is an IND-CCA1 SKES.
\end{theorem}
\begin{proof}[Proof (sketch)]
Being able to perform decryption queries (before the challenge phase) gives to the adversary the possibility to forge new ciphertexts different (but related in a known way) to some other ciphertext of his choice. However, before the challenge phase, this does not provide any extra power, except the possibility of performing (polynomially many) extra queries to the PRF.
\end{proof}

Then, recalling Corollary~\ref{cor:OWFiffPRNG} and Theorem~\ref{thm:PRNGimPRF}, we can state the following.

\begin{corollary}[IND-CCA1 SKES from OWF]\label{cor:INDCCA1fromOWF}
If OWFs exist, then IND-CCA1 SKES exist.
\end{corollary}

\subsection{Adaptive Chosen Ciphertext Attacks}

Finally, in the {\em adaptive chosen ciphertext attack} scenario, in addition to the IND-CCA1 capabilities, the adversary is able to query the decryption oracle also after the challenge query, {\em with an important exception}: he is not allowed to query $\Dec_k$ on the challenge ciphertext received. This restriction is necessary, as we have already discussed in the CCA1 case, otherwise the adversary could simply decrypt the challenge ciphertext and trivially win the game, and this would make the security notion unachievable. Formally, we have therefore to define a `modified' decryption oracle, which is able to {\em reject} certain `forbidden' decryption queries (those trying to decrypt the challenge ciphertext), by replying with a special symbol $\bot$ to those queries.

\begin{definition}[CCA2 Oracle]\label{def:CCA2oracle}
Let $\E := (\KGen,\Enc,\Dec)$ be a SKES, and $c \in \Supp{\Enc}$. The {\em CCA2 decryption oracle rejecting $c$} is defined by:
$$
\Dec^c_k (c') \longrightarrow
\begin{cases}
\Dec_k (c') &\text{ if } c' \neq c,\\
\bot &\text{ otherwise.}
\end{cases}
$$
\end{definition}

The new security game is defined as follows. 

\begin{experiment}[$\gameINDCCAA$]\label{expt:INDCCA2}
Let $\E$ be a SKES, and $\A:=(\M,\D)$ an IND adversary. The {\em IND-CCA2 experiment} proceeds as follows:
\begin{algorithmic}[1]
\State \textbf{Input:} $n \in \NN$
\State $k \from \KGen$
\State $(m^0,m^1,\state) \from \M^{\Enc_k,\Dec_k}$
\State $b \rand\bin$
\State $c \from \Enc_k(m^b)$
\State $b' \from \D^{\Enc_k,\Dec_k^c}(c,\state)$
\If{$b = b'$}
	\State \textbf{Output:} $1$
\Else
	\State \textbf{Output:} $0$
\EndIf
\end{algorithmic}
The {\em advantage of \A} is defined as:
$$
\advINDCCAA := \Pr \left[ \gameINDCCAA \to 1 \right] - \half .
$$
\end{experiment}

\begin{definition}[Indistinguishability of Ciphertexts under Adaptive Chosen Ciphertext Attack (IND-CCA2)]\label{def:INDCCA2}
A SKES $\E$ has {\em indistinguishable encryptions under adaptive chosen ciphertext attack (or, it is IND-CCA2 secure)} iff, for any IND adversary $\A$ it holds that: $\advINDCCAA \leq \negl$.
\end{definition}

IND-CCA2 is clearly at least as strong as IND-CCA1.

\begin{theorem}[IND-CCA2 $\implies$ IND-CCA1]\label{thm:INDCCA2toINDCCA1}
If a SKES is IND-CCA2 secure, then it is also IND-CCA1 secure.
\end{theorem}

But the converse is not true. There exist IND-CCA1 secure SKESs where an adversary able to decrypt ciphertexts which are different, but related, to the challenge ciphertext, can find out information about the underlying plaintext.

\begin{theorem}[IND-CCA1 \nimplies IND-CCA2]\label{thm:INDCCA1notoINDCCA2}
There exist SKES which are IND-CCA1 secure, but not IND-CCA2 secure.
\end{theorem}
\begin{proof}[Proof (sketch)]
The counterexample is given by Construction~\ref{constr:goldreich}, as already hinted in the proof of Theorem~\ref{thm:GoldreichINDCCA1}. Being able to forge a valid ciphertext related in a controlled way to a target challenge ciphertext $c$ allows the adversary to ask for decryptions of such ciphertexts without violating the CCA2 limitation that the ciphertext must be different from the challenge one. For example, the adversary might be able to ask for a decryption of $c \xor 1\ldots1$, therefore recovering $m \xor 1\ldots1$, where $m$ was the original plaintext.
\end{proof}

Finally, although we are not going to write it down formally, it is possible to extend the SEM security notion to CPA, CCA1, and CCA2 scenarios as well, obtaining the security notions SEM-CPA, SEM-CCA1, and SEM-CCA2 respectively. Each of them can be shown to be equivalent to their IND counterpart. The situation is summarized in Figure~\ref{fig:QS0relations}.

\begin{figure}[t]
\begin{center}
\includegraphics[width=\textwidth]{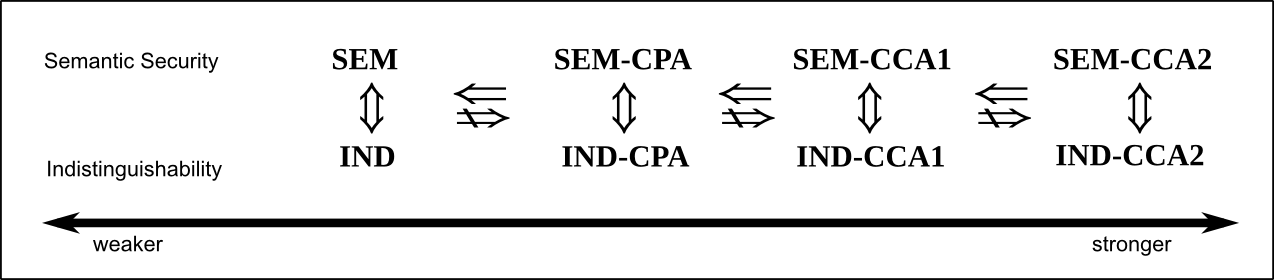}
\end{center}
\caption{Relations for SKES security notions in \QS0.}
\label{fig:QS0relations}
\end{figure}

\section{Public-Key Encryption Schemes}\label{sec:pke}

Another important cryptographic primitive are {\em public-key encryption schemes (PKES)}. Analogously to SKES, PKES work by encrypting messages from a plaintext space \X to a ciphertext space \Y, and decrypting ciphertexts the other way around. The difference this time is that the key generation algorithm generates {\em pairs} of keys: a {\em public-key \pk} which is only used to encrypt, and a {\em secret key \sk } which is only used to decrypt. W.l.o.g. we assume that, for security parameter \secpar, public keys are of bit size $\p(n)$, while secret keys are of bit size $\s(n)$, where \p,\s are polynomial functions determined by the scheme considered. Under this notation, we identify the (public, private) keyspace $\K$ as $\family{\K} = \family{\Kp} \times \family{\Ks} =: \Kp \times \Ks \subset \bin^{\p(n)} \times \bin^{\s(n)}$.

\begin{definition}[Public-Key Encryption Scheme (PKES)]\label{def:pkes}
A {\em public-key encryption scheme (PKES)} with plaintext space $\X$, ciphertext space $\Y$, and key space $\K\!:=\!\Kp \times \Ks$ \!is a tuple of\, \PPT algorithms $\E\!:=\!\E_{\K,\X,\Y}\!:=\!(\KGen,\Enc,\Dec)$:
\begin{enumerate}
\item $\KGen: \to \K$;
\item $\Enc: \Kp \times \X \to \Y$;
\item $\Dec: \Ks \times \Y \to \X \cup \set{\bot}$;
\end{enumerate}
such that $\foral n \in \NN, \foral  x \in \X, \foral (\pk,\sk) \from \KGen \implies \Dec(\sk, \Enc(\pk,x)) = x$.
\end{definition} 

As in the case of SKES, the following hold:

\begin{itemize}
\item we assume w.l.o.g. that $n$ is also appended to every \pk and every \sk such that $(\pk,\sk) \from \KGen$, so that every key also implicitly contains the security parameter.
\item As a shorthand notation, we will write $\Enc_{\pk}$ meaning the $\Enc$ algorithm with $\pk \in \Kp$ fixed as a first input; analogously for $\Dec_{\sk}$.
\item If $\Enc_{\pk}$ is probabilistic for all $\pk \in \Kp$, then we say that $\E$ is {\em randomized}, otherwise we say that $\E$ is {\em deterministic}.
\end{itemize}

The notions of security for PKES are basically the same as the ones for SKES, with two important differences:

\begin{enumerate}
\item because the public keys are, in fact, public, {\em all} the parties (including every stage of any adversary) can perform encryptions in polynomial time. Therefore, all parties have oracle access to $\Enc_\pk$. In particular, giving \M and \D the public key \pk as input also implies access to $\Enc_\pk$.
\item As an immediate consequence, notice that for PKES, IND-CPA is the {\em minimal} meaningful security notion. In fact, if \E is a PKES and \A an IND adversary for \E, notice that $\gameINDCPA=\gameIND[\A^{\Enc_\pk}]$.
\end{enumerate}

Finally, as in the SKES case, it is clear that for a PKES to be IND-CPA secure, $\Supp{\KGen(\secpar)}$ must be superpolynomial in $\secpar$ -- actually, both $\left| \set{ \pk \in \Kp_\secpar} \right|$ and $\left| \set{ \sk \in \Ks_\secpar} \right|$ must be superpolynomial in $n$.

IND-CPA secure PKES can be built from OWTPs. Assume for simplicity that $\X = \bin^\secpar$. Then we define the following.

\begin{construction}[PKES from OWTP]\label{constr:PKESfromOWTP}
Let $\P:=(\Gen,\Eval,\Invert)$ be a OWTP on \X, with index and trapdoor spaces \I and \T respectively, and let $\PRNG_\P:\X \to \X$ be the Goldreich-Levin PRNG for \P (seen as a OWF with hard-core predicates). Define $\E = \E_{\K,\X,\X^2}:=(\KGen,\Enc,\Dec)$ as a PKES with (public,private) key space $\K  = \Kp \times \Ks$ (where $\Kp := \I$ and $\Ks := \T$, plaintext space \X, and ciphertext space $\X^2$, in the following way:
\begin{enumerate}
\item $\KGen \to (\pk,\sk)$, with  $(\pk,\sk) := (i,t) \from \Gen$;
\item $\Enc_\pk(x) \to (y,z)$,\\with $y := x \xor \PRNG_\P(r)$ and $z \from \Eval(\pk,r)$, where $r \rand \X$;
\item $\Dec_\sk(y,z) := y \xor \PRNG_\P(s)$, where $s \from \Invert(\pk,\sk,z)$.
\end{enumerate}
\end{construction}

\begin{theorem}[IND-CPA PKES from OWTP]\label{thm:PKESfromOWTP}
Construction~\ref{constr:PKESfromOWTP} is an IND-CPA secure PKES.
\end{theorem}
\begin{proof}[Proof (sketch)]
If we omit the second half $z$ of the ciphertext, then the indistinguishability of the encryptions immediately follows from the information-theoretical security of the OTP, as the key $r$ of the OTP is always sampled indipendently and uniformly at random, and the output from the PRNG is computationally indistinguishable from random. So the only way to attack the scheme would be to extract information about the seed $r$ of the PRNG, by looking at the OWTP image $z$ obtained through \Eval. However, since $\PRNG_\P$ only outputs a sequence built from hard-core bits of \P, this would violate the one-wayness of the OWTP.
\end{proof}

\section{Digital Signature Schemes}\label{sec:sig}

{\em Digital signature schemes (DSS)} are another fundamental cryptographic building block for many other advanced constructions. In a DSS, each user has a unique private/public key pair, as in PKES. However, the goal is not to protect the secrecy of the message, but its {\em authenticity}, intended as assurance about the identity of the originator of the message, and {\em integrity}, intended as a guarantee that the original message sent by the originator has not been altered prior to being received. This is achieved by computing a piece of information (the {\em signature}) to attach to a message, in such a way that everyone can verify that such signature could not be computed without possession of a specific secret key. More in detail, the signature is computed by the {\em sender} of a message using the sender's private key, and it is attached to the message. The {\em verifier}, upon receiving the message, checks the validity of the signature by using the sender's public-key. The signature is a (short) message-- and secret-key-- specific bit string, with the following properties:
\begin{enumerate}
\item for any message and any secret-key, it is efficiently computable; and
\item for any message, it is hard to generate a valid signature for any public-key without having the corresponding secret-key.
\end{enumerate}

More formally, and borrowing the notation used in Section~\ref{sec:pke}, we define a DSS as follows.

\begin{definition}[Digital Signature Scheme (\!DSS)]\label{def:dss}
A {\em digital signature scheme (DSS)} with message space $\X$, signature space $\T$, and key space $\K := \Kp \times \Ks$  is a tuple of \PPT algorithms $\Sigscheme := \Sigscheme_{\K,\X,\T} := (\KGen,\Sign,\SVer)$:
\begin{enumerate}
\item $\KGen: \to \K$;
\item $\Sign: \Ks \times \X \to \T$;
\item $\SVer: \Kp \times \X \times \T \to \bin$;
\end{enumerate}
such that the following correctness condition holds:
$$
\foral n \in \NN, \foral x \in \X, \foral (\pk,\sk) \from \KGen, \foral \sig \from \Sign(\sk,x)
$$$$
\implies \SVer(\pk,x,\sig) = 1.
$$
\end{definition} 

As in the case of SKES, the following hold:

\begin{itemize}
\item we assume w.l.o.g. that \secpar is also appended to every \pk and every \sk such that $(\pk,\sk) \from \KGen$, so that every key also implicitly contains the security parameter.
\item As a shorthand notation, we will write $\Sign_{\sk}$ meaning the $\Sign$ algorithm with $\sk \in \Ks$ fixed as a first input; analogously for $\SVer_{\pk}$.
\end{itemize}

\subsection{Existential Unforgeability}

The notions of security for DSS is given in terms of {\em (strong) existential unforgeability under chosen message attack} (there are also weaker notions, but we will not use them here). An adversary is successful if he manages to create a valid signature for a message and public-key without having the corresponding secret key, even after observing a polynomial number of valid message/signature pairs.

\begin{experiment}[$\gameEUFCMA$]\label{expt:EUF-CMA}
Let $\Sigscheme$ be a DSS, and \A a \PPT algorithm. The {\em EUF-CMA experiment} proceeds as follows:
\begin{algorithmic}[1]
\State \textbf{Input:} $\secpar, q_s \in \NN$
\State $(\pk,\sk) \from \KGen$
\State $(x,\sig) \from \A^{\Sign_\sk}(\pk)$ after making at most $q_s$ queries to $\Sign_\sk$, receiving signatures $(x_1,\sig_1) , \ldots (x_{q_s}, \sig_{q_s})$
\If{$\SVer(\pk,x,\sig) = 1$ and $\x \neq \x_i \foral i = 1,\ldots,q_s$}
	\State \textbf{Output:} $1$
\Else
	\State \textbf{Output:} $0$
\EndIf
\end{algorithmic}
The {\em advantage of \A} is defined as:
$$
\advEUFCMA(\secpar,q_s) := \Pr \left[ \gameEUFCMA(\secpar,q_s) \to 1 \right].
$$
\end{experiment}

Sometimes we also consider a slightly different version of Experiment~\ref{expt:EUF-CMA}, where the public/private key pair is given as an input to the game instead of being generated randomly. This is useful if we want to target a specific public key during some security reduction.

\begin{definition}[Existential Unforgeability under Chosen Message Attack (EUF-CMA)]\label{def:EUFCMA}
A DSS $\Sigscheme$ is {\em existentially unforgeable under chosen message attack (or, it is EUF-CMA secure)} iff, for any \PPT algorithm \A it holds that:
$$
\advEUFCMA \leq \negl.
$$
\end{definition}

\subsection{Signatures in the Random Oracle Model}

For certain applications it makes sense to investigate the security properties of signature schemes in the random oracle model. Recall that, in the ROM, all the parties involved gain access to an oracle $\RO$, where $\h$ is a function chosen uniformly at random from the set of all functions on certain spaces. This also means, in particular, that Definition~\ref{def:dss} changes by allowing $\KGen,\Sign,\SVer$ oracle access to $\RO$. The resulting security model changes as follows.

\begin{experiment}[$\gameEUFCMARO$]\label{expt:EUF-CMA-RO}
Let $\Sigscheme$ be a DSS, \RO a random oracle (computing a function \h selected uniformly at random), and \A a \PPT algorithm. The {\em EUF-CMA-RO experiment} proceeds as follows:
\begin{algorithmic}[1]
\State \textbf{Input:} $\secpar, q_s, q_h \in \NN$
\State $(\pk,\sk) \from \KGen^\RO$
\State $(x,\sig) \from \A^{\Sign_\sk,\RO}(\pk)$ after making at most $q_h$ queries to \RO, and $q_s$ queries to $\Sign_\sk$ receiving signatures $(x_1,\sig_1) , \ldots (x_{q_s}, \sig_{q_s})$
\If{$\SVer(\pk,x,\sig) = 1$ and $\x \neq \x_i \foral i = 1,\ldots,q_s$}
	\State \textbf{Output:} $1$
\Else
	\State \textbf{Output:} $0$
\EndIf
\end{algorithmic}
The {\em advantage of \A} is defined as:
$$
\advEUFCMARO(\secpar,q_s,q_h) := \Pr \left[ \gameEUFCMARO(\secpar,q_s,q_h) \to 1 \right].
$$
\end{experiment}

\begin{definition}[Existential Unforgeability under Chosen Message Attack in the Random Oracle Model (EUF-CMA-RO)]\label{def:EUFCMARO}
A DSS $\Sigscheme$ is {\em existentially unforgeable under chosen message attack in the random oracle model (or, it is EUF-CMA-RO secure)} iff, for any \PPT algorithm $\A$ it holds that:
$$
\advEUFCMARO \leq \negl.
$$
\end{definition}

\section{The Fiat-Shamir Transformation}\label{sec:FS}

The Fiat-Shamir (FS) transformation \cite{FS86} is a well known method to remove interaction in three-move identification schemes between a prover and verifier (also called {\em$\Sigma$-protocol}), by letting the verifier's challenge $\ch$ be determined via a hash function $\h$ applied to the prover's first message $\com$. Currently, the only generic, provably secure instantiation is by modeling the hash function $\h$ as a random oracle \cite{BR93,PS00}. In this section, we will investigate the security of the FS transformation when applied to a $\Sigma$-protocol \sigmaproto in order to obtain a DSS \Sigscheme, which we call {\em the FS transform of \sigmaproto}.

\subsection{Hard Languages}

Let $\L \in \NP$ be a language with a (polynomially computable) relation $\R$, i.e., 
$\foral x : x \in \L \iff \exists \ w \in \W \subset \bin^{\poly(|x|)} : (x,w) \in \R$. In this case we also write that $x \in \L_\secpar$ and $(x,w) \in \R_\secpar$, for $\secpar = |x|$. For using $\L$ in cryptographic applications, we need to discuss the following two issues:
\begin{enumerate}
\item given a statement $x \in \L$, how hard is to find a valid witness for $x$? And,
\item is it possible at all to find valid pairs $(x,w) \in \R$ in an efficient way?
\end{enumerate}

For an interesting security notion, finding a witness from $x$ alone should be infeasible for computationally bounded adversaries. On the other hand, it is useful to have a way to efficiently sample elements from the relation.

To this end we assume the existence of an efficient {\em hard instance generator $\Inst$}, which on input the security parameter \secpar outputs a pair $(x,w) \in \R_\secpar$ such that no \PPT algorithm can find valid witnesses for the overwhelming majority of statements contained in any of $\Inst$'s output. 
If $\L$ admits a hard instance generator, we say that $\L$ is a {\em hard language}.

\begin{definition}[Hard Language and Instance Generator]
Let \R be an $\NP$ relation between language \L and witness space \W. A \PPT algorithm \Inst is a {\em hard instance generator for \R} iff the following hold:
\begin{enumerate}
\item $(x,w) \in \R_\secpar$, for any $(x,w) \from \Inst$; and
\item for any \PPT algorithm \A it holds:
$$
\Pr_{(x,w) \from \Inst} \left[ (x,\A(x)) \in \R \right] \leq \negl.
$$
If $\L$ admits a hard instance generator, we say that $\L$ is a {\em hard language}, and we denote it by $\hardL$.
\end{enumerate}
\end{definition}

Notice that the existence of a hard instance generator does not mean that it is hard to find a valid witness for {\em any} statement in \L. But this certainly holds for the vast majority of those statements in the subclass output by \Inst. Moreover, the cardinality of this subclass is at least superpolynomial in \secpar (otherwise \PPT algorithms with oracle access to \Inst could find valid witnesses by exhaustive search). This fact is used in the following paragraphs about the {\em Fiat-Shamir transformation} in order to show that large enough commitment spaces can be built from hard languages. Candidates for hard languages are at the base of many cryptographic constructions, and stem from \NP problems such as {\em graph isomorphism}~\cite{GMW86}, {\em decisional Diffie-Hellman for finite groups}~\cite{BonehDDH}, and many others.

\subsection{Identification Schemes}

An {\em identification scheme (IS)} between a {\em prover} $\P$ and a {\em verifier} $\V$ is an interactive protocol which allows $\P$ to prove his {\em identity} to $\V$. By `proving identity' we mean `proving a statement about one's identity'. This is usually done with the help of a hard language \hardL where every user identity is bound to a certain statement; in practice, identities are usually linked to some public key, and for the prover to succeed he must prove ownership of the corresponding private key. We write $d \from \left(\P(x,w),\V(x)\right)$ for the final outcome of the protocol, where $d \in \bin$ is a bit denoting the final decision (acceptance or rejection) of the verifier.

ISs are related to a class of cryptographic objects known as {\em interactive proofs of knowledge}. Traditionally, the security notion for an IS is based on {\em impersonation security}, which intuitively states that no efficient adversary should be able to make \V accept a statement $x$ without knowing a valid witness $w$. However, for the scope of this work, a weaker notion of security (which we call {\em weak impersonation security}) suffices. In this notion, additional effort is required for an adversary to be successful. Namely, given a statement $x$, the adversary should be able, after interacting with \idscheme, to output a valid witness for $x$, breaking the security of the hard language. This, in particular, would allow the adversary to make \V accept the execution of the scheme (but the converse is not necessarily true, that is why this notion is called `weak'). Moreover, weak security comes in two variants, depending on the level of interaction that the adversary is allowed to have with \idscheme. For {\em passive} adversaries, the only allowed interaction is given by observing and recording the executions of (at most polynomially many) sequential instances of the IS for a given statement. Therefore, passive weak security only relies on the hardness of the language \hardL.

\begin{definition}[Passively and Weakly Secure Identification Scheme (PWSIS)]\label{def:pwsis}
A {\em passively and weakly secure identification scheme (PWSIS)}, \idscheme for a hard language $\hardL$ is an interactive protocol between two \PPT algorithms \P and \V satisfying:
$$
\foral \secpar, \foral (x,w) \from \Inst \implies \left(\P(x,w),\V(x)\right) \to 1.
$$
\end{definition}

An {\em active} adversary, instead, is also allowed to interact directly with $\P(x,w)$ by impersonating \V, and its goal is to output a valid witness for $x$ given this interaction. That is, an active adversary $\A:=(\A_1,\A_2)$ is a passive adversary who has also access to $\P(x,w)$ (seen as oracles). However, in order to avoid trivial breaks of the identification scheme (e.g., by man-in-the-middle attacks), during the security game the adversary can only be active before actually seeing $x$, and becomes passive afterwards. We express this as $\A_1^{\P(x,w)}(x)$. Obviously, if an IS is weakly secure against active attacks, it is also secure against passive attacks, but the converse does not necessarily hold. More formally, we define the following.

\begin{definition}[Actively Weakly Secure Identification Scheme (AWSIS)]\label{def:awsis}
An {\em actively and weakly secure identification scheme (AWSIS)}, \idscheme for a hard language $\hardL$ is a PWSIS (according to Definition~\ref{def:pwsis}) such that, for every \PPT algorithms $\A_1, \A_2$, the following holds:
$$
\Pr_{(x,w) \from \Inst} \left[ (x,\A_2(x, \A_1^{\P(x,w)}(x))) \in \R \right] \leq \negl.
$$
\end{definition}

\subsection{$\Sigma$-Protocols}

A $\Sigma$-protocol for a hard language $\hardL$ between a {\em prover} $\P$ and a {\em verifier} $\V$ is a 3-step interactive protocol which allows $\P$ to convince $\V$ that he knows a witness $w$ for a public theorem $x \in \L$, without giving to $\V$ non-trivial information beyond this fact. Informally, a $\Sigma$-protocol $\sigmaproto$ consists of an interactive exchange of three messages $(\com,\ch,\resp)$ where the first message $\com$ (the {\em commitment}) is sent by $\P$, the second message $\ch$ (the {\em challenge}) is sampled uniformly from a challenge space by \V, and the last message $\resp$ (the {\em response}) is computed by \P by using the witness. We write $(\com,\ch,\resp) \from \left(\P(x,w),\V(x)\right)$ for the randomized output (the {\em communication transcript}) of an interaction between $\P$ and $\V$. We denote individual messages of the (stateful) prover in such an execution by $\com \from \P(x,w)$ and $\resp \from \P(x,w,\com,\ch)$, respectively. Analogously, we denote the verifier's steps by $\ch \from \V(x,\com)$ for the challenge step, and $d \from \V(x,\com,\ch,\resp)$ for the final decision, where $d \in \bin$ is a bit denoting acceptance or rejection.

More formally, we define the following.

\begin{definition}[$\Sigma$-Protocol]\label{def:sigmaproto}
A $\Sigma$-protocol ({\em `sigma-protocol'}) \sigmaproto for a hard language $\hardL$ is a $3$-move interactive protocol with exchange of messages $\com,\ch,\resp$ between two \PPT algorithms \P and \V satisfying the following properties:

\begin{enumerate}
\item \textbf{Completeness:} $\forall \secpar \in \NN, (x,w) \in \R_\secpar, (\com,\ch,\resp) \from (\P(x,w),\V(x))$ it holds that: $\V(x,\com,\ch,\resp)=1$.

\item \textbf{Public-Coin:} $\forall \secpar \in \NN, (x,w) \in \R_\secpar, \com \from \P(x,w)$, the challenge distribution $\ch \from \V(x,\com)$ is uniform on $\bin^{\poly(\secpar)}$.

\item \textbf{Special Soundness:} there exists a \PPT algorithm $\J$ (the {\em extractor}) such that, given two valid transcripts $(\com,\ch,\resp)$ and $(\com,\ch',\resp')$ for $x\!\in\!\L$ (with $\ch\!\neq\!\ch'$) and $\V(x,\com,\ch,\resp)\!=\!\V(x,\com,\ch',\resp')\!=\!1$, the extractor outputs a witness $w\from \J(x,\com,\ch,\resp,\ch',\resp')$ for $x$, satisfying $(x,w) \in \R$.

\item \textbf{Honest-Verifier Zero-Knowledge (HVZK):} there exists a \PPT algorithm $\S$ (the {\em zero-knowledge simulator}) which, on input $x \in\L$, outputs a transcript $(\com,\ch,\resp)$ that is computationally indistinguishable from a valid transcript derived in a $\sigmaproto$ interaction. That is, for any \PPT algorithm $\V=(\V_1,\V_2)$, the following two distributions are statistically indistinguishable:
\begin{multicols}{2}

\begin{algorithmic}[1]
\State \textbf{Input:} $n \in \NN$
\State $(x,w,\mathsf{state}) \from \V^*_1$
\If{$(x,w) \in \R$}
	\State \!\!\!\!\!\!$(\com,\!\ch,\!\resp)\!\from\!(\P(x,w),\!\V(x))$
\Else
	\State \!\!\!\!\!\!$(\com,\ch,\resp) := (\bot,\bot,\bot)$
\EndIf
\State $b \from \V^*_2(\com,\ch,\resp,\mathsf{state})$
\State \textbf{Output:} $b$
\end{algorithmic}
\vfill
\columnbreak
\begin{algorithmic}[1]
\State \textbf{Input:} $n \in \NN$
\State $(x,w,\mathsf{state}) \from \V^*_1$
\If{$(x,w) \in \R$}
	\State \!\!\!\!\!\!$(\com,\ch,\resp) \from \S(x)$
\Else
	\State \!\!\!\!\!\!$(\com,\ch,\resp) := (\bot,\bot,\bot)$
\EndIf
\State $b \from \V^*_2(\com,\ch,\resp,\mathsf{state})$
\State \textbf{Output:} $b$
\end{algorithmic}
\vfill
\end{multicols}
\end{enumerate}
\end{definition}

It turns out that $\Sigma$-protocols are also (passively, weakly-secure) identification schemes.

\begin{theorem}[$\Sigma$-Protocols as IS]\label{thm:sigmaIS}
Let \sigmaproto be a $\Sigma$-protocol. Then \sigmaproto is a PWSIS.
\end{theorem}

It is important to notice two things in the above theorem:
\begin{itemize}
\item HVZK is not necessary for Theorem~\ref{thm:sigmaIS} to hold; and
\item a $\Sigma$-protocol may or may not be also an AWSIS.
\end{itemize}

\subsection{The FS Transformation applied to $\Sigma$-Protocols}

The Fiat-Shamir transformation of a $\Sigma$-protocol $(\P,\V)$ is a modification of the protocol where the computation of $\ch$ is done as $\ch \from \h(x,\com)$ instead of $\from \V(x,\com)$. Here, $\h$ is a public hash function which is usually modeled as a random oracle $\RO$; in this case we speak of the {\em Fiat-Shamir (FS) transformation of $(\P,\V)$ in the random-oracle model}. Note that we include $x$ in the hash computation, but all of our results remain valid if $x$ is omitted from the input. If applying the FS transformation to a $\Sigma$-protocol, one obtains a signature scheme, if the hash computation also includes the message $m$ to be signed. We call the resulting signature scheme {\em FS transform of \sigmaproto in the ROM}, and we denote it by $\FSSigma$.

\begin{definition}[FS Transform of a $\Sigma$-Protocol]
Let \sigmaproto be a $\Sigma$-protocol for a hard language \hardL, with commitment space \X, challenge space \Y, and response space \Z. Let \RO be a random oracle for a random function $\h: \L \times \X \times \M \to \Y$. The {\em FS transform of \sigmaproto in the ROM, $\FSSigma$}, is a DSS with message space $\M$, signature space $\T :=\X \times \Y \times \Z$, and key space $\K := \L \times \W$, defined as follows:
\begin{enumerate}
\item $\KGen \to (\pk,\sk)$, where $(\pk,\sk):=(x,w) \from \Inst$
\item $\Sign^\RO(\sk,m) \to \sig := (\com,\ch,\resp)$,\\where $\com \from \P(\pk,\sk)$, $\ch := \h(\pk,\com,m)$,\\and $\resp \from \P(\pk,\sk,\com,\ch)$
\item $\SVer^\RO(\pk,m,\sig) \to b$,\\where $\sig := (\com,\ch,\resp)$, $b \from \V(\pk,\com,\h(\pk,\com,m),\resp)$
\end{enumerate}
\end{definition}

Notice that in the above signature the challenge $\ch$ can always be omitted (and it is infact ignored in the verification step), because it is recovered by computing $\h$ on the message, the commitment, and the public key. In this case we define the signature space of the DSS as $\T := \X \times \Z$. 
The following theorem states that the above construction yields secure DSSs in the ROM.

\begin{theorem}[Security of a Fiat-Shamir Transform~\cite{PS00}]\label{thm:FS}
Let \sigmaproto be a $\Sigma$-protocol. Then $\FSSigma$ is an EUF-CMA-RO secure DSS.
\end{theorem}
\begin{proof}[Sketch]
The proof of this theorem uses rewinding. Intuitively, given a statement $x$ and an adversary forging a signature for the DSS, this is used to extract a transcript $(\com,\ch,\resp)$ for the underlying $\Sigma$-protocol. After that, the adversary is rewound, and the random oracle reprogrammed, in such a way that letting the adversary run again with the new oracle yields a related transcript $(\com,\ch',\resp')$ for the same \com but $\ch'\neq\ch$. This, in turn, allows to use the special soundness property to extract a valid witness for $x$, therefore breaking the weak security of the underlying $\Sigma$-protocol, in contrast to Theorem~\ref{thm:sigmaIS}.
\endproof
\end{proof}

\section{ORAMs}\label{sec:ORAM}

In this chapter we have presented many different cryptographic objects, in order of growing technical complexity. As a last example, we conclude with the concept of {\em Oblivious Random Access Machine (ORAM)}, and we define and analyze security models against classical adversaries.

Defining ORAMs in a fully formal way is a long, delicate, and strenuous task~\cite{gooram}. Therefore, in the following we will use a simplified model (introduced in~\cite{GKK17}) which covers most of the existing ORAM constructions without delving too much into the fine print - but still retaining a reasonable level of formalism - and which has the advantage of being much easier to treat.

Informally, an ORAM is an interactive protocol between two parties: a {\em client \C} and a {\em server \S}, which we model as two \PPT Turing machines (or, in our case, uniform families of circuits) sharing a communication tape (circuit register) $\Xi$ to exchange data. In this scenario, a computationally limited \C wants to outsource a {\em database (DB)} to the more powerful \S. Moreover, \C wants to perform {\em operations} on the DB (by interactively communicating with \S) in such a way that \S, or any other honest-but-curious adversary \A having read-only access to $\Xi$ and \S's internal memory, cannot determine the nature of such operations. The security notion for ORAM schemes is therefore a particular notion of {\em privacy}.

More formally: we define {\em blocks}, the basic storage units used in an ORAM construction. A block is an area of memory (circuit register) storing a \bsize-bit value, for a fixed parameter $\bsize \in \NN$ which depends on \C's and \S's architectures. A {\em database} (DB) of size $\dbsize \in \NN$ is an area of \S's memory which stores an array $(\block_1, \ldots , \block_\dbsize)$ of such blocks. As we assume this database to reside on the server's side, we will denote it as $\S.\DB$. Notice that the precise way this array of blocks is represented in the database is unspecified, and left to the exact implementation of the ORAM scheme taken into account. For example, in the ORAM construction we are going to analyze in detail, the server's database $\S.\DB$ stores blocks in a binary tree structure. We will abuse notation and write that $\S.\DB(i) = \block$ if $\block$ is the $i$-th component of $\S.\DB$, and that $\block \in \S.\DB$ if $\block$ is stored at some position in the database \S.\DB. 

Next we define {\em data units} as the basic units of data that the client wants to access, read, or write. Formally, a data unit is an \dsize-bit value for a fixed parameter $\dsize \leq \bsize$ which depends on \C's and \S's architectures. Every block encodes (usually in an encrypted form) a data unit, plus possibly auxiliary information such as a block identifier, checksum, or hash value. Since every block can encode a single data unit, at any given time $t$ it is defined a function $\Data_t: \S.\DB \to \bin^\dsize$. With abuse of notation, we will denote by $\Data(\block)$ the data unit encoded in the block $\block$ at a certain time $t$. The client \C can operate on the database through {\em data requests}.

\begin{definition}[Data Request]\label{dfn:dr}
A {\em data request} to a database $\S.\DB$ of size \dbsize is a tuple $\dr = (\op, i, \data)$, where $\op\in\{\text{read},\text{write}\}, i \in \set{1,\ldots,\dbsize}$, and $\data \in \bin^\dsize$ is a data unit ($\data$ can also be $\bot$ if $\op=\text{read}$).
\end{definition}

Finally, we define the meaning of a {\em communication transcript} during an execution of an ORAM protocol. Since this also depends on the exact implementation of the ORAM scheme, we will use the following definition.

\begin{definition}[Communication Transcript]\label{def:comtrans}
A {\em communication transcript} $\com_t$ at time $t$ is the content of the communication channel $\Xi$ at time $t$ of the protocol's execution.
\end{definition}
Notice that the above defines the communication transcript as a function of time, but since an ORAM is a multi-round interactive protocol we will just consider $\com$ as a discrete function of the round $1,2,\ldots$ of the protocol.

\

We are now ready to give a definition of ORAM. We assume that a server's database is always initialized empty (usually with randomized encryptions of $0$ elements as blocks), and it is left up to the client the task of `populating' the database with appropriate {\em write} operations.

\begin{definition}[ORAM]\label{def:oram}
Let $\maxsize \in \NN, \msize \geq \dsize \in \NN$ be fixed parameters, and $\E=(\KGen,\Enc,\Dec)$ be a SKES mapping $\msize$-bit plaintexts to $\bsize$-bit ciphertexts. An {\em ORAM} $\oram_\E$ with parameters $(\maxsize,\dsize,\E)$ is a pair of two-party interactive randomized algorithms, $(\init,\access)$, such that:
\begin{itemize}
\item $\init(\secpar,\dbsize) \rightarrow (\C,\S)$ in the following way:
	\begin{enumerate}
	\item \secpar is the security parameter, $\dbsize \leq \maxsize$;
	\item $k \from \KGen(\secpar)$ is generated by \C;
	\item \S includes a database $\S.\DB=(\block_1,\ldots,\block_\dbsize)$,\\where $\foral i \implies \block_i \from \Enc_k(0)$;
	\end{enumerate}
\item $\access(\C,\S,\dr) \rightarrow (\C',\S',\com)$ in the following way:
	\begin{enumerate}
	\item \C issues a data request $\dr$;
	\item \C and \S communicate through $\Xi$ and produce the communication's transcript $\com$;
	\end{enumerate}
\end{itemize}
\end{definition}

One might wonder why it is necessary to explicitly condition the definition of an ORAM in respect to a symmetric-key encryption scheme \E. It is actually possible to use different primitives, such as PKES, but most of the known ORAM constructions work well with just a simple primitive such as SKES. One might also wonder why the definition does not depend on other cryptographic primitives, such as PRNGs or PRFs. The reason is that not all ORAM constructions use such primitives, for example the `trivial' ORAM scheme~\cite{gooram} (which consists in just transferring the whole encrypted database from \S to \C and back at every data request) does not use anything else than a SKES \E as a building block. On the other hand, notice that encryption of the database is a minimal requirement for security, as we will see, therefore it makes sense to explicitly specify the scheme \E in the notation.

An ORAM must satisfy {\em soundness} and {\em security}. We are going to define security in Section~\ref{sec:oramsec}. Regarding soundness, the exact specification depends on the particular ORAM construction considered. A simplified, game-based definition of soundness ({\em `correctness'}) can be found in~\cite{tworam}, but it is difficult to adapt to the model from~\cite{GKK17} which we consider here, and which is more aimed at studying ORAM security, while a general definition (that can be found in~\cite{gooram}) is rather involved, and goes outside the scope of this work. The meaning of the soundness property is that the ORAM protocol `should work', i.e., after any execution of $\init$ or $\access$ the two parties \C and \S must be left in such a state that allows them to continue the protocol in the next round. Despite the generality of this statement, in the model we consider here minimal soundness conditions can be identified, which must hold for {\em any} ORAM construction.

\begin{definition}[Minimal ORAM Soundness Conditions]\label{def:soundness}
An ORAM construction $\oram_\E$ has {\em minimal soundness} if the following hold:
\begin{enumerate}
\item \label{item:key1} for any $(\secpar,\dbsize)$, if $(\C,\S) \leftarrow \init(\secpar,\dbsize)$, then \C stores the secret key $k$ from Def.~\ref{def:oram};
\item for any $\dr = (\op, i, \data)$, if $(\C',\S',\com) \leftarrow \access(\C,\S,\dr)$, then:
	\begin{enumerate}
	\item \label{item:key2} if \C stores the secret key $k$, then also $\C'$ stores $k$;
	\item \label{item:encrypted} if $\op=\text{read}$ and $\S.\DB(i) = \block$, then $\C'$ stores $\Data(\block)$;
	\item \label{item:disappear} if $\op=\text{write}$ and $\S'.\DB(i) = \block$, then $\Data(\block) = \data$.
	\end{enumerate}
\end{enumerate}
\end{definition}

Notice that conditions~\ref{item:key1} and~\ref{item:key2} do not say anything about \S having access to the key $k$ or not: This is a property of {\em security}, not soundness, as we will see in Section~\ref{sec:oramsec}. An ORAM scheme \oram can have additional soundness conditions, depending on the particular construction. We assume that whenever \C (resp., \S) is modified during the execution of the protocol to \C' (resp., \S') after $\access$ calls, all these soundness conditions (the minimal ones above as well as the special ones) are always satisfied. In this case, we also say that \C' is a {\em sound evolution} of \C and that \S' is a {\em sound evolution} of \S.

\subsection{Classical Security of ORAMs}\label{sec:oramsec}

We now look at the security model for ORAMs against classical adversaries introduced in~\cite{GKK17}. Traditionally, the threat model in this case is defined by an {\em honest-but-curious adversary \A}. This means that \A is some entity who wants to compromise \C's privacy by having access to the 
communication channel $\Xi$ and \S's internal memory, but who is not allowed to modify the content of the channel or the database against the protocol, i.e., soundness must be preserved. In general, one does not lose generality by assuming that \S itself is the adversary: \S must behave `honestly' (in the sense that he follows the protocol, in particular related to the protocol's soundness), but at the same time he will use all the information he can get through the interaction with \C in order to compromise \C's privacy. In particular, this also implies that \S cannot know the key $k$ generated during $\oram.\init$, as noted above.

Formally, this model is defined in terms of {\em access patterns}, which are the adversarial views during an execution of data requests in $\oram.\access$. Security requires that the adversary's view over a certain run of the protocol does not leak any information about the data requests executed by \C, except the sequences' length. This formulation reminds of the definition of semantic security for encryption schemes. As in that case, equivalent but easier-to-deal-with formulations can be given in terms of {\em computational indistinguishability 
of access patterns}. Following the security model introduced in~\cite{GKK17}, we will consider an adaptive, game-based indistinguishability notion stating that for any two data requests, no computationally bounded adversary with knowledge of the access pattern of the client executing one of the two can distinguish which one was executed. This definition is equivalent~\cite{GKK17} to the simulation-based notion given in~\cite{tworam}, which states that no computationally bounded adversary can distinguish between the interaction with a real client or with a simulator that produces bogus transcripts.

More formally: when a data request is executed, we assume that the honest-but-curious adversary \A records all the communication between \C and \S, plus the changes in \S's internal status. Without loss of generality, as we assume that \A and \S coincide, we assume that the only meaningful changes in the database area $\S.\DB$ only happen between the beginning and the end of an $\access$ execution. The communications are polynomially bounded and, for simplicity, we assume that the channel $\Xi$ does not erase symbols, i.e., it is write-once. Hence, the adversarial view is composed of the communication transcript, and the server's database before and after the execution of the data request. We call this adversarial view, the {\em access pattern} of the execution.

\begin{definition}[Access Pattern]\label{dfn:ap}
Given ORAM client and server $\C$ and $\S$, and a data request $\dr$, the {\em access pattern} $\ap(\dr)$ is the tuple $\left( \S.\DB,\com,\S'.\DB \right)$, where $(\C',\S',\com) \from \access(\C,\S,\dr)$.
\end{definition}

Next, we define formally a {\em classical ORAM adversary}.

\begin{definition}[Classical ORAM Adversary]\label{dfn:oramadv}
A {\em classical ORAM adversary} \A is a \PPT algorithm which 
is computationally indistinguishable from an honest server \S for every ORAM client \C (in particular, soundness is preserved).
\end{definition}

Notice the following fact: this definition of adversary can generally be {\em stronger} than in the usual `honest-but-curious' meaning. In fact, such adversary could still manipulate the channel and the database in a malicious way, as long as the client \C cannot detect such manipulation -- in particular, the soundness of the protocol must be preserved. We define the security of an ORAM through the following indistinguishability game.

\begin{experiment}[\gameORAM]\label{expt:gameORAM}
Let $\oram=(\init,\access)$ be an ORAM construction with parameters $(\maxsize,\dsize,\E)$, \secpar a security parameter and \A a classical ORAM adverary. The {\em computational indistinguishability of access patterns game under adaptive chosen query attack} $\gameORAM$ proceeds as follows:
\begin{algorithmic}[1]
\State \textbf{Input:} $\secpar \in \NN$
\State $\A \to (\A_0,\dr_1,\dbsize \leq \maxsize)$
\State $(\C_0,\S_0) \leftarrow \init(\secpar,\dbsize)$
\Loop{ for $i=1,\ldots,q_1 \in \NN$:} \Comment{first CQA learning phase}
\State $\access(\C_{i-1},\S_{i-1},\dr_i) \to (\C_i,\S_i,\ap_i)$
\State $\A_{i-1}(\ap_i,\S_i) \to (\A_i,\dr_{i+1})$
\EndLoop
\State $\A_{q_1}(\dr_{q_1+1}) \to (\A',\dr^0,\dr^1)$ 
\State $b \rand \bin$
\State $\access(\C_{q_1},\S_{q_1},\dr^b) \to (\C_{q_1+1},\S_{q_1+1},\ap_{q_1+1})$ \Comment{AP-IND challenge query}
\State $\A'(\ap_{q_1+1},\S_{q_1+1}) \to (\A_{q_1+1},\dr_{q_1+2})$
\Loop{ for $i=q_1+2,\ldots,q_2 \geq q_1+2$:} \Comment{second CQA learning phase}
\State $\access(\C_{i-1},\S_{i-1},\dr_i) \to (\C_i,\S_i,\ap_i)$
\State $\A_{i-1}(\ap_i,\S_i) \to (\A_i,\dr_{i+1})$
\EndLoop
\State $\A_{q_2}(\dr_{q_2+1}) \to b' \in \bin$
\If{$b = b'$}
	\State \textbf{Output:} $1$
\Else
	\State \textbf{Output:} $0$
\EndIf
\end{algorithmic}
The {\em advantage of \A} is defined as:
$$
\advORAM := \Pr \left[ \gameORAM \to 1 \right] - \half .
$$
\end{experiment}

In this game the adversary, after selecting suitable ORAM parameters of his choice, is first allowed to see the access patterns originated by executions of \access for data requests of his choice, chosen adaptively one after the other (this is called `first CQA learning phase'.) At some point, the adversary issues a challenge query composed of two (w.l.o.g. different) data requests. One of the two is selected at random and executed through \access, and the adversary, after being allowed a second CQA learning phase, must guess which one of the two was executed. Notice that, since \A is polynomially bounded, $q_1$ and $q_2$ are at most polynomials in \secpar. We are now ready to define the classical security notion for ORAMs.

\begin{definition}[Access Pattern Indistinguishability Under Adaptive Chosen Query Attack]\label{def:AP-IND-CQA}
An ORAM construction $\oram$ has {\em computationally indistinguishable access patterns under adaptive chosen query attack} (or, it is AP-IND-CQA-secure) iff for any classical ORAM adversary \A it holds that $\advORAM \leq \negl$.
\end{definition}

\subsection{PathORAM}

As an example of ORAM construction, we recall here PathORAM, one of the most efficient ORAM constructions proposed to date, introduced by Stefanov et al. in~\cite{pathoram}. We only give a high-level explanation of PathORAM, and for a thorough description of the construction, as well as a detailed proof of its functionality, we refer to~\cite{pathoram}.

In PathORAM a client stores \dbsize blocks of bit size \bsize on a server, in a binary tree structure of height $\tsize = \lceil \log_2 \dbsize \rceil$. Each node of the tree can store a constant amount \zsize of blocks. Every block encodes (in an encrypted form, using an IND-CPA 
SKES) a data unit of bit size \dsize, and optionally additional information which is used to label the block for efficient retrieval. There are many different ways one can implement this labeling of the blocks. In our case we will use the simple approach of concatenating to the data unit $\data$ an \ksize-bit string encoding the block identifier $i \in \set{1,\ldots,\dbsize}$, that is, blocks are of the form $\block_i \from \Enc_k (i\|\data_i)$. This system is very general, and as we will see it has the advantage that it easily translates to the quantum setting, unlike other approaches such as identifying blocks by using a hash table. At the beginning, all the blocks in the tree are initialized in an `empty' state, which is defined by setting to $0$ the identifying label -- recall in fact that valid block identifiers are $1,\ldots,\dbsize$ only. Every block is mapped to a leaf of the tree, and this mapping is recorded in a correspondence table, called {\em position map}, by the client\footnote{Note that the size of the position map is linear in the number of blocks that the client has, and thus cannot be stored locally by the client. The authors of~\cite{pathoram} propose storing the position map recursively to smaller PathORAMs following an idea from~\cite{SSS12}. For ease of exposition however, we will assume here that the position map is stored locally.}. 

A read (or write) operation for a block $\block_i$ is performed by the client, by downloading the path (tree branch) from the root of the tree to the leaf indicated in the client's position map, and randomly remapping $\block_i$ to another leaf in the position map. Then the client decrypts and re-encrypts (re-randomizing) all the blocks in the downloaded path, and for every valid (non-empty) block $\block_j$ found, the client checks its corresponding leaf in the position map, and moves $\block_j$ (if there is enough available space) to the node in the path which is closest to the leaf level and that belongs both to the downloaded path and the path to the leaf of $\block_j$ given by the position map. If a block does not fit anywhere in the downloaded path, then an extra storage, called {\em `stash'} is used by the client to store this overflowing block locally. The blocks found in the stash are also examined during every read (or write) operation and checked if they can be evicted from the stash and placed in the tree. Since the stash must be stored locally by the client, the stash's size should be reasonably small; in fact, in~\cite{pathoram}, the authors show that the probability that the stash exceeds a size of $\O(\log \dbsize)$ is negligible in the number of queries. The intuition is to notice that the stash is only used if the tree root is full, but the average action of a data request is to push only $\block_i$ toward the tree root, and push many other blocks $\block_j$ toward the leaf level. In the following we will mostly ignore the use of the stash for simplicity.

More concretely, we give here a full description of PathORAM (which we denote as \pathoram) according to the formalism introduced.

\begin{construction}[\pathoram~{\cite[Definition~18]{GKK17}}]\label{dfn:pathoram}
For fixed parameters $\dsize,\maxsize \in \NN$, let $\ksize = \lceil \log_2 \maxsize \rceil,\zsize \in \NN,\msize = \dsize + \ksize,\bsize \geq \msize$. Let $\G$ be a PRNG outputting $\ksize$-bit values, and $\E=(\KGen,\Enc,\Dec)$ be a SKES with $\msize$-bit plaintexts and $\bsize$-bit ciphertexts. We define an ORAM construction called $\pathoram = \pathoram_{\E,\G}$ as follows:
\begin{itemize}
\item $\init(\secpar,\dbsize) \rightarrow (\C,\S)$ in the following way:
	\begin{algorithmic}[1]
	\State $\C$ generates a secret key $k\from\KGen$
	\State set $\tsize := \lceil \log_2 \dbsize \rceil$ \Comment{notice $\tsize \leq \ksize$}
	\State $\C$ initializes a position map of the form $\left((1,r_1),\ldots,(\dbsize,r_\dbsize)\right)$, where $r_i$ are $\tsize$-bit values generated 
	 by truncating bits from $\G$'s output
	\State $\S.\DB$ is stored in a binary tree of height $\tsize$, with root $\treeroot$ and leaves $\leaf_0,\ldots,\leaf_{2^\tsize-1}$, and such that:
		\begin{enumerate}
		\item each node of the tree stores up to $\zsize$ blocks;
		\item every block of every node is initialized to $\Enc_k(0^\ksize\|0^\dsize)$.
		\end{enumerate}
	\end{algorithmic}
\item If $\dr = (\op, i, \data)$, then $\access(\C,\S,\dr)\rightarrow (\C',\S',\com)$ as follows:
	\begin{algorithmic}[1]
	\State \C reads $r_i$ from his position map and sends it to \S
	\State $\S$ sends to $\C$ the path $\branch$ from $\treeroot$ to $\leaf_{r_i}$
	\State remap $(i,r_i)$ to $(i,r'_i)$ in the position map of $\C$, where $r'_i$ is a fresh pseudorandom $\tsize$-bit value (generated by truncating the first $\ksize-\tsize$ bits of $\G$'s output), obtaining~$\C'$
	\ForAll{$\block$ contained in $\branch$}
		\State $\C'$ decrypts $\Dec_k(\block) \rightarrow (j\|\data_j) \in \bin^\msize$, \newline \mbox{\ \ \ \ \ \ where $j\in\bin^\ksize, \data_j \in \bin^\dsize$}
		\If{$j=i$}
			\If{$\op = \text{`read'}$}
				\State $\C'$ reads $\data_j$ \Comment{$\C'$ now has {\em access} to $\data_j$}
			\ElsIf{$\op = \text{`write'}$}
				\State $\C'$ sets $\data_j = \data$ \Comment{$\block$ is updated}
			\EndIf
		\EndIf
		\State $\C'$ re-encrypts (re-randomizing) $\block$
		\State find in $\branch$ the common parent node $\node$ between $\leaf_{r_i}$ \newline \mbox{\ \ \ \ \ \ and $\leaf_{r_j}$, closer to the leaf level}
		\State set $b_\mathsf{swap} := \text{`false'}$
		\ForAll{$\block'$ in $\node$}\label{alg:repeatloop}
			\State $\C'$ decrypts $\Dec_k(\block') \rightarrow (j'\|\data'_j) \in \bin^\msize$ 
			\State $\C'$ re-encrypts (re-randomizing) $\block'' \from \Enc_k(j'\|\data'_j)$ 
			\If{$j'=0\ldots0$} \Comment{$\block''$ is empty, can be used}
				\State swap $\block$ and $\block''$
				\State set $b_\mathsf{swap} := \text{`true'}$
			\EndIf
		\EndFor
		\If{$b_\mathsf{swap} = \text{`false'}$} \Comment{no empty blocks in current $\node$}
			\If{$\node \neq \treeroot$}
				\State set $\node$ to be one level up in the tree (i.e., $\node$'s parent)
				\State go to step~\ref{alg:repeatloop}
			\Else
				\State store $\block$ in the $\stash$ \Comment{no empty blocks found}
			\EndIf
		\EndIf
		
	\EndFor
	\State $\C'$ sends back the updated tree branch, $\newbranch$, to $\S$
	\State update $\S.\DB$ with $\newbranch$, obtaining $\S'$
	\State produce $\com$, which contains $r_i, \branch, \newbranch$
	\end{algorithmic}
\end{itemize}

\end{construction}

In the above, we recap the meaning of the parameters as follows:

\begin{itemize}
\item $\secpar$ is the security parameter, used by the encryption scheme $\E$.
\item $\maxsize$ is the maximum number of blocks that the server's architechture can support (an upper bound to \S's tree storage).
\item $\dbsize$ is the maximum number of `real' blocks that the client \C wants to store (so, $\dbsize\leq\maxsize$). Unlike $\maxsize$ thus, $\dbsize$ can be chosen by the adversary in the security game.
\item $\ksize$ is the minimum number of bits that are needed to index all the `real' blocks in the limit scenario where $\dbsize=\maxsize$. Hence, $\ksize$ is also architecture-dependant, and not chosen by $\A$.
\item $\zsize$ is the maximum number of blocks that can be stored in every tree node. Lower values reduce the amount of memory used by \S to store the tree (for a fixed $\dbsize$), but increase the risk of using large amounts of memory by the client for the stash. This is a parameter of the particular \pathoram implementation: as we do not care about performance analysis here, we will leave $\zsize$ undefined, as any nonzero value works for us.
\item $\tsize$ is the minimum number of bits that are needed to index all the `real' $\dbsize$ blocks (hence, $\tsize\leq \ksize$). $\tsize$ also represents the minimum height of the tree necessary to store all blocks in the limit case $\zsize=1$.
\item $\dsize$ is the bit size of the data units used in the \pathoram implementation, and it is hence architecture-dependant.
\item $\msize$ is the total bit size of a data unit, plus the number of bits necessary to address the block where this data unit is encoded, so also this value is architecture-dependant. The encryption scheme $\E$ must be able to work with $\msize$-bit plaintexts.
\item $\bsize$ is the size of a ciphertext produced by the encryption scheme $\E$, and hence the total size of a block. The size of \S's tree storage memory is thus at most $\bsize \maxsize$ bits.
\end{itemize}

We now show the (classical) security of \pathoram.

\begin{theorem}[\!AP-IND-CQA Security of \pathoram]\label{thm:pathoramsec}
Let $\E\!=\!(\!\KGen,\!\Enc,\!\Dec)$ be an IND-CPA SKES, and let $\G$ be a PRNG. Then, \pathoram instantiated 
using $\E$ and $\G$ is an AP-IND-CPA secure ORAM.
\end{theorem}

\begin{proof}
By assumption, the outputs of $\G$ are indistinguishable from random. Therefore, in the following analysis, we can w.l.o.g. replace $\G$ with a real source of randomness.

Suppose that there exists an adversary $\A$ and a non-negligible $\ell$, such that:
$$
\Pr\left[ \gamePORAM = 1 \right] = \frac{1}{2} + \ell.
$$

We will use $\A$ in a black-box way to construct a PPT algorithm able to break the IND-CPA security of $\E$, against the assumption. The idea is to build an algorithm $\D$ which simulates a \pathoram client $\C$, playing the AP-IND-CQA game against $\A$ (w.l.o.g., we assume that $\A$ itself simulates the server $\S$, otherwise $\S$ can be also simulated by $\D$). Throughout the game, $\D$ also stores a copy of the server's database $\S.\DB$, in plaintext. This is allowed, as $\S.\DB$ is of size linear in $\dbsize$, and $\D$ is only simulating $\C$, so he is not limited by the storage constraints usually assumed in a normal ORAM client. Then $\D$ will use the interaction with $\A$ to win the IND-CPA game for scheme $\E$.

More in detail: first, $\D$ executes $\A$. Then $\A$ starts $\gamePORAM$ by choosing $\secpar$ and $\dbsize$, and $\D$ simulates a \pathoram client $\C$ created during $\init$, by initializing his own position map (populated with random values), but {\em without} generating a secret encryption key. Furthermore, $\D$ creates a tree memory structure of height $\tsize$, with leaves indexed $0,\ldots,2^\tsize-1$, where every node stores $\zsize$ plaintexts of bit size $\msize$, which are initialized to $(0^\ksize \| 0^\dsize)$ (the parameters are the same as in Construction~\ref{dfn:pathoram}). This structure will be used by \D to `mirror' $\S.\DB$ in cleartext throughout the execution of \pathoram.

$\D$ now starts $\gameINDCPA[\D]$, obtaining oracle access to $\Enc_k$ for an unknown secret key $k$, and 
choosing as security parameter the same $\secpar$ chosen by $\A$. At this point, notice that \D is able to perfectly simulate a valid client \C having access to the key $k$, in the following way:
\begin{itemize}
\item whenever \C downloads a branch of $\S.\DB$ identified by leaf $r$ by calling $\access$, \D does the same (although the blocks in such downloaded branch will be ignored, as we will see);
\item whenever \C decrypts a certain block in a downloaded branch, \D simulates the decryption oracle $\Dec_k$ by fetching the plaintext $(i\|\data)$ found at the corresponding position in the `mirrored' tree;
\item whenever \C swaps two blocks in a downloaded branch, \D swaps the two plaintexts found at the corresponding positions in the `mirrored' tree;
\item whenever \C encrypts a plaintext $(i\|\data)$ to obtain a new encrypted block, \D does so by using the encryption oracle $\Enc_k$ obtained from the IND-CPA game;
\item whenever \C updates his position map, or uploads an updated branch to $\S.\DB$, \D does the same.
\end{itemize}

Given the above, it is clear that now whenever $\A$ asks for the execution of a data request $\dr$, \D is able to simulate the correct communication transcript $\com$ and a correctly formed updated branch $\newbranch$. Therefore, for every data request performed during the first CQA phase, \A always receives the correct access pattern.

Eventually, at the challenge step \A produces two data requests $\dr^0,\dr^1$, and requests the execution of one of them. For $a \in \bin$, let $\dr^a = (\op^a, i^a, \data^a)$ be the two data requests 
and let $m^a \in \bin^\msize$ be formed as follows:
\begin{itemize}
\item if $\op^a = \text{`write'}$, then set $m^a = (i^a \| \data^a)$;
\item else, set $m^a = (i^a \| \data_{i^a})$, where $\data_{i^a}$ is retrieved by looking for identifier $i^a$ in the mirrored tree.
\end{itemize}

Now, it could happen that $m^0 = m^1$. For example, it might be that the two data requests are of the form $(\text{`write'}, i, \data)$ and $(\text{`read'}, i, \data')$ respectively, but $\block_i$ already encodes $\data$. If this happens we say that the challenge query is {\em non-meaningful}. It is easy to see that two data requests from a non-meaningful challenge query will produce the same statistical distributions of communication transcripts\footnote{Notice how this is not true anymore if the values in the position map are not totally random. Therefore, this step fails if the PRNG used is not secure.} and updated paths, because their effect on the database is equivalent. Therefore, since \A distinguishes the two resulting access patterns with non-negligible probability by assumption, it is clear that the challenge query must be {\em meaningful}, i.e., $m^0 \neq m^1$.

At this point \D executes the challenge IND query using $m^0, m^1$ as plaintexts, and receiving back an encryption $c \from \Enc_k(m^b)$ for a secret bit $b$. \D will also generate a random bit $b^* \rand \bin$ (a `guess'), and will answer \A's challenge query by simulating the execution of $\dr^{b^*}$ as in the CQA phase, but injecting $c$ as an updated block with identifier $i^{b^*}$ during the execution of $\dr^{b^*}$. Then \D keeps simulating \C during the second CQA phase as before, and waits until \A outputs a bit $\hat{b}$. Finally: if $\hat{b}=b^*$, then \D outputs $b^*$ in the IND-CPA game, otherwise \D outputs a new random bit.

Now, notice the following. In the case that \D's guess was correct, i.e., $b=b^*$, it means that $c$ was the right ciphertext at the right place, so that \A has received a correctly formed access pattern. This means that \A correctly guesses $\hat{b}=b^*$ with probability at least $\frac{1}{2} + \ell$, by assumption. In that case, also \D wins, so:
\begin{equation}\label{qoramwin}
\Pr \left[ \gameINDCPA[\D] = 1 \middle| b=b^* \right] \geq \frac{1}{2} + \ell.
\end{equation}
On the other hand, if $b\neq b^*$ we cannot say anything on \A's success probability, because now \A has a malformed access pattern. But we can say that, even if \A fails, \D still succeeds with probability $\frac{1}{2}$.
\begin{equation}\label{qoramfail}
\Pr \left[ \gameINDCPA[\D] = 1 \middle| b \neq b^* \right] \geq \frac{1}{2}.
\end{equation}
Thus, combining~\ref{qoramwin} and~\ref{qoramfail}, the reduction's overall 
success probability is:
$$
\Pr \left[ \gameINDCPA[\D] = 1 \right] \geq \frac{1}{2} + \frac{\ell}{2},
$$
which concludes the proof.
\endproof
\end{proof}

\chapter{QS1: Post-Quantum Security}\label{chap:QS1}

The next step in our analysis of quantum security notions is to consider what happens to classical encryption primitives when the adversaries have access to a quantum computing device. In this scenario, the cryptographic objects we are studying are still classical, as in the security class \QS0. However, since many constructions in \QS0 rely on computational hardness assumptions which do not hold anymore against quantum computers, new security models and constructions have to be considered in order to retain security in the new scenario. The branch of cryptography which aims at this goal has traditionally been called {\em post-quantum cryptography}. That is, post-quantum cryptography is about the security of {\em classical} primitives {\em after} (hence `post-') quantum computing becomes available\footnote{Admittedly, this naming is a bit misleading, because it might be meant as {\em `cryptography resistant against the more advanced model of computation which will conceivably come after quantum computing'}. We do not want to argue here about the term `post-quantum', which has become commonly accepted in the literature.}. The security class which we denote by \QS1 in our new labeling system covers this scenario.

But {\em how do we model post-quantum security exactly?} In the scientific community there has not always been mutual agreement on this. For example, one of the questions which most often cryptographers ask is: {\em ``When should one consider classical access to a function for a quantum adversary, and when should one consider quantum access instead?''}. As we will see, the answer to this question is: {\em ``Whenever the security model implies that the adversary computes the function on his local device, then quantum access should be used.''} We call this principle {\em the \QS1 principle}.

In this chapter we will discuss in detail the \QS1 principle and all the issues arising toward properly defining post-quantum security. Next, we introduce security models and definitions for post-quantum cryptographic primitives, starting from the very basic ones to more elaborated ones. We also discuss post-quantum assumptions, building blocks, and transformations from one primitive to another.

\subsection{My Scientific Contribution in this Chapter}

Theorem~\ref{thm:pqOWFtopqPRNG} is commonly considered folklore, but to the best of my knowledge the first fully formal proof appears in~\cite{ABF+16}, which is a joint work with Gorjan Alagic, Anne Broadbent, Bill Fefferman, Christian Schaffner, and Michael St. Jules.

All the material from Section~\ref{sec:FSQROM} appeared first in~\cite{DFG13}, which is a joint work with Özgür Dagdelen and Marc Fischlin.

Regarding post-quantum ORAMs, all of the results in Section~\ref{sec:pqORAM} first appeared in~\cite{GKK17}, which is a joint work with Nikolaos P. Karvelas and Stefan Katzenbeisser.

Finally, to best of my knowledge, the classification of quantum security reductions appearing in Section~\ref{sec:qred} has never been made explicit before, and it appears in this work for the first time, although single examples of any of those kind of reductions have appeared in the literature before.

\section{Issues in Post-Quantum Security}\label{sec:QS1basics}

Post-quantum security constructions are usually obtained by replacing some underlying hardness assumption with a different, quantum-hard assumption, and then repeating the construction process (i.e., the security proof) leading to the realization of a secure primitive as in \QS0. For example, when designing a post-quantum signature scheme, a natural option would be to consider a signature scheme in \QS0 based on, e.g., the DLP problem, and see if it is possible to obtain a new scheme by replacing the DLP problem with some other quantum hardness assumption, e.g., {\em learning with errors (LWE)} or {\em shortest vector problem (SVP)}. Alternatively, one could simply try to design a signature scheme from scratch by relying on a new security proof reducing the security of the scheme to the quantum hardness of one of the aforementioned mathematical problems. Traditionally, schemes produced by such approaches are labeled `post-quantum'. However, this labeling is sometimes inappropriate. The goal of this section is to give an overview of the many things that {\em could go wrong} when adopting too blindly the procedure described above, and to explain why one should take a more careful approach when defining meaningful notions of post-quantum security.

\subsection{Proof Failures}\label{sec:prooffailures}

The general issue when designing post-quantum primitives is that the classical security proofs might fail quantumly, even when only relying on quantum-hard assumptions. Common reasons for this are (but not limited to) the following.

\begin{itemize}
\item \textbf{No-Cloning:} when the security proof works by using the same value or element for two different purposes, care must be taken in making sure that this does not contradict the No-Cloning Theorem. If the element in question is a classical element, there is no problem. However, for quantum states, it is usually not possible to re-use the state for computing more than a single operation. Sometimes this can be solved by defining the operation in such a way that it does not destroy the input state.

\item \textbf{Memory Snapshots:} as a consequence of the previous point, problems may arise when the security proof requires recording a `snapshot' of an algorithm, or adversary, in order to execute it on different instances, or to analyze some internal area of memory. As the adversary is now a quantum machine, this cannot usually be done.

\item \textbf{Rewinding:} analogously, proofs that use rewinding are notoriously hard to translate to the quantum setting. Limited positive results have been achieved in this respect in the existing literature~\cite{QuantRewinding,WatrousZK}.

\item \textbf{Quantum Queries:} if the security proof requires `counting the number of queries' to a certain oracle, it will probably fail when the oracle is replaced by a quantum oracle. The reason is that a quantum oracle can, in some sense, be queried over {\em all} the domain elements at once.

\item \textbf{Lookup Tables:} analogously, if the proof requires storing a transcript of a protocol execution, including the query calls to some oracle, and if the oracle is quantum, problems may arise.

\item \textbf{Measurements:} conditional procedures such as ``if the value of $x$ is $y$, then do...'' are often an issue in the context of analyzing quantum states, because the information in the state is usually destroyed in the measurement process. This is particularly problematic when analyzing the values of queries to quantum oracles, or when comparing those values to those contained in some set.
\end{itemize}

Unfortunately, there is no general recipe to solve all of the above problems, and much of the existing literature erroneously advertises cryptographic constructions as `post-quantum' just because they are based on quantum-hard problems, without addressing the previous issues. We strongly argue against the use of the term `post-quantum' when describing the security of such constructions. Regardless, over the last few years many important tools have been developed in order to deal with these problems.

\subsection{Quantum-Classical Oracles}\label{sec:QS1princ}

The first important concept to define is what happens when an oracle $\O_\f$ computing a {\em classical function} $\f:\X \to \Y$ is invoked by a quantum algorithm. Two possible scenarios arise, depending on the {\em interaction}, or {\em access mode}, of the algorithm to the oracle:

\begin{enumerate}
\item the interaction is classical; in this case, the oracle is still a classical object which can be queried on classical inputs $x \in \X$ and returning outputs $y \in \Y$; or
\item the interaction is quantum; in this case the classical oracle $\O_\f$ must be replaced by a {\em quantum-classical oracle} (which we denote by $\ket{\O_\f}$).
\end{enumerate}

A quantum-classical oracle can be queried on a {\em quantum superposition of classical input values}, usually of the form:
$$
\sum_{x \in \X, y \in \Y} a_{x,y} \ket{x,y}, \text{ where } \sum_{x,y} |a_{x,y}|^2 = 1,
$$
and it returns a quantum state encoding somehow the evaluation of $\f$ on the inputs in the superposition query. The exact form of the input and output states can vary, and it depends on the type of quantum access considered, as mentioned in Section~\ref{sec:quantumoracles}. However, for most applications, and unless differently specified, we will denote by $\ket{\O_\f}$ the unitary operator acting as follows.

\begin{definition}[Canonical Quantum-Classical Oracle]\label{def:qcoracle}
Let \X,\Y be sets, and $\f:\X\to\Y$. The {\em (canonical) quantum-classical oracle for $\f$}, denoted by $\ket{\O_\f}$, is a unitary operator on $\Hilb{\X \otimes \Y}$, defined by:
$$
\ket{\O_\f} : \ket{x,y} \mapsto \ket{x,y \xor \f(x)}.
$$
\end{definition}
When not necessary to specify otherwise, in order to simplify notation we assume the ancilla register to be initialized with $\ket{0}$, so that:
$$
\ket{\O_\f} : \sum_{x \in \X} a_x \ket{x,0} \mapsto \sum_{x \in \X} a_x \ket{x,\f(x)}, \text{ where } \sum_{x} |a_{x}|^2 = 1.
$$

One important question regards quantum-classical oracles for randomized functions. For instance, if $\f$ is a randomized function, we can explicit the dependence from the randomness $r$ (sampled from some appropriate distribution \R) by writing: $y := \f(x;r)$. Then the question is: when considering $\ket{\O_\f}$, should we consider superpositions of evaluations using the same, fixed randomness $r$, or should we consider evaluations where a fresh new randomness $r$ is sampled for every element in the superposition? In other words, should we consider:
$$
\ket{\O_\f} : \sum_{x \in \X} a_x \ket{x,0} \mapsto \sum_{x \in \X} a_x \ket{x,\f(x;r)}, \text{ where } r \from \R,
$$
or should we consider the following instead?
$$
\ket{\O_\f} : \sum_{x \in \X, r \from \R} a_{x,r} \ket{x,0} \mapsto \sum_{x \in \X, r \from \R} a_{x,r} \ket{x,\f(x;r)}.
$$
As observed in~\cite{BZ13}, it turns out that the two cases are equivalent. The reason is that, using the first case, we can simulate the second case by first sampling a single $r$ from \R, and then applying a quantum-secure PRF (described in Section~\ref{sec:qPRF}) to generate independent pseudorandom values for every component of the superposition query. Because of the security properties of such PRF, the result would look the same to any \QPT adversary.

\subsection{Quantum Reductions}\label{sec:qred}

Another thing to discuss is the meaning of {\em quantum reductions}. As in the classical case, a quantum reduction \B from (the security of) a scheme $\Sigma$ to (the security of) a primitive, or (the hardness of) a problem $\Pi$, is an efficient algorithmic procedure which uses an hypothetical adversary \A against $\Sigma$ to attack $\Pi$. The existence of a reduction shows that: if an efficient adversary against $\Sigma$ exists, then an efficient algorithm breaking $\Pi$'s security must also exist. In this work we only consider black-box reductions, that is, reductions which do not have access to \A's or $\Sigma$'s internal code/circuit, but are only allowed to use the interactions between these components to attack $\Pi$.

Let us consider different possible scenarios in the quantum world. The following is a {\em classification} of possible (post-)quantum security reductions.

\begin{enumerate}
\item \A is classical but \B is quantum. In this case, \B is a \QPT algorithm using \A as a (classical) subroutine. These kind of reductions offer the weakest form of security guarantees because they basically say: ``if a {\em classical} adversary against $\Sigma$ exists, then a {\em quantum} algorithm breaking $\Pi$'s security exists''. They do not say anything about the possibility that a quantum adversary against $\Sigma$ might exist, so they are not really useful in our \QS1 setting. We call these {\em weak quantum reductions}.

\item\label{item:qr} \A is quantum and \B is quantum. This is the most common scenario. These reductions say: ``if a {\em quantum} adversary against $\Sigma$ exists, then a {\em quantum} algorithm breaking $\Pi$'s security exists''. In particular, this rules out classical adversaries against $\Sigma$, but the existence of any of these adversaries would not necessarily imply a {\em classical} algorithm against $\Pi$, only a quantum one. We call these {\em (standard) quantum reductions}.

\item\label{item:qcr} \A is quantum but \B is classical. These reductions offer the strongest security guarantees, because they say: ``if a {\em quantum} adversary against $\Sigma$ exists, then a {\em classical} algorithm breaking $\Pi$'s security exists, with only black-box access to the adversary''. Not only this rules out quantum and classical adversaries alike, but it also implies that the post-quantum security of $\Sigma$ can rely just on the post-quantum security of $\Pi$, so that in particular one does need to worry about oracle access modes. We call these {\em semi-classical reductions}.
\end{enumerate}

Finally, it should be discussed what `black-box' in the quantum setting means. Classically, this means that \B is allowed to interact with \A without accessing \A's internal code or state. In other words, \B can only act on \A's inputs, outputs, and oracle queries. Furthermore, in cryptographic reductions, one usually has to make sure that \B's action is computationally undetectable for \A, which means that the probability that \A's output is affected by this action is negligible. This is important, for example, in the case that \B injects or reads values inside \A's queries to an oracle.

In the quantum setting, we adopt the same principle: \B can tamper with \A's inputs, outputs, and queries, as long as \A's behaviour is only negligibly affected. So, for example, \B could measure (fully or partially) \A's queries to some quantum oracle, and even modify the queries and reprogram the oracle, as long as it can be proven that this action does not disturb \A's working behaviour too much.

However, one could also take a stricter approach. Since measuring unknown quantum states might destroy the information therein, we could also consider quantum reductions that {\em do not measure external quantum states at all}, and only rely on the classical interactions with \A (or other oracles) instead. For example, in the case of quantum oracle queries, such reductions would ignore those queries, and only interact classically with the (quantum) adversary. Clearly, these `careful' reductions are quite powerful, because they work even when ignoring some potential source of information (the quantum queries). They basically say: ``if a {\em quantum} adversary against $\Sigma$ exists, then a {\em quantum} algorithm breaking $\Pi$'s security exists, by using {\em only classical access} to some external quantum resources''. These kind of reductions are placed somewhere between points~\ref{item:qr} and~\ref{item:qcr} of the above hierarchy, and we call them {\em strong quantum reductions}.

\section{The Quantum Random Oracle Model}\label{sec:QROM}

One archetypical example of where the \QS1 principle comes into play is the {\em Quantum Random Oracle Model (QROM)}. Recall that, in \QS0, the Random Oracle Model (ROM) is a computation model where all parties have access to an oracle $\RO$ computing a function $\h$ picked uniformly at random from the set of all functions from some domain \X to some range \Y. This model is useful when analyzing the security of schemes employing PRFs or hash functions. In other words, the (truly) random function \h is just an abstraction, or a model, for a real-world function \g which we assume {\em behaving} like a random function.

But this also means that the random oracle \RO itself is an abstract model for the computation of the real-world, algorithmic function \g, {\em performed on some computer.} And since the code for \g is public, and can be run by anyone (after all, in the ROM the access to \RO is given to every participant in the scheme because of this reason), it is necessary to assume that a quantum adversary could implement the circuit computing \g on his quantum computer, therefore being able to query \g quantumly. Therefore, in the {\em Quantum Random Oracle Model (QROM)}, the random oracle \RO must be replaced by a {\em quantum random oracle} \QRO. It is important to stress the fact that there exist models where security is proven in the random oracle model against quantum adversaries. We strongly argue against the use of the term `post-quantum' when referring to those models.

So, in other words, in \QS1 the ROM {\em must} be replaced by the QROM, where every \QPT algorithm has access to a quantum oracle:
$$
\QRO : \ket{x,y} \mapsto \ket{x,y \xor \h(x)}.
$$
and where \h is chosen uniformly at random from the set of all functions from \X to \Y, as in the random oracle model.

\subsection{QROM Emulation}\label{sec:QROMemu}

Notice the following difficulty when defining the QROM operationally. Classically, as explained in Section~\ref{sec:rom}, during a cryptographic reduction a random oracle is {\em emulated} by a \PPT algorithm, for example through lazy sampling. But lazy sampling cannot work for quantum random oracles, for two reasons.

First of all, a single quantum query to \QRO could require the emulator to lazy-sample too many elements. E.g., a query of the form:
$$
\sum_{x \in \bin^\secpar} \frac{1}{\sqrt{2^\secpar}} \ket{x,0}
$$
would query all the exponentially-many input values at once, and so it would `force' the emulator to `fix' all those values at the same time. This is not compatible with what we require from an efficient cryptographic reduction.

The second problem is that the concept of lookup table, used in the classical ROM to answer consistently with the previous queries, becomes meaningless. Firstly because such table could quickly reach exponential size, as the previous query example shows; and secondly because, as discussed in Section~\ref{sec:QS1princ}, there might be no way to check whether the values of some query are in the table or not without destroying the query.

Luckily, there exist a few other techniques to solve the above issues and to make the QROM a meaningful tool in \QS1. If the number of queries performed by the adversary to the QRO is known a priori, then the QRO can be efficiently emulated by {\em $d$-wise independent functions}. These are families of functions that are statistically indistinguishable from random functions if queried (classically) no more than $d$ times. An example are polynomial functions of degree $d-1$. It is known~\cite{Zhandry12} that no quantum algorithm performing at most $q$ queries can distinguish between random oracles and distributions of $2q$-wise independent functions.

Another common technique is to emulate a RO with a PRF, which is useful if one does not know a priori an upper bound to the number of adversarial queries. In the QROM we need something analogous, but classical PRFs alone cannot work. One idea might be to use {\em post-quantum PRFs} (we will define them in the next section), but actually for emulating a QRO, classical access to the PRF is not enough, so we need something more: {\em quantum-secure (superposition-secure) PRFs} will be defined in the next chapter.

\subsection{QROM Reprogramming}

It is important to analyze what happens when reprogramming a quantun random oracle \QRO. In particular, a useful technique often consists in {\em injecting} some fixed value $y$ for a subset $\S \subset \X$ of possible input query values, so that $\h(x) := y$ for all $x \in \S$. Intuitively, if the set $S$ is `very small', it is going to be very hard for a quantum algorithm to distinguish the modified oracle from a true QRO. However, some proofs might rely explicitly on the probability of the adversary querying one of those values, so it is important to have a detailed quantitative analysis for these probabilities.

We start by recalling~\cite{ZhandryPRF} a tool known as {\em semi-constant distributions}.

\begin{definition}[Semi-Constant Distributions]\label{def:scdist}
Let $\H := \set{\h : \X \to \Y}$ be the family of functions between sets $\X$ and $\Y$, and let $\delta \in [0,1]$. We define the {\em $\delta$-fraction semi-constant distribution} $\scDist$ as the distribution over $\H$ resulting from the following procedure:
\begin{algorithmic}[1]
\State sample $y \rand \Y$
\ForAll{$x \in \X$}
	\State $p \rand [0,1]$
	\If{$p \leq \delta$}
		\State define $\h(x) := y$
	\Else
		\State sample $y' \rand \Y$
		\State define $\h(x) := y'$
	\EndIf
\EndFor
\State \textbf{Return:} $\h$
\end{algorithmic}
\end{definition}

Notice that $\U^0$ is the uniform distribution, while $\U^1$ is a constant distribution. Also note that the distribution, when used within an oracle, is consistent in the sense that the settings are chosen once at the outset. We will use this definition to describe a QRO which has been `reprogrammed' on a fraction $\delta$ of its possible inputs. The following lemma~\cite{Zhandry12} gives an upper bound on the probability that a quantum algorithm's behavior changes when switching from a truly QRO to a quantum oracle for a function drawn from $\scDist$ in terms of statistical distance.

\begin{lemma}[{\cite[Corollary 4.3]{Zhandry12}}]\label{lem:zhandrycustom}
Let $\A^\QRO$ be a \QPT algorithm making at most $q_h$ queries to the quantum random oracle $\QRO$. Let $\delta \in (0,1)$ and let $\scQRO$ be the classical-quantum oracle obtained by reprogramming $\RO$ on a fraction $\delta$ of its possible inputs, i.e., $\scQRO$ is described by the semi-constant distribution $\scDist$. Then, the following holds:
$$
\left| \A^\QRO - \A^\scQRO \right| \leq \frac{8}{3}\cdot q_h^4 \cdot  \delta^2.
$$
\end{lemma}

The above lemma is quite general, because it does not take into account the specific values where the reprogramming happens, but just a generic fraction $\delta$ of all possible values. Therefore, it is especially useful in those cases where the quantum random oracle is reprogrammed randomly, i.e., by just replacing some of its values with a certain probability $\delta$. However, in all those cases where it is possible to track the specific amplitudes (across the oracle queries) of the elements to be reprogrammed, then one can usually find better bounds, for example by using Lemma~\ref{lem:queryprob}.

\section{Post-Quantum Assumptions, Building Blocks}\label{sec:pqbuildingblocks}

In this section we redefine the basic assumptions and building blocks for the post-quantum setting.

\subsection{Post-Quantum OWFs}

As in the \QS0 case, the existence of {\em post-quantum one-way functions (pqOWF)} is a basic security assumptions. Because a OWF's code is public, and recalling the \QS1 principle, we expect quantum adversaries to be able to query a OWF on a superposition of values. However, for the same reason, since in the definition of OWF the quantifier is {\em `for all'} \PPT algorithms, without mentioning oracle access, it is enough to define post-quantum OWFs by just replacing \PPT adversaries with \QPT adversaries.

\begin{definition}[Post-Quantum One-Way Functions (pqOWF) and Permutations (pqOWP)]\label{def:pqowfowp}
Let $\F = \family{\F}$ be a \DPT algorithm, with $\F_\secpar : \X_\secpar \to \words$. \F is a (family of) {\em post-quantum one-way functions (pqOWF)} iff for any \QPT algorithm \A it holds:
$$
\Pr_{x \rand \X} \left[ \A(\F(x)) \to x' : \F(x) = \F(x') \right] \leq \negl.
$$
Moreover, in the special case where $\F_\secpar : \X_\secpar \to \X_\secpar$ are permutations on $\X_\secpar$ for every \secpar, \F is a (family of) {\em post-quantum one-way permutations (OWP)}.
\end{definition}

The definition of {\em post-quantum hard-core predicates} is as in the \QS0 case.

\begin{definition}[Post-Quantum Hard-Core Predicate]\label{def:pqhc}
Let $\F: \X \to \Y$ be a OWF. A polynomial-time computable function $\hc_\F : \X \to \bin$ is a {\em post-quantum hard-core predicate of \F} iff, for any \QPT algorithm \A it holds:
$$
\Pr_{x \rand \X} \left[ \A(\F(x)) \to \hc_\F(x) \right] \leq \half + \negl.
$$
\end{definition}

\begin{proposition}\label{prop:pqhc}
Let \F be a pqOWF (resp., pqOWP). Then it is possible to efficiently transform \F into a pqOWF (resp., pqOWP) \H such that at least one post-quantum hard-core predicate $\hc_\H$ exists.
\end{proposition}

Given the above, from now on for simplicity we assume that every pqOWF admits post-quantum hard-core predicates. In the case that $\F:\X\to\X$ (in particular, if \F is a pqOWP), the construction of hard-core bits can be iterated as in Proposition~\ref{prop:hcmult}.

\subsection{Post-Quantum OWTPs}

The same discussion in the case of post-quantum OWFs applies for the assumption of the existence of {\em post-quantum one-way trapdoor permutations (pqOWTP)}. As usual, we express a family of pqOWTPs as indexed through efficiently sampleable index family \I and associated trapdoor space \T.

\begin{definition}[Post-Quantum One-Way Trapdoor Permutation (pqOWTP)]\label{def:pqowtp}
A (family of) {\em post-quantum one-way trapdoor permutations (pqOWTP)} is a tuple $(\Gen,\Eval,\Invert)$ of \PPT algorithms:
\begin{enumerate}
\item $\Gen: \to \I \times \T$;
\item $\Eval: \I \times \X \to \X$;
\item $\Invert: \I \times \T \times \X \to \X \cup \set{\bot}$,
\end{enumerate}
and such that:
\begin{enumerate}
\item for any \QPT algorithm \A it holds:
$$
\Pr_{\substack{x\!\rand\!\!\X \\(i,t) \from \Gen}} \left[ \A(i,\Eval(i,x)) \to x \right] \leq \negl; \text{ and}
$$
\item $\Invert(i,t,y) = \Eval(i,x), \forall x \in \X , \foral (i,t) \from \Gen , \foral y \from \Eval(i,x)$.
\end{enumerate}
\end{definition}

As in the \QS0 case, the existence of pqOWTP implies the existence of pqOWP and pqOWF.

\begin{proposition}[pqOWTP $\implies$ pqOWP $\implies$ pqOWF]\label{prop:pqOWTPtopqOWF}
Let $\P := (\Gen,\Eval,\Invert)$ be a pqOWTP on \X. Then, for all but a negligible fraction of possible sequences $\left( (i_\secpar,t_\secpar) \right)_\secpar$ of outputs of $\Gen(\secpar) \implies \Eval(i_\secpar,.)$ is a pqOWP (and hence a pqOWF) on $\X$.
\end{proposition}

\subsection{Post-Quantum PRNGs}

Again, the same principle from OWF and OWTP applies when translating PRNGs to the post-quantum setting. Remember that the security property for PRNGs does not mention any kind of oracle access or code emulation, but it just says that no efficient adversaries, by looking at the stream of (classical) values output by the PRNG, can distinguish such stream from a random stream. So, the interaction is still classical, and the only change is that the adversary is now a quantum algorithm.

\begin{definition}[Post-Quantum PRNG (pqPRNG)]\label{dfn:pqPRNG}
Let $\p$ be a polynomial such that $\p(\secpar) \geq \secpar+1, \forall \secpar \in \NN$. A {\em post-quantum pseudorandom number generator (pqPRNG)} with expansion factor $\p$ is a \DPT algorithm $\PRNG$ such that:
\begin{enumerate}
\item given as input a bit string $s \in \bin^\secpar$, (the {\em seed}), outputs a bit string $\PRNG(s) \in \p(\secpar)$; and
\item for any \QPT algorithm $\D$:
$$
\left| \Pr \left[ \D(r) \to 1 \right] - \Pr \left[ \D(\PRNG(s)) \to 1 \right] \right| \leq \negl,
$$
where $r\rand\bin^{\p(\secpar)}, s\rand\bin^\secpar$, and the probabilities are taken over the choice of $r$ and $s$, and the randomness of $\D$.
\end{enumerate}
\end{definition}

Moreover, as noticed in Section~\ref{sec:PRNG}, the proof of Theorem~\ref{thm:OWFtoPRNG} still goes through in the post-quantum scenario, because it does not make any assumption on the query capabilities of the adversary.

\begin{theorem}[{\cite[Lemma 19]{ABF+16}}]\label{thm:pqOWFtopqPRNG}
If \F is a pqOWF, then $\PRNG_\F$ (defined as in Construction~\ref{constr:goldreichlevin}) is a pqPRNG.
\end{theorem}

\begin{corollary}[pqOWF $\iff$ pqPRNG]\label{cor:pqOWFiffPRNG}
pqOWFs exist iff pqPRNGs exist.
\end{corollary}

Clearly, a pqPRNG it is also a PRNG. However, the opposite is not believed to hold, as the following example shows.

\begin{lemma}\label{lem:blummicali}
Under the DLP hardness assumption, there exists a PRNG $\PRNGBM$ which is {\em quantumly predictable}. I.e., there exists a non-negligible function $\delta$ and a \QPT algorithm $\D$ which, on input $n$ sequential values output by $\PRNGBM$ on any random seed, predicts the $(n+1)$-th value output by $\PRNGBM$ with probability at least $\delta(n)$.
\end{lemma}
\begin{proof}
A counterexample $\PRNGBM$ is the modular exponentiation Blum-Micali generator~\cite{KatzLindell}, but many other similar variants work as well~\cite{GuedesAL13}. This construction is based on exponentiation of a public generator $g$ modulo a public large prime $p$, and it is a classically secure PRNG under the assumption that computing discrete logarithms is computationally hard. More specifically, if $s_i$ is the current state of the generator, one output bit is computed as a hardcore predicate of the value $s_{i+1} = g^{s_i} \mod p$ (where $s_{i+1}$ becomes the next state of the generator). Thus, starting from a secret seed $s_0$, a pseudorandom bit string can be generated by applying iteratively the procedure.

However, there exists a quantum attack~\cite{GuedesAL13} (based on variants of both Shor's and Grover's algorithms) which, given $p,g$ and a sequence $(r_1,\ldots,r_\secpar)$ of values output by $\PRNGBM$, can recover the initial state $s_0$ with probability $\delta$ non-negligible in $\secpar$. This, in turns, allows to predict the whole sequence of $\PRNGBM$. 
\endproof
\end{proof}

\subsection{Post-Quantum PRFs}\label{sec:pqPRF}

The case of pseudorandom functions, instead, is a bit different. Definition~\ref{def:PRF} specifically conditions the existence of (classical) PRFs to the query capabilities of the adversary, so we should make a distinction whether, in the post-quantum case, these queries should still be classical or not.

The \QS1 principle comes handy here. In a reasonable security model, should the adversary be able to implement the code of the PRF on his local computing device? The answer is: {\em ``normally, no, because he does not know the secret key''}. After all, the whole point of a PRF is that the adversary should not be able to distinguish the output of the PRF from the output of an (abstractly defined) completely random function, which in particular means that the adversary should not be able to see the PRF's code, because there might be {\em no code at all}. This is in striking contrast with the QROM, and the reason is that a QRO models a {\em public hash function}, which everybody can compute, while a PRF exists {\em as long as the key remains secret}.

In other words, {\em post-quantum pseudorandom functions (pqPRFs)} are defined by merely replacing the \PPT adversary with a \QPT adversary, and keeping the oracle access classical. {\em Quantum-secure PRFs} instead, as defined in~\cite{QROM,ZhandryPRF}, are a different object, and they will be presented in the next chapter in the context of the domain \QS2.

\begin{definition}[Post-Quantum Pseudorandom Function (pqPRF)]\label{def:pqPRF}
A (family of) {\em post-quantum pseudorandom functions (pqPRF)} from \X to \Y with key space \K is a \DPT algorithm $\PRF: (k \in \K, x \in \X) \mapsto y \in \Y$ such that for any \QPT algorithm \D it holds:
$$
\left| \Pr_{k \rand \K} \left[ \D^{\PRF_k} \to 1 \right] - \Pr_{\h \rand \Y^\X} \left[ \D^{\O_\h} \to 1 \right] \right| \leq \negl,
$$
where $\O_\h$ is an oracle for \h (i.e., a random oracle), and the probabilities are over the choice of $k$ and $\h$, and the randomness of $\D$.
\end{definition}

Moreover, the same proofs for Theorems~\ref{thm:PRFimPRNG} and~\ref{thm:PRNGimPRF} go through unchanged, because we are not modifying the oracle access mode, but just the adversary computation model. As a consequence, we can state the following.

\begin{theorem}[pqPRF $\iff$ pqPRNG]\label{thm:pqPRFiffpqPRNG}
pqPRFs exist iff pqPRNGs exist.
\end{theorem}

\begin{corollary}\label{cor:pqOWFiffpqPRF}
pqOWF exist iff pqPRF exist.
\end{corollary}

\subsection{Post-Quantum PRPs}

The case of post-quantum PRPs is analogous to the one for pqPRFs.

\begin{definition}[Post-Quantum Weak PRP (pqWPRP)]\label{def:pqWPRP}
A (family of) {\em post-quantum weak pseudorandom permutations (pqWPRP)} on \X with key space \K is a pair of \DPT algorithms $(\PRP,\PRP^{-1}): (k \in \K, x \in \X) \mapsto x' \in \X$ such that:
\begin{enumerate}
\item $\forall k \in \K \implies \PRP_k, \PRP^{-1}_k$ are permutations on $\X$;
\item $\forall k \in \K \implies (\PRP_k)^{-1} = \PRP^{-1}_k$; and
\item for any \QPT algorithm \D it holds:
$$
\left| \Pr_{k \rand \K} \left[ \D^{\PRP_k} \to 1 \right] - \Pr_{\p \rand S(\X)} \left[ \D^{\O_\p} \to 1 \right] \right| \leq \negl,
$$
where $\O_\p$ is an oracle for \p, and the probabilities are over the choice of $k$ and $\p$, and the randomness of $\D$.
\end{enumerate}
\end{definition}

\begin{definition}[Post-Quantum Strong PRP (pqSPRP)]\label{def:pqSPRP}
A (family of) {\em post-quantum strong pseudorandom permutations (pqSPRP)} on \X with key space \K is a pair of \DPT algorithms $(\PRP,\PRP^{-1}): (k \in \K, x \in \X) \mapsto x' \in \X$ such that:
\begin{enumerate}
\item $\forall k \in \K \implies \PRP_k, \PRP^{-1}_k$ are permutations on $\X$;
\item $\forall k \in \K \implies (\PRP_k)^{-1} = \PRP^{-1}_k$; and
\item for any \QPT algorithm \D it holds:
$$
\left| \Pr_{k \rand \K} \left[ \D^{\PRP_k,\PRP^{-1}_k} \to 1 \right] - \Pr_{\p \rand S(\X)} \left[ \D^{\O_\p,\O_{\p^{-1}}} \to 1 \right] \right| \leq \negl,
$$
where $\O_\p$ is an oracle for \p, $\O_{\p^{-1}}$ is an oracle for $\p^{-1}$, and the probabilities are over the choice of $k$ and $\p$, and the randomness of $\D$.
\end{enumerate}
\end{definition}

When left unspecified, by `pqPRP' we mean the strong version. A pqPRP is clearly also a pqPRF, but the converse does not necessarily hold. Again, as we are not modifying the oracle access mode, the classical constructions of PRPs from PRFs go through unchanged in the post-quantum setting. Therefore, the existence of pqPRPs is also equivalent to the existence of pqOWFs.

\begin{theorem}[pqPRF $\iff$ pqPRP]\label{thm:pqPRFiffpqPRP}
pqPRFs exist iff pqPRPs exist.
\end{theorem}

\section{Post-Quantum Encryption}\label{sec:QS1enc}

{\em Post-quantum encryption schemes} are classical encryption schemes meant to retain their security also against quantum adversaries. It is common for this scenario to just assume the same definitions and security notions we saw in Chapter~\ref{chap:QS0}, and just replacing \PPT adversaries with \QPT ones. However, in the case of public-key encryption, one must be a bit careful in doing so.

\subsection{Post-Quantum Secret-Key Encryption}

Following the \QS1 principle, in {\em post-quantum secret-key encryption} one can just `blindly' replace classical adversaries with quantum ones, because the adversary itself is never supposed to run encryption or decryption procedures locally (after all, he does not have the secret key). So we discuss here the modified security definitions as follows (we do it just for the IND and IND-CPA notions, but the same procedures yields equivalent post-quantum security notions for SEM, IND-CCA1, and IND-CCA2). As usual, $\E := \E_{\K,\X,\Y} := (\KGen,\Enc,\Dec)$ denotes a SKES with plaintext space \X, ciphertext space \Y, and key space \K.

\begin{definition}[Post-Quantum IND Adversary]\label{def:pqINDadv}
Let $\E$ be a SKES. A {\em post-quantum IND (pq-IND) adversary \A for \E} is a pair of \QPT algorithms $\A := (\M,\D)$, where:
\begin{enumerate}
\item $\M: \to \X \times \X \times \Hilbert$ is the {\em pq-IND message generator};
\item $\D: \Y \times \Hilbert \to \bin$ is the {\em pq-IND distinguisher},
\end{enumerate}
where \Hilbert is a Hilbert space of appropriate dimension, modeling the state communication register between \M and \D.
\end{definition}

\begin{experiment}[$\gamepqIND$]\label{expt:pqIND}
Let $\E$ be a SKES, and $\A:=(\M,\D)$ a pq-IND adversary. The {\em pq-IND experiment} proceeds as follows:
\begin{algorithmic}[1]
\State \textbf{Input:} $\secpar \in \NN$
\State $k \from \KGen$
\State $(m^0,m^1,\ket{\state}) \from \M$
\State $b \rand\bin$
\State $c \from \Enc_k(m^b)$
\State $b' \from \D(c,\ket{\state})$
\If{$b = b'$}
	\State \textbf{Output:} $1$
\Else
	\State \textbf{Output:} $0$
\EndIf
\end{algorithmic}
The {\em advantage of \A} is defined as:
$$
\advpqIND := \Pr \left[ \gamepqIND \to 1 \right] - \half .
$$
\end{experiment}

\begin{definition}[Post-Quantum Indistinguishability (pq-IND)]\label{def:pqIND}
A SKES $\E$ has {\em post-quantum indistinguishable encryptions (or, it is pq-IND secure)} iff, for any pq-IND adversary $\A$ it holds that: $\advpqIND \leq \negl$.
\end{definition}

\begin{experiment}[$\gamepqINDCPA$]\label{expt:pqINDCPA}
Let $\E$ be a SKES, and $\A:=(\M,\D)$ a pq-IND adversary. The {\em pq-IND-CPA experiment} proceeds as follows:
\begin{algorithmic}[1]
\State \textbf{Input:} $\secpar \in \NN$
\State $k \from \KGen$
\State $(m^0,m^1,\ket{\state}) \from \M^{\Enc_k}$
\State $b \rand\bin$
\State $c \from \Enc_k(m^b)$
\State $b' \from \D^{\Enc_k}(c,\ket{\state})$
\If{$b = b'$}
	\State \textbf{Output:} $1$
\Else
	\State \textbf{Output:} $0$
\EndIf
\end{algorithmic}
The {\em advantage of \A} is defined as:
$$
\advpqINDCPA := \Pr \left[ \gamepqINDCPA \to 1 \right] - \half .
$$
\end{experiment}

\begin{definition}[Post-Quantum Indistinguishability of Ciphertexts under Chosen Plaintext Attack (pq-IND-CPA)]\label{def:pqINDCPA}
A SKES $\E$ has {\em post-quantum indistinguishable encryptions under chosen plaintext attack (or, it is pq-IND-CPA secure)} iff, for any pq-IND adversary $\A$ it holds that: $\advpqINDCPA \leq \negl$.
\end{definition}

Clearly, pq-IND-CPA is at least as strong as IND-CPA.

\begin{theorem}[pq-IND-CPA $\implies$ IND-CPA]\label{thm:pqINDCPAtoINDCPA}
If a SKES is pq-IND-CPA secure, then it is also IND-CPA secure.
\end{theorem}

It is common folklore that, unlike some PKES, the most widely used constructions for SKES are actually also post-quantum secure. However, the converse of Theorem~\ref{thm:pqINDCPAtoINDCPA} does not hold, and it is important to remember that post-quantum notions for SKES are actually {\em strictly stronger} than the classical ones in \QS0.

\begin{theorem}[IND-CPA SKES \nimplies pq-IND-CPA SKES]\label{thm:INDCPAnotopqINDCPA}
Under standard hardness assumptions, there exist SKES which are IND-CPA secure, but not pq-IND-CPA secure.
\end{theorem}
\begin{proof}[Proof (sketch)]
It is sufficient to consider an IND-CPA SKES which appends to every ciphertext the secret key used, encrypted with another, IND-CPA but non--post-quantum secure PKES (e.g., some RSA variant) under a fixed, known public key. With the knowledge of the public key, a quantum adversary can emulate a quantum oracle for the encryption of the PKES, which can then be broken by, e.g., Shor's algorithm, thus revealing the SKES's secret key.
\end{proof}

The Goldreich scheme from Construction~\ref{constr:goldreich} is pq-IND-CPA when instantiated with a pqPRF, because the same arguments used in Theorem~\ref{thm:GoldreichINDCPA} go through as long as the adversary is unable to distinguish the PRF from a real source of randomness.

\begin{theorem}\label{thm:pqGoldreich}
Let $\E_\PRF$ be the SKES from Construction~\ref{constr:goldreich} implemented through a pqPRF \PRF. Then $\E_\PRF$ in a pq-IND-CPA SKES.
\end{theorem}

The same relations and separations examples between pq-IND, pq-IND-CPA, pq-IND-CCA1, and pq-IND-CCA2, hold as from Section~\ref{sec:SKES}, 
and with analogous separation examples from their classical counterparts as in Theorem~\ref{thm:INDCPAnotopqINDCPA}. 
Therefore, the relations between security notions for SKES in \QS0 and \QS1 are as summarized in Figure~\ref{fig:QS1relations}.

\begin{figure}[t]
\begin{center}
\includegraphics[width=\textwidth]{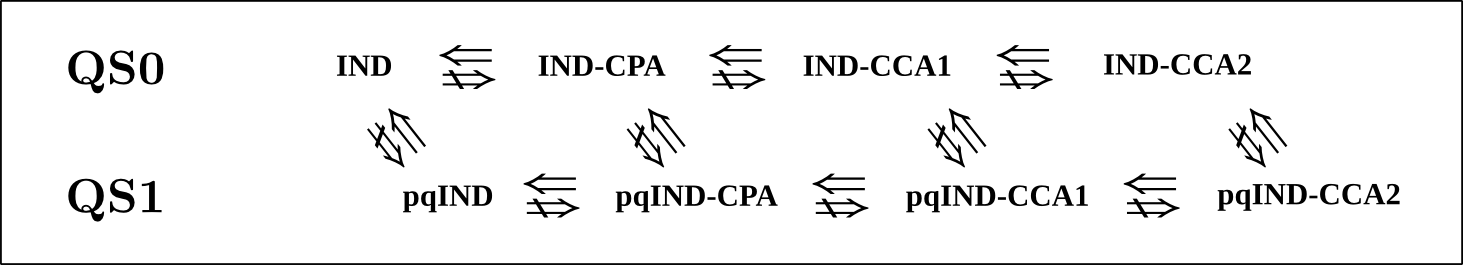}
\end{center}
\caption{Relations for SKES security notions in \QS0 and \QS1.}
\label{fig:QS1relations}
\end{figure}

\subsection{Post-Quantum Public-Key Encryption}\label{sec:pqPKES}

In {\em post-quantum public-key encryption schemes} the situation is quite different. The reason is that, in this case, the presence of a public-key allows the adversary to compute encryptions autonomously. In this scenario, following the \QS1 principle, the encryption oracle $\Enc_\pk$ should be replaced by the quantum counterpart $\ket{\Enc_\pk}$. However, this is only true for the {\em learning phases} during the security game (recall that, for PKES, IND security alone does not constitute a meaningful notion). The IND phase, on the other hand, models the attack of the adversary against the encryption of some unknown message, encryption that, therefore, is performed by some {\em classical} third party (the {\em IND challenger}). Moreover, as \M and \D are \QPT algorithms, giving them the public key \pk as input automatically implies access to $\ket{\Enc_\pk}$.

The resulting post-quantum IND-CPA security game is modified as follows.

\begin{experiment}[$\gamepqINDCPA$ for PKES]\label{expt:pqINDCPApk}
Let $\E$ be a PKES, and $\A:=(\M,\D)$ a pq-IND adversary. The {\em pq-IND-CPA experiment} (in the post-quantum public-key setting) proceeds as follows:
\begin{algorithmic}[1]
\State \textbf{Input:} $\secpar \in \NN$
\State $(\pk,\sk) \from \KGen$
\State $(m^0,m^1,\ket{\state}) \from \M(\pk)$
\State $b \rand\bin$
\State $c \from \Enc_\pk(m^b)$
\State $b' \from \D(c,\ket{\state},\pk)$
\If{$b = b'$}
	\State \textbf{Output:} $1$
\Else
	\State \textbf{Output:} $0$
\EndIf
\end{algorithmic}
The {\em advantage of \A} is defined as:
$$
\advpqINDCPA := \Pr \left[ \gamepqINDCPA \to 1 \right] - \half .
$$
\end{experiment}

Notice how only the encryption oracle during the learning phases is replaced by a quantum oracle, but it is still classical during the IND phase. This notion was introduced in~\cite{BZ13}, but we will discuss more the implications of this important difference in Section~\ref{sec:QS2enc}. Also notice how $\gamepqINDCPA=\gamepqIND[\A^{\ket{\Enc_\pk}}]$ only holds for the public-key setting.

The security notions pq-IND-CCA1 and pq-IND-CCA2 in the public-key setting are a straightforward modification of the ones for the SKES case, 
by giving to the adversary quantum oracle access to $\ket{\Enc_\pk}$ -- but the oracle $\Dec_\sk$ remains classical. It is well-known that certain PKES which are IND-CPA secure under standard assumptions are {\em not} pq-IND-CPA secure (examples are RSA, ElGamal EC-based schemes, etc.) Instead, pq-IND-CPA (or stronger) PKESs can be constructed under other quantum-hardness assumptions, as discussed in Section~\ref{sec:assumptions}.

\section{Post-Quantum Signatures}\label{sec:QS1sig}

In the case of {\em post-quantum signature schemes}, as the oracle access to $\Sign_\sk$ is kept classical according to the \QS1 principle, the definition of existential unforgeability is modified in the standard post-quantum way, e.g., by merely replacing \PPT adversaries with \QPT ones.

\begin{experiment}[$\gamepqEUFCMA$]\label{expt:pq-EUF-CMA}
Let $\Sigscheme$ be a DSS, and \A a \QPT algorithm. The {\em pq-EUF-CMA experiment} proceeds as follows:
\begin{algorithmic}[1]
\State \textbf{Input:} $\secpar, q_s \in \NN$
\State $(\pk,\sk) \from \KGen$
\State $(x,\sig) \from \A^{\Sign_\sk}(\pk)$ after making at most $q_s$ queries to $\Sign_\sk$, receiving signatures $(x_1,\sig_1) , \ldots (x_{q_s}, \sig_{q_s})$
\If{$\SVer(\pk,x,\sig) = 1$ and $\x \neq \x_i \foral i = 1,\ldots,q_s$}
	\State \textbf{Output:} $1$
\Else
	\State \textbf{Output:} $0$
\EndIf
\end{algorithmic}
The {\em advantage of \A} is defined as:
$$
\advpqEUFCMA(\secpar,q_s) := \Pr \left[ \gamepqEUFCMA(\secpar,q_s) \to 1 \right].
$$
\end{experiment}

\begin{definition}[Post-Quantum Existential Unforgeability under Chosen Message Attack (pq-EUF-CMA)]\label{def:pqEUFCMA}
A DSS $\Sigscheme$ is {\em post-quantum existentially unforgeable under chosen message attack (or, it is pq-EUF-CMA secure)} iff, for any \QPT algorithm \A it holds that:
$$
\advpqEUFCMA \leq \negl.
$$
\end{definition}

However, the situation changes in the case of signatures in the random oracle model: in this case, it would not make sense to define a notion of post-quantum security without switching to the quantum random oracle model. The resulting security notion should be called, for consistency with our naming conventions, pq-EUF-CMA-QRO. However, it is clear that the presence of QRO automatically implies \QPT adversaries, which in turn implies a post-quantum security notion {\em at least}. Therefore, for simplicity, we will call this new security notion just EUF-CMA-QRO.

\begin{experiment}[$\gameEUFCMAQRO$]\label{expt:EUF-CMA-QRO}
Let $\Sigscheme$ be a DSS, \RO a random oracle with corresponding quantum random oracle \QRO, and \A a \QPT algorithm. The {\em EUF-CMA-QRO experiment} proceeds as follows:
\begin{algorithmic}[1]
\State \textbf{Input:} $\secpar, q_s, q_h \in \NN$
\State $(\pk,\sk) \from \KGen^\RO$
\State $(x,\sig) \from \A^{\Sign_\sk,\QRO}(\pk)$ after making at most $q_h$ queries to \QRO, and $q_s$ queries to $\Sign_\sk$ receiving signatures $(x_1,\sig_1) , \ldots (x_{q_s}, \sig_{q_s})$
\If{$\SVer(\pk,x,\sig) = 1$ and $\x \neq \x_i \foral i = 1,\ldots,q_s$}
	\State \textbf{Output:} $1$
\Else
	\State \textbf{Output:} $0$
\EndIf
\end{algorithmic}
The {\em advantage of \A} is defined as:
$$
\advEUFCMAQRO(\secpar,q_s,q_h) := \Pr \left[ \gameEUFCMAQRO(\secpar,q_s,q_h) \to 1 \right].
$$
\end{experiment}

Notice how, in the above experiment, only the adversary has access to \QRO, while honest parties have only access to \RO.

\begin{definition}[(Post-Quantum) Existential Unforgeability under Chosen Message Attack in the Quantum Random Oracle Model (EUF-CMA-QRO)]\label{def:EUFCMAQRO}
A DSS $\Sigscheme$ is {\em (post-quantum) existentially unforgeable under chosen message attack in the quantum random oracle model (or, it is EUF-CMA-QRO secure)} iff, for any \QPT algorithm $\A$ it holds that:
$$
\advEUFCMAQRO \leq \negl.
$$
\end{definition}

\vfill

\section{Fiat-Shamir in the QROM}\label{sec:FSQROM}

The Fiat-Shamir transformation is a fascinating example of how things can go wrong when blindly switching to \QPT adversaries in defining post-quantum security notions. The presence of a random oracle and, especially, of rewinding in the security proof makes this a case to be treated carefully.

In the last few years, a few works have been presented dealing with the FS transformation in a quantum world. Here, we only discuss the results from Dagdelen et al.~\cite{DFG13}, which was hystorically the first work in the direction of assessing the security of FS in the quantum world.

\subsection{Preliminaries}

We start by defining {\em quantum-hard languages} as the `post-quantum analogue' of hard languages.

\begin{definition}[Quantum-Hard Language]
A hard language \hardL is a {\em quantum-hard language} iff for any \QPT algorithm \A it holds:
$$
\Pr_{(x,w) \from \Inst} \left[ (x,\A(x)) \in \R \right] \leq \negl.
$$
\end{definition}

Next, we identify a special class of $\Sigma$-protocols, where the prover's commitment $\com$ does not depend on the witness $w$.

\begin{definition}[$\Sigma$-Protocol with Witness-Independent Commitments]\label{def:witindcom}
A $\Sigma$-protocol \sigmaproto for a hard language \hardL has {\em witness-independent commitments} iff there exists a \PPT algorithm \Com which, on input a statement $x \in \L$, produces the same distribution as the prover’s first message $\com(x,w)$ for input $(x, w) \from \Inst$. In this case, we also write the first message as $\com \from \Com(x)$.
\end{definition}

Many $\Sigma$-protocols are actually of this type. Examples are the well known graph-isomorphism proof \cite{GMW86}, the Schnorr proof of knowledge \cite{Schnorr91}, or the protocol for lattices used in an anonymous credential system~\cite{CNR12}. A typical example of non--witness-independent commitment $\Sigma$-protocol is the graph $3$-coloring ZKPoK scheme \cite{GMW86}, where the prover commits to a random permutation of the coloring.

Finally, we define a class of $\Sigma$-protocols, where the prover's commitment $\com$ can be actually generated obliviously by the verifier instead.

\begin{definition}[$\Sigma$-Protocol with Oblivious Commitments]\label{def:oblcom}
A $\Sigma$-protocol \sigmaproto for a hard language \hardL has {\em oblivious commitments} iff there exist \PPT algorithms $\Com$ and $\SmplRnd$ such that the following distributions are statistically indistinguishable:
\begin{multicols}{2}
\begin{algorithmic}[1]
\State \textbf{Input:} $n \in \NN, (x,w) \in \R$
\State $r \rand \bin^{\poly(\secpar)}$
\State $\com \from \Com(x;r)$
\State $\ch \from \V(x,\com)$
\State $\resp \from \P(x,w,\com,\ch)$
\State \textbf{Output:} $(x,w,r,\com,\ch,\resp)$
\end{algorithmic}
\vfill
\columnbreak
\begin{algorithmic}[1]
\State \textbf{Input:} $n \in \NN, (x,w) \in \R$
\State $(\com,\ch,\resp) \from \left( \P(x,w),\V(x)\right)$
\State $r \from \SmplRnd(x,\com)$
\State \textbf{Output:} $(x,w,r,\com,\ch,\resp)$
\end{algorithmic}
\vfill
\end{multicols}
\end{definition}

Notice that a $\Sigma$-protocol with oblivious commitments has, in particular, witness-independent commitments. With oblivious commitments, the prover is able to compute a response from the given commitment $\com$ without knowing the randomness used to compute the commitment. This is usually achieved by placing some extra trapdoor into the witness $w$. For example, for the Guillou-Quisquater RSA based proof of knowledge~\cite{GQ88} where the prover shows knowledge of $w \in \ZZ_n^*$ with $w^e= y \bmod{n}$ for $x=(e,n,y)$, the prover would need to compute an $e$-th root for a given commitment $r \in \ZZ_n^*$. If the witness would contain the prime factorization of $n$, instead of the $e$-th root of $y$, this would indeed be possible.

$\Sigma$-protocols with oblivious commitments allow to move the generation of the commitment from the prover to the honest verifier. For most schemes this infringes on active security, because a malicious verifier could generate the commitment `non-obliviously'. However, the scheme remains honest-verifier zero-knowledge, and this suffices for deriving secure signature schemes through the FS transformation. We call such modified scheme a {\em $\Lambda$-protocol}\footnote{The choice of the symbol `$\Lambda$', in analogy to the choice of `$\Sigma$' in `$\Sigma$-protocol', is meant as a mnemonic graphical representation of the protocol flow. For $\Sigma$-protocols, in fact, the $\Sigma$ recalls a stylization of the left-to-right (and viceversa) arrows denoting exchange of messages between one `prover side' to the left and one `verifier side' to the right when representing the protocol as a workflow, with the direction of time going down. Analogously, $\Lambda$-protocols can be seen as $\Sigma$-protocols where part of the interaction (i.e., some `arrows') are removed. This is stylized by rotating the $\Lambda$ by 90 degrees.}.

\begin{definition}[$\Lambda$-Protocol]\label{def:lambdaproto}
Let \sigmaproto be a $\Sigma$-protocol for a hard language $\hardL$ with oblivious commitments. The {\em $\Lambda$-protocol \lambdaproto  associated to \sigmaproto} is a $3$-move interactive protocol with exchange of messages $r,(\com,\ch),\resp$ between two \PPT algorithms \PL and \VL such that:
\begin{enumerate}
\item $\PL(x,w) \to r$, where $r \rand \bin^{\poly(\secpar)}$
\item $\VL(x) \to (\com,\ch)$, where $\com \from \Com(x;r)$, and $\ch \from \V(x,\com)$
\item $\PL(x,w,\com,\ch) \to \resp$, where $\resp \from \P(x,w,\com,\ch;r')$, and\\ $r' \from \SmplRnd(x,\com)$
\item $\VL(x,\com,\ch,\resp) := \V(x,\com,\ch,\resp)$
\end{enumerate}
\end{definition}

The generation of the initial randomness $r$ can be performed by \VL himself, so that a $\Lambda$-protocol can generally be seen as a $2$-move interactive protocol.

\subsection{Impossibility Result for Post-Quantum Fiat-Shamir}

In this section, we use a meta-reduction technique to rule out the existence of strongly black-box reductions for the Fiat-Shamir transformation of actively secure $\Sigma$-protocols under certain conditions. That is: {\em it is not possible to find reductions with strong security guarantees for the Fiat-Shamir transformation in the QRO, by only relying on the active security of certain $\Sigma$-protocols}. Before assessing more in detail the strength of this result, we outline here the proof. Recall that, classically, if \sigmaproto is a $\Sigma$-protocol, then its FS transform in the ROM, \FSSigma, is an EUF-CMA-RO secure digital signature scheme (Theorem~\ref{thm:FS}).

\begin{enumerate}
\item First we describe a hypothetical, all-powerful adversary $\A^\QRO$ with quantum access to the random oracle (and no oracle access to the signing algorithm \Sign at all), able to break the EUF-CMA-RO security (generate forgeries) for \FSSigma for any input public key. This adversary does not need to exist in practice -- it is sufficient for our meta-reduction to successfully emulate it. The adversary $\A^\QRO$ uses his unbounded power to find a secret key \sk to its input \pk, and then uses a (single) query to the random oracle to generate a forgery. Moreover, such adversary uses the quantum access to the random oracle to `hide' his query in a superposition (this prevents any strong quantum reduction to apply the rewinding techniques of Pointcheval and Stern~\cite{PS00} as in the classical setting). Finally, this hypothetical adversary uses the secret key and the random oracle query to output a valid forgery.

\item Then we describe the behavior of a strongly black-box reduction \B reducing the EUF-CMA-RO security of \FSSigma to the weak security of an identification scheme \idscheme. We show how this is equivalent to finding valid witnesses for statements in a quantum-hard language \hardL by having only classical access to an efficient adversary for \FSSigma. We call these very powerful reductions {\em strong quantum extractors} (or, in short, just `extractors').

\item Then we build a reduction \M which breaks the active security of \sigmaproto by having classical access to an extractor \B.

\item Finally, we show how \M can successfully emulate the all-powerful adversary \A for \B by interacting with the honest prover \P and with the same random oracle \RO generated by \B. That is, \M is actually a {\em meta-reduction} which breaks the active security of \sigmaproto by using \B.
\end{enumerate}

We give such impossibility result in respect to the subclass of witness-independent $\Sigma$-protocols, while leaving open the other cases. Moreover, we assume that the strong quantum extractor is {\em input-preserving} (i.e., it forwards $x$ faithfully to the adversary). In this case we can easily derandomize the adversary (with respect to classical randomness) by `hardwiring' a key of a random function into it, which he initially applies to its input $x$ to recover the same classical randomness for each run. Since the strong extractor has to work for all adversaries, it in particular needs to succeed for those where we pick the function randomly but fix it from thereon.

\begin{theorem}[Impossibility Result for Fiat-Shamir]\label{thm:impossFS}
If \sigmaproto is an actively and weakly secure $\Sigma$-protocol with witness-independent commitments, then it does not admit any input-preserving strong quantum extractor.
\end{theorem}

\begin{proof}
We follow the proof sketch above by giving explicit descriptions of the adversary \A, the extractor \B, and the meta-reduction \M. At the beginning of the game, the honest prover \P generates a public/secret key pair $(\pk,\sk) \from \KGen$ for the DSS \FSSigma (which is actually a valid statement/witness pair $(x,w) \from \Inst$ for the quantum hard language \hardL). The public key \pk is also given to the honest verifier \V.

\

\textbf{The Adversary.} Our hypothetical, all-powerful adversary \A works as follows (see Figure~\ref{fig:adv}). He receives as input the public key $\pk=x$ and first uses its unbounded computational power to compute a random witness $w'$ (according to uniform distributions of coin tosses $\D$ subject to $\Inst(\secpar;\D) \to (x,w')$, but where $\D$ is a random function of~$x$). Then \A prepares all possible random strings $r \in \bin^{\r(\secpar)}$ (for some appropriate polynomial function \r) for the prover's algorithm in superposition, i.e., \A prepares the state:
$$
\sum_{r = 0}^{2^\r - 1} \frac{1}{\sqrt{2^\r}} \ket{r}
$$
(this can be done efficiently by using Hadamard gates). In the next step, \A evaluates (a unitary version of) the classical witness-independent algorithm \Com for (deterministically) computing the prover's commitment \com on this superposition (and on $x$) in order to obtain a superposition of all $\ket{r,\com := \Com(x;r)}$ plus an extra $\ket{0}$ ancilla register, i.e., the state:
$$
\ket{\phi} := \sum_{r = 0}^{2^\r - 1} \frac{1}{\sqrt{2^\r}} \ket{r, \com, 0 }.
$$
At this point, \A evaluates the QRO \QRO in superposition on the \com component of the above state (and using the public-key \pk and a chosen message $m$), thereby obtaining the state:
$$
\ket{\psi} := \sum_{r = 0}^{2^\r - 1} \frac{1}{\sqrt{2^\r}} \ket{r, \com, \ch := \h(\pk,\com,m) }.
$$
Then \A computes, in superposition, responses $\resp \from \P(x,w',\com,\ch,r)$ for all values in the superposition, by using $w'$ to emulate a valid prover, obtaining the state:
$$
\sum_{r = 0}^{2^\r - 1} \frac{1}{\sqrt{2^\r}} \ket{r, \com, \ch, \resp },
$$
Finally, \A measures  such state, obtaining a valid transcript $(\com,\ch,\resp)$, and hence a valid forgery \sig for \FSSigma.

\begin{center}
\begin{figure}[t]
\begin{center}
\includegraphics[width=\textwidth]{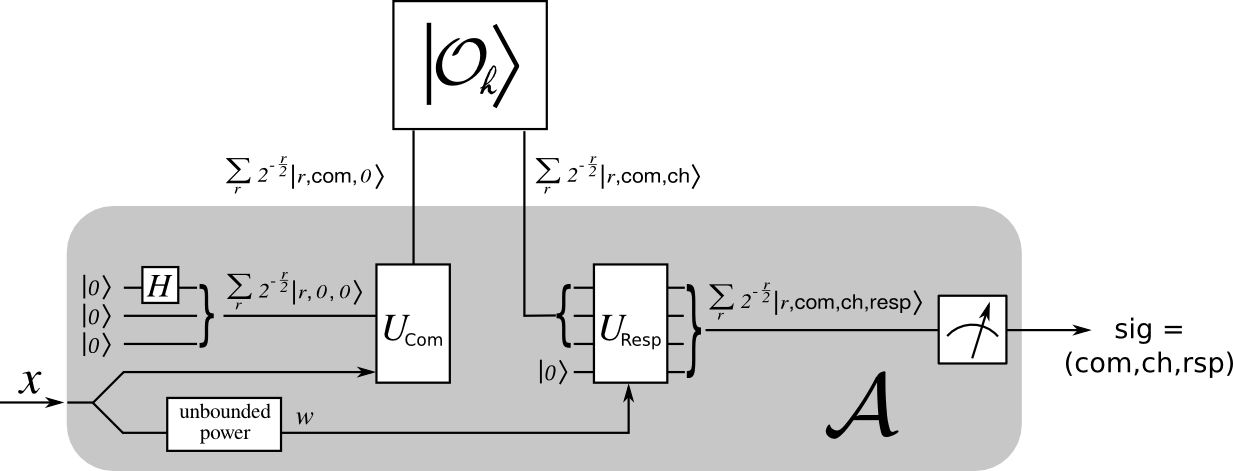}
\end{center}
\caption{the all-powerful adversary.}
\label{fig:adv}
\end{figure}
\end{center}

\textbf{The Extractor.} An extractor \B for \sigmaproto is a strong (black-box) quantum reduction which uses an adversary against \FSSigma in order to break the weak security of \sigmaproto. Therefore, it has the following characteristics.
\begin{itemize}
\item \B is a \QPT algorithm, taking as input a public-key \pk for \FSSigma (i.e., a statement $x$ in \hardL).
\item Because he wants to break the {\em weak} security of \sigmaproto, the goal of \B is eventually to output a valid witness $w''$ for $x$.
\item \B is a black-box reduction, so it works by interacting with {\em any} successful adversary against the EUF-CMA-RO security of \FSSigma, but without having any information about the internal workings of the adversary. In particular, it must work for the all-powerful adversary \A.
\item Because \A eventually wants to interact with a quantum random oracle, \B must also emulate a valid \QRO for \A. In particular, \B {\em \ must} be a quantum reduction.
\item However, since \B is a {\em strong} extractor, he is not allowed to tamper with \A's queries to \QRO. That is, \B cannot perform measurements or other quantum operations on those queries, except the evaluation through \QRO (but \B could, for example, reprogram the oracle, or rewind \A).
\end{itemize}

For example, such extractors might work by running \A twice, obtaining two distinct signature forgeries for the same messages, and then applying the special soundness property of \sigmaproto to extract a valid witness $w''$. These extractors can be passive or active (i.e., interacting with \P), there is no restriction on that as long as they output a valid $w''$.

On the other hand, we restrict our impossibility result to extractors with the two following additional properties:
\begin{enumerate}
\item\label{item:extr1} they are {\em input-preserving}, that is, the same statement $x$ (public-key \pk) input to \B is relayed as input to the black-box adversary; and
\item\label{item:extr2} they are {\em RO-broadcasting}, that is, they provide a public interface for evaluating \RO to be used by other external parties, not only exclusively by the black-box adversary.
\end{enumerate}

It is important to notice that this last condition is perfectly natural: recall that the ROM idealizes a publicly known hash function, so that it is reasonable to postulate that, once \B has set up the emulated \QRO, everyone can have access to it. Actually, for this reason, one could also assume that the extractor is {\em QRO-broadcasting} (i.e., providing a public quantum interface for evaluating \QRO), but for our result it is sufficient for the meta-reduction to query \RO classically, and a single query is enough.

\

\textbf{The Meta-Reduction.} We illustrate the meta-reduction \M in Figure~\ref{fig:metared}. Assume that there exists an extractor $\B$ with black-box access to an underlying quantum adversary $\A$, and which on input a statement (public-key) $x$ sampled according to $\Inst$, is able to extract a witness $w''$ to $x$ by running several resetting executions of $\A$, each time answering $\A$'s QRO queries $\ket{\phi}$ by emulating a QRO \QRO for a classical, possibly probabilistic function $\h$ for which \B also provides a public interface to be (at least classically) accessed by \M. Then \M can use $w''$ to break the (weak and strong) security of the underlying $\Sigma$-protocol \sigmaproto by impersonating a valid prover for $x$ against \V, against the assumption, and thereby concluding the proof.

It is left to show how \M can succesfully simulate a quantum adversary for \B. In particular, we describe here how \M can simulate the all-powerful adversary \A. Clearly, \M can produce the same query $\ket{\phi}$ that \A produces, because of the witness-independence of \sigmaproto. However, upon receiving back the reply $\ket{\psi}$ from \QRO, this state is discarded and ignored, and a valid forgery is instead generated in a different way. Namely, \M initiates a \sigmaproto execution with the valid prover \P for $x$, receiving a commitment \com. \M can now compute a valid challenge $\ch := \h(\com)$ by using the public interface provided by \B for evaluating \h, that is, \M is simulating a valid verifier \V for \P. At this point, a valid response \resp is computed by \P, and \M can use the transcript $(\com,\ch,\resp)$ to output a valid forgery for \B.
\endproof
\begin{center}
\begin{figure}[t]
\centering
\includegraphics[width=\textwidth]{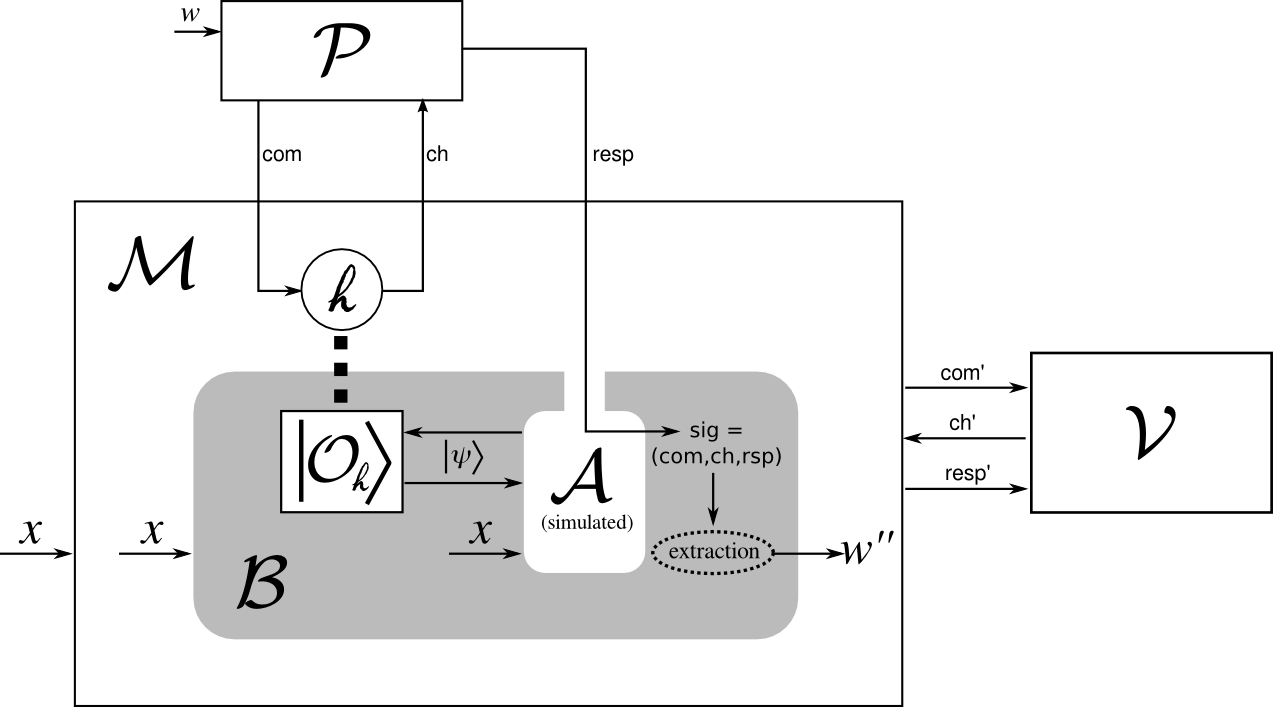}
\caption{An overview of our meta-reduction}
\label{fig:metared}
\end{figure}
\end{center}
\end{proof}

The above theorem is a special case of~\cite[Theorem~3.3]{DFG13} with DSS in mind, but the Fiat-Shamir transform can also be cast in the scenario of non-interactive zero-knowledge proofs. It is important to notice that the above impossibility result has the following limitations:
\begin{itemize}
\item it only holds for witness-independent commitment $\Sigma$-protocols.
\item It only holds for {\em strong black-box quantum extractors}. I.e., the extractor is not allowed to tamper with the adversary's queries to the QRO.
\item The extractors must be input-preserving, i.e., they use their underlying black-box adversary by giving as input the same public-key used to break the $\Sigma$-protocol.
\item It only holds for extractors breaking weak security, that is, {\em witness-extracting} -- they are stronger than extractors who just win the impersonation game in the $\Sigma$-protocol.
\item It is necessary that the extractor allows the meta-reduction to evaluate \RO at least once.
\end{itemize}

Before discussing more in detail some of the above limitations, it is important to put this result in hystorical perspective: this was the {\em first} impossibility result for Fiat-Shamir in the quantum world, and following works~\cite{QuantRewinding,Unruh15} rely on more advanced tools. As already discussed, the witness-independence of the commitments is not a strong limitation, as most $\Sigma$-protocols have this property. Finally, notice that the existence of strong black-box extractors is {\em not} an unreasonable assumption -- and therefore the above impossibility result is not unreasonably weak. In fact, Theorem~\ref{thm:FSLambda} in the next section shows that certain $\Sigma$-protocols {\em do} indeed admit such extractors.

As we have already noticed, the extractor has to choose and provide public classical access to a classical function $\h$ for answering random oracle queries. While this may be considered a `gray-box' restriction in general interactive quantum proofs, it seems to be inevitable in the QROM; it is rather a consequence of the approach where a quantum adversary mounts attacks in a classical setting. After all, both the honest parties as well as the adversary expect a classical hash function. The adversary is able to check this property easily, even if it treats the hash function otherwise as a black box (and may thus not be able to spot that the hash function uses (pseudo)randomness). We remark that this approach also complies with previous efforts \cite{QROM,BZ13auth,Zhandry12,ZhandryPRF} and the positive result in the next section to answer such hash queries. Moreover, notice that in the above proof technically \M only needs to evaluate \h {\em once}, i.e., it must not necessarily require unlimited access to \RO. For these reasons, the meta-reduction still qualifies as {\em black-box}.

Furthermore, the extractor \emph{can} rewind the quantum adversary to any point before the final measurement. Recall that for this impossibility result it is assumed, to the advantage of the extractor, that the adversary does not perform any measurement until the very end. Since the extractor can re-run the adversary from scratch for the same classical randomness, and the `no-cloning restriction' does not apply to our adversary with classical input, the extractor can therefore easily put the adversary in the same (quantum) state as in a previous execution, up to  the final measurement. However, because we consider \emph{strong black-box} extractors, the extractor can only influence the adversary's behavior via the answers it provides to \A's external communication. In this sense, the extractor may always rewind the adversary to such communication points. The extractor is also allowed to measure and abort at such communication points.

The extraction strategy by Pointcheval and Stern~\cite{PS00} in the purely classical case \emph{can} be cast in the strong black-box extractor framework. For this the extractor would run the adversary for the same classical randomness twice, providing a lazy-sampling--based hash function description, with different replies in the $i$-th answers in the two runs. The extractor then extracts the witness from two valid signatures. This shows that a different approach than in the classical setting is necessary for extractors in the QROM.

One might ask why the meta-reduction does not apply to the Fiat-Shamir transform when adversaries have only classical access to the random oracle. The reason is the following: if the adversary made a classical query about a single commitment (and so would the meta-reduction), then one could apply the rewinding technique of Pointcheval and Stern changing the random oracle answers, and extract the underlying witness via special soundness of the identification scheme. The quantum adversary here, however, queries the random oracle in a superposition. In this scenario, as we explained above, the extractor is not allowed to `read' the query of the adversary unless it makes the adversary stop. In other words, the extractor cannot measure the query and then keep running the adversary until a valid witness is output. This intrinsic property of strong black-box quantum extractors, hence, makes `quantum' rewinding impossible. Note that rewinding in the classical sense -- as described by Pointcheval and Stern -- is still possible, as this essentially means to start the adversary with the same random coins. This does not cover the case where \B measures (at least partially) the query state without disturbing \A's behavior significantly (i.e., {\em non-strong} extractors), but subsequent works~\cite{Unruh15} have also ruled out this possibility.

Finally, we briefly discuss that active security is basically necessary for an impossibility result as above. That is, we outline a three-move protocol for any quantum-hard language which, when applying the FS transformation, supports a straight-line extractor, and is honest-verifier zero-knowledge, but not actively secure. This holds as long as there are {\em post-quantum dense encryption schemes}, and {\em post-quantum non-interactive zero-knowledge proofs}. The latter are classical non-interactive zero-knowledge proofs (in the common random string model) for which simulated and genuine proofs are indistinguishable, even for \emph{quantum} distinguishers. The former are pq-IND-CPA encryption schemes where honestly generated public keys are quantum-indistinguishable from random strings. The construction is based on the (classical) non-interactive zero-knowledge proofs of knowledge of De Santis and Persiano \cite{DP92} and works as follows. The first message is irrelevant, e.g., we let the prover simply send the constant $0$ (potentially padded with redundant randomness). In the second message the verifier sends a random string which the prover interprets as a public key $\pk$ of the dense encryption scheme and a common random string $\crs$ for the NIZK. The prover encrypts the witness under $\pk$ and gives a NIZK that the encrypted value forms a valid witness for the public value $x$. The verifier only checks the NIZK proof. The protocol is clearly not secure against active (classical) adversaries because such an adversary can create a public key $\pk$ via the key generation algorithm, thus, knowing the secret key and allowing the adversary to recover the witness from a proof by the prover. It is, however, honest-verifier zero-knowledge against quantum distinguishers, because of the pq-IND-CPA security and the simulatability of the NIZK hide the witness and allow for a simulation.

\vfill

\subsection{Security Result for Post-Quantum Fiat-Shamir}\label{sec:positiveFS}

In this section, we show how it is possible to actually resurrect the security of the FS transformation for a certain class of $\Sigma$-protocols able to overcome the previous impossibility result. The intuition is the following: as such impossibility result works by exploiting the active security of the $\Sigma$-protocol, and since such property is not needed for the FS transformation to yield secure signature schemes, we can `patch' the $\Sigma$-protocol by removing its active security. That is, by {\em weakening} the security guarantees of a $\Sigma$-protocol (seen as an identification scheme) we work toward {\em strengthening} the properties of its FS transform (seen as a DSS).

We achieve this goal by considering the FS transform of $\Lambda$-protocols obtained by $\Sigma$-protocols with oblivious commitments. In particular, using random oracles one can hash directly into pairs $(\com,\ch)$ by first computing the output of the hash function obtaining a (public-coin) challenge \ch and some randomness $r'$, and then running $\Com(x;r')$ to sample a commitment \com obliviously. The existence of $\SmplRnd$ guarantees that we could `bend' this value back to an actual pre-image $r$ for $\com$. In the sequel we therefore often identify $r'$ with $\Com(x;r')$ in the sense that we assume that the hash function maps to $\Com(x;r')$ directly, and for a (randomized) hash function \h and message $m$ we write $(\com,\ch) \from \h(x,m,r)$. The modified FS transformation then looks as follows.
 
\begin{definition}[FS Transform of a $\Lambda$-Protocol]\label{def:FSLambda}
Let \lambdaproto be a $\Lambda$-protocol for a hard language \hardL, with commitment space \X (with associated randomness space $\X^{-1} = \set{r:r \from \SmplRnd} := \bin^{\poly(\secpar)}$), challenge space \Y, and response space \Z. Let \RO be a random oracle for a random function $\h: \L \times \M \times \X^{-1} \to \Y$. The {\em FS transform of \lambdaproto in the ROM, $\FSLambda$}, is a DSS with message space $\M$, signature space $\T := \Y \times \Z$, and key space $\K := \L \times \W$, defined as follows:
\begin{enumerate}
\item $\KGen \to (\pk,\sk)$, where $(\pk,\sk):=(x,w) \from \Inst$
\item $\Sign^\RO(\sk,m) \to \sig := (r,\resp)$,\\where $r \rand \X^{-1}$, $(\com,\ch) \from \h(\pk,m,r)$,\\and $\resp \from \PL(\pk,\sk,\com,\ch,r)$
\item $\SVer^\RO(\pk,m,\sig) \to b$,\\where $\sig := (r,\resp)$, $b \from \V(\pk,\h(\pk,m;r),\resp)$
\end{enumerate}
\end{definition}

As we have already discussed, this modified FS transformation eludes the impossibility result from the previous section. In order to show its security, we exploit the special soundness of the $\Lambda$-protocol: by reprogramming the QRO \QRO for a forgery-generating adversary \A, eventually we obtain two related transcripts $(\com^\star,\ch^\star,\resp^\star)$ and $(\com^\star,\ch',\resp')$ for $\ch^\star \neq \ch'$, and thus extracting a valid witness for $x$ and breaking the weak security of \lambdaproto. The idea of the proof is as follows.

\begin{enumerate}
\item First, we run the HVZK simulator \S of the $\Lambda$-protocol to obtain a valid transcript $(\com^\star,\ch^\star,\resp^\star)$.
\item We reprogram the QRO \QRO by `injecting' the value $(\com^\star,\ch')$ (for $\ch^\star \neq \ch'$) on a fraction $\delta$ of the possible oracle answers. That is, we replace \RO with a semi-constant distribution \scDist.
\item Then, we run the adversary \A against the modified quantum oracle, obtaining a forgery for \FSLambda for some message $m$, and hence a valid transcript $(\com,\ch,\resp)$ for \lambdaproto.
\item Finally, if it happens that $\com = \com^\star$ and $\ch \neq \ch^\star$, we can use the special soundness extractor \J to obtain a valid witness for $x$ and breaking the weak security of \lambdaproto, concluding the proof.
\end{enumerate}

In order for this proof strategy to work, the following two (seemingly contradictory) conditions have to be fulfilled:
\begin{itemize}
\item we need to ensure that \A eventually outputs a valid signature yielding a transcript for the commitment $\com^\star$ of our choice (the one we obtained from the zero-knowledge simulator of the underlying $\Sigma$-protocol). This requires that $\com^\star$ appears with sufficiently large probability in the responses for oracle queries.
\item On the other hand, we still require that \A has a small probability of distinguishing a true QRO $\QRO$ from the reprogrammed one. Otherwise, the adversary may refuse to give a valid signature at all.
\end{itemize}

The following technical lemma shows that both conditions can be satisfied simultaneously by choosing $\delta$ carefully.

\begin{lemma}\label{lem:output}
Let \lambdaproto be a $\Lambda$-protocol for a quantum-hard language \hardL, and let $\O'$ be the oracle obtained by reprogramming $\RO$ on a 
fraction $\delta$ of its possible inputs $(\pk,m,r)$ such that $\O'(\pk, m, r) = (\com^\star, \ch')$ with probability $\delta \in (0,1)$ for fixed values $\com^\star$ and $\ch'$. Let \A be a \QPT algorithm such that $\A^\QRO(\pk)$ outputs a valid forgery for \FSLambda for a public key \pk with probability at least $\epsilon$ after performing $q_h$ queries to \QRO, and let $(\com,\ch,\resp)$ the transcript obtained by the output of the same algorithm $\A^{\ket{\O'}}(\pk)$ running against the reprogrammed quantum oracle. Then:
$$
\Pr \left[ \VL^{\O'}(x,\com,\ch,\resp) \to 1 \land (\com,\ch) = (\com^\star,\ch') \right] \geq \delta \cdot \epsilon - \frac{8}{3} \cdot q_h^4  \delta^2.
$$
\end{lemma}
\begin{proof}
Consider the probability that we first run \A on the original oracle $\QRO$ and check if it successfully forges a signature $(r,\resp)$ for \pk and some message $m$ (leading to a transcript $(\com,\ch,\resp)$), and then, independently, we also verify that $(\pk,m,r)$ is mapped to $(\com^\star, \ch')$ under $\O'$. Then:
$$
\Pr \left[ \A^\QRO(\pk) \text{ succeeds } \land \O'(\pk,m,r) = (\com^\star,\ch') \right] \geq \delta \cdot \epsilon.
$$
This follows from the independence of the events: the oracle $\O'$ reprograms the output with probability $\delta$, independently of $\A$'s behavior, but at the same time we know that $\A^{\QRO}$ succeeds with probability at least $\epsilon$ by assumption. Next, we replace \QRO with $\ket{\O'}$ for \A, and we consider the new output $(m,r,\resp)$, arguing that:
$$
\Pr \left[ \A^{\ket{\O'}}(\pk) \text{ succeeds } \land \O'(\pk,m;r) = (\com^\star,\ch') \right] \geq \delta \cdot \epsilon - \frac{8}{3} \cdot q_h^4  \delta^2.
$$
This follows from Lemma~\ref{lem:zhandrycustom}: switching to the new oracle can change the distance of the output distribution of $\A$ by at most $\frac{8}{3} \cdot q_h^4 \delta^2$, and adding the verification step $\VL^{\O'}(x,\com,\ch,\resp) \to 1$ cannot increase this distance. Therefore, we conclude that the probability for the event
$$
\VL^{\O'}(x,\com,\ch,\resp) \to 1 \land (\com,\ch) = (\com^\star,\ch')
$$
cannot be smaller than the claimed bound, because $(\com,\ch) := \O'(\pk,m,r)$ by construction.
\end{proof}

The previous lemma informally tell us that, in order to succeed, we have to balance between a large $\delta$ to increase the chances of the adversary outputting a signature containing our desired $\com^\star$, and a small $\delta$ to avoid that the adversary detects the reprogrammed oracle. We are now ready to prove the main theorem.

\begin{theorem}[Security of a Fiat-Shamir Transform for $\Lambda$-Protocols]\label{thm:FSLambda}
Let \lambdaproto be a $\Lambda$-protocol for a quantum-hard language. Then $\FSLambda$ is an EUF-CMA-QRO secure DSS.
\end{theorem}
\begin{proof}
We assume towards contradiction the existence of an efficient quantum adversary $\A$ which, on input a public key $\pk$, outputs a valid forgery $(m,\sig)$ under $\pk$ with non-negligible probability $\epsilon$, hence breaking the existential unforgeability of \FSLambda. This adversary has access to a quantum-accessible random oracle $\QRO$ with $\h(\pk,m_i,r_j)= (\com_{i,j},\ch_{i,j})$, and to a signing oracle $\Sign_\sk$ for the secret key $\sk$ (where $(\pk,\sk) := (x,w) \in \R$) producing, on input a message $m$, a (classical) signature $\sig=(r,\resp) \from \Sign^\RO(\sk,m)$.

The adversary $\A$ gets $\pk$ as an input, and is then allowed to perform up to $q_h=\poly(\secpar)$ quantum queries to $\QRO$, and up to $q_s=\poly(\secpar)$ classical queries to $\Sign_\sk$. Then, after running for $\poly(\secpar)$ time, $\A$ produces (with non-negligible probability $\epsilon$) a forgery $(m,\sig)$ such that $m$ has never been asked to the signing oracle $\Sign_\sk$ throughout $\A$'s execution (i.e., $m$ is a fresh message). We assume that $q_h$ also covers a classical query of the verifier to check the signature. 

Under these assumptions we show how to build a strong black-box quantum extractor $\B$, with access to $\A$ as a subroutine, and which is able to break the hardness of \hardL with non-negligible probability. That is, $\B$ on input $x \in \L$ generated according to $\Inst$, is able to output a valid witness $w'$ such that $(x,w') \in \R$ by only interacting classically with \A. The quantum extractor $\B$ works as follows:
\begin{itemize}
\item on input statement $x$, it first runs the simulator \S of the underlying $\Lambda$-protocol to obtain a valid transcript $(\com^\star,\ch^\star,\resp^\star)$. This is possible because of the honest-verifier zero-knowledge property. Note also that this does not require access to the random oracle. As already explained, we assume for simplicity that the oblivious commitment is a random string; else we would need to run $\SmplRnd$ on $(\pk,\com^\star)$ to derive a preimage randomness $r$, and then use $r$ in the hash reply (and argue that this is indistinguishable).
\item Then, $\B$ simulates a quantum-classical oracle $\ket{\O_0} := \scQRO$ which is obtained by reprogramming a (simulated) quantum random oracle $\QRO$ over a fraction $\delta$ of its possible inputs $(\pk,m,r)$ with the value $(\com^\star, \ch')$. Here, $\delta$ is some non-negligible probability in the security parameter (whose optimal value will be computed later), and $\ch'$ is an arbitrarily chosen challenge different from $\ch^\star$. That is, $\O_0 (\pk,m,r) = (\com^\star,\ch')$ with probability $\delta$, and random elsewhere.
\item Next, $\B$ invokes $\A$ on input $\pk = x$. 
\item Whenever $\A$ performs the $i$-th (classical) query to $\Sign_\sk$ for signing a message $m_i$, $\B$ does the following:
	\begin{itemize}
	\item choose a random value $r_i \rand \X^{-1}$;
	\item execute the honest-verifier zero-knowledge simulator $\S$ of the $\Lambda$-protocol, obtaining a valid (simulated) transcript$(\com_i, \ch_i, \resp_i)$;
	\item reprogram $\O_{i-1}$ with value $(\com_i, \ch_i)$ for the input $(\pk,m_i,r_i)$. We denote by $\O_i$ the reprogrammed oracle after the $i$-th query to the signing oracle;
	\item then output $\sig_i := (r_i,\com_i,\ch_i,\resp_i)$ as $\Sign_\sk$'s reply to $\A$.
	\end{itemize}
\item Finally, when $\A$ outputs a (hopefully valid) fresh forgery $(m, \sig)$, where $\sig=(r,\resp)$ and $\O_{q_s}(\pk,m;r)=(\com,\ch)$, the extractor $\B$ aborts if $\com \neq \com^\star$ or $\ch = \ch^\star$. Otherwise, it uses the special soundness extractor $\J$ of the underlying $\Lambda$-protocol on input $(\com^\star,\ch^\star,\resp^\star)$ and $(\com,\ch,\resp)$ to obtain a valid witness $w'$ for $x$, concluding the attack.
\end{itemize}

Note that we can formally let $\B$ implement the dynamic reprogramming of the quantum-classical oracle, 
basically hardwiring all changes due to reprogramming into the code of the underlying classical algorithm. In a second step we can emulate the quantum oracle as explained in Section~\ref{sec:QROM}.

We next show that the success probability of our extraction procedure $\B$ is non-negligible given a successful $\A$. The proof follows the common game-hopping technique where we gradually deprive the adversary of (a negligible amount of) its success probability.

\

$\mathbf{Game_1}:$ this is $\game^{\mathsf{EUF-CMA-QRO}}_{\FSLambda,\A}$ describing $\A$'s original attack against \FSLambda constructed according to Definition~\ref{def:FSLambda}, played against a public key $\pk$. By assumption we have:
$$
\Pr \left[ \A \text{ wins } \game_1 \right] \geq \epsilon 
$$
for some non-negligible value $\epsilon$.

\

$\mathbf{Game_2}:$ this game is identical to $\game_1$, except that we abort if $\A$ outputs a valid fresh forgery $(m,\sig)$ where $\sig$ {\em does not} contain a randomness leading to 
the pre-selected commitment $\com^\star$ and challenge $\ch'$. Furthermore, we replace the random oracle $\RO$ with the 
oracle $\O_0$. Recall that $\O_0$ is obtained by reprogramming $\RO$ on a fraction $\delta$ of its entries with the value $(\com^\star,\ch')$. By Lemma~\ref{lem:output} we have:
$$
\Pr \left[ \A \text{ wins } \game_2 \right] \geq \delta\epsilon - \frac{8}{3} q_h^4 \delta^2.
$$

\

$\mathbf{Game_3}$ is actually a sub-sequence of $q_s$ different experiments denoted by $\game_3^{(i)}$ for $i = 1,\ldots,q_s$.

\

$\mathbf{Game_3^{(1)}}:$ this is as $\game_2$, but this time $\O_0$ is reprogrammed to $\O_1$ (i.e., $\O_1(\pk,m_1,r_1) := (\com_1,\ch_1)$) as soon as $\A$ performs its $1^{st}$ classical query $m_1$ to $\Sign_\sk$. From then on, the oracle $\O_1$ always answers consistently with this value. We need to show that this switching does not change the winning probability significantly. For this we basically need to show that, so far, the amplitudes of this value $(\pk,m_1,r_1)$ in the queries to the quantum oracle are small, or else the adversary may be able to spot some inconsistency.

Let $\X^{-1}$ the randomness space from \SmplRnd as from Definition~\ref{def:FSLambda}, and let $\card{\X^{-1}} = 2^{\r}$ for some function \r polynomial in the security parameter. We define the value $(\pk,m'_i,r'_j)$ to have {\em high amplitude} if there exists at least one of the quantum queries $\ket{\phi_1},\ket{\phi_2},\ldots$ to the quantum oracle $\ket{\O_0}$ {\em before the current ($1^{st}$) signing query}, where the amplitude $a_{i,j}$ associated to the corresponding basis element of $(\pk,m'_i,r'_j)$ is such that $|a_{i,j}|^2 \geq 2^{\frac{-{\r}}{2}}$. Otherwise, the tuple is said to have {\em low amplitude}. Note that each query to the quantum oracle can have at most $2^{\frac{\r}{2}}$ tuples with high amplitude, because the (square of the) amplitudes need to sum up to $1$.

When $\O_0$ is reprogrammed to $\O_1$, the choice of $m_1$ is fixed (i.e., determined by the $1^{st}$ query of $\A$ to $\Sign_\sk$), but $r_1$ is still chosen uniformly at random in $\X^{-1}$. Since $\A$ performs at most $q_h$ queries to $\ket{\O_0}$ before the signing query, we have thus at most $q_h \cdot 2^{\frac{{\r}}{2}}$ tuples with high amplitude before this query. The probability of hitting such a tuple is then given by:
\begin{equation}\label{eqn:highamp}
\Pr \left[ (\pk,m_1,r_1) \text{ has high amplitude} \right] \leq q_h \cdot 2^{\frac{-{\r}}{2}}.
\end{equation}
Moreover, provided $(\pk,m_1,r_1)$ has {\em low} amplitude, and since there are at most $q_h +q_s$ query steps, using Lemma~\ref{lem:distances} and 
Lemma~\ref{lem:queryprob} we obtain:
\begin{equation}\label{eqn:BBB}
\left| \A^{\ket{\O_0}} - \A^{\ket{\O_1}} \right| \leq 4 \sqrt{(q_h + q_h) \cdot 2^{\frac{-{\r}}{2}}}.
\end{equation}
Let us assume, on behalf of the adversary, that $\A$ fails whenever $(\pk,m_1,r_1)$ has high amplitude. Still, from equations~\eqref{eqn:highamp} and~\eqref{eqn:BBB}, we have:
\begin{align*}
\Pr \left[ \A \text{ wins } \game_3^{(1)} \right] & \geq \Pr \left[ \A \text{ wins } \game_2 \right] - 4 \sqrt{(q_h + q_s) \cdot 2^{\frac{-{\r}}{2}}} - q_H \cdot 2^{\frac{-{\r}}{2}} \\
& = \delta \epsilon - \frac{8}{3} q_h^4 \delta^2 - \negl.
\end{align*}
Here, we use the fact that reprogramming the oracle for $(\pk,m_1,r_1)$ does not change the adversary's success probability for a forgery {\em for a fresh message $m$}. That is, since the adversary's forgery is for $m \neq m_1,m_2,\dots$ it cannot simply copy a signature query as a forgery, but must still forge on the original oracle $\O_0$. So the argument about the winning probability applies as it did for $\O_0$.

We now repeat at most $q_s$ times the game hopping, from $\game_3^{(1)}$ to $\game_3^{(q_s)}$, every time repeating the previous game but switching from $\O_{i-1}$ to $\O_i$ during the $i^{th}$ query to $\Sign_\sk$, each time losing at most a negligible factor in the winning probability. Note that the probability of hitting a high amplitude with the signature generation in each hop increases to at most $q_h \cdot 2^{\frac{-{\r}}{2}}+q_s\cdot 2^{-\r}$ when taking into account the at most $q_s$ hash queries in the previous signature requests, but this remains negligible.

\

After $q_s$ steps we reach the following game.

$\mathbf{Game_3^{(q_s)}}:$ as $\game_2$, but now $\O_0$ is dynamically reprogrammed as a sequence $\O_1 , \ldots , \O_{q_s}$ throughout all of the $\A$'s queries to $\Sign_\sk$. We have:
$$
\Pr \left[ \A \text{ wins } \game_3^{(q_s)} \right]  \geq \delta \epsilon - \frac{8}{3} q_h^4 \delta^2 - \negl.
$$

\

$\mathbf{Game_4:}$ as before, but now $\Sign_\sk$ is just simulated through the zero-knowledge simulator $\S$ of the underlying $\Lambda$-protocol. If, by contradiction, $\A$'s winning probability is affected by more than a negligible amount in so doing, then we could use $\A$ to build an efficient distinguisher between `real' and `simulated' transcripts of the $\Lambda$-protocol. This would require a distinguisher with access to a random oracle, in order to simulate the game. According to~\cite[Theorem 6.1]{Zhandry12}, however, we can simulate the oracle via $q$-wise independent functions (which exists without requiring cryptographic assumptions). Furthermore, a hybrid argument can be used to reduce the case of $q_s$ proofs to a single proof. Therefore:
$$
\Pr \left[ \A \text{ wins } \game_4 \right]  \geq \delta \epsilon - \frac{8}{3} q_h^4 \delta^2 - \negl.
$$

\

$\mathbf{Game_5:}$ finally, in this game the special soundness extractor $\J$ is run on the transcript obtained from $\A$'s output from the previous game. Change the winning condition of $\A$ such that the adversary wins if this extraction yields a valid witness $w'$ for $x$. If the winning probability in this game is more than negligibly far from the winning probability of $\A$ in the previous game then this can only be due to the fact that the simulated proof with $(\com^\star,\ch^\star,\resp^\star)$ cannot be accepted by the verifier; else the extractor would be guaranteed to work for this proof and the (accepted) signature. But this would allow an easy distinguisher against the zero-knowledge property, similar to the previous games. Hence:
$$
\Pr \left[ \A \text{ wins } \game_5 \right]  \geq \delta \epsilon - \frac{8}{3} q_h^4 \delta^2 - \negl.
$$
Note that $\A$'s winning condition in the final game corresponds exactly to the probability of $\B$ successfully deriving a witness $w'$ for its input $x$. This winning probability can be maximized (by zeroing the first derivative in $\delta$) by choosing:
$$
\delta := {\frac{3\epsilon}{16 q_h^4}}.
$$
This yields:
$$
\Pr \left[ \A \text{ wins } \game_5 \right] \geq \frac{3 \epsilon^2}{16 q_H^4} - \negl,
$$
which is non-negligible. This concludes the proof of the theorem.
\endproof
\end{proof}

The results from this section regarding the security and impossibility results for the Fiat-Shamir transform of witness-independent commitments in the QROM is summarized in Figure~\ref{fig:FSresults}: a security proof can be found for $\Sigma$-protocols with oblivious commitments (that is, $\Lambda$-protocols), while strong extractors can be ruled out whenever the FS transformation is applied to $\Sigma$-protocols which are actively secure (seen as identification schemes). However, some of these schemes can be `patched' by using commitment trapdoors in order to make them oblivious commitment and remove their active security, yielding signature schemes in a way similar to the hash-and-sign paradigm~\cite{GPV08}. This is for example the situation in the lattice-based signature scheme by Lyubashevsky~\cite{Lyu12}, which can be patched in such a way to be rendered EUF-CMA-QRO secure according to Theorem~\ref{thm:FSLambda}, as explained in~\cite{DFG13}.

\begin{figure}[t]
\begin{center}
\includegraphics[width=0.8\textwidth]{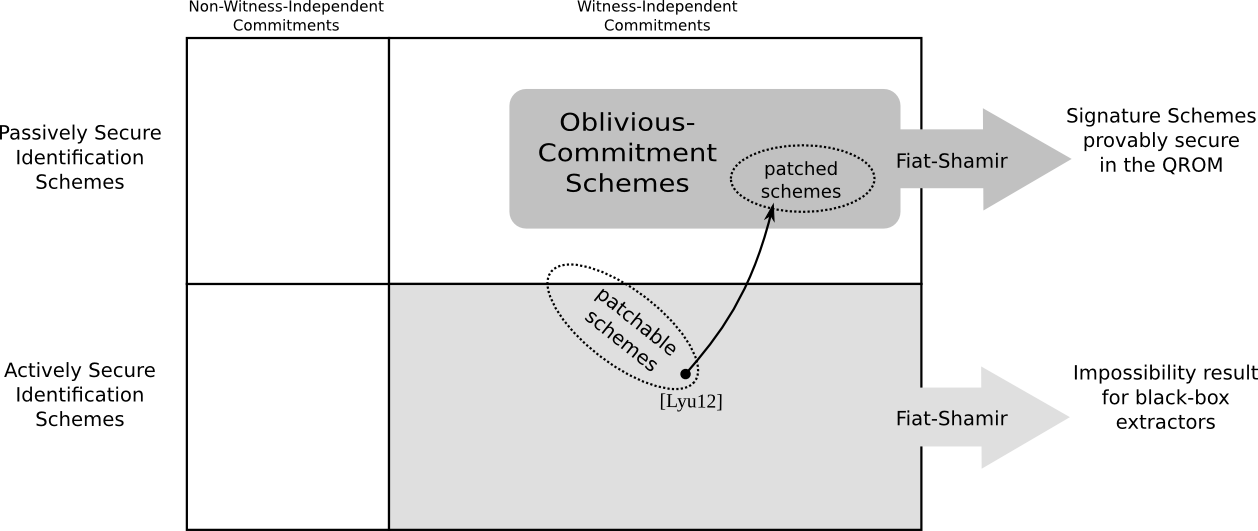}
\end{center}
\caption{Security results for the Fiat-Shamir transformation in the QROM.}
\label{fig:FSresults}
\end{figure}

\section{Post-Quantum ORAMs}\label{sec:pqORAM}

In this section we look at the post-quantum security of ORAMs. First of all, we define a suitable security model. Then we show that the extension of a classically secure ORAM to its post-quantum secure counterpart is not necessarily trivial. To this end, we examine \pathoram and we show that merely substituting the underlying encryption scheme with a post-quantum one does not generally yield a post-quantum ORAM. The idea is to exploit the weakness of other components of the ORAM construction under examination (in this case, the PRNG used). This is not surprising, because it has to be somewhat expected that post-quantum security can only be achieved by hardening {\em all} the underlying components of a cryptographic scheme, not only the encryption. However, it is important to keep this possibility in mind.

Then, we show that building post-quantum secure ORAMs is possible. We do it by showing that \pathoram, instantiated with a post-quantum secure SKES {\em and} a post-quantum PRNG, achieves post-quantum security. This is important from an application perspective, because it shows that efficient and post-quantum secure ORAMs can indeed be obtained in a straightforward way. Moreover, the proof of this fact is a straightforward adaptation from Theorem~\ref{thm:pathoramsec}, and the resulting security reduction is semi-classical, therefore offering very strong security guarantees, as discussed in Section~\ref{sec:qred}.

\subsection{Post-Quantum Security of ORAMs}\label{sec:oramsec}

Since the security model for ORAM only involves a classical communication channel and there is no oracle access involved, we can simply switch to a post-quantum model of security for ORAMs in the usual way: we keep the AP-IND-CQA game as from Experiment~\ref{expt:gameORAM}, but we switch to \QPT adversaries.

\begin{definition}[Quantum ORAM Adversary]\label{dfn:pqoramadv}
A {\em quantum ORAM adversary} \A is a \QPT algorithm which is computationally indistinguishable from an honest server \S for every ORAM client \C. In particular, the ORAM's soundness is preserved.
\end{definition}

\begin{definition}[Post-Quantum Access Pattern Indistinguishability Under Adaptive Chosen Query Attack]\label{def:pq-AP-IND-CQA}
An ORAM construction $\oram$ has {\em post-quantum computationally indistinguishable access patterns under adaptive chosen query attack} (or, it is pq-AP-IND-CQA-secure) iff for any quantum ORAM adversary \A it holds that $\advORAM \leq \negl$.
\end{definition}

Clearly, if an ORAM is pq-AP-IND-CQA-secure, then it is also AP-IND-CQA-secure. The converse does not hold (under standard hardness assumptions) as we will see.

\subsection{The Impossibility Result}

In order to show that one cannot in general obtain post-quantum ORAMs by just using a post-quantum SKES in a black-box way, we provide the following counterexample.

\begin{theorem}\label{thm:pqORAMsep}
Let $\E=(\KGen,\Enc,\Dec)$ be a pq-IND-CPA SKES according to Definition~\ref{def:pqINDCPA}, and let $\PRNGBM$ be the Blum-Micali PRNG from Lemma~\ref{lem:blummicali}. Let $\pathoramBM$ be the ORAM obtained by instantiating the \pathoram construction from Definition~\ref{dfn:pathoram} using $\E$ and $\PRNGBM$. Then, under the DLP hardness assumption, $\pathoramBM$ is an AP-IND-CQA secure ORAM, but not pq-AP-IND-CQA secure.
\end{theorem}

At the light of Theorem~\ref{thm:pathoramsec} and Definition~\ref{def:pq-AP-IND-CQA}, in order to prove Theorem~\ref{thm:pqORAMsep} we only need to show the following lemma.

\begin{lemma}\label{lem:pqORAMsep}
There exists a \QPT algorithm $\A$ winning 
\mbox{$\game_{\A,\pathoramBM}^{\texttt{AP-IND-CQA}}$} with 
non-negligible advantage over guessing.
\end{lemma}
\begin{proof}
We start by making a key observation concerning the access patterns produced in \pathoram. Let $\dr=(\op,i,\data)$ be a data request sent by \C. By only examining the communication transcript $\com$ resulting from the execution of this data request, one can see which path (branch of the tree) \S sent to \C, thus learning the leaf $r_i$ to which $i$ was mapped to, even without knowing $i$ itself. In normal circumstances, this is of no use to an adversary, because this value $r_i$ becomes immediately obsolete, being replaced by a new fresh value output by the PRNG in the position map. But it will be important in our attack as we will see.

Let $\D$ be the \BQP algorithm (the `PRNG predictor') of Lemma~\ref{lem:blummicali}. We build the adversary $\A$ with oracle access to $\D$. First of all $\A$ chooses $\secpar, \dbsize \leq \maxsize$ and starts the AP-IND-CQA game by calling $\init(\secpar,\dbsize)$. For his attack, \A fixes an arbitrary identifier 
$i \in \set{1,\ldots,\dbsize}$, and an arbitrary data unit $\data \in \bin^\dsize$.

During the first CQA learning phase, \A asks \C to execute $\k = \poly(\secpar)$ consecutive data requests of the form $(\text{`write'},i,\data)$. \A records the resulting access patterns from all these queries, $\ap_1, \ldots, \ap_{\k}$, which include the communication transcripts $\com_1,\ldots,\com_{\k}$ and then, by the observation made before, a `history' $(r_i^{(0)},\ldots,r_i^{(\k-1)})$ of the past mappings of block $i$ at the beginning of the execution of every data request from $1$ to $\k$. These mappings, in turn, are $\k$ outputs of $\PRNGBM$, and they are given as input to the algorithm $\D$, which then outputs a candidate prediction $r^*$ for the current secret leaf value $r_i^{(\k)}$.

Then \A executes his challenge query by using data requests $(\dr^0,\dr^1)$ with $\dr^0  = (\text{`write'},i,\data)$, and $\dr^1  = (\text{`write'},j,\data)$ for $j\neq i$, and records the resulting access pattern $\ap_{\k+1} = \ap(\dr^b)$ (where $b$ is the secret bit to be guessed). At this point, the adversary looks at this last communication transcript $\com_{\k+1}$ and, by the observation made at the beginning of the proof, checks the leaf index $r$ related to the tree branch exchanged during the execution of the challenge query. If $r = r^*$, then \A sets $b' = 0$ (where $b'$ is \A's current `guess' at $b$), otherwise \A sets $b' = 1$.

However, before outputting his guess $b'$ in order to win the AP-IND-CQA game, \A has to perform an additional check (during the second CQA challenge phase) in order to verify whether \D had correctly guessed the right value $r_i^{(\k)}$ or not. The problem here is that, if \D is unsuccessful (which happens with probability as high as $1-\delta$), we cannot say anything about the predicted value $r^*$. In fact, in that case \D could potentially act maliciously against \A, and output a value $r^*$ which maximizes the probability of $b'$ being wrong in the above strategy: for example, $r^* = r_j^{(0)}$. For this reason \A performs the following `sanity check' after the challenge query:
\begin{itemize}
\item if $b' = 1$, then \A demands the execution of an additional query of the form $(\text{`write'},i,\data)$, and verifies that the resulting path leads to leaf $r^*$. This guarantees that $r^*$ was actually correct, and it was not observed during the challenge query just because $\dr^1$ was chosen, as guessed.
\item Otherwise, if $b'=0$, then \A demands the execution of an additional query of the form $(\text{`write'},j,\data)$, and verifies that the resulting tree branch {\em does not} lead to leaf $r^*$. This guarantees with high probability that \D did not maliciously output the secret leaf state for element $j$ instead of $i$.
\end{itemize}
It is easy to see that in the case of misbehavior of \D, both of the above tests fail with high probability. In fact, in the case $b'=1$, the current mapping of element $i$ leads to leaf $r_i^{(\k)}$, which was {\em not} correctly predicted by \D by assumption. In the latter case instead, recall that \A had guessed $b'=0$ because during the execution of the challenge query he observed the leaf $r^*$; this could only lead to a fail in the case that $r_j^{(0)} = r_i^{(\k)}$, which only happens with negligible probability at most $\epsilon$, or if $r_j^{(0)} = r^*$, which is detected by the sanity check.

Finally, if the above sanity check is passed, \A outputs $b'$, otherwise he outputs a random bit.

Notice that (provided \D was successful) this strategy is always correct, {\em except} in the case that: $\dr_1$ was chosen (probability $\half$) {\em and} the initial mapping of $\block_j$ (which is $r_j^{(0)}$), coincides with $r_i^{(\k)}$. As already mentioned, the latter event can only happen at most with probability $\epsilon$ negligible in the bit size of $\PRNGBM$'s output, and hence in the security parameter $\secpar$ (it is easy to see that this is a minimum requirement for any classically secure PRNG, as $\PRNGBM$ is). Thus:
\begin{equation}\label{eqnwin}
\Pr \left[ \game_{\A,\pathoramBM}^{\texttt{AP-IND-CQA}} \to 1
\middle| \D \text{ succeeds} \right] \geq 
 1 - \frac{\epsilon}{2}.
\end{equation}
On the other hand, if $\D$ fails (which happens with probability $(1-\delta)$ at most) and predicts a wrong value $r^* \neq r_i^{(\k)}$, the above strategy still succeeds with probability at least $\half - \frac{\epsilon}{2}$ (again, because of the remote possibility that $r_j^{(0)} = r_i^{(\k)}$). Hence:
\begin{equation}\label{eqnfail}
\Pr \left[ \game_{\A,\pathoramBM}^{\texttt{AP-IND-CQA}} \to 0 
\middle| \D \text{ fails} \right] \leq 
\frac{1}{2} \left( 1 + \epsilon \right).
\end{equation}
Thus, combining~\ref{eqnwin} and~\ref{eqnfail}, the adversary's overall success probability is:
$$
\Pr \left[ \game_{\A,\pathoramBM}^{\texttt{AP-IND-CQA}} \to 1 \right] 
$$$$
= 
\Pr \left[ \hbox{$\A$ wins} \right] \cdot \Pr \left[ \hbox{$\D$ succeeds} \right] + 
\left( 1 - \Pr \left[ \hbox{$\A$ loses} \right] \cdot \Pr \left[ \hbox{$\D$ fails} 
\right]\right)
$$$$
\geq \delta \left( 1 - \frac{\epsilon}{2} \right) + \left( 1 - \left( 1 - \delta \right) 
\frac{1}{2} \left( 1 + \epsilon \right) \right) \geq 
\frac{1}{2} + \frac{1}{2} \delta - \frac{1}{2} \epsilon,
$$
which concludes the proof, because $\epsilon$ is negligible, while $\delta$ is not.
\endproof
\end{proof}

\subsection{Construction of a Post-Quantum ORAM}

A careful examination of \pathoram's construction details reveals that an important role in the security is played by the pseudorandom number generator used to map a block to a leaf during every access. As we have just shown, a PRNG which is not post-quantum secure is enough to break \pathoram's security in a quantum setting. It is natural then to wonder whether the attack on \pathoram can be avoided by using a post-quantum PRNG, {\em in addition} to a post-quantum secure encryption scheme, when instantiating \pathoram. Here, we give a positive answer to such question.

\begin{theorem}\label{thm:pqsecpathoram}
Let $\E$ be a pq-IND-CPA SKE according to Definition~\ref{def:pqINDCPA}, and let $\PRNG$ be a pq-PRNG as from Definition~\ref{dfn:pqPRNG}. Then, \pathoram instantiated using $\E$ and $\PRNG$ is a pq-AP-IND-CPA secure ORAM.
\end{theorem}
\begin{proof}
The proof follows step-by-step the proof of Theorem~\ref{thm:pathoramsec}. In fact this time, since $\PRNG$ is a pq-PRNG by assumption, the new output values used to update the position map in \pathoram are indistinguishable from random (and therefore, in particular, unpredictable) even for \QPT adversaries. As $\PRNG$ has an internal state which is completely unrelated to $\E$'s internal randomness, and because there is no quantum oracle access involved, the security arguments at every step in the proof of Theorem~\ref{thm:pathoramsec} remain unchanged. Therefore, any \QPT adversary who can distinguish the execution of two data request sequences with probability non-negligibly better than guessing, can be turned into a successful adversary against the pq-IND-CPA security of $\E$, or against the pqPRNG, against the security assumptions.
\endproof
\end{proof}

\chapter{QS2: Quantum (Superposition-Based) Security}\label{chap:QS2}

In this chapter we conclude our study of quantum security notions for classical cryptographic objects by presenting the quantum security class \QS2. In this domain, the schemes are classical and the adversaries are quantum, as in \QS1. However, unlike in \QS1, the adversaries are {\em always} given quantum access to classical oracles, not only when the `realistic' model requires it. So, for example, encryption schemes in \QS2 must provide security against adversaries with quantum access to the encryption oracle, even in the secret-key case, and digital signature schemes must be unforgeable toward adversaries with quantum access to the signing oracle, even if such schemes are still classical. What we call here the {\em \QS2 principle} states: {\em ``Whenever an adversary has access to a classical oracle, then such oracle should be accessible by the adversary in a quantum way.''}

As we will see, the resulting security notions can be strictly stronger than `post-quantum' notions as defined in the previous chapter. Constructions which are secure in \QS2 retain in particular their security in \QS1, but the converse does not always hold. \QS2 is, in a sense, quantum security {\em beyond} post-quantum security. When a cryptographic construction is secure in the \QS2 sense, we will just call it {\em quantum-secure}.

In the following sections first we discuss the motivations for considering this scenario, and then we introduce security models and definitions for quantum-secure cryptographic building blocks and secret-key encryption schemes.

\subsection{My Scientific Contribution in this Chapter}

Quantum-secure PRPs (Definitions~\ref{def:qWPRP} and~\ref{def:qSPRP}), and all the results in Section~\ref{sec:QS2enc} first appeared in~\cite{GHS16}, which is a joint work with Andreas Hülsing and Christian Schaffner. Theorem~\ref{thm:pqINDCPAnotoINDqCPA} is considered folklore but, to the best of my knowledge, the first formal proof appears in this thesis.

\section{Why Superposition Access?}\label{sec:whysuperposition}

The obvious question one might ask is: {\em ``why considering quantum access to classical primitives, in the case where the adversary does not implement the primitive's code himself? Doesn't this clash with the \QS1 principle?''} Actually, it does not: the \QS1 principle only states that whenever a quantum adversary can implement some code locally, this should be modeled as a quantum access, {\em but it does not say anything about the converse}. In fact, classical access to an oracle can be seen just as a special case of quantum access, where the adversary is limited to queries in the form of basis states. So, the first `trivial' reason why one should consider quantum access is the following.

\paragraph{Reason \#1: it is a more general model.} Nothing is lost, in terms~of security, by considering adversaries able to execute superposition queries. The resulting security notions will be at least {\em as strong} as the corresponding post-quantum security notions, and sometimes strictly so, as we will see. Of course this does not make post-quantum notions obsolete: for example it might be impossible (or much harder, or worse in performance) to achieve certain \QS2 notions in contrast with the analogous \QS1 notions. It will be the model and the circumstances to dictate whether post-quantum security is enough, or something more should be requested. But for sure, all other factors being equal, one does not lose anything by requesting security in the more challenging scenario considered in \QS2.

\

There are, of course, less `trivial' reasons. We have already met one in Section~\ref{sec:QROMemu} about the emulation of a quantum random oracle: since the QROM describes an object with quantum superposition access by definition, emulating it using post-quantum PRFs would not be enough, because post-quantum PRFs are only accessed classically, and their security model says nothing about what happens when the access is quantum. For this reason, if we want to emulate a quantum oracle with PRFs, we need a security model which covers the quantum superposition access, {\em even} if we are using the quantum random oracle `only' in a post-quantum security proof.

Another example is the case of post-quantum obfuscation, in particular {\em indistinguishability obfuscation (iO)}. This is a relatively recent branch of cryptographic techniques which, roughly speaking, achieves certain functionalities by `obfuscating' the code of some algorithm in a secure way. One typical example (which has also received interest~\cite{obfAES} from an application perspective) is how to build PKES from SKES. The idea is to hardcode the secret key of the SKES in the code of the encryption routine, and then obfuscate the code and distribute it as a public key. In the standard model, it is known~\cite{IR88} that it is impossible to achieve key-exchange and public-key encryption in a black-box way just from one-way functions. However, Corollary~\ref{cor:pqOWFiffpqPRF} and Theorem~\ref{thm:pqGoldreich} tell us that, using iO, it might be possible to build post-quantum PKES from pqOWF. Regardless whether iO is a reasonable assumption or not, it is clear that for this to work, the post-quantum security of the underlying SKES would {\em not} be enough because, as discussed in Section~\ref{sec:pqPKES}, post-quantum PKES can be queried in superposition. Therefore, for this application we also need a superposition-based security notion for SKES.

Summing up, we can say the following.

\paragraph{Reason \#2: it is useful for post-quantum security proofs.} If a security reduction for an object in \QS0 fails when `translating' it to \QS1, one of the reasons (in addition to the ones described in Section~\ref{sec:prooffailures}) might be that the security of some of the underlying building blocks should be `lifted' to \QS2, not just \QS1.

\

A less obvious reason regards the physical interaction between the adversary and the device where the cryptographic code is running. An adversary able to `trick' a classical computation device into quantum behavior might exploit such behavior to gain superposition access to the function computed by the device. In order to fix the ideas on what this actually means we give a motivating example. In this mind experiment, we consider a not-so-distant future where the target of an attack is a tiny encryption chip, e.g., integrated into an RFID tag or smart-card. It is reasonable to assume that it will include elements of technology currently researched but undeployed (i.e., extreme miniaturization, optical electronics, etc.) Regardless, the chip we consider is a purely classical device, performing classical encryption (e.g., AES) on classical inputs, and outputting classical outputs. Consider an adversary equipped with some future technology which subjects the device to a fault-injection environment, by varying the physical parameters (temperature, power, speed, etc.) under which the device usually operates. As a figurative example, our `quantum hacker' could place the chip into an isolation pod, which keeps the device at a very low temperature and shields it from any external electromagnetic or thermal interference. This situation would be analogous to what happens when security researchers perform side channel analysis on cryptographic hardware in nowaday’s labs, using techniques such as thermal or electromagnetic manipulation which were previously considered futuristic. There is no guarantee that, under these conditions, the chip does not start to show full or partial quantum behaviour. At this point, the adversary could query the device on a superposition of plaintexts by using, e.g., a laser and an array of beam splitters when feeding signals into the chip via optic fiber. It is unclear today what a future attacker might be able to achieve using such an attack. As traditionally done in cryptography, we assume the worst-case scenario where the attacker can actually query the target device in superposition. Classical and post-quantum security notions such as IND-CPA do not cover this scenario. This setting is an example of what we mean by
`tricking classical parties into quantum behaviour'.

Another example of a sort of `quantum fault attack' occurs in a situation where one party using a quantum computer encrypts messages for another party that uses a classical computer, and the adversary is able to observe the outcome of the quantum computation before measurement.

\paragraph{Reason \#3: it covers quantum fault attack scenarios.} Also notice that the threat deriving from these kind of attacks is potentially high considering that, unlike for the post-quantum scenario, they do not necessarily require the adversary to build a fully-fledged quantum computer.

\

Finally, it is important to consider superposition-based quantum security in all those cases where a classical cryptographic object is used as a building block for more complex quantum protocols (meant to run natively on quantum computing devices). Post-quantum guarantees alone are usually not enough to ensure secure composition in these scenarios.

\paragraph{Reason \#4: it might be necessary for securely composing fully quantum constructions.} For instance, we will see an example in the next chapter where schemes for securely encrypting quantum data can be built by adapting classical encryption schemes, but only if such schemes are (superposition-based) quantum-secure.

\section{Quantum-Secure Building Blocks}\label{QS2:bb}

We look first at the basic (superposition-based) quantum-secure building blocks. As already discussed in Section~\ref{sec:pqbuildingblocks}, there is nothing to say about quantum-secure OWF, OWTP, and PRNG. In the first two cases, the superposition access is already implied by the post-quantum definition, so that the post-quantum and the superposition-based quantum security notions coincide. We will use the two terms interchangeabily, as the meaning is the same. In the latter case instead, a superposition-based security notion for PRNG makes no sense, because PRNG security, by definition, is based on a stream of classical data, and there is no oracle access involved. As we mentioned already, the situation is instead quite different in the case of PRF and PRP.

\subsection{Quantum-Secure PRF}\label{sec:qPRF}

In the case of pseudorandom functions, an adversary might be able to distinguish the PRF \PRF from a random function by gaining quantum access to the oracle for \PRF, which we denote by $\ket{\qPRF}$. Since a PRF is a keyed family of functions, we write sometimes $\ket{\qPRF_k}$ to denote the quantum-classical oracle for \PRF keyed by $k$.

\begin{definition}[Quantum-Secure Pseudorandom Function (qPRF)]\label{def:qPRF}
A (family of) {\em quantum-secure pseudorandom functions (qPRF)} from \X to \Y with key space \K is a \DPT algorithm $\qPRF: (k \in \K, x \in \X) \mapsto y \in \Y$ such that for any \QPT algorithm \D it holds:
$$
\left| \Pr_{k \rand \K} \left[ \D^{\ket{\qPRF_k}} \to 1 \right] - \Pr_{\h \rand \Y^\X} \left[ \D^\QRO \to 1 \right] \right| \leq \negl,
$$
where \QRO is a quantum-classical oracle for \h (i.e., a quantum random oracle), and the probabilities are over the choice of $k$ and $\h$, and the randomness of $\D$.
\end{definition}

Obviously, a qPRF is also a pqPRF and, in particular, a PRF. As discussed in Section~\ref{sec:PRF}, and unlike in the case of pqPRFs in Section~\ref{sec:pqPRF}, the security proof of Theorem~\ref{thm:PRNGimPRF} does {\em not} go through, because of the impossibility of dealing with the quantum oracle access in the standard way required for such proof. However,~\cite{ZhandryPRF} shows that qPRFs {\em can} indeed be built from post-quantum OWF using standard constructions, so the analogue of Corollary~\ref{cor:pqOWFiffpqPRF} still holds. The following is a corollary of~\cite[Theorem~4.5]{ZhandryPRF}.

\begin{theorem}\label{cor:pqOWFiffqPRF}
pqOWF exist iff qPRF exist.
\end{theorem}

\subsection{Quantum-Secure PRP}\label{sec:qPRP}

Quantum-secure PRPs are defined in a similar way as qPRFs, denoting by $\ket{\qPRP_k}$ the quantum-classical oracle evaluating \qPRP with secret key $k$.

\begin{definition}[Quantum-Secure Weak PRP (qWPRP)]\label{def:qWPRP} \ 
A (family of) {\em quantum-secure weak pseudorandom permutations (qWPRP)} on \X with key space \K is a pair of \DPT algorithms $(\!\qPRP\!,\!\qPRP^{-1}\!)\!:\!(k\!\in\!\K, x\!\in\!\X)\!\mapsto\!x'\!\in\!\X$ such that:
\begin{enumerate}
\item $\forall k \in \K \implies \qPRP_k, \qPRP^{-1}_k$ are permutations on $\X$;
\item $\forall k \in \K \implies (\qPRP_k)^{-1} = \qPRP^{-1}_k$; and
\item for any \QPT algorithm \D it holds:
$$
\left| \Pr_{k \rand \K} \left[ \D^{\ket{\qPRP_k}} \to 1 \right] - \Pr_{\p \rand S(\X)} \left[ \D^{\ket{\O_\p}} \to 1 \right] \right| \leq \negl,
$$
where $\ket{\O_\p}$ is a quantum-classical oracle for \p, and the probabilities are over the choice of $k$ and $\p$, and the randomness of $\D$.
\end{enumerate}
\end{definition}

\begin{definition}[Quantum-Secure Strong PRP (qSPRP)]\label{def:qSPRP}
A (family of) \emph{quan\-tum-secure strong pseudorandom permutations (qSPRP)} on \X with key space \K is a pair of \DPT algorithms $(\qPRP,\qPRP^{-1}):(k\in\K, x\in\X)\mapsto x'\in\X$ such that:
\begin{enumerate}
\item $\forall k \in \K \implies \qPRP_k, \qPRP^{-1}_k$ are permutations on $\X$;
\item $\forall k \in \K \implies (\qPRP_k)^{-1} = \qPRP^{-1}_k$; and
\item for any \QPT algorithm \D it holds:
$$
\left| \Pr_{k \rand \K} \left[ \D^{\ket{\qPRP_k},\ket{\qPRP^{-1}_k}} \to 1 \right] - \Pr_{\p \rand S(\X)} \left[ \D^{\ket{\O_\p},\ket{\O_{\p^{-1}}}} \to 1 \right] \right| \leq \negl,
$$
where $\ket{\O_\p}$ is a quantum-classical oracle for \p, $\ket{\O_{\p^{-1}}}$ is a quantum oracle for $\p^{-1}$, and the probabilities are over the choice of $k$ and $\p$, and the randomness of $\D$.
\end{enumerate}
\end{definition}

It is important to notice that building provably secure qPRPs is not trivial. Kuwakado and Morii showed~\cite{KuwakadoM10,KuwakadoM12} that the two most commonly used constructions for building PRPs are actually {\em quantum-insecure}, in the sense that there exist specific quantum attacks (using a modified version of Simon's algorithm) able to distinguish such constructions from random. Their attacks are limited to the (3-round) Feistel construction (for building WPRPs from PRFs) and the (1-round) Even-Mansour construction (for building SPRPs from public random permutations). However, Zhandry~\cite{ZhandryPRP} shows that qSPRPs {\em can} indeed be built from qPRFs (and hence by post-quantum OWF) using constructions based on {\em format-preserving encryption}, so the analogue of the result from Theorem~\ref{thm:pqPRFiffpqPRP} still holds.

\begin{theorem}[qPRF $\iff$ qPRP]\label{thm:qPRFiffqPRP}
qPRFs exist iff qPRPs exist.
\end{theorem}

\section{Quantum-Secure Secret-Key Encryption}\label{sec:QS2enc}

In Section~\ref{sec:pqPKES}, we have seen how security notions for public-key encryption in the post-quantum setting should allow for an adversary to query the encryption oracle in superposition. Following the \QS2 principle, in this section we extend such a possibility to the secret-key scenario (we limit our analysis here to the CPA case). We start by considering indistinguishability notions for SKESs where the IND phase is still classical, but the adversary has oracle access to the encryption oracle (this would be the analogue, for SKESs, of the pq-IND-CPA notion for PKESs).

Then we look at what happens when also the IND query becomes quantum. 
We also discuss a modification of such scenario, which can be useful in certain situations, where the adversary is restricted to working with quantum messages having efficient classical representations.

We conclude with a brief discussion on the extension of the above models to the CCA1 and CCA2 scenarios.

\subsection{Classical IND, Quantum CPA}

The first indistinguishability notion with quantum CPA query phase, called IND-qCPA, was proposed in~\cite{BZ13}. Formally, the base adversarial model is the same pq-IND adversary from Definition~\ref{def:pqINDadv}.

\begin{experiment}[$\gameINDqCPA$]\label{expt:INDqCPA}
Let $\E$ be a SKES, and $\A:=(\M,\D)$ a pq-IND adversary. The {\em IND-qCPA experiment} proceeds as follows:
\begin{algorithmic}[1]
\State \textbf{Input:} $\secpar \in \NN$
\State $k \from \KGen$
\State $(m^0,m^1,\ket{\state}) \from \M^{\ket{\Enc_k}}$
\State $b \rand\bin$
\State $c \from \Enc_k(m^b)$ 
\State $b' \from \D^{\ket{\Enc_k}}(c,\state)$
\If{$b = b'$}
	\State \textbf{Output:} $1$
\Else
	\State \textbf{Output:} $0$
\EndIf
\end{algorithmic}
The {\em advantage of \A} is defined as:
$$
\advINDqCPA := \Pr \left[ \gameINDqCPA \to 1 \right] - \half .
$$
\end{experiment}

Notice how, as in the \QS0 case, we have: $\gameINDqCPA=\gameIND[\A^{\ket{\Enc_k}}]$.

\begin{definition}[Indistinguishability of Ciphertexts under Quantum Chosen Plaintext Attack (IND-qCPA)]\label{def:INDqCPA}
A SKES $\E$ has {\em indistinguishable encryptions under quantum chosen plaintext attack (or, it is IND-qCPA secure)} iff, for any pq-IND adversary $\A$ it holds that: $\advINDqCPA \leq \negl$.
\end{definition}

Clearly, IND-qCPA is at least as strong as pq-IND-CPA (and it is actually equivalent for PKES). But the converse is not true.

\begin{theorem}[IND-qCPA $\implies$ pq-IND-CPA]\label{thm:INDqCPAtopqINDCPA}
If a SKES is IND-qCPA secure, then it is also pq-IND-CPA secure.
\end{theorem}

\begin{theorem}[pq-IND-CPA SKES \nimplies IND-qCPA SKES]\label{thm:pqINDCPAnotoINDqCPA}
Under standard hardness assumptions, there exist SKES which are pq-IND-CPA secure, but not IND-qCPA secure.
\end{theorem}
\begin{proof}[Proof (sketch)]
Consider the same counterexample described in the proof of Theorem~\ref{thm:INDCPAnotopqINDCPA}, but where this time the public key used for the (IND-CPA but non--post-quantum secure) PKES is generated by $\KGen$ and kept secret. This way, in the post-quantum setting the adversary would lose access to the quantum encryption oracle for the PKES, and hence the pq-IND-CPA security notion coincides with the IND-CPA notion, which the resulting scheme achieves by construction. However, an adversary for the IND-qCPA security notion would still have access to such encryption oracle, thereby being able to break the security of the PKES, and thus recovering the SKES key.
\end{proof}

A simple modification from {\cite[Theorem~4.10]{BZ13}} shows that Construction~\ref{constr:goldreich} is IND-qCPA when instantiated with a {\em quantum-secure PRF}.

\begin{theorem}\label{thm:qGoldreich}
Let $\E_\PRF$ be the SKES from Construction~\ref{constr:goldreich} implemented through a qPRF \PRF. Then $\E_\PRF$ in an IND-qCPA SKES.
\end{theorem}

\subsection{Type-$(2)$ Oracles}

Before discussing other quantum security notions, we must provide a technical tool arising from the following consideration. In quantum computing, the `canonical' way of evaluating an oracle for a classical function $\f$ in superposition is, as discussed in Section~\ref{sec:QS1princ}, by using an auxiliary register and then the canonical quantum-classical oracle:
$$
\ket{\O_\f}: \sum_{x,y} a_{x,y} \ket{x,y} \mapsto \sum_{x,y}  a_{x,y} \ket{x,y \xor \f(x)}.
$$
This way ensures that the resulting operator is invertible, even if $\f$ itself is not. We call these {\em type-$(1)$ transformations}, and we denote them by $\ket{\O_\f}_{(1)}$ when necessary to specify (by default, we assume $\ket{\O_\f} = \ket{\O_\f}_{(1)}$). For SKES, if $\Enc_k$ is an encryption mapping $m$-bit plaintexts to $c$-bit ciphertexts, the resulting operator in this case will act on $m+c$ qubits in the following way:
$$
\ket{\Enc_k}_{(1)}: \sum_{x,y} a_{x,y} \ket{x,y} \mapsto \sum_{x,y} a_{x,y} \ket{x,y \xor \Enc_k(x)},
$$
where the $y$'s are ancillary values.

In our case, though, we do not consider arbitrary functions, but encryptions, which act as {\em bijections} on some bit-string spaces (assuming that the randomness, if in presence of a randomized SKES, is treated as an input, although never chosen by the adversary.) Therefore, provided that the encryption does not change the size of a message, the following transformation is also invertible:
\begin{equation}\label{eq:UEnc2}
\sum_{x} a_{x} \ket{x} \mapsto \sum_{x} a_{x} \ket{\Enc_k(x)}.
\end{equation}
For the more general case of arbitrary message expansion factors, we will consider transformations of the form:
$$
\sum_{x,y} a_{x,y} \ket{x,y} \mapsto \sum_{x,y} a_{x,y} \ket{\phi_{x,y}},
$$
where the length of the ancilla register is $|y|\!=\!|\Enc_k(x)|-|x|$ and ${\phi_{x,0}\!= \Enc_k(x)}$ for every $x$ -- i.e., initializing the ancilla register in the $\ket{0}$ state produces a correct encryption, which is what we expect from an honest execution of the encryption. As in the SKES case we always assume that encryption and decryption oracles are actually provided by an honest third party (usually the {\em challenger}), we will not consider cases where $y\neq0$. We call the resulting operator {{\em type-$(2)$ transformations}}\footnote{These are called {\em minimal quantum oracles} in~\cite{KKVB02}.}, and we denote them by $\ket{\O_\f}_{(2)}$ when necessary to specify), that is:
$$
\Encq: \sum_{x} a_{x} \ket{x,0} \mapsto \sum_{x} a_{x} \ket{\Enc_k(x)},
$$
where the ancillary $\ket{0}$ is of the necessary qubit-size.

Notice that, in general, type-$(1)$ and type-$(2)$ transformations are very different: having quantum gate access to a type-$(2)$ unitary encryption oracle (that is, quantum oracle access to $\Encq$ and its adjoint $\Encqd$) also gives access to the related type-$(2)$ {\em decryption oracle} $\Decq : \sum_x a_x \ket{\Enc_k(x)} \mapsto \sum_x a_x \ket{x}$. In fact, notice that $\Encqd = \Decq$, while the adjoint of a type-$(1)$ encryption operator, $\ket{\Enc_k}_{(1)}^{\dagger}$, is generally {\em not} a type-$(1)$ decryption operator. In particular, type-$(2)$ operators are `more powerful' in the sense that knowledge of the secret key is required in order to build any efficient quantum circuit implementing them. However, we stress the fact that whenever access to a decryption oracle is allowed, the two models are completely equivalent, because then we can simulate a type-$(2)$ operator by using ancilla qubits and `uncomputing' the resulting garbage lines (see Figure~\ref{fig:equiv12}). This is in fact the case in our security model for SKES, as it is not the adversary himself who computes the encryptions, but they are instead provided by a challenger who, in particular, already knows the secret key.

\begin{figure}[h]
\center
 \includegraphics[width=\textwidth,keepaspectratio]{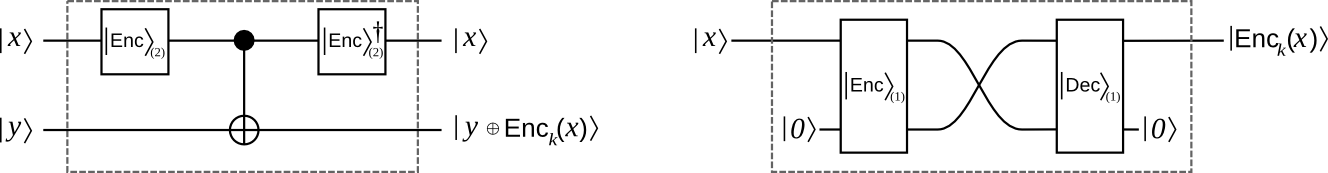}
 \caption{Equivalence between type-$(1)$ and type-$(2)$ in the case of \mbox{$1$-qubit} messages. Left: building a \mbox{type-$(1)$} encryption oracle by using a \mbox{type-$(2)$} encryption oracle (and its inverse) as a \mbox{black-box}. Right: building a \mbox{type-$(2)$} encryption oracle by using \mbox{type-$(1)$} encryption and decryption oracles as \mbox{black-boxes}.\label{fig:equiv12}}
\end{figure}

\subsection{Quantum Indistinguishability}\label{sec:qIND}

When trying to apply the \QS2 principle at its fullest in the context of security notions for SKESs, the main difficulty is how to properly define a quantum version of the IND notion. As shown in~\cite{BZ13}, trying to define a new notion where all the communication and interaction of IND is blindly moved into quantum registers does not work, because the resulting notion would be trivially unachievable. The work~\cite{GHS16} presents an in-depth discussion about other possible strategies spanning a {\em `security tree'} of definitions. Most of these strategies lead to quantum indistinguishability notions that are either unachievable, or equivalent to IND-qCPA. However, some of them lead to more meaningful notions for the \QS2 setting. These notions, in~\cite{GHS16}, are called {\em quantum indistinguishability (qIND)} and {\em general quantum indistinguishability (gqIND)}. However, for the purpose of this work, we rename them as {\em weak quantum indistinguishability (wqIND)} and {\em quantum indistinguishability (qIND)} respectively, because the latter is of more direct interest to our framework. That is, what we call `qIND' in this work was originally called `gqIND' in~\cite{GHS16}, and what we call `wqIND' was originally called `qIND' in~\cite{GHS16}. We will use such denomination from now on, and we will discuss qIND in this section, while presenting wqIND at a later point.

We give the qIND model for the most general case of adversaries able to query oracles on mixed states. This can happen if, for example, the adversary queries the oracle on a state which is entangled with another state kept by the adversary. Basically, what happens in the qIND experiment is the following:
\begin{enumerate}
\item first, the adversary outputs two quantum states $\phi^0,\phi^1$ representing the challenge plaintexts of his choice. These states can be thought as superpositions of classical plaintexts, but in general can also be mixed states, possibly entangled together or with some other state kept by the adversary.
\item Then, these two states are sent over a quantum channel to some abstract challenger algorithm. This challenger selects at random one of the two states and traces out the other one. The selected state is encrypted according to \Encq with a secret key $k$ generated by the challenger, and sent back to the adversary.
\item Finally the adversary, upon receiving such encrypted state, has to guess which of the two states was selected.
\end{enumerate}

More formally, we define the following.

\begin{definition}[Quantum IND Adversary]\label{def:qINDadv}
Let $\E$ be a SKES with plaintext space \X and ciphertext space \Y. A {\em quantum IND (qIND, or QIND) adversary \A for \E} is a pair of \QPT algorithms $\A := (\M,\D)$, where:
\begin{enumerate}
\item $\M: \to \states{\Hilbert_\X} \times \states{\Hilbert_\X} \times \QEnv$ is the {\em qIND (or QIND) message generator};
\item $\D: \states{\Hilbert_\Y} \times \QEnv \to \bin$ is the {\em qIND (or QIND) distinguisher},
\end{enumerate}
where $\Hilbert_{\com}$ is a Hilbert space of appropriate dimension, modeling the state communication register (or, {\em environment}) between \M and \D.
\end{definition}

\begin{experiment}[$\gameqIND$]\label{expt:qIND}
Let $\E$ be a SKES, and $\A:=(\M,\D)$ a qIND adversary. The {\em qIND experiment} proceeds as follows:
\begin{algorithmic}[1]
\State \textbf{Input:} $\secpar \in \NN$
\State $k \from \KGen$
\State $(\phi^0,\phi^1,\sigma) \from \M$
\State $b \rand\bin$
\State $\psi \from \Encq(\phi^b)$
\State trace out $\phi^{1-b}$
\State $b' \from \D(\psi,\sigma)$
\If{$b = b'$}
	\State \textbf{Output:} $1$
\Else
	\State \textbf{Output:} $0$
\EndIf
\end{algorithmic}
The {\em advantage of \A} is defined as:
$$
\advqIND := \Pr \left[ \gameqIND \to 1 \right] - \half .
$$
\end{experiment}

\begin{definition}[Quantum Indistinguishability of Ciphertexts (qIND)]\label{def:qIND}
A SKES $\E$ has {\em quantum indistinguishable encryptions (or, it is qIND secure)} iff, for any qIND adversary $\A$ it holds that: $\advqIND \leq \negl$.
\end{definition}

We can strengthen this security notion by adding quantum CPA capabilities to the adversary.

\begin{experiment}[$\gameqINDqCPA$]\label{expt:qINDqCPA}
Let $\E$ be a SKES, and $\A:=(\M,\D)$ a qIND adversary. The {\em qIND-qCPA experiment} proceeds as follows:
\begin{algorithmic}[1]
\State \textbf{Input:} $\secpar \in \NN$
\State $k \from \KGen$
\State $(\phi^0,\phi^1,\sigma) \from \M^\Encq$
\State $b \rand\bin$
\State $\psi \from \Encq(\phi^b)$
\State trace out $\phi^{1-b}$
\State $b' \from \D^\Encq(\psi,\sigma)$
\If{$b = b'$}
	\State \textbf{Output:} $1$
\Else
	\State \textbf{Output:} $0$
\EndIf
\end{algorithmic}
The {\em advantage of \A} is defined as:
$$
\advqINDqCPA := \Pr \left[ \gameqINDqCPA \to 1 \right] - \half .
$$
\end{experiment}

\begin{definition}[Quantum Indistinguishability of Ciphertexts Under Quantum Chosen Plaintext Attack (qIND-qCPA)]\label{def:qINDqCPA}
A SKES $\E$ has {\em quantum indistinguishable encryptions under quantum chosen plaintext attack (or, it is qIND-qCPA secure)} iff, for any qIND adversary $\A \Rightarrow \advqINDqCPA \leq \negl$.
\end{definition}

Clearly, qIND-qCPA is at least as strong as IND-qCPA, because a classical IND query is a special case of a quantum IND query.

\begin{theorem}[qIND-qCPA $\implies$ IND-qCPA]\label{thm:qINDqCPAtoINDqCPA}
If a SKES is qIND-qCPA secure, then it is also IND-qCPA secure.
\end{theorem}

However, as we will show later, the converse is not necessarily true.

\begin{corollary}[of Theorem~\ref{thm:impossqIND} and Corollary~\ref{cor:qINDtowqIND}]\label{cor:impossqIND}
There exist SKES which are IND-qCPA secure, but not qIND-qCPA secure.
\end{corollary}

In particular, Construction~\ref{constr:goldreich} (which is IND-qCPA secure according to Theorem~\ref{thm:qGoldreich}) is not qIND-qCPA secure, because it is covered in the impossibility result from Section~\ref{sec:impossqIND}. However,~\cite{GHS16} shows how to build qIND-qCPA secure SKES from qPRPs.

\begin{construction}[{\cite[Construction 6.4]{GHS16}}]\label{constr:qIND}
Let $(\qPRP,\qPRP^{-1})$ be a qPRP over $\X\times\R$ with key space \K, where \X and \R are both of size superpolynomial in \secpar. Define $\E = \E_{\K,\X,\X\times\R}:=(\KGen,\Enc,\Dec)$ as a SKES with key space $\K $, plaintext space \X, and ciphertext space $\X\times\R$, in the following way:
\begin{enumerate}
\item $\KGen \to k$, with  $k \rand \K$;
\item $\Enc_k(x) \to \qPRP_k(x \| r)$, where $r \rand \R$;
\item $\Dec_k(y) := \qPRP^{-1}(y)\project{\X}$.
\end{enumerate}
\end{construction}

Instead of proving the qIND-qCPA security of this construction directly, we prove it instead for another construction which {\em generalizes} it. Construction~\ref{constr:qIND} has the drawback that the message length is upper bounded by the input length of the qPRP (minus the bit length of the randomness). However, like in the case of block ciphers, we can overcome this issue with a {\em mode of operation}. More specifically, we can handle arbitrary message lengths by splitting the message into blocks of a fixed length and applying the encryption algorithm of Construction~\ref{constr:qIND} independently to each message block (using the same key but new randomness for each block). This procedure is akin to a `randomized ECB mode', in the sense that each message block is processed separately, like in the ECB (Electronic Code Book) mode, but in our case the underlying cipher is inherently randomized (since we use fresh randomness for each block), so we can still achieve qCPA security. For simplicity we consider only message lengths which are multiples of the chosen blocksize. The construction can be generalized to arbitrary message lengths using standard padding techniques. Moreover, the randomness for every block can be generated efficiently using a single random seed and a pqPRNG.

\begin{construction}[{\cite[Construction 6.6]{GHS16}}]\label{constr:extension}
Let $(\qPRP,\qPRP^{-1})$ be a qPRP over $\X\times\R$ with key space \K, where \X and \R are both of size superpolynomial in \secpar. For a polynomial function $\l$, let $\M:= \X^\l$ and $\C:= (\X \times \R)^\l$. Define $\E = \E_{\K,\M,\C}:=(\KGen,\Enc,\Dec)$ as a SKES with key space $\K $, plaintext space \M, and ciphertext space $\C$, in the following way:
\begin{enumerate}
\item $\KGen \to k$, with  $k \rand \K$;
\item $\Enc_k(x_1 \| \ldots \| x_\l) \to \qPRP_k(x_1 \| r_1) \| \ldots \| \qPRP_k(x_\l \| r_\l)$,\\where $r_i \rand \R, \foral i = 1,\ldots,\l$;
\item $\Dec_k(y_1 \| \ldots \y_\l) := \qPRP^{-1}(y_1)\project{\X} \| \ldots \| \qPRP^{-1}(y_\l)\project{\X}$.
\end{enumerate}
\end{construction}

Before proving the security of Construction~\ref{constr:extension}, we need a technical lemma. Let us assume w.l.o.g. that $\X = \bin^{\m(\secpar)}, \R = \bin^{\r(\secpar)}$ for polynomial functions \m and \r, so that $\C = \bin^{\l\cdot(\m+\r)}$.

\begin{lemma}[{\cite[Lemma~6.7]{GHS16}}]\label{lem:diamond}
Let $\Phi$ be the quantum channel that takes as input an arbitrary $\m$-qubit state, attaches other $\r$ qubits in state $\ket{0}$, and then applies a permutation picked uniformly at random from $S(\bin^{\m+\r})$ to the computational basis space. Let $\Psi$ be the constant quantum channel which maps any $\m$-qubit state to the totally mixed state $\tau:=\frac{\Id}{2^{\m+\r}}$ on $\m+\r$ qubits. Then, $\| \Phi - \Psi \|_{\diamond} \leq 2^{-\r+2}$.
\end{lemma}
\begin{proof}
In order to consider the fact that the $\m$-qubit input state might be entangled with something else, we have to start with a purification of such a state. This is a bipartite pure $2\m$-qubit state $\ket{\phi}_{XY} = \sum_{x,y} a_{x,y} \ket{x}_X \ket{y}_Y$ whose $\m$-qubit $Y$ register is input into the channel and gets transformed into $\Id_X \otimes \Phi( \ketbra{\phi} ) = \tr_{Z} \ketbra{\psi}$, where:
$$
\ket{\psi} := \sum_{x \in \bin^\m,y \in \bin^\m,\p \in S(\bin^{\m+\r})} a_{x,y} \ket{x}_X \ket{\p(y \| 0 \ldots 0)}_C \ket{\pi}_{Z} .
$$
By definition of the diamond norm, we have to show that for any $2\m$-qubit state $\rho$, we have that $\|(\Id \otimes \Phi)(\rho) - (\Id \otimes \Psi)(\rho)\|_{\tr} \leq 2^{-\r+2}$. Due to the convexity of the trace distance, we may assume that $\rho = \ketbra{\phi}$ is pure with $\ket{\phi}_{XY} = \sum_{x,y} a_{x,y} \ket{x}_X \ket{y}_Y$. Hence, we obtain:
\begin{align*}
(\Id_X &\otimes \Phi)( \ketbra{\phi} ) = \tr_{Z} \ketbra{\psi} \\
&=\frac{1}{2^{\m+\r}!} \sum_{x,x',y,y',\p} a_{x,y} \overline{a_{x',y'}} \ketbra{x}{x'}_X \otimes \ket{\p(y \| 0 \ldots 0)}\bra{ \p(y' \| 0 \ldots 0)}_C\\
&=\frac{1}{2^{\m+\r}!} \sum_{x,x',y} a_{x,y} \overline{a_{x',y}} \ketbra{x}{x'}_X \otimes \sum_\p \ket{\p(y \| 0 \ldots 0)}\bra{ \p(y \| 0 \ldots 0)}_C\\
&\quad + \frac{1}{2^{\m+\r}!} \sum_{x,x',y \neq y'} a_{x,y} \overline{a_{x',y'}} \ketbra{x}{x'}_X \otimes \sum_\p \ket{\p(y \| 0 \ldots 0)}\bra{ \p(y' \| 0 \ldots 0)}_C \\
\end{align*}
\begin{align*}
&=\sum_{x,x',y} a_{x,y} \overline{a_{x',y}} \ketbra{x}{x'}_X \otimes \frac{1}{2^{\m+\r}} \sum_z \ketbra{z}{z}_C\\
&\quad + \sum_{x,x',y \neq y'} a_{x,y} \overline{a_{x',y'}} \ketbra{x}{x'}_X \otimes \frac{1}{2^{\m+\r}(2^{\m+\r}-1)}\sum_{z \neq z'} \ketbra{z}{z'}_C\\
&=\tr_Y \ketbra{\phi} \otimes \tau_C + \chi_{XC}\\
&=(\Id_X \otimes \Psi) (\ketbra{\phi}) + \chi_{XC},
\end{align*}
where we defined the `difference state':
$$
\chi_{XC} := \sum_{x,x',y \neq y'} a_{x,y} \overline{a_{x',y'}} \ket{x}\!\!\bra{x'}_X \otimes \frac{1}{2^{\m+\r}(2^{\m+\r}-1)}\sum_{z \neq z'} \ket{z}\!\!\bra{z'}_C .
$$
In order to conclude, it remains to show that $\| \chi_{XC} \|_{\tr} \leq 2^{-\r+2}$. For the $C$-register $\chi_C = \frac{1}{2^{\m+\r}(2^{\m+\r}-1)}\sum_{z \neq z'} \ket{z}\!\!\bra{z'}_C$, one can verify that the $2^{\m+\r}$ eigenvalues are $(\lambda \cdot (2^{\m+\r}-1),-\lambda,-\lambda,\ldots,-\lambda)$ where $\lambda:=\frac{1}{2^{\m+\r}(2^{\m+\r}-1)}$. Hence, the trace norm (which is the sum of the absolute eigenvalues) is exactly $\lambda \cdot 2(2^{\m+\r}-1) = 2^{-\m-\r+1}$.

For the $X$-register, we split $\chi_X$ into two parts $\chi_X = \xi_X - \xi'_X$ where:
\begin{align*}
\xi_X &:= \sum_{x,x'} \ket{x}\!\!\bra{x'} \sum_{y,y'} a_{x,y} \overline{a_{x',y'}};\\
\xi'_X &:= \sum_{x,x'} \ket{x}\!\!\bra{x'} \sum_{y} a_{x,y} \overline{a_{x',y}},
\end{align*}
and use the triangle inequality for the trace norm $\|\chi_X\|_{\tr} = \|\xi_X - \xi'_X\|_{\tr} \leq \|\xi_X\|_{\tr} + \|\xi'_{X}\|_{\tr}$. Observe that $\| \xi_X \|_{\tr} = \| \sum_{x,y} a_{x,y} \ket{x} \sum_{x',y'} \overline{a_{x',y'}} \bra{x'} \|_{\tr} = \| \ketbra{s} \|_{\tr}$ for the (non-normalized) vector $\ket{s} := \sum_{x,y} a_{x,y} \ket{x}$. Hence, the trace norm $\| \xi_X \|_{\tr} = | \braket{s|s} | = \sum_x | \sum_y a_{x,y} |^2 \leq \sum_x \sum_y |a_{x,y}|^2 \cdot 2^\m = 2^\m$ by the Cauchy-Schwarz inequality and the normalization of the $a_{x,y}$'s. Furthermore, we note that $\xi'_X$ is exactly the reduced density matrix of $\ket{\phi}_{XY}$ after tracing out the $Y$ register. Hence, $\xi'_X$ is positive semi-definite and its trace norm is equal to its trace which is $1$. In summary, we have shown that:
\begin{align*}
\| \chi_{XC} \|_{\tr} &= \| \chi_X \|_{\tr} \cdot \| \chi_C \|_{\tr} \leq (\|\xi_X - \xi'_X \|_{\tr} ) \cdot 2^{-\m-\r+1}\\
&\leq (\|\xi_X\|_{\tr} + \| \xi'_X \|_{\tr} ) \cdot 2^{-\m-\r+1} \leq (2^\m + 1) \cdot 2^{-\m-\r+1} \leq 2^{-\r+2}.
\end{align*}\endproof
\end{proof}

If we consider a slightly different encryption channel $\Phi^T$ which still maps $\m$ qubits to $\m+\r$ qubits but where the permutation $\p$ is not picked uniformly from the whole set $S(\bin^{\m+\r})$, but instead we are guaranteed that a certain subset $T \subset \bin^{\m+\r}$ of outputs never occurs in these permutations, we can see such permutations as picked uniformly at random from a smaller set $S(\bin^{\m+\r} \setminus T)$. In this setting, we are interested in the distance of the channel (modeling the encryption operation) $\Phi^T$ from the slightly different constant channel $\Psi^T$ which maps all inputs to the $(\m+\r)$-qubit state $\tau^T$ which is completely mixed on the smaller set $\bin^{\m+\r} \setminus T$ of basis elements. The set $T$ represents `forbidden' values that the encryption algorithm does never produce if we assume certain conditions on the randomness used. This technique will be used in the proof of the next theorem. By modifying slightly the proof of Lemma~\ref{lem:diamond} we get the following.

\begin{corollary}[{\cite[Corollary~6.8]{GHS16}}]\label{cor:diamond}
Let $\Phi^T,\Psi^T$ be quantum channels described as above. Then:
\begin{equation}\label{eq:corbound}
\| \Phi^T - \Psi^T \|_{\diamond} \leq \frac{4}{2^{\r} - |T|/2^\m}.
\end{equation}
\end{corollary}

We can now prove the qIND-qCPA security of Construction~\ref{constr:extension}.

\begin{theorem}[{\cite[Theorem 6.9]{GHS16}}]\label{thm:extension}
Let $\E$ be the SKES from Construction~\ref{constr:extension} implemented through a (weak) qPRP family $(\qPRP,\qPRP^{-1})$. Then $\E$ in a qIND-qCPA SKES.
\end{theorem}
\begin{proof}
We want to show that no \QPT adversary can win the qIND-qCPA game with probability substantially better than guessing. We first transform the game through a short game-hopping sequence into a computationally equivalent game for which we can bound the success probability of the quantum distinguisher \D.

\

$\mathbf{Game_0:}$ this is the original qIND-qCPA game.

\

$\mathbf{Game_1:}$ this is like $\game_0$, but instead of using a permutation drawn from the qPRP family $\qPRP$, a random permutation $\p \in S(\bin^{\m+\r})$ is chosen from the set of all permutations over $\bin^{\m+\r}$. The difference in the success probability of \D winning one or the other of these two games is negligible, otherwise, we could use \D to distinguish a random permutation drawn from $\qPRP$ from one drawn from $S(\bin^{\m+\r})$. This would contradict the assumption that $\qPRP$ is a qPRP. 

\

$\mathbf{Game_2:}$ this is like $\game_1$, but \D is guaranteed that the randomness used for each encryption query are $\l$ new random $\r$-bit strings that were not used before. In other words, the challenger keeps track of all random values used so far and excludes those when sampling a new randomness. Since in $\game_1$ the same randomness is sampled twice only with negligible probability, the probabilities of winning these two games differ at most negligibly.

\

$\mathbf{Game_3:}$ this is like $\game_2$, except that the answer to each query asked by \D also contains the randomness $r_1,\ldots,r_{\l}$ used by the challenger for answering that query. Clearly, \D's probability of winning this game is at least the probability of winning $\game_2$. 

\

When $\game_3$ starts, the qIND message generator \M (where $\A=(\M,\D)$ is the qIND adversary as in Definition~\ref{def:qINDadv}) chooses two different plaintext states. One of them is chosen at random and sent back encrypted with fresh randomness values $\hat{r}_1, \ldots,\hat{r}_{\l}$. Let $Q$ denote the set of $q \cdot \l = \poly(\secpar)$ query values used during the previous $q$ queries to $\ket{\Enc_k}$ in the first learning qCPA-phase. We have to consider that from this phase, \D knows a set $T \subset \bin^{\m+\r}$ of `taken' outputs (ciphertexts), i.e., he knows that any $\p(x\|\hat{r}_i)$ will not take one of these values, as $\hat{r}_i$ has not been used before. So, from the adversary's point of view, $\p$ is a permutation randomly chosen from $S'$, the set of those permutations over $\bin^{\m+\r}$ that fix these $|T|$ values. In order to simplify the proof, we will consider a very conservative bound where $|T| = q \cdot \l \cdot 2^\m$, and the size of $S'$ is $|S'| = (2^{\m+\r}-|T|)!$. Notice that this bound is very conservative because it assumes that the adversary learns $2^\m$ different (classical) ciphertexts for each one of the $q \cdot \l$ `taken' randomness values but, as we will see, this knowledge is still insufficient to win the game.

By construction, the encryption of an $(\l \cdot \m)$-qubit (possibly mixed) state $\rho$ is performed in $\l$ separate blocks of $\m$ qubits each. We are guaranteed that fresh randomness is used in each block, hence it follows from Corollary~\ref{cor:diamond} that $\Enc_k(\rho)$ is negligibly close to the ciphertext state where the first $\m+\r$ qubits are replaced with the completely mixed state (by noting that $\frac{|T|}{2^\m} = \m \cdot q$ is polynomial in $\secpar$ in our case, and hence the right-hand side of~\eqref{eq:corbound} is negligible). Another application of Corollary~\ref{cor:diamond} gives negligible distance to the ciphertext state where the first $2(\m+\r)$ qubits are replaced with the completely mixed state, etc. After $\l$ applications of Corollary~\ref{cor:diamond}, we have shown that $\Enc_k(\rho)$ is negligibly close to the totally mixed state on $\l(\m+\r)$ qubits. As this argument can be made for any plaintext state $\rho$, we have shown that, from $\D$'s point of view, all encrypted states have negligible distance from the totally mixed state, and therefore cannot be distinguished. This holds regardless of any additional query during the second qCPA phase, because a polynomial number of such queries cannot change this distance by more than a negligible amount.
\endproof
\end{proof}

\begin{corollary}[{\cite[Theorem 6.9]{GHS16}}]\label{cor:qIND}
Let $\E$ be the SKES from Construction~\ref{constr:qIND} implemented through a (weak) qPRP family $(\qPRP,\qPRP^{-1})$. Then $\E$ in a qIND-qCPA SKES.
\end{corollary}

Notice how the security of Constructios~\ref{constr:qIND} and~\ref{constr:extension} does not require {\em strong} qPRPs. The reason is that, even if we are considering type-$(2)$ transformations (which could be used to compute $\p^{-1}$), these transformations are never implemented directly by the adversary, but only evaluated as oracles. And since we only consider CPA quantum oracles here, and not CCA, the adversary is never granted access to the decryption oracle. Hence, $\p^{-1}$ is not needed by the reduction. However, extending the constructions to CCA1 security {\em would} require strong qPRPs.

\subsection{Weak Quantum Indistinguishability}

Before providing further results related to the qIND notion, we introduce here a slight relaxation of qIND which might be of use in certain contexts which we explain in this section. The idea is to restrict the power of the adversary in the qIND notion, by only allowing quantum states of a certain form for the qIND challenge phase. This notion was originally introduced in~\cite{GHS16} as `qIND' but, as already mentioned at the beginning of this section, we relabel it as `wqIND' (where `w' stands for `weak') for consistency with our framework.

We start by defining the `restricted' quantum states which can be used by the adversary in the new security notion.

\begin{definition}[Classical Description of Quantum States]\label{def:classrepr}
A {\em classical description} of a quantum state $\rho$ is a (classical) bit string $\desc{\rho}$ describing a quantum circuit which (takes no input but starts from a fixed initial state $\ket{0}$ and) outputs $\rho$.
\end{definition}

We deviate here from the traditional meaning of `classical description' referring to individual numerical entries of the density matrix. The reason is that Definition~\ref{def:classrepr} also covers the cases where those numerical entries are not easily computable, as long as we can give an explicit constructive procedure for that state. Clearly, every pure quantum state $\ket{\phi}$ has a classical description given by a description of the quantum circuit which implements the unitary that maps $\ket{0}$ to $\ket{\phi}$. The classical description of a mixed state $\rho_A$ is given by the circuit which first creates a purification $\ket{\phi}_{AR}$ of $\rho_A$ and then only outputs the $A$ register. Note that a state admitting a classical description cannot be entangled with any other system. We say that a state has an {\em efficient classical representation} if it has a classical representation, and such representation has a bit size at most polynomial in some security parameter \secpar. In this case, we assume the existence of a (fixed, public, canonical) \QPT algorithm \qbuild which, given as input a classical description of a quantum state, outputs that state, i.e., $\qbuild(\desc{\rho}) \to \rho$ (the notation for the output is probabilistic, because $\rho$ could be a mixed state, i.e., a distribution on pure states).

In classical models, there is no difference between sending a description of a message or the message itself. In the quantum world, there is a big difference between these two cases, as the latter allows the adversary \A to establish entanglement of the message(s) with other registers. This is not possible when using classical descriptions. It might intuitively appear that the more general model considered for the qIND notion is more natural. However, the above scenario models the case where \A is well aware of the message that is encrypted, but the message is not constructed by \A himself. Giving \A the ability to choose the challenge messages for the qIND game models the worst case that might happen: \A knows that the ciphertext he receives is the encryption of one out of the two messages that he can distinguish best. This closely reflects the intuition behind the classical IND notion: in that game, the adversary is allowed to send the two messages not because in the real world he would be allowed to do so, but because we want to achieve security even for the best possible choice of messages from the adversary's perspective. Hence, the model using classical descriptions of quantum states is a valid alternative.

\begin{experiment}[$\gamewqINDqCPA$]\label{expt:wqINDqCPA}
Let $\E$ be a SKES, and $\A:=(\M,\D)$ a qIND adversary. The {\em wqIND-qCPA experiment} proceeds as follows:
\begin{algorithmic}[1]
\State \textbf{Input:} $\secpar \in \NN$
\State $k \from \KGen$
\State $(\desc{\phi^0},\desc{\phi^1},\sigma) \from \M^\Encq$
\State $b \rand\bin$
\State $\phi^b \from \qbuild(\desc{\phi^b})$
\State $\psi \from \Encq(\phi^b)$
\State $b' \from \D^\Encq(\psi,\sigma)$
\If{$b = b'$}
	\State \textbf{Output:} $1$
\Else
	\State \textbf{Output:} $0$
\EndIf
\end{algorithmic}
The {\em advantage of \A} is defined as:
$$
\advwqINDqCPA := \Pr \left[ \gamewqINDqCPA \to 1 \right] - \half .
$$
\end{experiment}

\begin{definition}[Weak Quantum Indistinguishability of Ciphertexts Under Quantum Chosen Plaintext Attack (wqIND-qCPA)]\label{def:wqINDqCPA}
A SKES $\E$ has {\em weakly quantum indistinguishable encryptions under quantum chosen plaintext attack (or, it is wqIND-qCPA secure)} iff, for any qIND adversary $\A$ it holds: $\advwqINDqCPA \leq \negl$.
\end{definition}

Clearly, qIND-qCPA is at least as strong as wqIND-qCPA, because quantum states admitting an efficient classical description (used in wqIND) are just a special case of arbitrary quantum plaintext states (used in qIND).

\begin{theorem}[{\cite[Theorem 3.3]{GHS16}}]\label{thm:qINDqCPAtowqINDqCPA}
If a SKES is qIND-qCPA secure, then it is also wqIND-qCPA secure.
\end{theorem}

\begin{corollary}[qIND $\implies$ wqIND]\label{cor:qINDtowqIND}
If a SKES is qIND secure, then it is also wqIND secure.
\end{corollary}

Finding a separation between wqIND and qIND is an open problem, as explained in~\cite{GHS16}. Morally, the notion wqIND-qCPA should lie somewhere between IND-qCPA and qIND-qCPA, because it covers indistinguishability for messages which are not necessarily classical, but not arbitrarily quantum. The reason for considering the seemingly artificial wqIND is that in the context of classical encryption schemes resistant to superposition quantum access, it is important to not lose focus of what the capabilities of a `reasonable' adversary should be. Namely, recall the following classical IND argument: {\em `allowing the adversary to send plaintexts to the challenger is equivalent to the fact that indistinguishability must hold even for the most favorable case from the adversary's perspective'}. Such an argument does {\em not} hold anymore quantumly. In fact, the qIND model presents the following issues:
\begin{enumerate}
\item it allows entanglement between the adversary and the IND challenger: \A could prepare a state of the form $\rho_{AB} = \frac{1}{\sqrt{2}}\ket{00}+\frac{1}{\sqrt{2}}\ket{11}$, sending $\rho_A$ as a plaintext but keeping $\rho_B$; and
\item it allows the adversary to create certain non-reproduceable states. For example, consider the state $\ket{\psi} = \sum_{x \in \X} \frac{1}{\sqrt{|\X|}} \ket{x,\h(x)}$, where $\h$ is a collision-resistant hash function. \A could measure the second register, obtaining a random outcome $y$, and knowing therefore that the remaining state is the superposition of the preimages of $y$, i.e.:
$$
\ket{\psi_y} = \sum_{x \in \X :\h(x)=y} \frac{1}{\sqrt{|\set{x \in \X:\h(x)=y}|}} \ket{x}.
$$
\A could then use $\ket{\psi_y}$ as a plaintext in the challenge phase, but note that $\A$ cannot reproduce $\ket{\psi_y}$ for a given value $y$.
\end{enumerate}
Both of the above examples highlight adversary capabilities which might be considered unreasonably strong in certain scenarios. Entanglement between \A and the IND challenger \C represents a sort of `quantum watermarking' of messages, which goes beyond what a meaningful notion of indistinguishability should achieve. Knowledge of intermediate, unpredictable measurements also renders \A too powerful, because it gives \A access to information not available to \C itself; e.g., in the example above \C would not even know the value of $y$. As it is \C who prepares the state to be encrypted by running \qbuild, it is reasonable to assume that it is \C who should know these intermediate measurements, not \A. In the example above, what \A could see instead (provided he knows the circuit generating the state, as we assume in wqIND) is that the plaintext is a mixture $\Psi=\sum_y \psi_y$ for all possible values of $y$.

The possibility offered by qIND of allowing the adversary to play the IND game with arbitrary states is certainly elegant from a theoretical point of view, but from the perspective of the quantum security of the kind of schemes we are considering, it is sometimes useful to consider the restricted notion wqIND, because it  inherently provides guidelines and reasonable limitations on what a quantum adversary can or cannot do. Also, wqIND is often easier to deal with: notice that in such a model, unlike in the qIND model, \A always receives back an unentangled state from a challenge query. In security reductions, this means that we can more easily simulate the challenger, and that we do not have to take care of measures of entanglement when analyzing the properties of quantum states - for example, indistinguishability of states can be shown by only resorting to the {\em trace norm} instead of the more general {\em diamond norm} as in the proof of Theorem~\ref{cor:qIND}.

Finally, it is important to notice that it is actually unclear whether a separation between qIND and wqIND can be found at all in the realm of classical encryption schemes. In fact, all the positive results present in~\cite{GHS16} hold for the more general qIND notion, while the impossibility result we present in Section~\ref{sec:impossqIND} holds for both qIND and wqIND.

\subsection{Quantum Semantic Security}

In this section, we discuss notions of semantic security in \QS2. All of them have been presented before in~\cite{GHS16}. We start by defining a semantic security equivalent of IND-qCPA, called {\em SEM-qCPA}. This is just the usual notion of SEM, augmented by giving to the adversary qCPA capabilities. In order to not overload notation, we refer to `adversary' and `simulator' simply as \QPT versions of the \PPT algorithms from Definition~\ref{def:SEMadvsim}.

\begin{definition}[{\cite[Definition 4.1]{GHS16}}]\label{def:SEM-qCPA}
A SKES $\E$ is {\em semantically secure under quantum chosen plaintext attack (or, it is SEM-qCPA secure)} iff, for any \QPT adversary \A there exists a \QPT simulator \S such that, for every efficiently computable $\f,\h:\bin^* \to \bin^*$ polynomially bounded in the input bit size, for every probability ensemble $\M := \family{\M}$, where $\M_n$ are probability distributions over $\X_\secpar$ with $\card{\M_n} = \poly(n)$, such that:
$$
\left| \Pr \left[ \gameSEMA[\A^{\ket{\Enc_k}}](\M,\f,\h) \to 1\right] - \Pr \left[ \gameSEMS[\S^{\ket{\Enc_k}}](\M,\f,\h) \to 1\right] \right| \leq \negl,
$$
where $k \from \KGen$ is the secret key generated during the experiments, and the probabilities are taken over the randomness of $\A,\E,\M,\S$.
\end{definition}

Unsurprisingly, the above notion is equivalent to IND-qCPA. The proof is a straightforward modification of Theorem~\ref{thm:INDiffSEM} by also accounting for the quantum CPA queries.

\begin{theorem}[{\cite[Theorem 5.1]{GHS16}}]\label{thm:INDqCPAiffSEMqCPA}
A SKES is IND-qCPA secure iff it is SEM-qCPA secure.
\end{theorem}

We might ask what happens if the above definition is strenghtened by providing the adversary (and the simulator) {\em quantum} advice, instead of a classical advice $\h(x)$ for some plaintext $x$. The following two cases appear.
\begin{itemize}
\item We might replace the classical function $\h$ with a unitary operator $U$ which, acting on a basis element $\ket{x}$ for a (classical) plaintext $x$, produces a quantum advice state $\ket{\xi}$. The resulting security notion is called {\em quantum advice semantic security under quantum chosen plaintext attack (qaSEM-qCPA)}~\cite[Definition D.1]{GHS16}, and it turns out to be meaningless, because trivially achievable by {\em any} SKES. The reason is that a unitary $U$ can always be inverted as $U^\dagger$ by both adversary {\em and} simulator. Both of them are then able to recover the plaintext given the quantum advice.
\item To fix the above problem, we might allow more general quantum circuits $\U$ that can somehow provide non-reversible information, for example by applying some partial measurement at the end, or by providing \A (resp. \S) only with some output qubits, while tracing out the others. Towards this end let $\U$ be an arbitrary quantum circuit (the {\em advice circuit}) that takes as input a basis element $\ket{x}$ and a quantum state $\rho$ provided by \A (resp. \S) (that includes possibly needed auxiliary registers), and computes a (possibly mixed) quantum advice state $\xi$. The resulting security notion is called {\em ideal quantum advice semantic security under quantum chosen plaintext attack (iqSEM-qCPA)}~\cite[Definition D.2]{GHS16}, and it turns out to be equivalent to IND-qCPA. The reason is that the proof in Theorem~\ref{thm:INDqCPAiffSEMqCPA} only uses the advice function to transmit classical information, and therefore iqSEM-qCPA can be reduced to IND-qCPA.
\end{itemize}

It seems therefore that introducing quantum advice states is not meaningful as long as the messages are still classical. We proceed now instead to present a quantum security notion equivalent to the wqIND-qCPA notion. First of all, we redefine the meaning of quantum SEM adversary and simulator.

\begin{definition}[Quantum SEM Adversary, Quantum SEM Simulator]\label{def:qSEMadvsim}
Let $\E := \E_{\K,\X,\Y}$ be a SKES, and $\Hilbert_\f,\Hilbert_\h$ two Hilbert spaces of appropriate dimension (exponential in the security parameter). A {\em quantum SEM adversary \A for \E} is a \QPT algorithm $\A: \states{\Hilb{\Y}} \times \states{\Hilbert_\h} \to \states{\Hilbert_\f}$. A {\em quantum SEM simulator \S for \E} is a \QPT algorithm $\S: \states{\Hilbert_\h} \to \states{\Hilbert_\f}$.
\end{definition}

The wqSEM notion is given by replacing classical functions $\f$ and $\h$ with quantum CPTP maps $\Sigma,\Xi$, which are quantum circuits taking as input $\m$-qubit quantum states (where $\m$ is the bit size of plaintexts, polynomial in \secpar) and outputting $\poly(\m)$-qubit quantum states. The idea is that, since we are using quantum states with efficient classical representations, we can sample some classical randomness once, and reuse it with \qbuild to create many copies of the same plaintext state.

\begin{experiment}[$\gamewqSEMA$]\label{expt:wqSEMA}
Let $\E$ be a SKES, and $\A$ a quantum SEM adversary. The {\em wqSEM experiment} proceeds as follows:
\begin{algorithmic}[1]
\State \textbf{Input:} $n \in \NN$, CPTP maps $\Sigma,\Xi$ with $\m$-qubit input and $\poly(\m)$-qubit output, $\M := \family{\M}$, where $\M_\secpar$ are probability distributions over a family of randomness spaces $\family{\R}$ ,with $\card{\M_\secpar} = \poly(\secpar)$
\State $k \from \KGen$
\State $r \from \M_n$
\State $\phi \from \qbuild(r)$ \Comment{$\qbuild$ is invoked with randomness $r$}
\State $\psi \from \Encq(\phi)$
\State $\phi \from \qbuild(r)$ \Comment{a second copy of $\phi$ is generated, using the same $r$}
\State $\xi \from \Xi(\phi)$ \Comment{this is the quantum advice state}
\State $\sigma \from \A(\psi,\xi)$
\If{$\sigma$ is computationally indistinguishable from $\Sigma(\phi)$}
	\State \textbf{Output:} $1$
\Else
	\State \textbf{Output:} $0$
\EndIf
\end{algorithmic}
\end{experiment}

We use `computationally indistinguishable' as a shorthand for: `for every \QPT algorithm \D with outputs in \bin (a quantum distinguisher), the probability that the output differs on the two states given as input is negligible'. As usual, a third copy of $\phi$ (to be processed by $\Sigma$) can be generated using the same randomness $r$ and the \qbuild algorithm.

\begin{experiment}[$\gamewqSEMS$]\label{expt:wqSEMS}
Let $\E$ be a SKES, and $\S$ a quantum SEM simulator. The {\em simulated wqSEM experiment} proceeds as follows:
\begin{algorithmic}[1]
\State \textbf{Input:} $n \in \NN$, CPTP maps $\Sigma,\Xi$ with $\m$-qubit input and $\poly(\m)$-qubit output, $\M := \family{\M}$, where $\M_\secpar$ are probability distributions over a family of randomness spaces $\family{\R}$ ,with $\card{\M_\secpar} = \poly(\secpar)$
\State $k \from \KGen$
\State $r \from \M_n$
\State $\phi \from \qbuild(r)$
\State $\xi \from \Xi(\phi)$
\State $\sigma \from \S(\xi)$ \Comment{\S only gets the quantum advice, not the ciphertext}
\If{$\sigma$ is computationally indistinguishable from $\Sigma(\phi)$}
	\State \textbf{Output:} $1$
\Else
	\State \textbf{Output:} $0$
\EndIf
\end{algorithmic}
\end{experiment}

\begin{definition}[Weak Quantum Semantic Security (wqSEM)]\label{def:wqSEM}
A SKES $\E$ is {\em weakly quantumly semantically secure (wqSEM)} iff, for any quantum SEM adversary \A there exists a quantum SEM simulator \S such that, for every CPTP maps $\Sigma,\Xi$ with $\m$-qubit input and $\poly(\m)$-qubit output, for every probability ensemble $\M := \family{\M}$ with polynomial-size support over some randomness space, it holds:
$$
\left| \Pr \left[ \gamewqSEMA(\M,\Sigma,\Xi) \to 1 \right] - \Pr \left[ \gamewqSEMS(\M,\Sigma,\Xi) \to 1 \right] \right| \leq \negl,
$$
where the probabilities are taken over the randomness of $\A,\E,\M,\S$.
\end{definition}

\begin{definition}[Weak Quantum Semantic Security Under Quantum Chosen Plaintext Attack (wqSEM-qCPA)]\label{def:wqSEMqCPA}
A SKES $\E$ is {\em weakly quantumly semantically secure under quantum chosen plaintext attack (wqSEM-qCPA)} iff, for any quantum SEM adversary \A there exists a quantum SEM simulator \S such that, for every CPTP maps $\Sigma,\Xi$ with $\m$-qubit input and $\poly(\m)$-qubit output, for every probability ensemble $\M := \family{\M}$ with polynomial-size support over some randomness space, it holds:
$$
\left| \Pr \! \left[ \! \gamewqSEMA[\A^\Encq]\!(\M,\Sigma,\Xi) \!\to\! 1\! \right] \!-\! \Pr\! \left[ \!\gamewqSEMS[\S^\Encq]\!(\M,\Sigma,\Xi)\! \to\! 1 \!\right] \right| \!\leq \!\negl,
$$
where the probabilities are taken over the randomness of $\A,\E,\M,\S$.
\end{definition}

The resulting wqSEM-qCPA notion is equivalent to wqIND-qCPA.

\begin{theorem}[{\cite[Theorem 5.4]{GHS16}}]\label{thm:wqINDqCPAiffwqSEMqCPA}
A SKES is wqIND-qCPA secure iff it is wqSEM-qCPA secure.
\end{theorem}
\begin{proof}
The proof closely follows the one for Theorem~\ref{thm:INDiffSEM}, with some careful modifications. We prove the theorem by splitting it in two parts.

$\mathbf{wqIND-qCPA \implies wqSEM-qCPA}.$ Let \A be an efficient quantum SEM adversary. We want to show that a quantum SEM simulator \S exists, with roughly the same success probability as \A, by exploiting the wqIND-qCPA security of the encryption scheme. The idea of the proof is to hand \A's circuit as non-uniform advice to the simulator \S. This is allowed, because \A is a \QPT adversary against the wqSEM-qCPA game, and hence \A's circuit has a short classical representation. \S can then build and run \A's circuit, and simulate a qSEM-qCPA experiment for \A by generating a new key and answering all of \A's queries using this key. When \S performs his `real' wqSEM challenge query (using the challenge query generated by \A), he does not receive back a valid ciphertext. However, \S can generate a bogus ciphertext by encrypting (with his own key) the $\ket{1\ldots1}$ basis element of the same size as the original plaintext state. It follows from the indistinguishability of encryptions that \A's success probability in this game must be negligibly close to its success probability with a real ciphertext, otherwise \A would be an efficient distinguisher for the scheme \E.

$\mathbf{wqSEM-qCPA \implies wqIND-qCPA}.$ Assume there exists an efficient wqIND-qCPA distinguisher \D for the scheme \E. Then we show how to construct a \QPT algorithm $\A$ that has oracle access to \D and breaks the wqSEM-qCPA security of the scheme, in the sense that no simulator \S can do better than \A. The construction works as follows: \A starts the \gamewqSEMA game, and then he runs \D, emulating the quantum encryption oracle by simply forwarding all the qCPA queries performed by \D to its own oracle (the $\Encq$ oracle of the wqSEM-qCPA game). When \D executes the wqIND challenge query by sending classical descriptions of two states $\phi^0$ and $\phi^1$, $\A$ produces the wqSEM template $(\M,\Xi,\Sigma)$, with $\M$ such that $\qbuild(r)$ outputs $\phi^0$ for half of the possible values $r \from \M$ and $\phi^1$ for the other half, $\Xi$ is the constant map outputting $\ket{1\ldots 1}$, and $\Sigma$ is the identity map $\Sigma(\rho)=\rho$. Then $\A$ performs a qSEM challenge query with this template. Given challenge ciphertext state $\Encq(\phi^b)$ (for $b\in\bin$), \A forwards it as an answer to \D's wqIND challenge query. As \D distinguishes $\Encq(\phi^0)$ from $\Encq(\phi^1)$ with non-negligible success probability by assumption, \D returns the correct value of $b$ with non-negligible advantage over guessing. Then $\A$, having recorded a copy of the classical descriptions of $\phi^0$ and $\phi^1$, is able to create another copy of $\phi^b$ through \qbuild and compute the state $\Sigma(\phi^b)$ exactly, and consequently win the wqSEM-qCPA game with non-negligible advantage. However, as $\Xi$ generates the same (constant, useless) advice state $\ket{1\ldots 1}$ independently of the encrypted message, no simulator can do better than guessing the plaintext. This concludes the proof.
\endproof
\end{proof}

In this work, we will not explicitly define a notion of quantum semantic security related to qIND. However, we will show in the next chapter how the qIND notion is equivalent to the quantum indistinguishability notion Q-IND (introduced in~\cite{BJ15}) for quantum encryption schemes, when these are obtained by implementing a classical SKES in unitary type-$(2)$ mode. In~\cite{ABF+16}, on the other hand, notions of quantum semantic security are presented, which are proven to be equivalent to Q-IND, and therefore easily adaptable to the case of quantumly-accessible SKES that we consider here.

\subsection{Impossibility Result}\label{sec:impossqIND}

In this section we show how the qIND security notion cannot be achieved by a large class of SKESs: namely, all those schemes which do {\em not} substantially expand the message during encryption. First we formally define what it means for a cipher to expand or keep constant the message size by defining the {\em core function} of a SKES. Intuitively, the definition splits the ciphertext into the randomness and a part carrying the message-dependent information. This definition covers most encryption schemes in the literature.

\begin{definition}[Core Function~{\cite[Definition 6.1]{GHS16}}]\label{def:core}
Let $\E = \E_{\K,\X,\Y}$ be a SKES, and let \R be the randomness space of $\Enc$. Let $\f:\K \times \R \times \X \rightarrow \Y$ be a function such that:
\begin{itemize}
\item for all $k\in\K$ and for all $x \in \X$, $\Enc_k(x)$ can be written as $(r,\f(k,r,x))$, where $r\in\R$ is independent of the message; and
\item there exists a function $\g$ such that for all $k\in\K, r\in\R, x\in\X$ it holds: $\g(k,r,\f(k,x,r)) = x$.
\end{itemize}
Then, we call $\f$ the {\em core function} of the encryption scheme.
\end{definition}

For example, in case of Construction~\ref{constr:goldreich} (where $\Enc_k(x)$ is defined as $(r,\PRF_k(r) \xor x)$ for a PRF $\PRF$) the core function would be $\f(k,r,x) := \PRF_k(r) \xor x$, with associated $\g(k,r,z) := z \xor \PRF_k(r)$.

\begin{definition}[Quasi--Length-Preserving Encryption~{\cite[Definition 6.2]{GHS16}}]
We call a SKES with core function $\f$ {\em quasi--length-preserving} iff:
$$
\foral x \in \X, \foral r \in \R, \foral k \in \K \implies |\f(k,x,r)| = |x|,
$$
i.e., the output of the core function has the same bit length as the plaintext.
\end{definition}

For example, Construction~\ref{constr:goldreich} is quasi--length-preserving.

The crucial observation for our impossibility result is the following: for a quasi--length-preserving encryption scheme, the space of possible input and (core function) output bit strings (with respect to plaintext and ciphertext) coincide, therefore these ciphers act as permutations on these spaces. This means that, if we start with an input state which is a superposition of {\em all} the possible basis states, all of them with the {\em same} amplitude, this state will be left unmodified by the unitary type-$(2)$ encryption operation (because such operator will just `shuffle' in the space of computational basis-states amplitudes which are exactly the same).

\begin{theorem}[{\cite[Theorem 6.3]{GHS16}}]\label{thm:impossqIND}
If a SKES is quasi--length-preserving, then it is not wqIND secure.
\end{theorem}
\begin{proof}
Let $(\Gen,\Enc,\Dec)$ be a quasi--length-preserving scheme. We give an explicit, efficient distinguisher attack.
\begin{enumerate}
\item For $\m$-bit message strings, the distinguisher \D sets the two plaintext states for the qIND- game to be: $\ket{\phi^0} = H \ket{0^m}, \ket{\phi_1} = H \ket{1^m}$, where $H$ is the $\m$-fold tensor Hadamard transformation. Notice that both these states admit efficient classical representations, and are thus allowed in the wqIND game.
\item A random bit $b$ is flipped, and the challenge ciphertext state $\ket{\psi} = \Encq \ket{\phi^b}$ is returned to \D.
\item \D applies $H$ to the core-function part of the ciphertext $\ket{\psi}$ and measures it in the computational basis. \D outputs $0$ iff the outcome is $0^\m$, and outputs $1$ otherwise.
\end{enumerate}

Notice that applying $\Encq$ to $H \ket{0^m}$ leaves the state untouched: since the encryption oracle merely performs a permutation in the basis space, and since $\ket{\phi_0}$ is a superposition of every basis element with the same amplitude, it follows that whenever $b$ is equal to $0$, the ciphertext state will be left unchanged. In this case, after applying the self-inverse transformation $H$ again, $\D$ obtains measurement outcome $0^\m$ with probability $1$.

On the other hand, if $b=1$, then $\ket{\phi^1} = \frac{1}{2^{\m/2}}\sum_y (-1)^{y \cdot 1^\m} \ket{y}$ where $a \cdot b$ denotes the bitwise inner product between $a$ and $b$. Hence, $\ket{\phi^1}$ is a superposition of every basis element where (depending on the parity of $y$) half of the elements have a positive amplitude and the other half have a negative one, but all of them will be equal in absolute value. Applying $\Encq$ to this state results in $ \frac{1}{2^{\m/2}}\sum_y (-1)^{y \cdot 1^\m} \ket{\Enc(y)}$. After re-applying $H$, the amplitude of the basis state $\ket{0^\m}$ becomes $\sum_y (-1)^{y \cdot 1^\m + \Enc(y) \cdot 0^\m} = \sum_y (-1)^{\|y\|}$ (where $\|y\|$ is the {\em Hamming weight of $y$}) which is $0$. Hence, the probability for \D of observing $0^\m$ after the measurement is $0$. This gives \D a way of distinguishing between encryptions of the two plaintext states.
\endproof
\end{proof}

Notice that the above attack works also against qIND, because of Theorem~\ref{thm:qINDqCPAtowqINDqCPA}. In particular, Theorem~\ref{thm:impossqIND} shows that Construction~\ref{constr:goldreich}, which is IND-qCPA secure if the used PRF is quantum secure, does not fulfill qIND, nor wqIND. This attack is a consequence of the well-known fact~\cite{QOTP1,QOTP2} that, in order to perfectly (information-theoretically) encrypt a single quantum bit, {\em two} bits of classical information are needed: one to hide the basis bit, and one to hide the phase (i.e., the signs of the amplitudes). The fact that we are restricted to quantum operations of the form  $\Encq$  (that is, quantum instantiations of classical encryptions) means that we cannot afford to hide the phase as well, and this restriction allows for an easy distinguishing procedure in the case of a quasi--length-preserving SKES.

Summing up up, all the semantic security notions presented in this section are summarized in Figure~\ref{fig:QS2sem}.

\begin{figure}[t]
\begin{center}
\includegraphics[width=\textwidth]{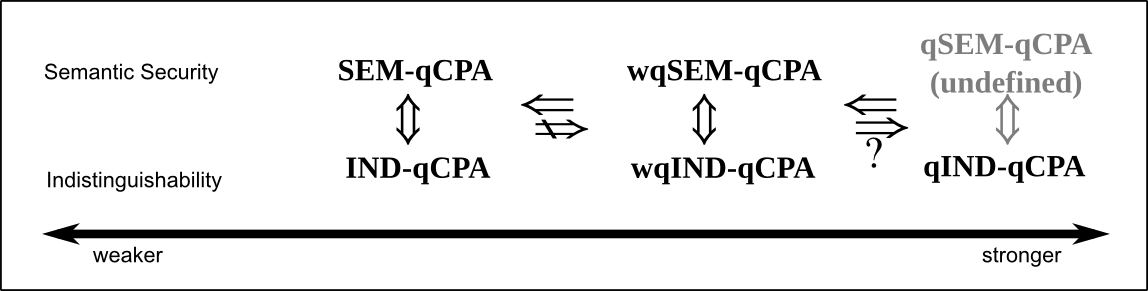}
\end{center}
\caption{Relations for semantic security notions in \QS2.}
\label{fig:QS2sem}
\end{figure}

\subsection{Quantum CCA}

Finally, here we give a brief discussion about the possibility of extending the \QS2 framework of security notions for SKES to the quantum chosen ciphertext attack (qCCA) case. The resulting notions, when applicable, are always stronger than the related qCPA notions, with counterexamples closely matching the classical ones. However, a few issues arise.

The case of quantum CCA1 is straightforward for the classical IND case. The resulting IND-qCCA1 notion is just as the IND-qCPA notion, augmented by a quantum CCA query {\em before} the classical IND query. This is modeled in the security game by giving to the first stage IND adversary oracle access to the quantum decryption oracle $\ket{\Dec_k}$.

The case of wqIND-qCCA1 and qIND-qCCA1 are also straightforward, as the decryption queries only happen before the qIND query. It is just necessary to define the type-$(2)$ decryption oracle $\Decq$, but this is trivial considered that $\Decq=\Encqd$. However, Construction~\ref{constr:extension} will require strong qPRPs in order to be secure under the new notion, as already discussed.

The case of qCCA2, instead, is much more delicate. For the classical IND case,~\cite{BZ13} shows how to correctly define IND-qCCA2 (and how to achieve it), by carefully defining the decryption oracle {\em after} the IND query. For the `fully quantum case' qIND-qCCA2, however, it is unclear whether such a notion is even possible to define. The problem is that in the CCA2 game it is necessary to ensure that the adversary does not ask for a decryption of the challenge ciphertext, leading to a trivial break. While this is easily demanded in the classical world, it raises several issues in the quantum world. What does it mean for a quantum ciphertext state to be different from the challenge ciphertext?
And, more importantly: how can the challenger check? There might be several reasonable ways to solve the first issue but, as long as the queries are not classical, it is not known how to solve the second issue without disturbing the challenge ciphertext and the query states. Defining CCA2 security notions in the quantum world is an outstanding open problem~\cite{GHS16,ABF+16}.

All the indistinguishability notions for classical SKES in the quantum world are summarized in Figure~\ref{fig:QS2relations}.

\begin{figure}[t]
\begin{center}
\includegraphics[width=\textwidth]{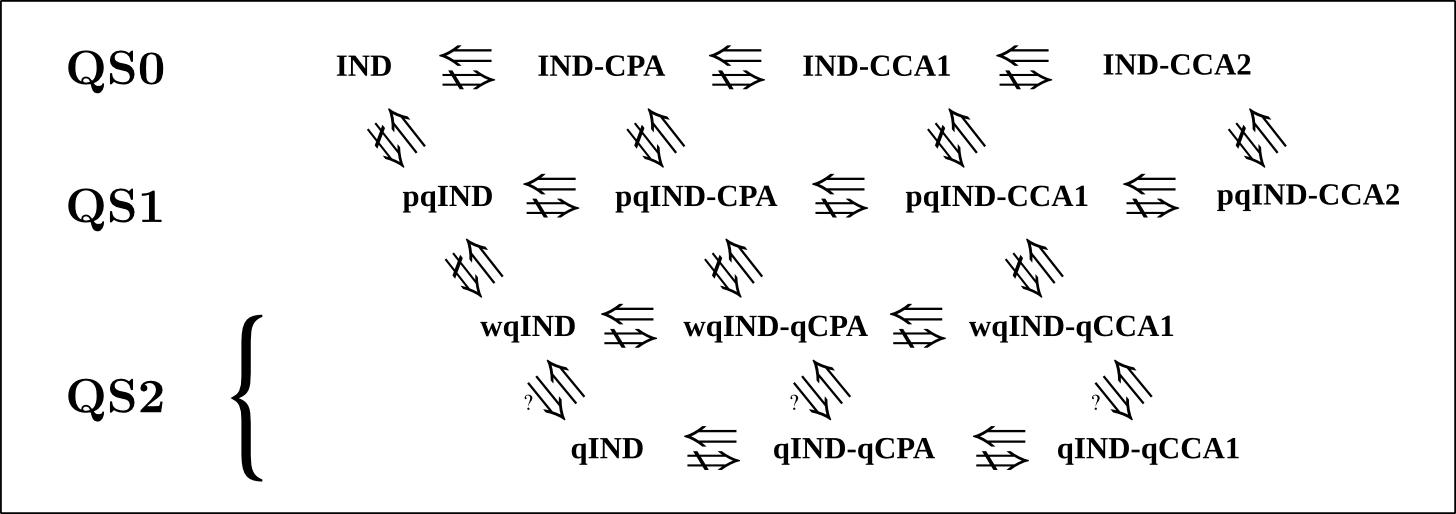}
\end{center}
\caption{Relations for SKES security notions in the quantum world.}
\label{fig:QS2relations}
\end{figure}

\chapter{QS3: Fully Quantum Security}\label{chap:QS3}

In the previous chapters, we studied the security of {\em classical} cryptographic primitives in different quantum scenarios. In this chapter, instead, we focus on the security of {\em quantum} cryptographic primitives, that is, cryptographic primitives which are meant to be natively run on a quantum computing device. The quantum security class \QS3 encompasses all those cryptographic objects which deal mainly with the manipulation and protection of {\em quantum data}. As such, one can see \QS3 as a natural extension of \QS0 to a `fully quantum computing world', that is, a world where quantum computing has become ubiquitous, and honest users have access to quantum devices.

One could consider \QS3 to be somehow `the last step', from a chronological point of view, in the study of computer security, in the sense that the models therein only concern possible future scenarios, somehow far away from the contemporary era of classical devices. However, such interpretation has not to be taken too literally. \QS3 is about {\em security of cryptographic primitives which natively deal with quantum information}, and this does not necessarily involve computation performed on some futuristic, fully-fledged quantum computer. As an example, {\em quantum key distribution (QKD)}~\cite{BB84} is a well-studied area in modern cryptography, where honest parties want to establish a shared secret by using quantum communication channels\footnote{Remarkably, most often than not, the term `quantum cryptography' is (incorrectly) used a synonym for `QKD' in scientific literature.}. As such, QKD perfectly fits in the \QS3 domain; however it is far from being futuristic: commercial implemementations of large-scale QKD systems have been available for a few years already~\cite{idquantique}, and have been deployed in many real-world scenarios.

Despite this, in the rest of this chapter we will focus on the study of the quantum security of cryptographic primitives natively designed to run on a fully-fledged quantum computer. We will first introduce the concept of {\em quantum encryption} (that is, cryptographic schemes meant to protect quantum data), and then we will see an application by extending ORAMs to the case where the database to be protected is composed of quantum data.

\subsection{My Scientific Contributions in this Chapter}

Regarding quantum encryption, most of the material from sections~\ref{sec:QS3ske} and~\ref{sec:QS3pke} first appears in~\cite{ABF+16}, which is a joint work with Gorjan Alagic, Anne Broadbent, Bill Fefferman, Christian Schaffner, and Michael St. Jules.

The part about QORAMs in Section~\ref{sec:QORAM} first appeared in~\cite{GKK17}, which is a joint work with Nikolaos P. Karvelas and Stefan Katzenbeisser.

\section{Secret-Key Quantum Encryption}\label{sec:QS3ske}

In this section, we study the computational security of {\em quantum encryption schemes}, that is, schemes which are meant to protect quantum data. In this sense, plaintexts and ciphertexts are pure quantum states from Hilbert spaces of appropriate dimension, or mixed states of such. In fact, the schemes described in this section are meant to work on {\em arbitrary} quantum states, even those who might be entangled with external systems, therefore it is crucial to use the density matrix formalism. Accordingly, (families of) classical plaintext and ciphertext spaces \X and \Y are replaced with quantum operator spaces \QX and \QY respectively, where $\Hilbert_\X$ and $\Hilbert_\Y$ are (families of) complex Hilbert spaces of dimension $\card{\X} = 2^\m$ and $\card{\Y} = 2^\c$ respectively, for functions $\m$ and $\c$ polynomial in the security parameter \secpar.

However, the encryption keys used will still be classical. This is actually a feature, as these schemes require for honest parties to be able to encrypt and decrypt several times with the same keys, and classical keys can be stored and managed more easily.

\subsection{Definitions, and the Quantum One-Time Pad}

We start by defining {\em secret-key quantum encryption schemes (SKQES)}, as introduced in~\cite{ABF+16}. We assume that the secret-key space is defined as $\K = \family{\K} := \bin^\secpar$, so that the key-length is~$\secpar$ bits. Later, we will define an additional Hilbert space $\HEnv$ (the {\em environment space}) in order to model auxiliary information used by some adversary. Encryption accepts a classical key and a quantum plaintext, and outputs a quantum ciphertext; decryption accepts a classical key and a quantum ciphertext, and outputs a quantum plaintext. The correctness guarantee is that plaintexts are preserved (up to negligible error) under encryption followed by decryption under the same key.

\begin{definition}[Secret-Key Quantum Encryption Scheme (SKQES)]\label{def:qskes}
A {\em se\-cret-key quantum encryption scheme (SKQES)} with plaintext space \QX, ciphertext space \QY, and (classical) key space \K is a tuple of \QPT algorithms $\E := \E_{\K,\QX,\QY} := (\KGen,\QEnc,\QDec)$:
\begin{enumerate}
\item $\KGen: \to \K$;
\item $\QEnc: \K \times \QX \to \QY$;
\item $\QDec: \K \times \QY \to \QX$;
\end{enumerate}
such that $\left| \QDec_k \circ \QEnc_k - \Id_\HX \right|_\diamond \leq \negl$ for all $k \from \KGen$.
\end{definition} 

As usual, we denote by $\QEnc_k$ the action of $\QEnc$ on a specific, fixed key $k \from \KGen$, and analogously for $\QDec_k$. However, unlike in the case of Definition~\ref{def:skes}, for simplicity we will omit the possibility that the decryption algorithm answers (a quantum analogue of) $\bot$ to some decryption queries. One of the most basic examples of SKQES is the {\em quantum one-time pad (QOTP)}. The QOTP takes as input an \secpar-qubit plaintext spaces and a $2\secpar$-bit secret key. Every pair of bits from the key selects one over four possible single-qubit Pauli operators $\Id,X,Y,Z$ as $X^{\text{(first bit)}} Z^{\text{(second bit)}}$. Thus, the secret key defines a sequence of \secpar independent single-qubit Pauli operators, each of them to be applied separately to each of the \secpar qubits of the plaintext (that is, the key defines an element of the $\secpar$-qubit Pauli group), resulting in the ciphertext. Since Pauli operators are self-adjoint, decryption just applies the same procedure to the ciphertext state.

\begin{construction}[Quantum One-Time Pad (QOTP)\cite{QOTP1,QOTP2}]\label{constr:qotp}
Let $\HX = \HY$ of dimension $\bin^\secpar$, and let $\K = \bin^{2\secpar}$. Define the {\em quantum one-time pad (QOTP) on \secpar qubits} $\QOTP:=(\KGen,\QEnc,\QDec)$ as the SKQES with key space $\K$, plaintext space \QX, and ciphertext space \QY, defined as:
\begin{enumerate}
\item $\KGen \to k$, with  $k \rand \K$;
\item $\QEnc_k(\phi) := P(k) \phi P(k)^\dagger$;
\item $\QDec_k(\rho) := P(k) \psi P(k)^\dagger$,
\end{enumerate}
where $P(k) := \prod_{j=1}^{\secpar} X_j^{k_{2j-1}} Z_j^{k_{2j}} \in \pauli_\secpar$, and $k_j$ is the $j$-th bit of $k$.
\end{construction}

Notice how {\em two} bits of key are needed for every qubit of plaintext. The QOTP is known~\cite{QOTP1,QOTP2} to be quantum information-theoretically secure, as long as the key is completely random and only used once.

\subsection{Quantum Indistinguishability}

We use a definition of {\em computational quantum indistinguishability} introduced in~\cite{BJ15}, which we relabel here as QIND for our purposes (notice the capital `Q', unlike Definition~\ref{def:qIND}), and which is the analogue of the classical IND notion, by keeping in mind that a quantum adversary for a SKQES could try to distinguish states that he has previously entangled with the environment. Intuitively, the adversary produces a tripartite system, composed of two plaintext states and an environment state. The environment state is passed to the second stage adversary, who also receives an encryption of one of the two other states, selected at random, while the other one is traced out. As usual, the goal of the adversary is to guess which one of the two plaintext system was selcted for encryption. Formally, we define the following.

\begin{experiment}[$\gameQIND$]\label{expt:QIND}
Let $\E$ be a SKQES, and $\A:=(\M,\D)$ a QIND adversary as from Definition~\ref{def:qINDadv}. The {\em QIND experiment} proceeds as follows:
\begin{algorithmic}[1]
\State \textbf{Input:} $\secpar \in \NN$
\State $k \from \KGen$
\State $(\phi^0,\phi^1,\sigma) \from \M$
\State $b \rand\bin$
\State $\psi \from \QEnc(\phi^b)$
\State trace out $\phi^{1-b}$
\State $b' \from \D(\psi,\sigma)$
\If{$b = b'$}
	\State \textbf{Output:} $1$
\Else
	\State \textbf{Output:} $0$
\EndIf
\end{algorithmic}
The {\em advantage of \A} is defined as:
$$
\advQIND := \Pr \left[ \gameQIND \to 1 \right] - \half .
$$
\end{experiment}

\begin{definition}[Indistinguishability of Quantum Ciphertexts (QIND)]\label{def:QIND}
A SKQES $\E$ has {\em indistinguishable quantum encryptions (or, it is QIND secure)} iff, for any QIND adversary $\A$ it holds that: $\advQIND \leq \negl$.
\end{definition}

Notice how this definition and the related experiment are exactly the same as Experiment~\ref{expt:qIND} and Definition~\ref{def:qIND}, even the adversarial model is the same as in the qIND case from Chapter~\ref{chap:QS2}. This is not incidental: historically, notions of computational indistinguishability for encrypted quantum states have been introduced in~\cite{BJ15} and~\cite{GHS16} as concurrent and independent works (although~\cite{BJ15} was published earlier), but for different purposes and with slightly different flavors. What we call here QIND was originally called q-IND-CPA-2 in~\cite{BJ15} (minus the CPA part), while qIND was originally called $(\C Q n 2 e)$-IND in~\cite{GHS16}. However, the former notion was given in the context of {\em fully homomorphic quantum encryption} (which, according to our framework, belongs to the \QS3 setting), while the latter was given in the context of {\em superposition-resistant quantum encryption} (as we mean it in the \QS2 sense). Further developments on the topic appeared in~\cite{ABF+16} and in the proceedings version of~\cite{GHS16}, which led to the conclusion that this indistinguishability model for quantum encryption is virtually the same, which can be used both in the setting of classical encryption resistant to quantum queries (\QS2) or `fully' quantum encryption (\QS3). In this work, in the attempt of providing a unified notation to work with, we use respectively `qIND' and `QIND' (with different capitalization of the first letter) in order to highlight the specific domain we are talking about, but making clear that, technically, it is the same model.

As usual, we can extend the QIND notion to CPA and non-adaptive CCA attacks. Since we are in the \QS3 domain, it is not ambiguous to write (e.g.) QIND-CPA instead of QIND-QCPA, because the plaintexts we are considering are inherently quantum, so a CPA notion in this scenario {\em must} be quantum. Hence, without need of specifying further, we call the resulting notions QIND-CPA and QIND-CCA1. This is also useful in order to understand `at first glance' that we are talking about a \QS3 notion.

\begin{experiment}[$\gameQINDCPA$]\label{expt:QINDCPA}
Let $\E$ be a SKQES, and $\A:=(\M,\D)$ a QIND adversary. The {\em QIND-CPA experiment} proceeds as follows:
\begin{algorithmic}[1]
\State \textbf{Input:} $\secpar \in \NN$
\State $k \from \KGen$
\State $(\phi^0,\phi^1,\sigma) \from \M^\QEnc$
\State $b \rand\bin$
\State $\psi \from \QEnc(\phi^b)$
\State trace out $\phi^{1-b}$
\State $b' \from \D^\QEnc(\psi,\sigma)$
\If{$b = b'$}
	\State \textbf{Output:} $1$
\Else
	\State \textbf{Output:} $0$
\EndIf
\end{algorithmic}
The {\em advantage of \A} is defined as:
$$
\advQINDCPA := \Pr \left[ \gameQINDCPA \to 1 \right] - \half .
$$
\end{experiment}

\begin{definition}[Indistinguishability of Quantum Ciphertexts Under Chosen Plaintext Attack (QIND-CPA)]\label{def:QINDCPA}
A SKQES $\E$ has {\em indistinguishable quantum encryptions under chosen plaintext attack (or, it is QIND-CPA secure)} iff, for any QIND adversary $\A$ it holds: $\advQINDCPA \leq \negl$.
\end{definition}

Clearly, QIND-CPA is at least as strong as QIND.

\begin{theorem}[QIND-CPA $\implies$ QIND]\label{thm:QINDCPAtoQIND}
If a SKQES is QIND-CPA secure, then it is also QIND secure.
\end{theorem}

However, the converse is not necessarily true. For example, the QOTP (Construction~\ref{constr:qotp}) is information-theoretically secure for random, unrelated keys, and thus it is also QIND. However, as in the classical OTP analogue, security is compromised if the same key is used more than once.

\begin{theorem}[QIND $\nimplies$ QIND-CPA]\label{thm:impossQIND}
There exist SKQES which are QIND secure, but not QIND-QCPA secure.
\end{theorem}

As usual, extending the above security notion to the QCCA1 case is straightforward.

\begin{experiment}[$\gameQINDCCA$]\label{expt:QINDCCA}
Let $\E$ be a SKQES, and $\A:=(\M,\D)$ a QIND adversary. The {\em QIND-CCA1 experiment} proceeds as follows:
\begin{algorithmic}[1]
\State \textbf{Input:} $\secpar \in \NN$
\State $k \from \KGen$
\State $(\phi^0,\phi^1,\sigma) \from \M^{\QEnc,\QDec}$
\State $b \rand\bin$
\State $\psi \from \QEnc(\phi^b)$
\State trace out $\phi^{1-b}$
\State $b' \from \D^\QEnc(\psi,\sigma)$
\If{$b = b'$}
	\State \textbf{Output:} $1$
\Else
	\State \textbf{Output:} $0$
\EndIf
\end{algorithmic}
The {\em advantage of \A} is defined as:
$$
\advQINDCCA := \Pr \left[ \gameQINDCCA \to 1 \right] - \half .
$$
\end{experiment}

\begin{definition}[Indistinguishability of Quantum Ciphertexts Under Non\-Adaptive Chosen Ciphertext Attack (QIND-CCA1)]\label{def:QINDCCA}
A SKQES $\E$ has {\em indistinguishable quantum encryptions under non-adaptive chosen ciphertext attack (or, it is QIND-CCA1 secure)} iff, for any QIND adversary $\A$ it holds:
$$
\advQINDCCA \leq \negl.
$$
\end{definition}

As in the classical case, in a completely specular way to Theorems~\ref{thm:INDCCA1toINDCPA} and~\ref{thm:INDCPAnotoINDCCA1}, one can show that QIND-CCA1 is strictly stronger than QIND-CPA.

\begin{theorem}[QIND-CCA1 $\implies$ QIND-CPA]\label{thm:QINDCCA1toQINDCPA}
If a SKQES is QIND-CCA1 secure, then it is also QIND-CPA secure.
\end{theorem}

\begin{theorem}[QIND-CPA \nimplies QIND-CCA1]\label{thm:QINDCPAnotoQINDCCA1}
There exists a SKQES which is QIND-CPA secure, but not QIND-CCA1 secure.
\end{theorem}

\subsection{Secure Construction}

QIND-CCA1 secure SKQES can be constructed given the existence of pqPRF (and hence from pqOWF, as from Corollary~\ref{cor:pqOWFiffpqPRF}), as shown in~\cite{ABF+16}. The idea of the construction is analogous to the one for Construction~\ref{constr:goldreich}: a (classical) randomness is processed by the keyed pqPRF, and the output is used as a key for the QOTP; the ciphertext is composed by the output of the QOTP, plus the classical randomness.

\begin{construction}[{\cite[Scheme 1]{ABF+16}}]\label{constr:qskes1}
Let $\PRF:\K \times \bin^{2\secpar} \to \bin^{2\secpar}$ be a pqPRF as from Definition~\ref{def:pqPRF}, and let $\HX$ be a complex Hilbert space of dimension $2^\secpar$. Define $\E = \E_{\K,\QX,\QX}:=(\KGen,\QEnc,\QDec)$ as the SKQES with key space $\K$, plaintext and ciphertext space \QX, as follows:
\begin{enumerate}
\item $\KGen \to k$, with  $k \rand \K$;
\item $\QEnc_k(\phi) \to \psi \otimes \ketbra{r}$, with $\psi := \QOTP[\PRF_k(r)](\phi)$, where $r \rand \bin^{2\secpar}$; 
\item $\QDec_k(\rho) \to \QOTP[\PRF_k(s)](\sigma)$, where $s$ is obtained by measuring the last $2\secpar$ qubits of $\rho$, while $\sigma$ is the reduced state left after such a measurement.
\end{enumerate}
\end{construction}

The above construction is QIND-CCA1 secure.

\begin{theorem}[{\cite[Lemma 14]{ABF+16}}]\label{thm:QCCA1constr}
Let \E be the SKQES from Construction~\ref{constr:qskes1} built using a pqPRF \PRF. Then \E is QIND-CCA1 secure.
\end{theorem}
\begin{proof}
First, we analyze the security of the scheme in an idealized scenario where \PRF is replaced by a function $\f: \bin^{2\secpar} \to \bin^{2\secpar}$ selected truly at random. We show that, in this case, \A correctly guesses the challenge state in the QIND-CCA1 game with probability at most $\half + \negl$. In fact, this bound holds for a stronger adversary $\A^*$, who has access to a classical oracle for~$\f$ prior to the challenge, and access to polynomially-many pairs $(r_i, \f(r_i))$ where $r_i \rand \bin^{2\secpar}$ for $1 = 1,\ldots, q= \poly(\secpar)$, after the challenge. This adversary is stronger than $\A$ since it can simulate~$\A$ by implementing the oracles $\Enc_\f$ and $\Dec_\f$ using its $\f$ oracles. Since the input $r$ into $\f$ in the challenge ciphertext is uniformly random, the probability that any of the polynomially-many oracle calls of $\A^*$ uses the same $r$ is negligible. In the case that no oracle calls use $r$, the mixtures of the inputs to $\A^*$ (including the pairs $(r_i, \f(r_i))$) are the same for any of the two original challenge states. This fact can be verified by first averaging over the values of $\f(r)$: since $\f$ is uniformly random, $\f(r)$ is also uniformly random as well as independent of the other values of $\f$. In both cases, applying the quantum one-time pad results in the state:
\begin{equation*}
\dfrac{1}{2^n} \Id \otimes \ketbra{r} \otimes  \sigma \otimes \ketbra{r_1} \otimes \ketbra{\f(r_1)} \otimes \dots \otimes \ketbra{r_q} \otimes \ketbra{\f(r_q)},
\end{equation*}
where $\sigma$ is the state in the `environment register' of $\A^*$ (communication channel in Experiment~\ref{expt:QINDCCA}), and hence indistinguishability follows.

Next, we consider the case that $\f$ is replaced by a post-quantum pseudorandom function $\PRF_k$ for a random key $k$. We show that a successful QIND-CCA1 adversary $\A$ (i.e., one that distinguishes challenges with probability at least $\half + \epsilon$ for non-negligible $\epsilon$) can be used to construct a successful adversary $\B$ for the pqPRF, i.e., one that distinguishes $\PRF_k$ from random with non-negligible advantage over guessing. The adversary $\B$ is a \QPT algorithm with classical oracle access to a function $\h : \bin^{2\secpar} \to \bin^{2\secpar}$, and his goal is to output $0$ if $\h$ is selected perfectly at random, and $1$ if $\h = \PRF_k$ for some $k$. Define the simulated oracles:
\begin{equation*}
\begin{aligned}
& \QEnc_\h : \phi \mapsto \QOTP[\h(r)](\phi) \otimes \ketbra{r} \text{ for } r \rand \bin^{2\secpar}; \quad \text{and} \\
& \QDec_\h : \psi \otimes \ketbra{r'} \mapsto \QOTP[\h(r')](\psi),
\end{aligned}
\end{equation*}
where, as before, we assume that $\QDec_\h$ measures the second register before decrypting the first one. Note that if $\h =\PRF_k$ then these are exactly the encryption and decryption oracles (with key $k$) of the real SKQES scheme.

The algorithm $\B$ proceeds as follows. First, it executes $\A$, and replies to \A's encryption queries with $\QEnc_\h$ and to \A's decryption queries with $\QDec_\h$. When \A performs the QIND challenge query with plaintext states $\phi^0$ and $\phi^1$, $\B$ replies with the encryption of either of the two, each with probability \half, and traces out the other one. Then \B keeps answering \A's encryption queries as before with his simulated oracle. If eventually \A correctly guesses the plaintext selcted by \B, then \B outputs $1$; otherwise it outputs a random bit. If $\h = \PRF_k$ then we have exactly simulated the QIND-CCA1 game with adversary $\A$; otherwise, \B still correctly distinguishes the PRF from random with probability \half. So, the overall success probability of \B is $\half + \frac{\epsilon}{2}$, which is non-negligible over guessing. This concludes the proof.
\endproof
\end{proof}

Notice how the security of Construction~\ref{constr:qskes1} only relies on the post-quantum security of the PRF, in the \QS1 sense. In particular, from Corollary~\ref{cor:pqOWFiffpqPRF}, this gives a construction of QIND-CCA1 secure SKQES from the existence of pqOWF.

\begin{corollary}[of Theorem~\ref{thm:QCCA1constr} and Corollary~\ref{cor:pqOWFiffpqPRF}]\label{cor:QINDCCA1frompqOWF}
If pqOWF exist, then QIND-CCA1 SKQES exist.
\end{corollary}

Another way to build secure SKQES is to rely on the security of some (classical) SKES in \QS2, and `lift' the SKES construction to the \QS3 scenario through the use of type-$(2)$ operators. The following theorem is not found in the literature, but is a direct consequence of~\cite[Theorem 3.4]{GHS16} and the observation made after Definition~\ref{def:QIND}, i.e., the adversarial model (and corresponding security notions) for (\QS2) qIND and (\QS3) QIND are basically the same.

\begin{theorem}\label{thm:QINDfromqIND}
Let $\E = \E_{\K,\X,\Y} := (\KGen,\Enc,\Dec)$ be a SKES, and let $\E' = \E'_{\K,\QX,\QY} := (\KGen',\QEnc,\QDec)$ be a SKQES constructed as follows:
\begin{enumerate}
\item $\KGen' \to k$, with  $k \from \KGen$;
\item $\QEnc_k(\phi) \to \Encq \phi \Encqdr$; 
\item $\QDec_k(\psi) \to \Decq \phi \Decqdr$,
\end{enumerate}
where $\Encq,\Decq$ are type-$(2)$ unitary operators associated to $\Enc,\Dec$. If \E is qIND(-qCPA/qCCA1), then $\E'$ is QIND(-CPA/CCA1).
\end{theorem}
\begin{proof}[Proof (sketch)]
The proof follows from~\cite[Appendix C]{GHS16}, but it basically boils down to what already discussed after Definition~\ref{def:QIND}. Namely, the experiments for qIND-qCCA1 and QIND-CCA1 are fundamentally the same, the only difference is that in the qIND- version, encryption and decryption oracles are specifically type-$(2)$ operators derived from classical SKES. So the only thing left to show is that the scheme defined by such encryption/decryption operators as in the statement of the theorem is actually a SKQES. This is trivially shown by observing that:
$$
\left( \QDec_k \circ \QEnc_k \right) (\phi) = \Decq \Encq \phi \Encqdr \Decqdr = \phi
$$
so that \QEnc and \QDec respect Definition~\ref{def:qskes}.
\end{proof}

The above is a typical example of what discussed in {\em Reason \#4} of Section~\ref{sec:whysuperposition}, about the necessity of superposition-based quantum security for composition results in fully quantum scenarios. Notice in fact that in the above theorem it is {\em crucial} that \E is a scheme secure in the \QS2 sense: a `simply' post-quantum \E (in the \QS1 sense) would not work, because the same impossibility result described in Section~\ref{sec:impossqIND} would apply.

\section{Public-Key Quantum Encryption}\label{sec:QS3pke}

When we move to the public-key scenario for quantum encryption schemes, intuitively we want the same kind of functionality offered by classical PKES, but with the possibility of encrypting arbitrary quantum states. As usual, we assume classical public/private key pairs $(\pk,\sk)$, where w.l.o.g. we assume that, for security parameter \secpar, public keys are of bit size $\p(n)$, while secret keys are of bit size $\s(n)$, for polynomial functions \p,\s. Under this notation, we identify the keyspace $\K$ as $\family{\K} = \family{\Kp} \times \family{\Ks} =: \Kp \times \Ks \subset \bin^{\p(n)} \times \bin^{\s(n)}$. We define a {\em quantum public-key encryption scheme (PKQES)} as in~\cite{ABF+16}, in the following way.

\begin{definition}[Public-Key Quantum Encryption Scheme (PKQES)]\label{def:pkqes}
A {\em public-key quantum encryption scheme (PKQES)} with plaintext space \QX, ciphertext space \QY, and key space $\K := \Kp \times \Ks$ is a tuple of \QPT algorithms $\E := \E_{\K,\QX,\QY} := (\KGen,\QEnc,\QDec)$:
\begin{enumerate}
\item $\KGen: \to \K$;
\item $\QEnc: \Kp \times \QX \to \QY$;
\item $\QDec: \Ks \times \QY \to \QX$;
\end{enumerate}
such that $\left| \QDec_\sk \circ \QEnc_\pk - \Id_\HX \right|_\diamond \leq \negl$ for all $(\pk,\sk) \from \KGen$.
\end{definition} 

For the security model, as usual, we use the same QIND indistinguishability notion for SKQES, but recalling that (as explained in Section~\ref{sec:pke} for classical SKES) in the public-key scenario the minimum meaningful security notion is QIND-CPA as from Definition~\ref{def:QINDCPA}.

\subsection{Secure Construction}

QIND-CPA secure PKQES can be constructed given the existence of pqOWTP, as shown in~\cite{ABF+16}. The idea of the construction is analogous to the one for Construction~\ref{constr:PKESfromOWTP}: a (classical) randomness is sampled, and used as an input to the Goldreich-Levin PRNG to generate a key for the QOTP on the plaintext state; then the pqOWTP is applied to that randomness, and the result appended to the output of the QOTP. For the decryption, the trapdoor of the pqOWTP is used to recover the randomness, and hence the key for the QOTP, inverting the encryption. Assume for simplicity that $\X = \bin^\secpar$. Then we define the following.

\begin{construction}[PKQES from pqOWTP]\label{constr:PKQESfrompqOWTP}
Let $\P:=(\Gen,\Eval,\Invert)$ be a pqOWTP on $\X^2$, with index and trapdoor spaces \I and \T respectively, and let $\PRNG_\P:\X^2 \to \X^2$ be the Goldreich-Levin PRNG for \P (seen as a OWF with hard-core predicates). Define $\E = \E_{\K,\QX,\states{\HX^{\otimes 3}}}:=(\KGen,\QEnc,\QDec)$ as a PKQES with (public,private) key space $\K  = \Kp \times \Ks$ (where $\Kp := \I$ and $\Ks := \T$, plaintext space \QX, and ciphertext space $\states{\HX^{\otimes 3}}$, in the following way:
\begin{enumerate}
\item $\KGen \to (\pk,\sk)$, with  $(\pk,\sk) := (i,t) \from \Gen$;
\item $\QEnc_\pk(\phi) \to \psi \otimes \ketbra{z}$,\\with $\psi := \QOTP[\PRNG_\P(r)](\phi)$ and $z \from \Eval(\pk,r)$, where $r \rand \X^2$;
\item $\QDec_\sk(\rho) \to \QOTP[\PRNG_\P(s)](\sigma)$,\\with $s \from \Invert(\pk,\sk,z)$, where $z$ is obtained by measuring the last $2\secpar$ qubits of $\rho$, while $\sigma$ is the reduced state left after such a measurement.
\end{enumerate}
\end{construction}

The above construction is a simplified version of~\cite[Scheme 2]{ABF+16}, and it can be shown to be QIND-CPA secure.

\begin{theorem}[{\cite[Lemma 14]{ABF+16}}]\label{thm:QCPAconstr}
Construction~\ref{constr:PKESfromOWTP} is a QIND-CPA secure PKQES.
\end{theorem}
\begin{proof}[Proof (sketch)]
The proof is as in Theorem~\ref{thm:PKESfromOWTP}: recall that the QOTP is information-theoretically secure for independent keys sampled uniformly at random, and hence computationally secure for keys output by the pqPRNG in the construction here. Then, the only way for an adversary to attack the scheme would be to extract information about $r$ by looking at the OWTP image $z$ obtained through \Eval, but this is impossible because \P is a pqOWTP family, and $\PRNG_\P$ only outputs (post-quantum) hard-core bits.
\endproof
\end{proof}

\section{Quantum ORAM}\label{sec:QORAM}

In this section we study {\em quantum ORAMs (QORAM)}, that is, ORAM constructions operating on {\em quantum data}. This new cryptographic primitive defined in~\cite{GKK17} considers the same scenario as in the ORAM case, but where all the parties have quantum computing and communication capabilities. As we will see, many difficulties arise in modeling this scenario.

In the QORAM model, the client \C and the server \S are both \QPT algorithms, sharing a {\em quantum communication channel} (quantum register) $\Psi$. Since such a quantum channel can also be used to share classical information, we assume without loss of generality that \A and \S also share a classical channel $\Xi$. In the following, if not otherwise stated, we will always assume that all the classical communication between \A and \S happens through $\Xi$, and all the quantum communication happens through $\Psi$. In this scenario, a computationally limited \C wants to outsource a {\em quantum database (QDB)} to the more powerful \S, and perform operations on the QDB in a secure way, as in the ORAM case.

We have first to define what it means to have a `quantum database'. In our case, this will be a structure of {\em quantum blocks}. A quantum block is a $\bsize$-qubit quantum state $\psi \in \states{\Hilbert_\bsize}$ for a fixed parameter $\bsize \in \NN$ which depends on \C's and \S's architectures. A {\em quantum database} (QDB) of size $\dbsize \in \NN$ is a quantum register of \S which stores $\dbsize$ quantum blocks. It is important to notice that we impose no restriction on the nature of the states stored in the quantum blocks, i.e., these states could be mixed or entangled, amongst them or with states stored in other, external registers. As explained in the preliminaries, in the following, for simplicity, we abuse notation and denote such multipartite system with a tuple of quantum blocks $(\psi_1,\ldots,\psi_\dbsize)$. Since we assume this quantum register to reside on the server's side, we will denote it as $\S.\QDB$. As in the ORAM case, the precise way this system of quantum blocks is represented in the quantum database is unspecified, and left to the exact implementation of the QORAM scheme taken into account. As usual, we will abuse notation and write that $\S.\QDB(i) = \psi$ if $\psi$ is the state obtained by tracing out all but the $i$-th subsystem of $\S.\QDB$, and that $\psi \in \S.\QDB$ if $\S.\QDB(i) = \psi$ for some $i \in \NN$.

A quantum block encodes (usually in an encrypted form) a {\em quantum data unit}, which is another quantum state representing the information that the client actually wants to access or modify, and possibly additional (quantum or classical) auxiliary information. Formally, a quantum data unit is a quantum state $\phi \in \states{\Hilbert_\dsize}$ of $\dsize$ qubits, where $\dsize \leq \bsize$ depends on \C's and \S's architecture. As before, no assumption is made about the nature of these quantum states. Every quantum block can encode a single quantum data unit, therefore at any given time $t$ it is defined a CPTP map $\QData_t: \S.\QDB \to \states{\Hilbert_\dsize}$. With abuse of notation, we will denote by $\QData(\psi)$ the quantum data unit encoded in the block $\psi$ at a certain time. The client \C can operate on the quantum database through {\em quantum data requests}.

\begin{definition}[Quantum Data Request]
A {\em quantum data request} to a database $\S.\QDB$ of size \dbsize is a tuple of the form $\qdr = (\op, i, \phi)$, where $\op\in\{\text{read},\text{write}\}, i \in \set{1,\ldots,\dbsize}$, and $\phi \in \states{\Hilbert_\dsize}$ is a quantum data unit ($\phi$ can also be $\ket{\bot}$ if $\op=\text{read}$).
\end{definition}

Finally, we define the meaning of a {\em quantum communication transcript} during an execution of a QORAM protocol. As in the ORAM case, we will use the following definition.

\begin{definition}[Quantum Communication Transcript]
A {\em quantum communication transcript} $\qcom$ at time $t$ is the content of the communication registers $(\Xi,\Psi)$ at time $t$ of the protocol's execution.
\end{definition}

As in the ORAM case, in the following we will consider $\qcom$ as a discrete function of the round $1,2,\ldots$ of the protocol. Notice the following difference from the classical case: as this time \C and \S are also allowed to exchange quantum data through $\Psi$, it might not be possible for an adversary to obtain a full transcript of $\qcom$ without disturbing the protocol. We will address this issue in the next section about security.

From now on, $\bsize$ and $\dsize$ will be fixed constants (the {\em quantum block size}, and {\em quantum data unit size}, resp.) As in the classical case, we assume that a server's QDB is always initialized empty (that is, with randomized encryptions of $\ket{0\ldots0}$ as data), and it is left up to the client the task of `populating' the database. We are now ready to define a QORAM as follows.

\begin{definition}[QORAM]\label{def:qoram}
Let $\maxsize \in \NN, \msize \geq \dsize \in \NN$, and $\E=(\KGen,\QEnc,\QDec)$ be a SKQES scheme mapping $\msize$-qubit plaintext states to $\bsize$-qubit ciphertext states. A {\em QORAM (quantum oblivious random access machine)} $\qoram_{\E}$ with parameters $(\maxsize,\dsize,\E)$ is a pair of two-party interactive \QPT algorithms $(\qinit,\qaccess)$, such that:
\begin{itemize}
\item $\qinit(\secpar,\dbsize) \rightarrow (\C,\S)$ in the following way:
	\begin{enumerate}
	\item \secpar is the security parameter, $\dbsize \leq \maxsize$;
	\item $k \from \KGen(\secpar)$ is generated by \C;
	\item \S includes a QDB $\S.\QDB=(\psi_1,\ldots,\psi_\dbsize)$,\\where $\forall i \implies \psi_i \from \QEnc_k(\ketbra{0})$;
	\end{enumerate}
\item $\qaccess(\C,\S,\qdr) \rightarrow (\C',\S',\qcom)$ in the following way:
	\begin{enumerate}
	\item \C issues a quantum data request $\qdr$;
	\item \C and \S communicate via $(\Xi,\Psi)$ and produce the quantum communication transcript~$\qcom$.
	\end{enumerate}
\end{itemize}
\end{definition}

The same considerations about soundness hold as in the classical case.

\subsection{QORAM Security}

We now look at the security model for QORAMs. As in the classical model, security will be given in terms of adaptive access pattern indistinguishability.

Our threat model considers a quantum adversary \A, which we identify as \S himself, and who wants to compromise \C's privacy by having access to the communication channel $(\Xi,\Psi)$ and \S's internal memory, but who is not allowed to modify the content of the channel against the soundness of the protocol. Without loss of generality, we assume that the only meaningful changes in the database area $\S.\QDB$ only happen between the beginning and the end of a $\qaccess$ execution.

As it often happens in the quantum world, there is a caveat here: it is unclear what a `honest-but-curious' quantum adversary is. 
In fact, the problem is even more general: we do not have a notion of `read-only' for quantum channels, as the mere act of observing the data in transit through $\Psi$ can destroy such data. For example, suppose that a quantum state $\phi$ is sent through $\Psi$. 
Because of the No-Cloning Theorem, \S cannot store a local copy of $\phi$; at the same time, measuring $\phi$ in transit through $\Psi$ without any knowledge of such state, would disturb it with high probability. Therefore, it seems hard to justify the inclusion of the state $\phi$ in the adversarial view (the {\em quantum access pattern}) of a honest-but-curious quantum adversary.

Nevertheless, it is important to allow the adversary \A to know some information about the quantum state $\phi$. There are many reasons for this choice. First of all, remember that we are defining QORAMs in a very abstract and general way, and the exact details of how the communication and storage of quantum information works is left to the particular QORAM construction. For example, there might be constructions which only use quantum states from a finite, fixed set of orthogonal states, or which only use subsets of quantum states admitting efficient classical representations (and encoding them in a classical way during the communication). Moreover, it might be possible that the adversary \A at some point obtains access to some side-information which allows him to know 
something about the content of the database or the data transferred in a sound way, e.g., by applying some quantum operation or partial measurement which does not disturb the state too much. As we need to cover all these possibilities, the option of not including at all the quantum data in the access pattern would be too restrictive. On the other hand, the adversary \A should not be able to modify too much (from \C's point of view) any quantum state, as this would go beyond the notion of honest-but-curious adversary usually considered in the ORAM scenario.

We solve this issue by introducing a {\em safe extractor}. The intuition behind this technique is to allow our adversary to extract any kind of (quantum) information he wants from a certain physical system, {\em as long as such extraction is hardly noticeable by any other party}. In this case we say that the action of the adversary on the physical system is {\em computationally undetectable}, meaning that no \QPT algorithm can reliably distinguish whether a quantum operation takes place or not by just looking at the processed quantum state, even in presence of auxiliary information such as, e.g., additional entangled registers. More 
formally we define the following.

\begin{definition}[Computational Undetectability of Quantum Action]
Let $\Hilbert_\Lambda,\Hilbert_\Sigma, \HEnv$ be Hilbert spaces of dimension polynomial in $2^\secpar$ associated to quantum register $\Lambda,\Sigma,\Env$ respectively, and let $\phi_\Sigma$ be an arbitrary quantum state on register $\Sigma$. A quantum algorithm $\B:\states{\Hilbert_\Lambda \otimes \Hilbert_\Sigma} \to \states{\Hilbert_\Lambda \otimes \Hilbert_\Sigma}$ acting on registers $\Lambda$ and $\Sigma$ has {\em computationally undetectable action on $\phi_\Sigma$} iff for any bipartite state $\phi_{\Sigma\Env}$ such that $(\phi_{\Sigma\Env})_\Sigma = \phi_\Sigma$, and for any \QPT algorithm $\D$ acting on registers $\Sigma$ and $\Env$ and outputting $0$ or $1$, it holds:
$$
\left| \Pr \left[ \D \left( \phi_{\Sigma\Env} \right) \to 1 \right] - \Pr \left[ \D \left( \left( \B \otimes \Id_\HEnv \right) \left( \ketbra{0}_\Lambda \otimes \phi_{\Sigma\Env} \right)_{\Sigma\Env} \right) \to 1 \right] \right| \leq \negl.
$$
\end{definition}

\begin{definition}[Safe Extractor]\label{dfn:extractor}
Let $\phi_\Sigma \in \states{\Hilbert_\Sigma}$ be the quantum state contained in a quantum register $\Sigma$. A {\em safe extractor} for $\Sigma$ in the state $\phi_\Sigma$ is a \QPT algorithm $\safex$ with additional classical input $x$ of size polynomial in $\secpar$, acting on $\Sigma$ and outputting a quantum state $\psi$ of qubit size polynomial in $\secpar$, and such that the action of $\safex$ on $\phi_\Sigma$ is computationally undetectable.
\end{definition}

Notice that Definition~\ref{dfn:extractor} depends on the state contained in the quantum register considered. That is, $\safex$ might be a safe extractor for a given quantum register if that register is in a certain state, but not in a different one. Of course one could define $\safex$ to be a safe extractor for a register {\em `tout-court'} if it is a safe extractor for {\em any} state of that register according to Definition~\ref{dfn:extractor}, but this would considerably reduce the power of the adversary. Instead, this definition allows the adversary to use $\safex$ adaptively, only at certain points of his execution, when he knows that the action of $\safex$ on the current state of the QORAM will be computationally undetectable. The additional classic input to $\safex$ serves a useful purpose here, as it can be seen as a way for the adversary to communicate instructions to $\safex$ about how to perform the extraction in a safe way (for example, \A might encode a certain measurement basis through this classical input.) With abuse of notation, and without loss of generality, we will write $\psi \from \safex(\qcom,\S.\QDB)$ to denote that $\safex$ performs the following:
\begin{itemize}
\item as a classical input, $\safex$ gets the classical part of a quantum communication transcript $\qcom$ (that is, the content of the classical channel $\Xi$) and additional classical information by the adversary $\A$;
\item $\safex$ acts on the quantum registers $\Psi$ and $\S.\QDB$;
\item finally, $\safex$ produces a quantum output $\psi$.
\end{itemize}

The intuition of a safe extractor is that we need a way to formalize the (quantum or classical) information that an adversary is able to extract by observing the changes in the quantum database and communication channel. However, we still require that such extraction does not lead to a meaningful deviation from the `regular' execution of the QORAM protocol. Computational undetectability of quantum action is a {\em strong} guarantee, because if such action is undetectable, in particular it means that such action cannot modify the QORAM soundness. The converse does not hold: it might be the case that an adversary manipulates the quantum channel or database in such a way that it is {\em theoretically possible} to detect this manipulation (for some distinguisher \D), but {\em not} for any QORAM client, and therefore the QORAM soundness would be still preserved. However, for our purposes the above restriction on the power of the QORAM adversary is sufficient to define meaningful notions of security, and it is analogous to the (classical) restriction of a honest-but-curious adversary in the ORAM case commonly used in the literature.

More formally, we define a QORAM adversary as follows.

\begin{definition}[QORAM Adversary]\label{dfnqadv}
Let $\Hilbert_\QDB, \Hilbert_\Psi, \Hilbert_\Lambda$ be complex Hilbert spaces associated to quantum registers \QDB (the quantum database), $\Psi$ (the quantum communication channel) and $\Lambda$ (the quantum access pattern register). A {\em QORAM adversary} is a \QPT algorithm $\A^\safex$ with quantum oracle access to a CPTP map $\safex: \Xi \times \states{\Hilbert_\QDB \otimes \Hilbert_\Psi} \to \states{\Hilbert_\Lambda}$, such that:
\begin{enumerate}
\item\label{item:safex} \safex is a {\em safe extractor} for the joint register $(\QDB,\Psi)$ for any of its states during any invocation of \safex by $\A$;
\item\label{item:sound} $\A^\safex$ is computationally indistinguishable from an honest server \S for every QORAM client \C.
\end{enumerate}
\end{definition}

As already discussed notice that, in the definition above, conditions~\ref{item:safex} and~\ref{item:sound} are independent: if \safex is {\em not} a safe extractor during the execution, it means that there exists {\em some} quantum distinguisher \D able to detect \safex's action on the joint register $(\QDB,\Psi)$, but $\A^\safex$ might still remain indistinguishable from an honest server for any honest quantum client. On the other hand, \A might be a misbehaving adversary which deviates `too much' from the execution of an honest server (and therefore might compromise the QORAM's soundness), even if \safex behaves always as a safe extractor. For a meaningful notion of security akin to the \QS0 case, we require that a QORAM adversary respects both conditions.

We are now able to define {\em quantum access patterns}, as the outputs of the safe extractor before and after the execution of a quantum data request.

\begin{definition}[Quantum Access Pattern]
Given QORAM client and server \C and \S, a quantum data request $\qdr$, and a QORAM adversary $\A = \A^\safex$, the {\em quantum access pattern observed by $\A$}, denoted by $\qap_\A(\qdr)$, is the pair of quantum states $(\psi,\psi')$, where:
\begin{itemize}
\item $\psi \from \safex(\qcom,\S.\QDB)$;
\item $(\C',\S',\qcom') \leftarrow \qaccess(\C,\S,\qdr)$
\item $\psi' \from \safex(\qcom',\S'.\QDB)$.
\end{itemize}
\end{definition}

Notice that, since the action of the safe extractor is computationally undetectable, running it on two consecutive quantum data requests does not allow, in any case, to clone arbitrary quantum states. We define the new security game as follows.

\begin{experiment}[\gameQORAM]\label{expt:gameQORAM}
Let $\qoram=(\qinit,\qaccess)$ be a QORAM construction with parameters $(\maxsize,\dsize,\E)$, \secpar a security parameter and $\A=\A^\safex$ a QORAM adversary. The {\em computational indistinguishability of quantum access patterns under adaptive chosen query attack game} $\gameQORAM$ proceeds as follows:
\begin{algorithmic}[1]
\State \textbf{Input:} $\secpar \in \NN$
\State $\A \to (\A_0,\qdr_1,\dbsize \leq \maxsize)$
\State $(\C_0,\S_0) \leftarrow \qinit(\secpar,\dbsize)$
\Loop{ for $i=1,\ldots,q_1 \in \NN$:} \Comment{first quantum CQA phase}
\State $\qaccess(\C_{i-1},\S_{i-1},\qdr_i) \to (\C_i,\S_i,\qap_i)$
\State $\A_{i-1}(\qap_i,\S_i) \to (\A_i,\qdr_{i+1})$
\EndLoop
\State $\A_{q_1}(\qdr_{q_1+1}) \to (\A',\qdr^0,\qdr^1)$ 
\State $b \rand \bin$
\State $\access(\C_{q_1},\S_{q_1},\qdr^b) \to (\C_{q_1+1},\S_{q_1+1},\qap_{q_1+1})$ \Comment{QAP-IND challenge}
\State trace out the quantum data contained in $\qdr^{1-b}$
\State $\A'(\ap_{q_1+1},\S_{q_1+1}) \to (\A_{q_1+1},\qdr_{q_1+2})$
\Loop{ for $i=q_1+2,\ldots,q_2 \geq q_1+2$:} \Comment{second quantum CQA phase}
\State $\access(\C_{i-1},\S_{i-1},\qdr_i) \to (\C_i,\S_i,\qap_i)$
\State $\A_{i-1}(\qap_i,\S_i) \to (\A_i,\qdr_{i+1})$
\EndLoop
\State $\A_{q_2}(\qdr_{q_2+1}) \to b' \in \bin$
\If{$b = b'$}
	\State \textbf{Output:} $1$
\Else
	\State \textbf{Output:} $0$
\EndIf
\end{algorithmic}
The {\em advantage of \A} is defined as:
$$
\advQORAM := \Pr \left[ \gameQORAM \to 1 \right] - \half .
$$
\end{experiment}

The idea of the above game follows specularly the classical intuition: the adversary is first allowed to enforce (adaptively) the execution of quantum data requests of his choice, and to observe the related access patterns. Then he issues the challenge query, composed of two different quantum data requests, one of which is executed, and the other discarded. After that, the adversary is allowed another adaptive learning phase, and finally he has to output a bit indicating the challenge data request which was executed. We are now ready to define the security notion for QORAMs.

\begin{definition}[Quantum Access Pattern Indistinguishability Under Adaptive Chosen Query Attack]
A QORAM construction $\qoram$ has {\em computationally indistinguishable quantum access patterns under adaptive chosen query attack} (or, it is QAP-IND-CQA-secure) iff for any QORAM adversary $\A^\safex$ it holds: $\advQORAM \leq \negl$.
\end{definition}

\subsection{PathQORAM}

In this section we describe the construction for a novel QAP-IND-CQA-secure QORAM scheme, which we call {\em $\pathqoram$}, and which has the interesting property that read and write operations are {\em inherently equivalent}. The idea is to modify \pathoram with the SKQES from Construction~\ref{constr:qskes1}, but we need some additional care for ensuring soundness. In fact, we have the following problem. Suppose the client issues a quantum data request for block $i$. This will be translated to a leaf in \S's quantum database, and the resulting tree branch $\qbranch$ will be sent to \C. Now \C knows that the data he is looking for is encoded in one of $\qbranch$'s nodes, but he does not know which one. Classically, \C would proceed by decrypting and inspecting every node in $\qbranch$ until he finds what he is looking for, then he would perform some operation on that element, before re-encrypting it again, and then complete the re-randomization of $\qbranch$ before re-sending the whole branch to \S. This operation might be problematic in the quantum world though: inspecting an unknown quantum state will destroy it with high probability. We have therefore to find a way to signal \C when he reaches the right node in the path without disturbing the quantum data unit itself.

The solution is to notice that, in our formalization of PathORAM, the client stores the classical identifier $i$ together with the data unit in the block. In the quantum version \pathqoram, this identifier is still classical, and of a fixed length $\ksize$. Once a node in $\qbranch$ is decrypted, it will be transformed to $\ket{i}\bra{i} \otimes \phi$. The first register can then be measured in the computational basis without being disturbed, and without disturbing the state $\phi$ (which is not entangled with $\ket{i}$). So the trick for \C is to find out when he is decrypting the right element by {\em only} measuring the first $\ksize$ qubits of the decrypted block, and then only act on the quantum data unit when the right identifier is found. Notice how other different approaches used classically to instantiate PathORAM, such as identifying blocks by storing a local table with the hash values of the data units, might not work so smoothly when translated to the quantum world.

More concretely, we give here a full description of PathQORAM (which from now on we denote as \pathqoram) according to our new formalism. The meaning of the parameters is as in Definition~\ref{dfn:pathoram}. 

\begin{construction}[\pathqoram~{\cite[Definition 36]{GKK17}}]\label{constr:pathqoram}
For fixed parameters $\dsize,\maxsize \in \NN$, let $\ksize = \lceil \log_2 \maxsize \rceil,\zsize \in \NN,\msize = \dsize + \ksize$, and $\bsize \geq \msize$. Let $\PRNG$ be a pqPRNG outputting $\ksize$-bit pseudorandom values, and let $\E=(\KGen,\QEnc,\QDec)$ be a QIND-CPA  SKQES with $\msize$-qubit plaintexts and $\bsize$-qubit ciphertexts. We define a QORAM construction called $\pathqoram = \pathqoram_{\E,\G}$ as follows:
\begin{itemize}
\item $\qinit(\secpar,\dbsize) \rightarrow (\C,\S)$ in the following way:
	\begin{algorithmic}[1]
	\State $\C$ generates a secret key $k\from\KGen$
	\State set $\tsize := \lceil \log_2 \dbsize \rceil$
	\State $\C$ initializes a lookup table (the \em{position map}) of the form\newline $\left((1,r_1),\ldots,(\dbsize,r_\dbsize)\right)$, where $r_i$ are $\tsize$-bit values generated by truncating the first $\ksize-\tsize$ bits of $\PRNG$'s output
	\State $\S.\QDB$ is stored in a binary tree of height $\tsize$, with root $\qtreeroot$ and leaves $\qleaf_0,\ldots,\qleaf_{2^\tsize-1}$, and such that:
		\begin{enumerate}
		\item each node of the tree stores up to $\zsize$ quantum blocks;
		\item every quantum block of every node is initialized\newline to $\QEnc_k(\ketbra{0^\msize})$.
		\end{enumerate}
	\end{algorithmic}
\item If $\qdr = (\op, i, \phi)$, then $\qaccess(\C,\S,\qdr) \rightarrow (\C',\S',\qcom)$ in the following way:
	\begin{algorithmic}[1]
	\State \C reads $r_i$ from his position map and sends it to \S
	\State $\S$ sends to $\C$ the quantum system containing the path $\qbranch$ from $\qtreeroot$ to $\qleaf_{r_i}$
	\State remap $(i,r_i)$ to $(i,r'_i)$ in the position map of $\C$, where $r'_i$ is a fresh pseudorandom $\tsize$-bit value (generated by truncating the first\newline $\ksize-\tsize$ bits of $\G$'s output), obtaining~$\C'$
	\ForAll{quantum block $\psi$ contained in $\qbranch$}
		\State $\C'$ decrypts $\QDec_k(\psi) \to (\ketbra{j} \otimes \sigma)$, \newline \mbox{\ \ \ \ \ \ where $\ket{j} \in \Hilbert_\ksize$, and $\sigma \in \states{\Hilbert_\dsize}$}
		\State $\C'$ measures the first $\ksize$ qubits of the decrypted state in the\newline \mbox{\ \ \ \ \ \ computational 
			basis, obtaining $j$}
		\If{$j=i$}
			\State swap $\sigma$ with $\phi$			
		\EndIf
		\State $\C'$ re-encrypts (re-randomizing) the current quantum block,\newline \mbox{\ \ \ \ \ \ obtaining $\psi'$}
		\State find in $\qbranch$ the common parent node $\qnode$ between \newline \mbox{\ \ \ \ \ \ $\qleaf_{r_i}$ and $\qleaf_{r_j}$, closer to the leaf level}
		\State set $b_\mathsf{swap} := \text{`false'}$
		\ForAll{$\rho$ in $\qnode$}\label{alg:qrepeatloop}
			\State $\C'$ decrypts $\QDec_k(\rho) \to (\ketbra{j'} \otimes \sigma')$ 
			\State $\C'$ re-encrypts (re-randomizing) $\rho' \from \QEnc_k(\ketbra{j'} \otimes \sigma')$ 
			\If{$j'=0\ldots0$} \Comment{$\rho'$ is empty, can be used}
				\State swap $\psi'$ and $\rho'$
				\State set $b_\mathsf{swap} := \text{`true'}$
			\EndIf
		\EndFor
		\If{$b_\mathsf{swap} = \text{`false'}$} \Comment{no empty blocks in current $\qnode$}
			\If{$\qnode \neq \qtreeroot$}
				\State set $\qnode$ to be one level up in the tree
				\State go to step~\ref{alg:qrepeatloop}
			\Else
				\State store the current quantum block in the $\qstash$
			\EndIf
		\EndIf
	\EndFor
	\State $\C'$ sends back the updated tree branch, $\newqbranch$, to $\S$
	\State update $\S.\QDB$ with $\newqbranch$, obtaining $\S'$
	\State produce $\qcom$, which contains $r_i, \qbranch, \newqbranch$
	\end{algorithmic}
\end{itemize}
\end{construction}

Notice that the following interesting property holds: the operations of `write' and `read' have the {\em same} effect. Namely: since qubits from the server's database cannot be copied, and cannot be removed or added (otherwise this would compromise indistinguishability), the action of a read or write operation is simply to swap a state in the database with a state in \C's memory. In fact, $\qaccess$ swaps $\phi$ known by \C with $\sigma$ stored in $\S$. Also notice how $\qcom$ containing $\qbranch, \newqbranch$ would imply a cloning of quantum states. This is just a formal artifice, because in the case of QORAMs as we defined them, $\qcom$ is only used in respect to a safe extractor $\safex$, which processes \newqbranch only after \C has processed \qbranch, so information is never copied. For the soundness of the $\pathqoram$ construction we have left unexplained the use of a {\em quantum stash \qstash}. This is an area of quantum memory basically used as the classical stash of $\pathoram$, but every time an element is `written' in the stash, it is actually `swapped' with an empty block in the tree. The security of the construction follows from the QIND-CPA security of the SKQES $\E$, and from the security of the pqPRNG $\PRNG$.

\begin{theorem}[{\cite[Theorem 34]{GKK17}}]\label{thm:pathqoram}
Let $\E$ be a QIND-CPA SKQES, and let $\PRNG$ be a pqPRNG. Then, \pathqoram instantiated using $\E$ and $\PRNG$ is a QAP-IND-CQA secure QORAM.
\end{theorem}
\begin{proof}
The proof follows step-by-step the proof of Theorem~\ref{thm:pathoramsec} with some important differences. First of all, \D cannot store a local mirrored tree of plaintexts of the form $(\ket{i}\bra{i} \otimes \sigma)$ because of the No-Cloning Theorem, so he cannot simulate \C perfectly. But he can store a mirrored tree which contains {\em only} the classical identifiers $i$, at the right positions of every block throughout the execution of the protocol. 

At this point, \D can simulate a decryption oracle for a certain block $\psi$ in a downloaded branch by fetching the cleartext identifier $i$ found at the corresponding position in the `mirrored' tree, and creating a `simulated' plaintext of the form $(\ket{i}\bra{i} \otimes \ket{0^\dsize}\bra{0^\dsize})$, i.e., replacing the `real' quantum data unit $\sigma$ with a zero state. Since \A never `sees' a decrypted block, this substitution is not immediately apparent to him. Moreover, whenever \C would create a block by encrypting $\psi \from \QEnc_k(\ket{i}\bra{i} \otimes \sigma)$, \D can simulate this by doing $\psi \from \QEnc_k(\ket{i}\bra{i} \otimes \ket{0^\dsize}\bra{0^\dsize})$. By the QIND-CPA security of $\E$, \A cannot detect this substitution with more than negligible advantage over guessing. Therefore, \D can still simulate \C (with overwhelming, albeit not $100\%$, probability) at any data request.

Another issue appears during the challenge phase, as this time the concept of {\em non-meaningful} challenge must be redefined. For the same argument as above, from \A's perspective it does not matter whether two data requests lead to two `different' quantum data units $\sigma^0,\sigma^1$ (the analogue of data units $\data^0,\data^1$ in the classical proof) or not. Therefore, \D can ignore the quantum data units at all. Moreover, as discussed above, in \pathqoram there is no difference between `read' and `write' operations. It follows, from the same argument as in the proof of Theorem~\ref{thm:pathoramsec}, that the two challenge quantum data requests $\qdr^0,\qdr^1$ must differ on the identifiers $i^0,i^1$. Then, \D plays the QIND-CPA game with challenge plaintexts $\phi^a = \ket{i^a}\bra{i^a} \otimes \ket{0^\dsize}\bra{0^\dsize}$ for $a \in \bin$, following the same strategy as in the classical case (by guessing a bit, injecting the challenge ciphertext, and observing \A's output), with only a negligible loss in the success probability because he is simulating fake plaintexts. This concludes the proof.
\endproof
\end{proof}

\backmatter
\bibliographystyle{myalphasort}
\bibliography{local}

\providecommand{\etalchar}[1]{$^{#1}$}
\begin{thebibliography}{AMTdW00}

\bibitem[Aar09]{quantmoney}
Scott Aaronson.
\newblock Quantum copy-protection and quantum money.
\newblock In {\em Proceedings of the 24th Annual {IEEE} Conference on
  Computational Complexity, {CCC} 2009, Paris, France, 15-18 July 2009}, pages
  229--242, 2009.

\bibitem[ABF{\etalchar{+}}16]{ABF+16}
Gorjan Alagic, Anne Broadbent, Bill Fefferman, Tommaso Gagliardoni, Christian
  Schaffner, and Michael~St. Jules.
\newblock Computational security of quantum encryption.
\newblock In {\em Information Theoretic Security - 9th International
  Conference, {ICITS} 2016, Tacoma, WA, USA, August 9-12, 2016, Revised
  Selected Papers}, pages 47--71, 2016.

\bibitem[AM16]{AlagicAuth}
Gorjan Alagic and Christian Majenz.
\newblock Quantum non-malleability and authentication.
\newblock {\em CoRR}, abs/1610.04214, 2016.

\bibitem[AR16]{gorjanPRP}
Gorjan Alagic and Alexander Russell.
\newblock Quantum-secure symmetric-key cryptography based on hidden shifts.
\newblock {\em {IACR} Cryptology ePrint Archive}, 2016:960, 2016.

\bibitem[ARTL15]{Dwave}
Tameen Albash, Troels~F. R{{\o}{}}nnow, Matthias Troyer, and Daniel~A. Lidar.
\newblock Reexamining classical and quantum models for the {D-Wave One}
  processor.
\newblock {\em The European Physical Journal Special Topics}, 224(1):111--129,
  2015.

\bibitem[AMTdW00]{QOTP1}
Andris Ambainis, Michele Mosca, Alain Tapp, and Ronald de~Wolf.
\newblock Private quantum channels.
\newblock In {\em 41st Annual Symposium on Foundations of Computer Science,
  {FOCS} 2000, 12-14 November 2000, Redondo Beach, California, {USA}}, pages
  547--553, 2000.

\bibitem[ARU14]{QuantRewinding}
Andris Ambainis, Ansis Rosmanis, and Dominique Unruh.
\newblock Quantum attacks on classical proof systems: The hardness of quantum
  rewinding.
\newblock In {\em 55th {IEEE} Annual Symposium on Foundations of Computer
  Science, {FOCS} 2014, Philadelphia, PA, USA, October 18-21, 2014}, pages
  474--483, 2014.

\bibitem[ATTU16]{Anand+}
Mayuresh~Vivekanand Anand, Ehsan~Ebrahimi Targhi, Gelo~Noel Tabia, and
  Dominique Unruh.
\newblock Post-quantum security of the {CBC}, {CFB}, {OFB}, {CTR}, and {XTS}
  modes of operation.
\newblock In {\em Post-Quantum Cryptography - 7th International Workshop,
  PQCrypto 2016, Fukuoka, Japan, February 24-26, 2016, Proceedings}, pages
  44--63, 2016.

\bibitem[AB09]{AroraBarak}
Sanjeev Arora and Boaz Barak.
\newblock {\em Computational Complexity - {A} Modern Approach}.
\newblock Cambridge University Press, 2009.

\bibitem[BCG{\etalchar{+}}02a]{quantauth}
Howard Barnum, Claude Cr{\'{e}}peau, Daniel Gottesman, Adam~D. Smith, and Alain
  Tapp.
\newblock Authentication of quantum messages.
\newblock In {\em 43rd Symposium on Foundations of Computer Science {(FOCS}
  2002), 16-19 November 2002, Vancouver, BC, Canada, Proceedings}, pages
  449--458, 2002.

\bibitem[BCG{\etalchar{+}}02b]{Barnum+}
Howard Barnum, Claude Cr{\'{e}}peau, Daniel Gottesman, Adam~D. Smith, and Alain
  Tapp.
\newblock Authentication of quantum messages.
\newblock In {\em 43rd Symposium on Foundations of Computer Science {(FOCS}
  2002), 16-19 November 2002, Vancouver, BC, Canada, Proceedings}, pages
  449--458, 2002.

\bibitem[Bel98]{BellareROM}
Mihir Bellare.
\newblock Practice-oriented provable security.
\newblock In {\em Lectures on Data Security, Modern Cryptology in Theory and
  Practice, Summer School, Aarhus, Denmark, July 1998}, pages 1--15, 1998.

\bibitem[BR93]{BR93}
Mihir Bellare and Phillip Rogaway.
\newblock Random oracles are practical: {A} paradigm for designing efficient
  protocols.
\newblock In {\em {CCS} '93, Proceedings of the 1st {ACM} Conference on
  Computer and Communications Security, Fairfax, Virginia, USA, November 3-5,
  1993.}, pages 62--73, 1993.

\bibitem[BBBV97]{BBBV97}
Charles~H. Bennett, Ethan Bernstein, Gilles Brassard, and Umesh~V. Vazirani.
\newblock Strengths and weaknesses of quantum computing.
\newblock {\em {SIAM} J. Comput.}, 26(5):1510--1523, 1997.

\bibitem[BB14]{BB84}
Charles~H. Bennett and Gilles Brassard.
\newblock Quantum cryptography: Public key distribution and coin tossing.
\newblock {\em Theor. Comput. Sci.}, 560:7--11, 2014.

\bibitem[BBD09]{pqcrypto}
Daniel~J. Bernstein, Johannes Buchmann, and Erik Dahmen.
\newblock {\em Post-Quantum Cryptography}.
\newblock Springer-Verlag Berlin Heidelberg, 2009.

\bibitem[BHH{\etalchar{+}}15]{sphincs}
Daniel~J. Bernstein, Daira Hopwood, Andreas H{\"{u}}lsing, Tanja Lange, Ruben
  Niederhagen, Louiza Papachristodoulou, Michael Schneider, Peter Schwabe, and
  Zooko Wilcox{-}O'Hearn.
\newblock {SPHINCS:} practical stateless hash-based signatures.
\newblock In {\em Advances in Cryptology - {EUROCRYPT} 2015 - 34th Annual
  International Conference on the Theory and Applications of Cryptographic
  Techniques, Sofia, Bulgaria, April 26-30, 2015, Proceedings, Part {I}}, pages
  368--397, 2015.

\bibitem[BV97]{BV97}
Ethan Bernstein and Umesh~V. Vazirani.
\newblock Quantum complexity theory.
\newblock {\em {SIAM} J. Comput.}, 26(5):1411--1473, 1997.

\bibitem[Bon98]{BonehDDH}
Dan Boneh.
\newblock The decision {Diffie-Hellman} problem.
\newblock In {\em Algorithmic Number Theory, Third International Symposium,
  ANTS-III, Portland, Oregon, USA, June 21-25, 1998, Proceedings}, pages
  48--63, 1998.

\bibitem[BDF{\etalchar{+}}11]{QROM}
Dan Boneh, {\"{O}}zg{\"{u}}r Dagdelen, Marc Fischlin, Anja Lehmann, Christian
  Schaffner, and Mark Zhandry.
\newblock Random oracles in a quantum world.
\newblock In {\em Advances in Cryptology - {ASIACRYPT} 2011 - 17th
  International Conference on the Theory and Application of Cryptology and
  Information Security, Seoul, South Korea, December 4-8, 2011. Proceedings},
  pages 41--69, 2011.

\bibitem[BZ13a]{BZ13auth}
Dan Boneh and Mark Zhandry.
\newblock Quantum-secure message authentication codes.
\newblock In {\em Advances in Cryptology - {EUROCRYPT} 2013, 32nd Annual
  International Conference on the Theory and Applications of Cryptographic
  Techniques, Athens, Greece, May 26-30, 2013. Proceedings}, pages 592--608,
  2013.

\bibitem[BZ13b]{BZ13}
Dan Boneh and Mark Zhandry.
\newblock Secure signatures and chosen ciphertext security in a quantum
  computing world.
\newblock In {\em Advances in Cryptology - {CRYPTO} 2013 - 33rd Annual
  Cryptology Conference, Santa Barbara, CA, USA, August 18-22, 2013.
  Proceedings, Part {II}}, pages 361--379, 2013.

\bibitem[BCD{\etalchar{+}}16]{Frodo}
Joppe~W. Bos, Craig Costello, L{\'{e}}o Ducas, Ilya Mironov, Michael Naehrig,
  Valeria Nikolaenko, Ananth Raghunathan, and Douglas Stebila.
\newblock Frodo: Take off the ring! practical, quantum-secure key exchange from
  {LWE}.
\newblock In {\em Proceedings of the 2016 {ACM} {SIGSAC} Conference on Computer
  and Communications Security, Vienna, Austria, October 24-28, 2016}, pages
  1006--1018, 2016.

\bibitem[BR03]{QOTP2}
P.~Oscar Boykin and Vwani Roychowdhury.
\newblock Optimal encryption of quantum bits.
\newblock {\em Phys. Rev. A}, 67:042317, Apr 2003.

\bibitem[BHT98]{BHT98}
Gilles Brassard, Peter H{{\o{}}}yer, and Alain Tapp.
\newblock {\em Quantum cryptanalysis of hash and claw-free functions}, pages
  163--169.
\newblock Springer Berlin Heidelberg, Berlin, Heidelberg, 1998.

\bibitem[BJ15]{BJ15}
Anne Broadbent and Stacey Jeffery.
\newblock Quantum homomorphic encryption for circuits of low {T}-gate
  complexity.
\newblock In {\em Advances in Cryptology - {CRYPTO} 2015 - 35th Annual
  Cryptology Conference, Santa Barbara, CA, USA, August 16-20, 2015,
  Proceedings, Part {II}}, pages 609--629, 2015.

\bibitem[BS16]{BS16}
Anne Broadbent and Christian Schaffner.
\newblock Quantum cryptography beyond quantum key distribution.
\newblock {\em Des. Codes Cryptography}, 78(1):351--382, 2016.

\bibitem[BFM15]{BFM15}
Christina Brzuska, Pooya Farshim, and Arno Mittelbach.
\newblock Random-oracle uninstantiability from indistinguishability
  obfuscation.
\newblock In {\em Theory of Cryptography - 12th Theory of Cryptography
  Conference, {TCC} 2015, Warsaw, Poland, March 23-25, 2015, Proceedings, Part
  {II}}, pages 428--455, 2015.

\bibitem[CNR12]{CNR12}
Jan Camenisch, Gregory Neven, and Markus R{\"{u}}ckert.
\newblock Fully anonymous attribute tokens from lattices.
\newblock In {\em Security and Cryptography for Networks - 8th International
  Conference, {SCN} 2012, Amalfi, Italy, September 5-7, 2012. Proceedings},
  pages 57--75, 2012.

\bibitem[CGH98]{CGH98}
Ran Canetti, Oded Goldreich, and Shai Halevi.
\newblock The random oracle methodology, revisited (preliminary version).
\newblock In {\em Proceedings of the Thirtieth Annual {ACM} Symposium on the
  Theory of Computing, Dallas, Texas, USA, May 23-26, 1998}, pages 209--218,
  1998.

\bibitem[CEJvO02]{obfAES}
Stanley Chow, Philip~A. Eisen, Harold Johnson, and Paul~C. van Oorschot.
\newblock White-box cryptography and an {AES} implementation.
\newblock In {\em Selected Areas in Cryptography, 9th Annual International
  Workshop, {SAC} 2002, St. John's, Newfoundland, Canada, August 15-16, 2002.
  Revised Papers}, pages 250--270, 2002.

\bibitem[DFG13]{DFG13}
{\"{O}}zg{\"{u}}r Dagdelen, Marc Fischlin, and Tommaso Gagliardoni.
\newblock The {Fiat-Shamir} transformation in a quantum world.
\newblock In {\em Advances in Cryptology - {ASIACRYPT} 2013 - 19th
  International Conference on the Theory and Application of Cryptology and
  Information Security, Bengaluru, India, December 1-5, 2013, Proceedings, Part
  {II}}, pages 62--81, 2013.

\bibitem[DFNS13]{superposition}
Ivan Damgaard, Jakob Funder, Jesper~Buus Nielsen, and Louis Salvail.
\newblock Superposition attacks on cryptographic protocols.
\newblock In {\em Information Theoretic Security - 7th International
  Conference, {ICITS} 2013, Singapore, November 28-30, 2013, Proceedings},
  pages 142--161, 2013.

\bibitem[DH76]{DH}
Whitfield Diffie and Martin~E. Hellman.
\newblock New directions in cryptography.
\newblock {\em {IEEE} Trans. Information Theory}, 22(6):644--654, 1976.

\bibitem[DFPR14]{delegated}
Vedran Dunjko, Joseph Fitzsimons, Christopher Portmann, and Renato Renner.
\newblock Composable security of delegated quantum computation.
\newblock In {\em Advances in Cryptology - {ASIACRYPT} 2014 - 20th
  International Conference on the Theory and Application of Cryptology and
  Information Security, Kaoshiung, Taiwan, R.O.C., December 7-11, 2014,
  Proceedings, Part {II}}, pages 406--425, 2014.

\bibitem[ES15]{ES15}
Edward Eaton and Fang Song.
\newblock Making existential-unforgeable signatures strongly unforgeable in the
  quantum random-oracle model.
\newblock In {\em 10th Conference on the Theory of Quantum Computation,
  Communication and Cryptography, {TQC} 2015, May 20-22, 2015, Brussels,
  Belgium}, pages 147--162, 2015.

\bibitem[FJP14]{isogeny}
Luca~De Feo, David Jao, and J{\'{e}}r{\^{o}}me Pl{\^{u}}t.
\newblock Towards quantum-resistant cryptosystems from supersingular elliptic
  curve isogenies.
\newblock {\em J. Mathematical Cryptology}, 8(3):209--247, 2014.

\bibitem[Fey82]{Feynman}
Richard~P. Feynman.
\newblock Simulating physics with computers.
\newblock {\em International Journal of Theoretical Physics}, 21(6):467--488,
  1982.

\bibitem[FS86]{FS86}
Amos Fiat and Adi Shamir.
\newblock How to prove yourself: Practical solutions to identification and
  signature problems.
\newblock In {\em Advances in Cryptology - {CRYPTO} '86, Santa Barbara,
  California, USA, 1986, Proceedings}, pages 186--194, 1986.

\bibitem[GHS16]{GHS16}
Tommaso Gagliardoni, Andreas H{\"{u}}lsing, and Christian Schaffner.
\newblock Semantic security and indistinguishability in the quantum world.
\newblock In {\em Advances in Cryptology - {CRYPTO} 2016 - 36th Annual
  International Cryptology Conference, Santa Barbara, CA, USA, August 14-18,
  2016, Proceedings, Part {III}}, pages 60--89, 2016.

\bibitem[GKK17]{GKK17}
Tommaso Gagliardoni, Nikolaos~P. Karvelas, and Stefan Katzenbeisser.
\newblock {ORAMs} in a quantum world.
\newblock {\em {IACR} Cryptology {ePrint} Archive}, 2017.

\bibitem[Gam84]{ElGamal}
Taher~El Gamal.
\newblock A public key cryptosystem and a signature scheme based on discrete
  logarithms.
\newblock In {\em Advances in Cryptology, Proceedings of {CRYPTO} '84, Santa
  Barbara, California, USA, August 19-22, 1984, Proceedings}, pages 10--18,
  1984.

\bibitem[GMP16]{tworam}
Sanjam Garg, Payman Mohassel, and Charalampos Papamanthou.
\newblock {TWORAM:} efficient oblivious {RAM} in two rounds with applications
  to searchable encryption.
\newblock In {\em Advances in Cryptology - {CRYPTO} 2016 - 36th Annual
  International Cryptology Conference, Santa Barbara, CA, USA, August 14-18,
  2016, Proceedings, Part {III}}, pages 563--592, 2016.

\bibitem[GYZ16]{Garg+16}
Sumegha Garg, Henry Yuen, and Mark Zhandry.
\newblock New security notions and feasibility results for authentication of
  quantum data.
\newblock {\em CoRR}, abs/1607.07759, 2016.

\bibitem[GPV08]{GPV08}
Craig Gentry, Chris Peikert, and Vinod Vaikuntanathan.
\newblock Trapdoors for hard lattices and new cryptographic constructions.
\newblock In {\em Proceedings of the 40th Annual {ACM} Symposium on Theory of
  Computing, Victoria, British Columbia, Canada, May 17-20, 2008}, pages
  197--206, 2008.

\bibitem[Gol01]{Goldreich1}
Oded Goldreich.
\newblock {\em The Foundations of Cryptography - Volume 1, Basic Techniques}.
\newblock Cambridge University Press, 2001.

\bibitem[Gol04]{Goldreich2}
Oded Goldreich.
\newblock {\em The Foundations of Cryptography - Volume 2, Basic Applications}.
\newblock Cambridge University Press, 2004.

\bibitem[Gol11]{GolBPP}
Oded Goldreich.
\newblock In a world of {P=BPP}.
\newblock In {\em Studies in Complexity and Cryptography. Miscellanea on the
  Interplay between Randomness and Computation - In Collaboration with Lidor
  Avigad, Mihir Bellare, Zvika Brakerski, Shafi Goldwasser, Shai Halevi, Tali
  Kaufman, Leonid Levin, Noam Nisan, Dana Ron, Madhu Sudan, Luca Trevisan,
  Salil Vadhan, Avi Wigderson, David Zuckerman}, pages 191--232. 2011.

\bibitem[GGH97]{GGH97}
Oded Goldreich, Shafi Goldwasser, and Shai Halevi.
\newblock Public-key cryptosystems from lattice reduction problems.
\newblock In {\em Advances in Cryptology - {CRYPTO} '97, 17th Annual
  International Cryptology Conference, Santa Barbara, California, USA, August
  17-21, 1997, Proceedings}, pages 112--131, 1997.

\bibitem[GGM84]{GGM84}
Oded Goldreich, Shafi Goldwasser, and Silvio Micali.
\newblock How to construct random functions (extended abstract).
\newblock In {\em 25th Annual Symposium on Foundations of Computer Science,
  West Palm Beach, Florida, USA, 24-26 October 1984}, pages 464--479, 1984.

\bibitem[GL89]{GL89}
Oded Goldreich and Leonid~A. Levin.
\newblock A hard-core predicate for all one-way functions.
\newblock In {\em Proceedings of the 21st Annual {ACM} Symposium on Theory of
  Computing, May 14-17, 1989, Seattle, Washigton, {USA}}, pages 25--32, 1989.

\bibitem[GMW86]{GMW86}
Oded Goldreich, Silvio Micali, and Avi Wigderson.
\newblock How to prove all np-statements in zero-knowledge, and a methodology
  of cryptographic protocol design.
\newblock In {\em Advances in Cryptology - {CRYPTO} '86, Santa Barbara,
  California, USA, 1986, Proceedings}, pages 171--185, 1986.

\bibitem[GO96]{gooram}
Oded Goldreich and Rafail Ostrovsky.
\newblock Software protection and simulation on oblivious {RAMs}.
\newblock {\em J. {ACM}}, 43(3):431--473, 1996.

\bibitem[Gro96]{Grover96}
Lov~K. Grover.
\newblock A fast quantum mechanical algorithm for database search.
\newblock In {\em Proceedings of the Twenty-Eighth Annual {ACM} Symposium on
  the Theory of Computing, Philadelphia, Pennsylvania, USA, May 22-24, 1996},
  pages 212--219, 1996.

\bibitem[GdAJ13]{GuedesAL13}
Ello{\'{a}}~B. Guedes, Francisco~Marcos de~Assis, and Bernardo~Lula Jr.
\newblock Quantum attacks on pseudorandom generators.
\newblock {\em Mathematical Structures in Computer Science}, 23(3):608--634,
  2013.

\bibitem[GQ88]{GQ88}
Louis~C. Guillou and Jean{-}Jacques Quisquater.
\newblock A "paradoxical" indentity-based signature scheme resulting from
  zero-knowledge.
\newblock In {\em Advances in Cryptology - {CRYPTO} '88, 8th Annual
  International Cryptology Conference, Santa Barbara, California, USA, August
  21-25, 1988, Proceedings}, pages 216--231, 1988.

\bibitem[HILL99]{HILL99}
Johan Haastad, Russell Impagliazzo, Leonid~A. Levin, and Michael Luby.
\newblock A pseudorandom generator from any one-way function.
\newblock {\em {SIAM} J. Comput.}, 28(4):1364--1396, 1999.

\bibitem[HSS11]{HSS11}
Sean Hallgren, Adam~D. Smith, and Fang Song.
\newblock Classical cryptographic protocols in a quantum world.
\newblock In {\em Advances in Cryptology - {CRYPTO} 2011 - 31st Annual
  Cryptology Conference, Santa Barbara, CA, USA, August 14-18, 2011.
  Proceedings}, pages 411--428, 2011.

\bibitem[IR88]{IR88}
Russell Impagliazzo and Steven Rudich.
\newblock Limits on the provable consequences of one-way permutations.
\newblock In {\em Advances in Cryptology - {CRYPTO} '88, 8th Annual
  International Cryptology Conference, Santa Barbara, California, USA, August
  21-25, 1988, Proceedings}, pages 8--26, 1988.

\bibitem[JMV01]{ECDSA}
Don Johnson, Alfred Menezes, and Scott~A. Vanstone.
\newblock The elliptic curve digital signature algorithm {(ECDSA)}.
\newblock {\em Int. J. Inf. Sec.}, 1(1):36--63, 2001.

\bibitem[KLLN16]{Kapl+}
Marc Kaplan, Ga{\"{e}}tan Leurent, Anthony Leverrier, and Mar{\'{\i}}a
  Naya{-}Plasencia.
\newblock Breaking symmetric cryptosystems using quantum period finding.
\newblock In {\em Advances in Cryptology - {CRYPTO} 2016 - 36th Annual
  International Cryptology Conference, Santa Barbara, CA, USA, August 14-18,
  2016, Proceedings, Part {II}}, pages 207--237, 2016.

\bibitem[KKVB02]{KKVB02}
Elham Kashefi, Adrian Kent, Vlatko Vedral, and Konrad Banaszek.
\newblock Comparison of quantum oracles.
\newblock {\em Phys. Rev. A}, 65:050304, May 2002.

\bibitem[KL07]{KatzLindell}
Jonathan Katz and Yehuda Lindell.
\newblock {\em Introduction to Modern Cryptography}.
\newblock Chapman and Hall/CRC Press, 2007.

\bibitem[KPG99]{OilVinegar}
Aviad Kipnis, Jacques Patarin, and Louis Goubin.
\newblock Unbalanced oil and vinegar signature schemes.
\newblock In {\em Advances in Cryptology - {EUROCRYPT} '99, International
  Conference on the Theory and Application of Cryptographic Techniques, Prague,
  Czech Republic, May 2-6, 1999, Proceeding}, pages 206--222, 1999.

\bibitem[KM10]{KuwakadoM10}
Hidenori Kuwakado and Masakatu Morii.
\newblock Quantum distinguisher between the 3-round {Feistel} cipher and the
  random permutation.
\newblock In {\em {IEEE} International Symposium on Information Theory, {ISIT}
  2010, June 13-18, 2010, Austin, Texas, USA, Proceedings}, pages 2682--2685,
  2010.

\bibitem[KM12]{KuwakadoM12}
Hidenori Kuwakado and Masakatu Morii.
\newblock Security on the quantum-type {Even-Mansour} cipher.
\newblock In {\em Proceedings of the International Symposium on Information
  Theory and its Applications, {ISITA} 2012, Honolulu, HI, USA, October 28-31,
  2012}, pages 312--316, 2012.

\bibitem[Lyu12]{Lyu12}
Vadim Lyubashevsky.
\newblock Lattice signatures without trapdoors.
\newblock In {\em Advances in Cryptology - {EUROCRYPT} 2012 - 31st Annual
  International Conference on the Theory and Applications of Cryptographic
  Techniques, Cambridge, UK, April 15-19, 2012. Proceedings}, pages 738--755,
  2012.

\bibitem[LPR13]{LPR13}
Vadim Lyubashevsky, Chris Peikert, and Oded Regev.
\newblock On ideal lattices and learning with errors over rings.
\newblock {\em J. {ACM}}, 60(6):43:1--43:35, 2013.

\bibitem[{McE}78]{McEliece}
Robert~J. {McEliece}.
\newblock A public-key cryptosystem based on algebraic coding theory.
\newblock {\em Deep Space Network Progress Report}, 44:114--116, January 1978.

\bibitem[Mic11]{Micciancio11a}
Daniele Micciancio.
\newblock Lattice-based cryptography.
\newblock In {\em Encyclopedia of Cryptography and Security, 2nd Ed.}, pages
  713--715. 2011.

\bibitem[NC00]{NC00}
Michael~A. Nielsen and Isaac~L. Chuang.
\newblock {\em Quantum Computation and Quantum Information}.
\newblock Cambridge University Press, Cambridge, New York, 2000.

\bibitem[OBK{\etalchar{+}}16]{Google}
P.~J.~J. O'Malley, R.~Babbush, I.~D. Kivlichan, J.~Romero, J.~R. McClean,
  R.~Barends, J.~Kelly, P.~Roushan, A.~Tranter, N.~Ding, B.~Campbell, Y.~Chen,
  Z.~Chen, B.~Chiaro, A.~Dunsworth, A.~G. Fowler, E.~Jeffrey, E.~Lucero,
  A.~Megrant, J.~Y. Mutus, M.~Neeley, C.~Neill, C.~Quintana, D.~Sank,
  A.~Vainsencher, J.~Wenner, T.~C. White, P.~V. Coveney, P.~J. Love, H.~Neven,
  A.~Aspuru-Guzik, and J.~M. Martinis.
\newblock Scalable quantum simulation of molecular energies.
\newblock {\em Phys. Rev. X}, 6:031007, Jul 2016.

\bibitem[PS00]{PS00}
David Pointcheval and Jacques Stern.
\newblock Security arguments for digital signatures and blind signatures.
\newblock {\em J. Cryptology}, 13(3):361--396, 2000.

\bibitem[RSA78]{RSA}
Ronald~L. Rivest, Adi Shamir, and Leonard~M. Adleman.
\newblock A method for obtaining digital signatures and public-key
  cryptosystems.
\newblock {\em Commun. {ACM}}, 21(2):120--126, 1978.

\bibitem[SP92]{DP92}
Alfredo~De Santis and Giuseppe Persiano.
\newblock Zero-knowledge proofs of knowledge without interaction (extended
  abstract).
\newblock In {\em 33rd Annual Symposium on Foundations of Computer Science,
  Pittsburgh, Pennsylvania, USA, 24-27 October 1992}, pages 427--436, 1992.

\bibitem[SS17]{SS16}
Thomas Santoli and Christian Schaffner.
\newblock Using {Simon's} algorithm to attack symmetric-key cryptographic
  primitives.
\newblock {\em Quantum Information {\&} Computation}, 17(1{\&}2):65--78, 2017.

\bibitem[Sch91]{Schnorr91}
Claus{-}Peter Schnorr.
\newblock Efficient signature generation by smart cards.
\newblock {\em J. Cryptology}, 4(3):161--174, 1991.

\bibitem[Sha01]{Shannon01}
Claude~E. Shannon.
\newblock A mathematical theory of communication.
\newblock {\em Mobile Computing and Communications Review}, 5(1):3--55, 2001.

\bibitem[Sho94]{Shor94}
Peter~W. Shor.
\newblock Algorithms for quantum computation: Discrete logarithms and
  factoring.
\newblock In {\em 35th Annual Symposium on Foundations of Computer Science,
  Santa Fe, New Mexico, USA, 20-22 November 1994}, pages 124--134, 1994.

\bibitem[Sim97]{Simon97}
Daniel~R. Simon.
\newblock On the power of quantum computation.
\newblock {\em {SIAM} J. Comput.}, 26(5):1474--1483, 1997.

\bibitem[Son14]{Song14}
Fang Song.
\newblock A note on quantum security for post-quantum cryptography.
\newblock In {\em Post-Quantum Cryptography - 6th International Workshop,
  PQCrypto 2014, Waterloo, ON, Canada, October 1-3, 2014. Proceedings}, pages
  246--265, 2014.

\bibitem[SSS12]{SSS12}
Emil Stefanov, Elaine Shi, and Dawn~Xiaodong Song.
\newblock Towards practical oblivious {RAM}.
\newblock In {\em 19th Annual Network and Distributed System Security
  Symposium, {NDSS} 2012, San Diego, California, USA, February 5-8, 2012},
  2012.

\bibitem[SvDS{\etalchar{+}}13]{pathoram}
Emil Stefanov, Marten van Dijk, Elaine Shi, Christopher~W. Fletcher, Ling Ren,
  Xiangyao Yu, and Srinivas Devadas.
\newblock Path {ORAM:} an extremely simple oblivious {RAM} protocol.
\newblock In {\em 2013 {ACM} {SIGSAC} Conference on Computer and Communications
  Security, CCS'13, Berlin, Germany, November 4-8, 2013}, pages 299--310, 2013.

\bibitem[SLB{\etalchar{+}}11]{idquantique}
D.~{Stucki}, M.~{Legr{\'e}}, F.~{Buntschu}, B.~{Clausen}, N.~{Felber},
  N.~{Gisin}, L.~{Henzen}, P.~{Junod}, G.~{Litzistorf}, P.~{Monbaron},
  L.~{Monat}, J.-B. {Page}, D.~{Perroud}, G.~{Ribordy}, A.~{Rochas},
  S.~{Robyr}, J.~{Tavares}, R.~{Thew}, P.~{Trinkler}, S.~{Ventura},
  R.~{Voirol}, N.~{Walenta}, and H.~{Zbinden}.
\newblock {Long-term performance of the SwissQuantum quantum key distribution
  network in a field environment}.
\newblock {\em New Journal of Physics}, 13(12):123001, December 2011.

\bibitem[TCM{\etalchar{+}}16]{IBM}
Maika Takita, Antonio~D. C\'orcoles, Easwar Magesan, Baleegh Abdo, Markus
  Brink, Andrew Cross, Jerry~M. Chow, and Jay~M. Gambetta.
\newblock Demonstration of weight-four parity measurements in the surface code
  architecture.
\newblock {\em Phys. Rev. Lett.}, 117:210505, Nov 2016.

\bibitem[Unr12]{UnruhZK}
Dominique Unruh.
\newblock Quantum proofs of knowledge.
\newblock In {\em Advances in Cryptology - {EUROCRYPT} 2012 - 31st Annual
  International Conference on the Theory and Applications of Cryptographic
  Techniques, Cambridge, UK, April 15-19, 2012. Proceedings}, pages 135--152,
  2012.

\bibitem[Unr13]{UnruhEverlasting}
Dominique Unruh.
\newblock Everlasting multi-party computation.
\newblock In {\em Advances in Cryptology - {CRYPTO} 2013 - 33rd Annual
  Cryptology Conference, Santa Barbara, CA, USA, August 18-22, 2013.
  Proceedings, Part {II}}, pages 380--397, 2013.

\bibitem[Unr15]{Unruh15}
Dominique Unruh.
\newblock Non-interactive zero-knowledge proofs in the quantum random oracle
  model.
\newblock In {\em Advances in Cryptology - {EUROCRYPT} 2015 - 34th Annual
  International Conference on the Theory and Applications of Cryptographic
  Techniques, Sofia, Bulgaria, April 26-30, 2015, Proceedings, Part {II}},
  pages 755--784, 2015.

\bibitem[VW16]{VW16}
Thomas Vidick and John Watrous.
\newblock Quantum proofs.
\newblock {\em Foundations and Trends in Theoretical Computer Science},
  11(1-2):1--215, 2016.

\bibitem[Wat01]{Watrousalgos}
John Watrous.
\newblock Quantum algorithms for solvable groups.
\newblock In {\em Proceedings on 33rd Annual {ACM} Symposium on Theory of
  Computing, July 6-8, 2001, Heraklion, Crete, Greece}, pages 60--67, 2001.

\bibitem[Wat06]{WatrousZK}
John Watrous.
\newblock Zero-knowledge against quantum attacks.
\newblock In {\em Proceedings of the 38th Annual {ACM} Symposium on Theory of
  Computing, Seattle, WA, USA, May 21-23, 2006}, pages 296--305, 2006.

\bibitem[Wie83]{Wiesner}
Stephen Wiesner.
\newblock Conjugate coding.
\newblock {\em SIGACT News}, 15(1):78--88, January 1983.

\bibitem[Yao82]{YaoPRNG}
Andrew~Chi{-}Chih Yao.
\newblock Theory and applications of trapdoor functions (extended abstract).
\newblock In {\em 23rd Annual Symposium on Foundations of Computer Science,
  Chicago, Illinois, USA, 3-5 November 1982}, pages 80--91, 1982.

\bibitem[Zha12a]{ZhandryPRF}
Mark Zhandry.
\newblock How to construct quantum random functions.
\newblock In {\em 53rd Annual {IEEE} Symposium on Foundations of Computer
  Science, {FOCS} 2012, New Brunswick, NJ, USA, October 20-23, 2012}, pages
  679--687, 2012.

\bibitem[Zha12b]{Zhandry12}
Mark Zhandry.
\newblock Secure identity-based encryption in the quantum random oracle model.
\newblock In {\em Advances in Cryptology - {CRYPTO} 2012 - 32nd Annual
  Cryptology Conference, Santa Barbara, CA, USA, August 19-23, 2012.
  Proceedings}, pages 758--775, 2012.

\bibitem[Zha16]{ZhandryPRP}
Mark Zhandry.
\newblock A note on quantum-secure {PRPs}.
\newblock {\em CoRR}, abs/1611.05564, 2016.

\end{thebibliography}

\end{document}